\documentclass[-eps-converted-to.pdfonefignum,onetabnum]{siamart190516}

\usepackage{braket,amsfonts}

\usepackage[caption=false]{subfig}

\usepackage{pgfplots}
\usepackage{multicol}
\usepackage{hhline}
\usepackage{multirow}
\usepackage{enumerate}
\usepackage[shortlabels]{enumitem}
 \usepackage{caption}
 \usepackage{float}
\usepackage{amssymb}
\usepackage{pifont} 
\usepackage{color}
\usepackage{chngcntr}
\counterwithin{table}{section}
\usepackage{amsbsy, mathrsfs, mathdots}
\newsiamthm{claim}{Claim}
\newsiamremark{remark}{Remark}
\newsiamremark{hypothesis}{Hypothesis}
\crefname{hypothesis}{Hypothesis}{Hypotheses}

\usepackage{algorithmic}

\usepackage{graphicx,epstopdf}

\Crefname{ALC@unique}{Line}{Lines}

\usepackage{amsopn}

\usepackage{xspace}
\usepackage{bold-extra}
\usepackage[most]{tcolorbox}

\colorlet{texcscolor}{blue!50!black}
\colorlet{texemcolor}{red!70!black}
\colorlet{texpreamble}{red!70!black}
\colorlet{codebackground}{black!25!white!25}

\newcommand  \stack[2]   {\overset{\text{#1}}{#2}}

\lstdefinestyle{siamlatex}{%
  style=tcblatex,
  texcsstyle=*\color{texcscolor},
  texcsstyle=[2]\color{texemcolor},
  keywordstyle=[2]\color{texemcolor},
  moretexcs={cref,Cref,maketitle,mathcal,text,headers,email,url},
}

\tcbset{%
  colframe=black!75!white!75,
  coltitle=white,
  colback=codebackground, 
  colbacklower=white, 
  fonttitle=\bfseries,
  arc=0pt,outer arc=0pt,
  top=1pt,bottom=1pt,left=1mm,right=1mm,middle=1mm,boxsep=1mm,
  leftrule=0.3mm,rightrule=0.3mm,toprule=0.3mm,bottomrule=0.3mm,
  listing options={style=siamlatex}
}

\newtcblisting[use counter=example]{example}[2][]{%
  title={Example~\thetcbcounter: #2},#1}

\newtcbinputlisting[use counter=example]{\examplefile}[3][]{%
  title={Example~\thetcbcounter: #2},listing file={#3},#1}

\newcommand{\cV}{\mathcal{V}}
\newcommand{\cE}{\mathcal{E}}

\def\cV{\mathcal{V}}

\newcommand{\R}{\mathbb{R}}

\newcommand{\mG}{\mathcal{G}}
\def\diag{{\mathrm{diag}}}

\def\G{\mathcal {G}}

\patchcmd\newpage{\vfil}{}{}{}
\begin{tcbverbatimwrite}{tmp_\jobname_header.tex}
\title{Robust recovery of bandlimited graph signals via randomized dynamical sampling \thanks{Submitted to the editors DATE.
\funding{L.Huang and D.Needell were supported in part by NSF DMS \#2011140 and the Dunn Family Endowed Chair fund. 
S. Tang was partially supported by Regents Junior Faculty fellowship and Faculty Early Career Acceleration grant sponsored by University of California Santa Barbara and NSF DMS \#2111303.  }}}

\author{Longxiu Huang\thanks{Department of  Mathematics, University of California Los Angeles, CA (\email{huangl3@math.ucla.edu}).}
\and Deanna Needell\thanks{Department of  Mathematics, University of California Los Angeles, CA (\email{deanna@math.ucla.edu}).}
\and Sui Tang\thanks{Department of Mathematics,  University of California Santa Barbara, CA (\email{suitang@math.ucsb.edu}).}
}

\headers{Random Dynamical Sampling on Graph Signals }{Longxiu Huang, Deanna Needell  and Sui Tang}
\end{tcbverbatimwrite}
\title{Robust recovery of bandlimited graph signals via randomized dynamical sampling \thanks{Submitted to the editors DATE.
\funding{L.Huang and D.Needell were supported in part by NSF DMS \#2011140 and the Dunn Family Endowed Chair fund. 
S. Tang was partially supported by Regents Junior Faculty fellowship and Faculty Early Career Acceleration grant sponsored by University of California Santa Barbara and NSF DMS \#2111303.  }}}

\author{Longxiu Huang\thanks{Department of  Mathematics, University of California Los Angeles, CA (\email{huangl3@math.ucla.edu}).}
\and Deanna Needell\thanks{Department of  Mathematics, University of California Los Angeles, CA (\email{deanna@math.ucla.edu}).}
\and Sui Tang\thanks{Department of Mathematics,  University of California Santa Barbara, CA (\email{suitang@math.ucsb.edu}).}
}

\headers{Random Dynamical Sampling on Graph Signals }{Longxiu Huang, Deanna Needell  and Sui Tang}

\ifpdf
\hypersetup{ pdftitle={Guide to Using  SIAM'S \LaTeX\ Style} }
\fi

\begin{document}
\maketitle

\begin{tcbverbatimwrite}{tmp_\jobname_abstract.tex}
\begin{abstract}
Heat diffusion processes have found wide applications  in modelling dynamical systems over graphs. In this paper, we consider the recovery of a $k$-bandlimited graph signal that is an initial signal of a heat diffusion process from its space-time samples. We propose three random space-time sampling regimes, termed  dynamical sampling techniques,   that consist in selecting a small subset of space-time nodes at random according to some probability distribution. We show that the number of space-time samples required to ensure stable recovery for each regime depends  on a parameter called the spectral graph weighted coherence, that depends on the interplay between the dynamics over the graphs and sampling probability distributions. In optimal scenarios, no more than $\mathcal{O}(k \log(k))$ space-time samples are sufficient to ensure  accurate and stable recovery of all $k$-bandlimited signals. In any case, dynamical sampling typically requires much fewer spatial samples than the static case by leveraging the temporal information. Then, we propose a computationally efficient method to reconstruct $k$-bandlimited signals from their space-time samples. We prove that it yields accurate reconstructions and that it is also stable to noise. Finally, we test dynamical sampling techniques on a wide variety of graphs. The numerical results support our theoretical findings and demonstrate the efficiency. 
  \end{abstract}


\end{tcbverbatimwrite}
\begin{abstract}
Heat diffusion processes have found wide applications  in modelling dynamical systems over graphs. In this paper, we consider the recovery of a $k$-bandlimited graph signal that is an initial signal of a heat diffusion process from its space-time samples. We propose three random space-time sampling regimes, termed  dynamical sampling techniques,   that consist in selecting a small subset of space-time nodes at random according to some probability distribution. We show that the number of space-time samples required to ensure stable recovery for each regime depends  on a parameter called the spectral graph weighted coherence, that depends on the interplay between the dynamics over the graphs and sampling probability distributions. In optimal scenarios, no more than $\mathcal{O}(k \log(k))$ space-time samples are sufficient to ensure  accurate and stable recovery of all $k$-bandlimited signals. In any case, dynamical sampling typically requires much fewer spatial samples than the static case by leveraging the temporal information. Then, we propose a computationally efficient method to reconstruct $k$-bandlimited signals from their space-time samples. We prove that it yields accurate reconstructions and that it is also stable to noise. Finally, we test dynamical sampling techniques on a wide variety of graphs. The numerical results support our theoretical findings and demonstrate the efficiency.
  \end{abstract}




\section{Introduction}
\label{sec:intro}

\label{sec:intro}

\subsection{Dynamical sampling problems over graphs}  Signals that arise from social, biological, and sensor networks  are typically structured and  interconnected, which are  modelled  as residing on weighted
graphs \cite{ceccon2017graph}.   The demand for large-scale data processing tasks over graphs such as graph-based filtering \cite{chen2013adaptive,loukas2016frequency} and  semi-supervised classifications  \cite{ekambaram2013wavelet,chen2014semi} has inspired an emerging field of graph signal processing \cite{shuman2013emerging,sandryhaila2014big}, where a cornerstone problem  is the robust recovery of graph signals from its sampled values on a proper subset of nodes \cite{pesenson2008sampling,pesenson2009variational,chen2015discrete,segarra2016reconstruction,puy2018random}.

In general, this type of inverse problem  is ill-posed and the exact reconstruction is impossible. Furthermore, the task of selecting a subset of nodes whose values enable reconstruction of the values in the entire graph with minimal loss is known to be NP- hard \cite{davis1997adaptive}. 
 However, many signals in real world applications exhibit structural features such as smoothness. In classical signal processing,  the smoothness assumptions are often defined in terms of the signal’s Fourier transform.  An important class of smooth signals is that of  bandlimited signals, i.e., signals whose Fourier transforms are compactly supported.  Bandlimited graph signals are defined in a similar way using the graph Fourier transform. That is to say, the signal is spanned by the eigenvectors corresponding to the small eigenvalues of the graph Laplacian. Such signals have appeared in various applications such as 
 temperature readings across U.S.,  wind speed across Minnesota, and natural image reconstruction  \cite{chen2016signal}.  Motivated by these applications, the sampling and reconstruction of  bandlimited graph signals have drawn considerable attention in recent years.

Most previous work in this area considered the static setting:  various sampling and reconstruction algorithms for bandlimited graph signals, inspired  from classical Shannon-sampling   \cite{chen2015discrete,pesenson2008sampling,pesenson2010sampling} and compressed sensing \cite{puy2018random},  have been proposed.  For $k$-bandlimited graph signals, one has to select at least $k$ nodes to ensure exact reconstruction  in the noiseless case, and it has been proved that $\mathcal{O}(k \log k)$ random samples are sufficient to recover all  $k$-bandlimited signals with high probability \cite{puy2018random}.

In many cases, the bandlimited signals over graphs are evolving and the underlying dynamical process can be modelled as a heat diffusion process over graphs, such as  rumor propagation over social networks \cite{xiao2019rumor}, traffic movement over transportation networks \cite{yu2017spatiotemporal}, spatial temperature profiles measured by a wireless sensor network \cite{thanou2017learning}, and neural activities at different regions of the brain \cite{sporns2010networks}. The existing sampling and reconstruction algorithms, however, do not take dynamical process into account.  In practice, we can obtain the space-time samples of the diffusive field by placing the sensors  at various spatial locations and   activating them at different time instances.  On one hand, due to application-specific restrictions,  it is often the case that one can not take a sufficient number of spatial samples at a single time instance to recover a bandlimited signal. On the other hand, the acquisition cost of space-time samples 
is much cheaper than the same amount of spatial samples, since the latter requires the physical presence of more sensors in the network, whereas the former is, in theory, only constrained by the communication capacity and energy budget of the sensor network.

The above mentioned practical concerns motivated  a new sampling approach termed ``dynamical sampling" \cite{aldroubi2013dynamical,aldroubi2017dynamical}.  Dynamical sampling deals with processing a time series of evolving signals. The goal is to recover the dynamical systems from a union of coarse spatial samples at multiple time instances. A core problem to address is where to put sensors, when to take samples and how to reconstruct the dynamical system from measured samples.

In this paper, we consider the dynamical sampling of a bandlimited diffusion field over the graph,  where we assume the physical law is known. Specifically, the homogeneous heat diffusion process over $\mathcal{G}$ is  governed by 
\begin{equation}\label{linearsystem}
\frac{\partial }{\partial t}\mathbf{x}_t =-L\mathbf{x}_t,  t\geq 0,
\end{equation} where $L$ is the normalized graph Laplacian of $\mG$.  The diffusive field $\mathbf{x}_t$   is completely determined by the initial condition $\mathbf{x}_0$:
$\mathbf{x}_t =e^{-tL}\mathbf{x}_0=:A^{t/\Delta t}\mathbf{x}_0, t\geq 0$
where $A \triangleq e^{-\Delta t}$ and $\Delta t>0$. In other words, $\mathbf{x}_t$ is just a filtered version of the initial state $\mathbf{x}_0$, with the filter being a time varying Laplacian kernel.

We observe the dynamics up to  $T$ time instances. For simplicity, let $\Omega \subset [n] \times \{0, \Delta t,\cdots, (T-1) \Delta t \}$ denote the space-time locations  where the entries of the diffusive field $\{\mathbf{x}_t: t\geq 0\}$ are observed. Now we decompose $\Omega=\bigcup_{\ell=0}^{T-1}\Omega_{\ell} \times \{\ell\}$ where $\Omega_{\ell}\subset [n]$ is the set of observational spatial locations at $t_{\ell}=\ell\Delta t$. We use $S_{\Omega_{\ell}}$ to denote the observation matrix that 
 $S_{\Omega_{\ell}}\mathbf{x}=\sum_{i\in \Omega_{\ell}}\mathbf{x}(i)\mathbf{e}_i$ where $\{\mathbf{e}_i\}_{i=1}^{n}$ is the standard orthonormal basis in $\mathbb{C}^n$.  Motivated by the way we gather the information from the social networks, we will investigate three random space-time sampling regimes with an example shown in Fig.~\ref{fig:sampling}: { \ding{172} Random spatial locations  at  initial  time, \ding{173} Random spatial locations at consecutive times and \ding{174} Random space-time locations}. In the first   regime, the sampling locations are the same for all time instances 
 and are randomly selected according to a predefined probability distribution $\mathbf{p}^{(1)}$ on $[n]$.  For the second  regime, the sampling locations at different time instances are  randomly selected according to their own predefined probability distributions $\mathbf{p}^{(2)}(t)$. Different from the first two sampling regimes, the evolved signals are treated as a signal in $\mathbb{R}^{nT}$ in the third sampling regime and the locations are randomly selected according to a probability distribution $\mathbf{p}^{(3)}$ in $[nT]$.

\begin{figure}
\centering
\begin{minipage}{.3\linewidth} \centering \small  Regime 1  \end{minipage}
\begin{minipage}{.28\linewidth} \centering \small  Regime 2 \end{minipage}
\begin{minipage}{.28\linewidth} \centering \small  Regime 3 \end{minipage}\\
\includegraphics[width=.9\linewidth, keepaspectratio]{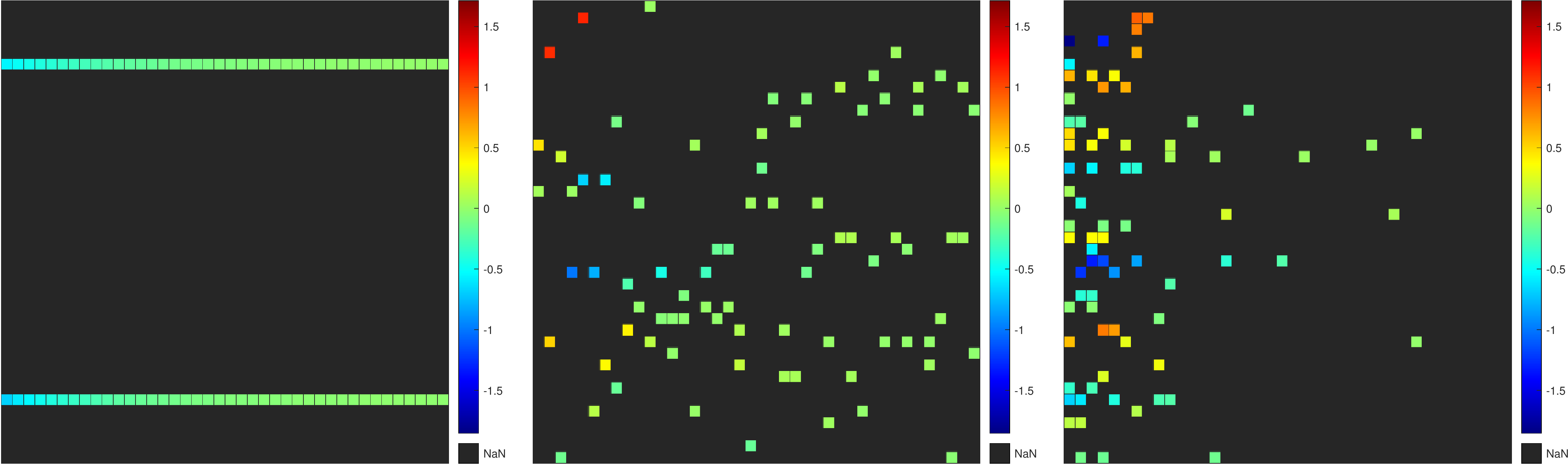}
\caption{\footnotesize\label{fig:sampling} An example of three space-time sampling regimes using the same number of samples. Here the $x$ axis denotes the time instances and  the $y$ axis denotes the spatial values of evolving signals with colors representing the magnitude.  Only values at colored spots are observed. }
\end{figure}

The goal is to investigate the conditions on $\Omega$ to achieve stable recovery of    the initial signal $\mathbf{x}_0$   from   (possibly noisy) partial observations  of states from a single trajectory $\{S_{\Omega_\ell}\mathbf{x}_{t_\ell}: \ell=0,1,\cdots, T-1\}$ and propose robust reconstruction strategies.

\subsection {Summary of contributions}
In this paper, we propose three  random  space-time sampling regimes to sample {bandlimited graph diffusive fields}. 
For  each regime, we introduce a quantity called  the \textit{spectral graph weighted coherence of order $(k,T)$}.  Such a quantity depends on the interplay between the sampling distribution and the localization of the  bandlimited diffusion fields over the space-time locations. We  provide sampling complexity bounds to ensure robust recovery of initial signal with high probability and show that they scale linearly with the square of  spectral graph weighted coherence of order $(k,T)$. We show that temporal information can reduce the spatial sampling density and one can  achieve comparable (even better) results given the same amount of space-time samples with the static case. We  further provide a detailed  comparison of three regimes in terms of their sample complexity  and   the influence of sampling probability  distributions on the performance. More details are discussed in Section \ref{subs:properties}.

Finally, we propose a robust method to reconstruct any $k$-bandlimited signal from its noisy samples, with a detailed error analysis. The idea is to use an  $\ell_2$ regularization term that utilizes the graph Laplacian and  bandlimited property of the graph signals. The results show that our methods can recover $k$-bandlimited signals exactly in the absence of noise and  the method is also robust to measurement acquisition noise and model errors. In addition, extensive numerical simulations are conducted on various graphs including an example of deblurring a real image from the space-time samples of its diffusion process, which further support our theoretical results and confirm that  random dynamical sampling offers a robust alternative for the standard spatial sampling.

 Note that our space-time sampling theorems are applicable to any  dynamical system over the graph  of the form $x_t=Ax_{t-1}$ with the evolution operator   $A$ being  a diagonalizable operator. For example, $A$ can be a polynomial function of any symmetrical Laplacian or adjacency matrix of any weighted undirected graph. It is also possible to generalize the current framework to bandlimited limited  graph processes governed by a linear time-invariant system \cite{isufi2020observing}.

 Nevertheless, the reconstruction methods  proposed here are specifically designed to take advantage of the semi-definite positivity of  the evolution operator. For  simplicity and conciseness, we therefore concentrate on the normalized Laplacians.

 Our work is built upon recent progress on sampling of bandlimited graph signals and dynamical sampling over regular domains. In particular, we generalize the work in \cite{puy2018random} to bandlimited diffusion fields with three random-space sampling regimes. In Section \ref{relatedwork}, we give  more detailed comparisons with these relevant works.

 \section{Related work} \label{relatedwork}
 
 \subsection{Dynamical sampling over regular domain}  Previous work on dynamical sampling has focused on deterministic sampling over regular domains.  Dynamical sampling was first proposed  in \cite{aldroubi2013dynamical,aldroubi2017dynamical} by Aldroubi et al for linear time invariant systems, motivated by the pioneering work of \cite{lu2009spatial} that considered the space-time sampling of bandlimited diffusion fields over the real line, see also \cite{ranieri2011sampling,dokmanic2011sensor, murray2015estimating,murray2016physics} for various extensions.   The majority of the work \cite{aceska2013dynamical,cabrelli2017dynamical,aldroubi2021sampling,aldroubi2020phaseless,tang2017universal,ulanovskii2021reconstruction,aldroubi2019frames,lai2017undersampled,christensen2019frame} has focused on developing deterministic sampling theorems to ensure  exact reconstruction, where the number and the location of sensors are  fixed for successive time instances.   Our paper is the first one to extend the dynamical sampling framework to graphs and consider the random space-time sampling regimes: the locations and numbers of sensors can change over time, and the observation time period is not necessarily successive.  We will show that this  allows us to adapt the graph topology to the sampling such that one needs much fewer number of space-time samples yet yields a more robust reconstruction.

 \subsection{Sampling and reconstruction of bandlimited graph signals} 
 
Sampling and reconstruction of bandlimited graph signals (corresponding to $T=1$  in our case)  have been extensively studied and achieved remarkable progress. The existing approaches can be mainly categorized into selection sampling and aggregation sampling. In the first category, deterministic sampling  (see e.g.~\cite{chen2015discrete,chamon2017greedy, pesenson2008sampling,pesenson2009variational, anis2016efficient, huang2020reconstruction}) and random sampling scheme (see e.g.~\cite{chen2015signal,chen2016signal,hashemi2018accelerated,puy2018random}) were proposed.  Using  ideas of variable density sampling from compressed sensing,  Vandergheynst et al. \cite{puy2018random}  derive random sampling schemes based on probability distributions over the graph vertexes.  In this paper, we generalize their work to the bandlimited diffusion field and propose three different  random space-time sampling schemes so that the work \cite{puy2018random} becomes a special case of our regimes  with $T=1$.  In the second category,   aggregation sampling  aims to recover signals from (random weighted) linear combinations of signal values on a  neighborhood of selected nodes. It corresponds to sampling a time-varying signal driven by a  graph shift operator, whose matrix form has the same sparsity pattern with the adjacent matrix.   The deterministic scheme proposed in \cite{marques2015sampling} (see  \cite{wang2016local,yang2021orthogonal} for extensions)  can be viewed a special case of our regime 1 with a pre-selected set of nodes; the random scheme \cite{valsesia2018sampling} used a graph shift operator with i.i.d.~Gaussian entries, corresponding to the case of $T=2$  and observations are only made at $T=2$. 
 
 Compared to the static setting,  there are relatively fewer work considering the sampling  and reconstruction problem in  dynamical system setting.  Different models for the time series over graph have been proposed. The examples include   distributed reconstruction algorithms for graph low-pass filtering processes  \cite{wang2015distributed} and  batch reconstruction methods   \cite{qiu2017time} for  graph signals whose temporal difference is bandlimited.
 \cite{isufi2020observing} considered 
 both deterministic and Bernoulli random space-time sampling for bandlimited graph processes. Necessary conditions on probabilities are derived to ensure the exact reconstruction of the process and a  convex optimization approach was proposed to choose the optimal sampling design.  The proposed regime is most related to our regime 3. However, our work is different: we also consider regimes 1 and 2, and we use a different observation model by extending the idea of variable density sampling.

\section{Preliminaries and Notation} \label{sec: NP}

\subsection{Operators associated to graphs} We consider an undirected weighted graph $\mG=(\cV,\cE,W)$ where $\cV=\{v_1,v_2,\cdots, v_n\}$ is  set of $n$ vertices and $ \cE \subset\mathcal{ V\times V}$ is a set of edges. {Let $\{v_i,v_j\}$ denote the edge if $v_i$ and $v_j$ are connected}. 
The weighted adjacent matrix $W$ is defined as 
\[
W(i,j) \triangleq 
\begin{cases}
   \alpha_{ij}&\mbox{if $\{v_i,v_j\} \in \cE$}\\
       0 &\mbox{otherwise}
\end{cases}; \alpha_{ij} \in \mathbb{R}_{+}; \text{ for all } v_i,v_j \in \mathcal{V}.\]

The weight associated with each edge in the graph often {indicates  similarity or dependency between the corresponding vertices}. The connectivity and edge weights are dictated by the physics of the problem at hand, for instance, the weights  are inversely proportional to the physical distance between nodes in the network.  In particular,  $W$ is symmetric.  The \textit{degree} $\deg(v_i)$ of a vertex $v_i \in V$ is defined as $\deg(v_i)=\sum_{j=1}^n W(i,j)$.  ${D} :=\diag(\deg(v_i)_{i=1}^{n})$ denotes the diagonal degree matrix. In the following, we introduce important operators associated with the graph $\mathcal{G}$. 

\begin{definition}The normalized diffusion operator of a graph $\mathcal{G}$ with  weighted adjacent matrix ${W} \in \R^{n \times n}$ is defined by 
$N =  {D}^{-\frac{1}{2}}W  {D}^{-\frac{1}{2}}$. We call $N$ the \textit{diffusion matrix} of $\mathcal{G}$.   The  normalized graph Laplacian operator  is $L =I-N$. 
\end{definition}

\subsection{Bandlimited graph signal}  The matrix $L$ is real symmetric, and so it admits an eigen-decomposition ${L}=U\Sigma U^{\top}$ where the columns of $U\in \mathbb{R}^{n\times n}$ are  orthonormal eigenvectors  and the diagonal entries of  $\Sigma$ consist of   $n$ eigenvalues $\sigma_1,\cdots, \sigma_n$. Furthermore, semi-definite positivity of $L$ implies that all eigenvalues are non-negative. Without loss of generality, we assume that $\sigma_1\leq\cdots\leq\sigma_n$. {Note that  $A=\exp(-\Delta t L)$, we have $A= U\Lambda U^\top $ with $\Lambda:=\exp(-\Delta t \Sigma)$. Let's denote the diagonal entries of $\Lambda$ as $\lambda_1=\exp(-\Delta t \sigma_1), \cdots,\lambda_n=\exp(-\Delta t\sigma_n)$. Thus we have $\lambda_1\geq \cdots\geq \lambda_n$. }

In graph signal processing, the matrix $U$ is often viewed as the graph Fourier transform. For any signal $\mathbf{x}\in \mathbb{R}^n$ defined on the nodes of the graph $\mathcal{G}$, $\widehat{\mathbf{x}}=U^\top\mathbf{x}$  {represents} the Fourier coefficients of $\mathbf{x}$ ordered in increasing frequencies. The smooth graph signals are modelled 
as  $k$ bandlimited  signals: for $\mathbf{x}\in \mathbb{R}^n$ on $\mathcal{G}$ with bandwidth $k>0$ it satisfies
$\mathbf{x}=U_k\widehat{\mathbf{x}}^k$
where $\widehat{\mathbf{x}}^k:=\widehat{\mathbf{x}}(1:k) \in \mathbb{R}^k$ and
$U_k:=(\mathbf{u}_1,
\cdots,\mathbf{u}_k) \in \mathbb{R}^{n\times k}$,
i.e., $U_k$ is the restriction of $U$ to its first $k$ vectors. This yields the following formal definition of a $k-$bandlimited signal.

\begin{definition}[Bandlimited signal on $\mG$] A signal $\mathbf{x} \in \mathbb{R}^n$ defined on the nodes of the graph $\mathcal{G}$ is
$k$-bandlimited with $k\in \mathbb{N}\setminus\{0\}$ if $\mathbf{x}\in \mathrm{span}(U_k)$.\end{definition}

Note that we use $U_k$ in our definition of $k$-bandlimited signals to handle the case where the eigen-decomposition 
is not unique.

\subsection{Notation}
 In this work, capital letters are used for matrices, lower boldface letters for vectors, and regular letters for scalars.  The set of the first $d$ natural numbers is denoted by $[d]$. For integers $i,n, T$, the notation $ i:n:Tn$ stands for the set  $\{k=i+jn\in[i, Tn]: \text{for some }j\in\mathbb{Z}\}$ and $I_d$ stands for the $d\times d$ identity matrix.
 
    We let $\diag(\mathbf{x})$ denote the diagonal matrix with diagonal entries given by  vector $\mathbf{x}$ and $\diag(X)=\diag(X(1,1),\cdots,X(d,d))$ for $X\in\mathbb{R}^{d\times d}$.
     Let $\mathbf{x}_1, \cdots,\mathbf{x}_t$ be vectors and $X_1,\cdots,X_t$ be  matrices, we define $\diag(\mathbf{x}_1;\cdots;\mathbf{x}_{t})$  and $\diag(X_1;\cdots;X_{t})$  as the block diagonal matrix with
    the $(i,i)$-th block of $\diag(\mathbf{x}_1; \cdots;\mathbf{x}_{t})$ being $\diag(\mathbf{x}_i)$ and  the $(i,i)$-th block of $\diag(X_1;\cdots;X_{t})$ being $X_i$ respectively.

      Let $T$ be a positive integer and we define 
      \begin{align}\label{map:f_T}
      f_{T}(\lambda)=\sqrt{\sum_{t=0}^{T-1}\lambda^{2t}},\; 
      F_{k,T}= \diag\left(\frac{1}{f_{T}(\lambda_1)},\cdots,\frac{1}{f_{T}(\lambda_k)}\right),\; \Lambda_k=\diag(\lambda_1,\cdots,\lambda_k).
      \end{align}

\section{Sampling of $k$-bandlimited diffusion fields} In this section, we describe three regimes that select a subset of space-time locations to sample a $k$-bandlimited diffusion field.  Then, we prove that the corresponding space-time sampling procedure stably embeds the set of initial signals that consists of $k$-bandlimited signals. 
We describe how to reconstruct initial signals  from the space-time samples in Section 5.

\subsection{The space-time sampling procedure}
 To select the space-time sampling set $\Omega$, we introduce various probability distributions $\mathbf{p}$ serving as sampling distributions for  three random  regimes.  
\begin{definition}\label{def:distributions} 
 The probability distributions can be defined as follows:
\begin{itemize}
    \item[\ding{172}] \textbf{Random spatial locations at initial time:} let $\mathbf{p}^{(1)}$ be a probability distribution on $[n]$. We define  $P_1 \in\mathbb{R}^{Tn\times Tn}$ as $P_1=\diag({\mathbf{p}^{(1)};\cdots;\mathbf{p}^{(1)}})$.
    \item[\ding{173}]  {\bf Random spatial locations at consecutive times:} let   $\mathbf{p}_t^{(2)}$ be a probability distribution on $[n]$ describing sampling at time $t$, for all $t=0,1,\cdots,T-1$. The associated matrix $P_2\in\mathbb{R}^{Tn\times Tn}$ is defined as $P_2=\diag(\mathbf{p}_0^{(2)}; \cdots;\mathbf{p}_{T-1}^{(2)} )$.
    \item[\ding{174}]  {\bf Random space-time locations:} we let  $\mathbf{p}^{(3)}$ be a probability distribution on $[Tn]$ and  define   $P_3\in\mathbb{R}^{Tn\times Tn}$ as $\mathrm{diag}(\mathbf{p}^{(3)})$. 
\end{itemize}

\end{definition}

We assume that the probability distributions above are all non-degenerate, i.e., entry-wise nonzero.  Once the probability distributions are assigned, we  obtain the nodes in  the space-time sampling sets $\Omega$ 
independently (with replacement) according to the  probability distributions. Note that $\Omega$ only needs to be selected once to sample all $k$-bandlimited diffusion fields over $\mathcal{G}$.  After the sampling distributions and sampling set $\Omega$ are determined, we define  the associated variables as follows.  
\begin{definition}\label{def_MPW}
For three regimes, we define the  associated sampling matrices and weighting matrices as follows. 
\begin{itemize}
 \item[\ding{172}] \textbf{Random spatial locations at initial time:}  $\Omega^{(1)}:= \Omega^{(1)}_0 \times [T]$ where $\Omega^{(1)}_0 =\{\omega_1^{(1)},\cdots,\omega_{m}^{(1)}\}$ is constructed by drawing $m$ indices independently (with replacement)  from set $[n]$ according to  probability $\mathbf{p}^{(1)}$, i.e.,
  $\mathbb{P}(\omega_j^{(1)}=i)=\mathbf{p}^{(1)}(i), \text{ for all } j\in[m] \text{ and }i\in[n]$.
  We define the  sampling matrix  
  \begin{equation}\label{eqn:samp_mat_fixed}
  S_1=\diag( \underbrace{S_1^{(0)};\cdots; S_1^{(0)}}_{T}),
  \end{equation}
        where  $ S_1^{(0)}\in\mathbb{R}^{m\times n}$ is defined as 
  $ S_1^{(0)}(i,j):=\begin{cases}
     1, &\text{ if }j=\omega_i\\
     0,& \text{otherwise}
  \end{cases}
  $.  Thus, we can define the probability matrix on the sampling set $P_{\Omega^{(1)}}=S_1P_1S_1^\top$. 
  To average the samples' information at each time instance,  we also introduce $M=Lm$ and the weighting matrix $W_1\in\mathbb{R}^{M\times M}$ as 
  \begin{equation}\label{eqn:wgt_mat_fixed}
     W_1=\frac{1}{m}\diag(\underbrace{I_{m};\cdots;I_{m}}_{T}).
  \end{equation}
    \item[\ding{173}]  \textbf{Random  spatial locations at consecutive times}: let $\Omega^{(2)}= \cup_{t=0}^{T-1}(\Omega_{t}^{(2)}\times \{t\})$ with $\Omega_t^{(2)}:=\{\omega_{t,1}^{(2)},\cdots,\omega_{t,m_t}^{(2)}\}$, where $\Omega_t^{(2)}$ is constructed by drawing  $m_t$ indices independently (with replacement) from the set $[n]$ according to the probability distribution $\mathbf{p}_t^{(2)}$, i.e.,
  $\mathbb{P}(\omega_{t,j}^{(2)}=i)=\mathbf{p}^{(2)}_t(i)$, for all $j\in[m_t] \text{ and }i\in[n]$.
  We set $M=\sum_{t=0}^{T-1}m_t$ and define the  sampling matrix $S_2\in\mathbb{R}^{M\times Tn}$  as
  \begin{equation}\label{eqn:samp_mat_rand_space}
  S_2=\diag( S_2^{(0)}; S_2^{(1)};\cdots; S_2^{(T-1)})
  \end{equation}
 where $S_2^{(t)}\in\mathbb{R}^{m_t\times n}$ is defined as $S_2^{(t)}(i,j):=\begin{cases}
     1, &j=\omega_{t,i}^{(2)}\\
     0,& \text{otherwise},
  \end{cases}
  $ for $t=0,\cdots,T-1$.  Thus, we could define the probability matrix on the sampling set $P_{\Omega^{(2)}}=S_2P_2S_2^\top$. 
  To average the samples at each time instance, we define the weighting matrix $W_2\in\mathbb{R}^{M\times M}$ as
  \begin{equation}\label{eqn:wgt_mat_fixed}
  W_2=\diag(\frac{1}{m_0}I_{m_0}; \frac{1}{m_1}I_{m_1};\cdots; \frac{1}{m_{T-1}}I_{m_{T-1}}).
  \end{equation}
    \item[\ding{174}] \textbf{Random space-time locations}:  $\Omega^{(3)}:=\{\omega_1^{(3)},\cdots,\omega_{m}^{(3)}\}$ is constructed by drawing $m$ indices independently (with replacement) from  $[Tn]$ according to probability distribution $\mathbf{p}^{(3)}$, i.e.,
  $\mathbb{P}(\omega_j=i)=\mathbf{p}^{(3)}(i), \forall j\in[m] \text{ and }i\in[Tn]$.
  For this sampling regime, we define the  sampling matrix $S_3\in\mathbb{R}^{m\times Tn}$ as
  \begin{equation}\label{eqn:samp_mat_rand_tsp}
  S_3(i,j):=\begin{cases}
     1, &\text{ if }j=\omega_i^{(3}\\
     0,& \text{otherwise}
  \end{cases},
  \end{equation}
    the probability matrix on the sampling set $P_{\Omega^{(3)}}=S_3P_3S_3^\top$, and weighting matrix $W_3=\frac{1}{M}I_{M}$ with $M=m$.
\end{itemize}
\end{definition}

\subsection{Stable embedding of $k$-bandlimited signals via space-time samples}
First let's introduce a map   
   $\pi_{A,T}:\mathbb{R}^n \rightarrow \mathbb{R}^{Tn}$ that   maps  $k$-bandlimited signals to    $k$-bandlimited diffusion fields and is defined by $\pi_{A,T}(\mathbf{x})=[\mathbf{x}^\top,(A\mathbf{x})^\top, \cdots, (A^{T-1}\mathbf{x})^\top]^\top$.
Notice that if $\mathbf{x}\in \mathrm{span}(U_k)$, then $\pi_{A,T}(\mathbf{x})\in \mathrm{span}(\widetilde{U}_{k,T})$, where  $\widetilde{U}_{k,T}  \in \mathbb{R}^{Tn\times k}$ consists of orthonormal columns vectors and is defined below
 \begin{equation}\label{eqn:tilde_UkL}
 \widetilde{U}_{k,T}=\begin{bmatrix}\frac{1}{f_{T}(\lambda_1)}\mathbf{u}_1& \frac{1}{f_{T}(\lambda_2)} \mathbf{u}_2&\cdots&\frac{1}{f_{T}(\lambda_k)} \mathbf{u}_k\\\frac{\lambda_1}{f_{T}(\lambda_1)} \mathbf{u}_1& \frac{\lambda_2}{f_{T}(\lambda_2)} \mathbf{u}_2&\cdots & \frac{\lambda_k}{f_{T}(\lambda_k)} \mathbf{u}_k\\ \vdots&\vdots&\cdots&\vdots\\ \frac{\lambda_1^{T-1}}{f_{T}(\lambda_1)} \mathbf{u}_1 &\frac{\lambda_2^{T-1}}{f_{T}(\lambda_2)}\mathbf{u}_2&\cdots&\frac{\lambda_k^{T-1}}{f_{T}(\lambda_k)} \mathbf{u}_k\end{bmatrix}.
 \end{equation}
This observation implies that $\pi_{ A,T}:\mathrm{span}(U_k) \rightarrow \mathrm{span}(\widetilde{U}_{k,T})$. The following lemma shows that $\pi_{A,T}$  is a stable embedding  and specifies lower and  upper embedding constants. 

\begin{lemma}\label{lmm: embedding}
For any $\mathbf{x}\in \mathrm{span}(U_k)$, we have that 
\begin{equation*}
f_{T}(\lambda_k)\|\mathbf{x}\|_2 \leq  \|\pi_{{A},T}(\mathbf{x})\|_2  \leq f_{T}(\lambda_1)\|\mathbf{x}\|_2,
\end{equation*}
where the map $f_T$ is defined in \eqref{map:f_T}. 
\end{lemma}

The key observation that $\pi_{ A,T}(\mathrm{span}(U_k))= \mathrm{span}(\widetilde{U}_{k,T})$ inspires us to interpret $k$-bandlimited diffusion fields  as $k$-bandlimited signals over an extended space-time graph $\mathcal{G}_T$,  whose top $k$ Laplacian eigenvectors are given by $ \widetilde{U}_{k,T}$.  Leveraging this interpretation, we connect the space-time sampling of bandlimited diffusion fields over $\mathcal{G}$ with the classical spatial sampling of  $k$-bandlimited graph signals over  $\mathcal{G}_T$.

 Similar to compressed sensing, the number of space-time samples needed to guarantee stable  reconstruction of $k$-bandlimited diffusion fields will depend on a quantity,   the spectral graph weighted coherence of order $(k,T)$, that represents how the energy of bandlimited diffusion fields spreads over the space-time nodes.  We will  derive and understand this quantity via sampling on $\mathcal{G}_T$. 

 Let  $\delta_i \in \mathbb{R}^{Tn}$ denote the Dirac at the space-time node $i$.
 The quantity 
\begin{equation}\label{gcoherence}
\|\widetilde{U}_{k,T}(i,:)\|_2 = \|\widetilde{U}_{k,T}^{\top}\delta_i\|_2 
\end{equation} describes how much the energy of $\delta_i$ is concentrated on the first $k$ Fourier modes of $\mathcal{G}_T$, or equivalently, the distribution of the energy of the $k$-bandlimited diffusion fields over the space-time nodes.

 When the quantity \eqref{gcoherence}  is one, then there exist $k$-bandlimited diffusion fields whose energy is solely concentrated on the $i$th space-time node, and not sampling the  $i$th node jeopardises the chance of reconstructions. When it is zero, then  no $k$-bandlimited diffusion fields over $\mathcal{G}$ has a part of its energy on the $i$th node; one can safely remove this node from the sampling set. We thus see that  the quality of our sampling method will depend on the interplay between the sampling distribution and the quantities  $\|\widetilde{U}_{k,T}^{\top}\delta_i\|_2$  for $i \in [nT]$. Ideally, the larger  $\|\widetilde{U}_{k,T}^{\top}\delta_i\|_2$ is, the higher probability of the node $i$ is observed.  Below, we introduce the spectral graph weighted coherence of order $(k,T)$ to characterize the interplay between the sampling distribution for each regime and space-time nodes, where the case of $T=1$ corresponds to the graph weighted coherence introduced in \cite[Definition 2.1]{puy2018random}   for the static case. 

\begin{definition}[Graph spectral  weighted coherence of order $(k,T)$ $\mathbf{\nu}_{j,\mathbf{p}^{(j)}}^{k,T}$ ]
\label{def:gswc}
Let $\mathbf{p}^{(j)}$ represent a sampling distribution for $j=1,2,3$. The graph spectral weighted coherence $\mathbf{\nu}_{j,\mathbf{p}^{(j)}}^{k,T}$ of order $(k,T)$ for the triple $(\G,T,\mathbf{p}^{(j)})$ is defined as following
\begin{itemize}
  \item[\ding{172}] \textbf{Random  spatial locations at the initial time:}  \vspace{-1.5mm}
  \begin{equation}\label{eqn:co_fixed}
  \begin{aligned}
      \mathbf{\nu}_{1,\mathbf{p}^{(1)}}^{k,T}:=&
      \max_{1\leq i\leq n}\left\{\frac{1}{\sqrt{\mathbf{p}^{(1)}(i)}}\left\| \widetilde{U}_{k,T}((i:n:Tn),:) \right\|_2\right\},
       \end{aligned}
  \end{equation}
  where $\widetilde{U}_{k,T}((i:n:Tn),:)$ is the row submatrix of $\widetilde{U}_{k,T}$ with row indices  in $(i:n:Tn)$ 
  
    \item[\ding{173}]  \textbf{Random  spatial locations at consecutive times:}  
     \vspace{-1.5mm}
    \begin{equation}\label{eqn:co_rand_space}
        \mathbf{\nu}_{2,\mathbf{p}^{(2)}}^{k,T}:=\begin{bmatrix}
        \mathbf{\nu}_{\mathbf{p}_0^{(2)}}\\
         \vdots\\
          \mathbf{\nu}_{\mathbf{p}_{T-1}^{(2)}}
        \end{bmatrix}=\begin{bmatrix}\max_{1\leq i\leq n}  \frac{1}{\sqrt{\mathbf{p}_0^{(2)}(i)}}\|\widetilde{U}_{k,T}^\top\delta_{i}\|_2\\
        \vdots\\
       \max_{1\leq i\leq n}  \frac{1}{\sqrt{\mathbf{p}_{T-1}^{(2)}(i)}}\|\widetilde{U}_{k,T}^\top\delta_{i+(T-1)n}\|_2\\
        \end{bmatrix},
    \end{equation}
     
    \item[\ding{174}] \textbf{Random space-time locations: }
    \begin{equation}\label{eqn:co_tsp}
      \mathbf{\nu}_{3,\mathbf{p}^{(3)}}^{k,T}:=\max_{1\leq i\leq Tn}\left\{\frac{1}{\sqrt{\mathbf{p}_{i}^{(3)}}}\|\widetilde{U}_{k,T}^{\top}\delta_i\|_2 \right\},
  \end{equation}
\end{itemize}

\end{definition}

From the definition,   the spectral graph weighted coherence of order $(k,T)$ for each regime is thus a characterization of the interaction among the $k$-bandlimited signals, spectrum of the graph Laplacian, and the respective sampling distributions. For the ease of the readers, we summarize three sampling regimes in Table \ref{tab:var_SR}.
\begin{table}[t]\footnotesize
\centering
\begin{tabular}{|l|l|l|l| }
\hline
Sampling  &\multirow{2}{*}{ \ding{172}} &\multirow{2}{*}{ \ding{173}} &\multirow{2}{*}{ \ding{174} } \\
Regimes &&&\\\hline
 $\mathbf{p}^{(j)}$&   $\mathbf{p}^{(1)}\in\mathbb{R}^{n}$ & $\mathbf{p}^{(2)}=[\mathbf{p}_0^{(2)};
\cdots;
\mathbf{p}_{T-1}^{(2)}]\in\mathbb{R}^{n\times T}$&  $\mathbf{p}^{(3)}\in\mathbb{R}^{Tn}$ \\
\hline
$P_j$ &$P_1=\diag(\underbrace{\mathbf{p}^{(1)};\cdots;\mathbf{p}^{(1)}}_{T})$&$P_2=\diag\left(\mathbf{p}_0^{(2)};\mathbf{p}_1^{(2)};\cdots;\mathbf{p}_{T-1}^{(2)}\right) $ &$P_3=\diag(\mathbf{p}^{(3)})$\\
 \hline
  \multirow{2}{*}{ $\Omega$}& { $\Omega^{(1)}=\Omega^{(1)}_0 \times [T]$} & $\Omega^{(2)}= \bigcup_{t=0}^{T-1}\Omega_{t}^{(2)} \times \{t\}$,  &  \multirow{2}{*}{ $\Omega^{(3)}=\{\omega_1^{(3)},\cdots,\omega_{m}^{(3)}\}$}\\
  & $\Omega^{(1)}_0 =\{\omega_1^{(1)},\cdots,\omega_{m}^{(1)}\}$& $\Omega_t^{(2)}=\{\omega_{t,1}^{(2)},\cdots,\omega_{t,m_t}^{(2)}\}$.&\\
  \hline
\multirow{4}{*}{ $S_j$}& $S_{1}=\diag( \underbrace{S_1^{(0)}; \cdots; S_1^{(0)}}_{T})$,&  $S_2=\diag( S_2^{(0)}; \cdots; S_2^{(T-1)})$ &\multirow{2}{*}{ $S_3\in\mathbb{R}^{m\times Tn}$ with}\\
&$ S_1^{(0)}\in\mathbb{R}^{m\times n}$ with &  $S_2^{(t)}\in\mathbb{R}^{m_t\times n}$ with &\\
&  $ S_1^{(0)}(i,j)=\begin{cases}
     1, & j=\omega_i\\
     0,& \text{otherwise}.
  \end{cases}$ & $S_2^{(t)}(i,j)=\begin{cases}
     1, & j=\omega_{t,i}\\
     0,& \text{otherwise}
  \end{cases}$
  &   $S_3(i,j)=\begin{cases}
     1, & j=\omega_i\\
     0,& \text{otherwise}.\\
     \end{cases}$\\
  &&for $t=0,\cdots,T-1$. &\\
 \hline
   $P_{\Omega^{(j)}}$& $P_{\Omega^{(1)}}=
   \diag(S_1 P_1 S_1^{\top})$& $P_{\Omega^{(2)}}=\diag(S_2P_2S_2^{\top})$ & $P_{\Omega^{(3)}}=\diag(S_3P_3S_3^{\top})$\\
 \hline
 $W_j$ & $W_1=\frac{1}{m}I_{mT}$&$W_2=\diag(\frac{1}{m_0}I_{m_0};  \cdots; \frac{1}{m_{T-1}}I_{m_{T-1}})$& $W_3=\frac{1}{m}I_m $\\
\hline 
$\mathbf{\nu}_{j,\mathbf{p}^{(j)}}^{k,T}$&$\mathbf{\nu}_{1,\mathbf{p}^{(1)}}^{k,T}$ in \eqref{eqn:co_fixed} &  $\mathbf{\nu}_{2,\mathbf{p}^{(2)}}^{k,T}$ in \eqref{eqn:co_rand_space}& $\mathbf{\nu}_{3,\mathbf{p}^{(3)}}^{k,T}$ in \eqref{eqn:co_tsp}\\
\hline
\end{tabular}
\caption{\footnotesize Relevant notations and matrices for different random space-time sampling regimes.}\label{tab:var_SR}
\end{table}

\subsection{Properties of spectral graph weighted coherences}\label{subs:properties} As shown later, the sampling size required to achieve  stable embedding  essentially depends on the  spectral graph weighted coherence. Next we  analyze its properties in Proposition \ref{prop:gcoherence}. 

\begin{proposition}\label{prop:gcoherence} We list several properties of the spectral graph weighted coherences of order $(k,T)$ that are useful for later analysis. 

\begin{itemize}

\item[(a)]For any probability distribution $ \mathbf{p}^{(j)}$ for $j=1,2,3$, we have that 

\begin{itemize}
    \item [(a.1)] $ \left(\mathbf{\nu}_{1,\mathbf{p}^{(1)}}^{k,T}\right)^2\geq \sum_{i=1}^{n}\left\| \widetilde{U}_{k,T}((i:n:Tn),:) \right\|_2^2 \geq k$ and the first inequality becomes an equality  when 
    \begin{equation}\label{opt1}
      \mathbf{p_{opt}}^{(1)}(i)=\frac{ \left\| 
    \widetilde{U}_{k,T}((i:n:Tn),:)) \right\|_2^2}{\sum_{\ell=1}^{n}\left\| 
     \widetilde{U}_{k,T}((\ell:n:Tn),:) \right\|_2^2} \quad  \text{ for } i=1,\cdots,n.
     \end{equation}

     \item[(a.2)] $(\nu_{2,\mathbf{p}^{(2)}}^{k,T}(t))^2 \geq  \sum\limits_{\ell=1}^{k}\frac{\lambda_{\ell}^{2t}}{f_{T}^2(\lambda_\ell)} $ for $t=0,\cdots,T-1$ and the equality holds when 
     \begin{equation}\label{opt2}
     \mathbf{p_{t,opt}}^{(2)}(i)= \frac{\|\widetilde{U}_{k,T}^{\top}\delta_{i+tn}\|_2^2}{\sum_{\ell=1}^{k}\lambda_{\ell}^{2t}/f_{T}^2(\lambda_\ell) }.
     \end{equation}
     As a result, 
     $\sum_{t=1}^{T} \left(\nu_{2,\mathbf{p}^{(2)}}^{k,T}(t)\right)^2 \geq k$.  
     \item[(a.3)] $\left(\mathbf{\nu}_{3,\mathbf{p}^{(3)}}^{k,T}\right)^2\geq k$ and the equality holds when 
     \begin{equation}\label{opt3}
     \mathbf{p_{opt}}^{(3)}(i)=\frac{ \|\widetilde{U}_{k,T}^{\top}\delta_i\|_2^2}{k}
     \end{equation}
\end{itemize}

\item[(b)] For any sampling distributions, we have that
\begin{equation*}
     T\left(\nu_{1,\mathbf{p}^{(1)}}^{k,T}\right)^2 \geq   \left(\nu_{3,\mathbf{p}^{(3)}}^{k,T}\right)^2  \geq  \sum_{t=0}^{T-1}\left(\nu_{2,\mathbf{p}^{(2)}}^{k,T}(t)\right)^2. \end{equation*}
 In particular, we call the probability  defined in \eqref{opt1}, \eqref{opt2} and \eqref{opt3} the optimal distribution for the regime 1,2 and 3. In this case, we have 
 \begin{equation*} T\left(\nu_{1,\mathbf{p}^{(1)}}^{k,T}\right)^2 \geq   \left(\nu_{3,\mathbf{p}^{(3)}}^{k,T}\right)^2  =  \sum_{t=0}^{T-1}\left(\nu_{2,\mathbf{p}^{(2)}}^{k,T}(t)\right)^2=k.  \end{equation*}

\end{itemize}
\end{proposition}

We are now ready to present the main theorem in this section. We represent  three random space-time sampling procedures of $k$-bandlimited diffusion fields as three operators $P_{\Omega^{(j)}}^{-\frac{1}{2}}W_j^{\frac{1}{2}}S_j$ for $j=1,2,3$. We provided a unified theory to show that their composition with the map $\pi_{ A,T}$ yields  stable embeddings of $k$-bandlimited signals with high probability, as long as the  number of space-time samples taken is above a threshold depending on the respective spectral graph coherence.

\begin{theorem}\label{thm:embedding_all}
Let $P_{\Omega^{(j)}},S_j,W_j$ be defined in Table \ref{tab:var_SR}  with the sampling distribution $\mathbf{p}^{(j)}$. 
For any $\delta,\epsilon \in (0,1)$, the following inequality
\begin{equation}\label{eqn:frame}
(1-\delta)f_{T}^2(\lambda_k)\|\mathbf{x}\|^2\leq   \| P_{\Omega^{(j)}}^{-\frac{1}{2}}W_j^{\frac{1}{2}}S_j\pi_{ A,T}(\mathbf{x})\|_2^2   \leq (1+\delta)f_{T}^2(\lambda_{1})\|\mathbf{x}\|_2^2,
\end{equation}
 holds with probability at least $1-\epsilon$,
provided that $m\geq \frac{3}{\delta^2}\left(\mathbf{\nu}_{j,\mathbf{p}^{(j)}}^{k,T}\right)^2\log\left(\frac{2k}{\epsilon} \right)$ for $j=1,3$, and $m_t\geq \frac{3}{\delta^2}\left(\mathbf{\nu}_{2,\mathbf{p}^{(2)}}^{k,T}(t) \right)^2\log\left(\frac{2k}{\epsilon} \right)$ for  $j=2$, where $m_t$ is the number of space-time samples taken at time $t$.
\end{theorem}

\eqref{eqn:frame} represents the stable embedding of $k$-bandlimited signals into $\mathbb{R}^m$. Let us take two different vectors $\mathbf{x}_1 \neq \mathbf{x}_2$ in $\mathrm{span} (U_k)$. Notice that $\mathbf{x}_1-\mathbf{x}_2 \in \mathrm{span}(U_k)$, then   $\| P_{\Omega^{(j)}}^{-\frac{1}{2}}W_j^{\frac{1}{2}}S_j\pi_{ A,T}(\mathbf{x}_1-\mathbf{x}_2)\|_2^2\geq (1-\delta)f_{T}^2(\lambda_k)\|\mathbf{x}_1-\mathbf{x}_2\|^2>0.$ 
\begin{remark}We make  several important comments  on  Theorem \ref{thm:embedding_all}.
\begin{itemize}
\item In the case of $T=1$,   $f_{T}^2(\lambda_{1})=f_{T}^2(\lambda_{k})=1$ and all the spectral graph weighted coherences equal to $(\nu_{\mathbf{p}}^{k})^2=\max_{1\leq i\leq n} \frac{\|U_k^{\top}\delta_i\|^2_2}{\mathbf{p}(i)}$, which coincides with  \cite[ Theorem 2.2]{puy2018random},  that the embedding operator satisfies the Restricted Isometry Property (RIP). Note that for regime 1, the total number of space-time samples is $Tm$ and  our analysis in  Proposition \ref{prop:gcoherence}  shows that 
$T(\nu_{1,\mathbf{p}^{(1)}}^{k,T})^2$ is  the largest among three regimes.  Even in this case, due to the temporal dynamics, one would expect fewer requirements of spatial sensors than the static case: Indeed, one can verify that $ \left(\nu_{1,\mathbf{p}^{(1)}}^{k,T}\right)^2 \leq \left(\nu_{\mathbf{p}}^{k}\right)^2$.
We also refer to the numerical section (Section \ref{sec:experiments}) for more empirical evidences.

\item In practical situations, one just needs to measure the bandlimited diffusion fields  by $S_j\pi_{ A,T}(\mathbf{x})$. The re-weighting by
 $P_{\Omega^{(j)}}^{-\frac{1}{2}}W_j^{\frac{1}{2}}$ can be done offline.

\item The spectral graph weighted coherences $(\mathbf{\nu}_{j,\mathbf{p}^{(j)}}^{k,T})^2\geq k$ for $j=1,3$ and $\sum_{t=1}^{T}(\mathbf{\nu}_{j,\mathbf{p}^{(2)}}^{k,T})^2(t)\geq k$ (see Proposition \ref{prop:gcoherence}). It indicates that, we need to have at least $k$ space-time samples in all three regimes. Note that $k$ is also the minimum number of measurements that one must take to ensure the reconstruction of the $k$ bandlimited signals. {Furthermore, the optimal spectral graph weighted coherence for regimes 2 and 3 is $k$ and therefore is  independent of the total sampling time. This indicates that, one can expect to use fewer number of spatial sensors at each time instance on average as $T$ increase without hurting too much  reconstruction accuracy, we refer to Theorem \ref{thm: err_dirct}.}
\item By part~(b) of   Proposition \ref{prop:gcoherence}, the sampling complexity bounds satisfy that  regime 1 $\geq$ regime 3 $\geq$ regime 2. This is  consistent with our numerical results in Section 6. 
\end{itemize}
\end{remark}

Theorem \ref{thm:embedding_all} generalizes the well-known compressed sensing results in bounded orthonormal systems to the sampling of bandlimited diffusion fields with various space-time sampling regimes. Here we are in a different setting where the signal model is the bandlimited diffusion field. We leverage this temporal structure  to connect with the orthonormal matrix $\widetilde U_{k,T}$. This  connection allows us to utilize the same concentration inequalities used in compressed sensing yet  refine and tighten the  space-time sampling conditions.

 In part $(a)$ of Proposition \ref{prop:gcoherence}, we derive $\mathbf{p}_{\textbf{opt}}^{(j)}$
 for regime $j$, that minimizes the  corresponding $\nu_{j,\mathbf{p}^{(j)}}^{k,T}$. 
 Combining Theorem \ref{thm:embedding_all} and part $(a.2)$, part $(a.3)$ of Proposition \ref{prop:gcoherence}, we can see that there always exists a sampling distribution for regime 2 and regime 3 such that $\mathcal{O}(k\log k)$ space-time samples are enough to capture all $k$-bandlimited signals. For regime 1, due to the inequality in part $(a.1)$, one would expect   more space-time samples than other two regimes.

\subsection{Intuitive links between graph diffusion process and the spectral graph weighted coherence}\label{graphdiffusion} We will  interpret $i\in [nT]$ 
as a space-time node in $[n]\times [T]$. We give here some examples showing how the spectral graph weighted coherence changes for diffusion processes over different graphs.

Consider circulant graphs with periodic boundary conditions, whose adjacent matrices are real and symmetric circulant matrices. In this case, the eigenvectors of their Laplacian are the $n$ dimensional classical Fourier modes. For simplicity, assume that the top $k$ eigenvalues are distinct. Let $\lambda_1,\lambda_2,\cdots,\lambda_k$ denote the eigenvalues of evolution operator $A$ for the heat diffusion process. { In this case, one can show that  $\mathbf{p_{opt}}^{(1)}$ is  the uniform distribution, and $\mathbf{p_{opt}}^{(2)}$ and $\mathbf{p_{opt}}^{(3)}$ have the identical components for a fixed time instance. These conclusions come from the following two facts:}\\
(i)For regime 1, 
$\left\| \widetilde{U}_{k,T}((i:n:Tn),:) \right\|_2 $  is the same for all $1\leq i\leq n$ and it equals to the operator norm of the matrix 
 \begin{equation*}\frac{1}{\sqrt{n}}\begin{bmatrix}1&1&\cdots&1\\ \lambda_1&\lambda_2&\cdots&\lambda_k\\ \vdots&\vdots&\cdots&\vdots\\  \lambda_1^{T-1}&\lambda_2^{T-1}&\cdots&\lambda_k^{T-1} \end{bmatrix} \begin{bmatrix} {1}/{f_T(\lambda_1)}&0&\cdots&0\\ 0& {1}/{f_T(\lambda_2)} &\cdots&0\\ \vdots&\vdots&\cdots&\vdots\\ 0 & 0 &\cdots& {1}/{f_T(\lambda_k)}\end{bmatrix}.
 \end{equation*}
(ii) For regime 2 and regime 3, 
$\left\| \widetilde{U}_{k,T}^\top \delta_i \right\|_2$ is the same for all node $i$ at the  time instance $t$ and it equals to 
$\frac{1}{\sqrt{n}}\sqrt{\sum_{j=1}^{k} {\lambda_j^{2(t-1)}}/{f_T(\lambda_j)^2}}$.


Let us consider  the graph made of $k$ disconnected components of size $n_1,\cdots, n_k$ such that $\sum_{i=1}^{k}n_i=n$.  Let the Laplacian $L$ be the  graph Laplacian with edge weights all equal to one. Then a basis of $\mathrm{span}(U_k)$ is the concatenation of the square root of the degree operator times the indicator vectors of each component. Moreover, $\mathrm{span}(U_k)$  is the eigenspace associated to the eigenvalue 0 of $L$. So the heat diffusion process over this eigenspace is the trivial identity. So the temporal dynamics play no role in sampling and it will be the same with the static case.  By calculation, for the node $i$ at component $j$, we have 
\begin{equation*}
\left\| \widetilde{U}_{k,T}((i:n:Tn),:) \right\|_2 =\frac{\sqrt{d_{ij}}}{\sqrt{d_j}} ~\text{ and }~
\left\| \widetilde{U}_{k,T}^\top \delta_i \right\|_2 =\frac{\sqrt{d_{ij}}}{\sqrt{d_jT}}, 
\end{equation*}where $d_{ij}$ is the degree of  node $i$ and $d_j$ is the total degree of nodes in component $j$. 

To keep things simplified, let us assume each component is regular, i.e., each node has the same degree. In this case,  in $\mathbf{p_{opt}}^{(1)}$, the probability of choosing the spatial node $1\leq i \leq n$ at component $j$ is $\frac{1}{kn_j}$, coinciding with the case $T=1$. For regime 2 and 3, the optimal probability of choosing the spatial node $i$ at time $t$ is $\frac{1}{k{n_j}}$ and $\frac{1}{k{n_jT}}$ respectively. Thus, if all components have the same size, then uniform sampling is the optimal sampling. If the components have different sizes, {the smaller  a component is, the larger  the probability of sampling one of its nodes is}.  We can see that optimal sampling requires that each component is sampled with probability $\frac{1}{k}$, independent of its size. The probability that each component is sampled at least once, a necessary condition for perfect recovery, is thus higher than that using uniform sampling. In the numerical section (Section \ref{sec:experiments}), we consider the loosely defined
community-structured graph, which is a  relaxation of the strictly  disconnected component example. In this case, one also expects that $\mathbf{p}_{\textbf{opt}}^{(j)}$ to sample a node is inversely proportional to the size of its community.

\section{Reconstruction of $k$-bandlimited signals from space-time samples}
\label {Main}
In this section, we are interested in designing a robust procedure to recover the initial signal in $\mathrm{span}(U_k)$ from $M$  space-time samples. Specifically, 
  the space-time samples $\mathbf{y}\in\mathbb{R}^{M}$ are defined  as 
\[\mathbf{y}=S_j(\pi_{ A,T}(\mathbf{x}))+\mathbf{e}
\]
where  $\mathbf{e}\in\mathbb{R}^{M}$ models the observational noise and $S_j$ is defined in Definition \ref{def_MPW}. 

When one knows a basis for $\mathrm{span}(U_k)$, we adopt the standard least squares method to estimate $\mathbf{x}\in\text{span}(U_k)$ from $\mathbf{y}$. This is done by solving 
\begin{equation}\label{eqn:obj1}
    {\mathbf{x}}^*=\arg\min_{\widetilde{\mathbf{x}}\in\text{span}({U}_{k})}\|P_{\Omega^{(j)}}^{-1/2}W_{j}^{\frac{1}{2}}(S_j\pi_{ A,T}(\widetilde{\mathbf{x}})-\mathbf{y})\|_2.
\end{equation}
Here we introduce the weighting   matrix $P_{\Omega^{(j)}}^{-1/2}W_{j}^{\frac{1}{2}}$ in \eqref{eqn:obj1} to ensure the stable embedding property in Theorem \ref{thm:embedding_all}.  The  solution to \eqref{eqn:obj1} is obtained by solving 

\[\mathrm{z}=\arg\min_{\widetilde{\mathrm{z}}\in\mathbb{R}^{k}}\|P_{\Omega^{(j)}}^{-1/2}W_{j}^{\frac{1}{2}}(S_j\pi_{ A,T}U_{k}(\widetilde{\mathbf{z}})-\mathbf{y})\|_2
\]
and setting ${\mathbf{x}}^*=U_k \mathbf{z}$.  We generalize the idea of   \cite[Theorem 3.1]{puy2018random} to prove that ${\mathbf{x}}^*$ is a faithful estimation of $\mathbf{x}$ so that   \cite[Theorem 3.1]{puy2018random}  becomes a special case of the following theorem:

\begin{theorem}\label{thm: err_dirct}
Let $\Omega^{(j)}$ be a sampling set  according to a sampling distribution $\mathbf{p}^{(j)}$, and $S_j$ be the sampling matrix associated to $\Omega^{(j)}$. Let $\varepsilon ,\delta\in(0,1)$ and suppose that $m\geq \frac{3}{\delta^2}(\nu_{j,\mathbf{p}^{(j)}}^{k,T})^2\log(\frac{2k}{\varepsilon})$ for $j= 1, 3$ or $m_t\geq \frac{3}{\delta^2}(\nu_{2,\mathbf{p}^{(2)}}^{k,T}(t))^2\log(\frac{2k}{\varepsilon})$ for $j=2$. With probability at least $1-\varepsilon$, the following holds for all $\mathbf{x}\in\text{span}( {U}_{k})$ and all $\mathbf e\in\mathbb{R}^{M}$.
\begin{enumerate}[label=\roman*)]
    \item Let ${\mathbf{x}}^*$ be the solution of Problem \eqref{eqn:obj1} with $\mathbf{y}=S_j(\pi_{ A,T}(\mathbf{x}))+\mathbf{e}$. Then 
    \begin{equation}\label{err:1}
        \|\mathbf{x}^*-\mathbf{x}\|_2\leq \frac{2}{\sqrt{1-\delta}f_{T}(\lambda_k)}\|P_{\Omega^{(j)}}^{-1/2}W_{j}^{\frac{1}{2}}\mathbf{e}\|_2,
    \end{equation}
    where $f_T$ is defined in \eqref{map:f_T}
    \item There exist particular vectors $\mathbf e_0\in\mathbb{R}^{M}$ such that the solution ${\mathbf{x}}^*$ of Problem \eqref{eqn:obj1} with $\mathbf{y}=S_j\pi_{ A,T}(\mathbf{x})+\mathbf e_0$ satisfies 
    \begin{equation}\label{err:2}
        \|\mathbf{x}^*- {\mathbf{x}}\|_2\geq \frac{1}{\sqrt{1+\delta} f_{T}(\lambda_1)}\|P_{\Omega^{(j)}}^{-1/2}W_{j}^{\frac{1}{2}}\mathbf e_0\|_2.
    \end{equation}
\end{enumerate}
\end{theorem}

The estimate \eqref{err:1} tells us that if $\mathbf{e}=0$, then ${\mathbf{x}}^*=\mathbf{x}$ with high probability.   In the presence of noise,  \eqref{err:1}  shows that  $\|\mathbf{x}^*- {\mathbf{x}}\|_2 $ scales linearly with $\|P_{\Omega^{(j)}}^{-1/2}W_{j}^{\frac{1}{2}}\mathbf{e}\|_2$, and the estimate \eqref{err:2} shows that the bound on  $\|\mathbf{x}^*- {\mathbf{x}}\|_2 $  is sharp up to a constant. In the case of uniform sampling, we have 
\begin{eqnarray}
\|P_{\Omega^{(j)}}^{-1/2}W_{j}^{\frac{1}{2}}\mathbf{e}\|_2^2&=&{\frac{nT}{M}}\|\mathbf{e}\|_2^2, j=1,3 \label{bd1}\\
\|P_{\Omega^{(2)}}^{-1/2}W_{2}^{\frac{1}{2}}\mathbf{e}\|_2^2&=& \sum_{t=0}^{T-1}\frac{n}{m_t}\|\mathbf{e}_t\|_2^2\label{bd2},
\end{eqnarray} where $\mathrm{e}_t$ denotes the projection of $\mathrm{e}$ into the components at time instance $t$.

In the case of non-uniform sampling, the noise may be amplified a lot at  some particular draws of $\Omega$ due to some probability weights could be very close to 0.  Non-uniform sampling could thus be very sensitive to noise unlike uniform sampling. Fortunately, this is a worst case scenario and it is unlikely to draw samples from locations where the probability is small by our sampling procedures. For each regime, the expectation $\mathbb{E}\|P_{\Omega^{(j)}}^{-1/2}W_{j}^{\frac{1}{2}}\mathbf{e}\|_2^2$ over the choice of $\Omega^{(j)}$
is as in the right hand side of \eqref{bd1} and \eqref{bd2} respectively.  Therefore  the noise term is not too large on average over the draw of $\Omega$.

Notice that to solve \eqref{eqn:obj1}, we need to know $U_k$ in advance, which could be computationally expensive for large scale graphs. When $U_k$ is not given,  we propose to estimate $\mathbf{x}$ by solving the following {regularized least square} problem
\begin{equation}\label{eqn:obj2}
    \min_{\widetilde{\mathbf{ x}}\in\mathbb{R}^{n}}\|P_{\Omega^{(j)}}^{-1/2}W_{j}^{\frac{1}{2}}(S_j\pi_{ A,T}(\widetilde{\mathbf {x}})-\mathbf{y})\|_2^2+\gamma \widetilde{\mathbf{ x}}^\top g(L) \widetilde{{\mathbf{x}}},
\end{equation}
where $\gamma>0$ and $g:\mathbb{R}\rightarrow\mathbb{R}$ is a nonnegative and nondecreasing polynomial function. One can find a solution to \eqref{eqn:obj2} by solving the following equation 
\[ (\pi_{ A,T}^\top S_j^\top W_j^{\frac{1}{2}}P_{\Omega^{(j)}}^{-1}W_{j}^{\frac{1}{2}}S_j\pi_{ A,T}+\gamma g(L) ){\mathbf{x}}=\pi_{ A,T}^\top S_j^\top W_j^{\frac{1}{2}}P_{\Omega^{(j)}}^{-1}W_{j}^{\frac{1}{2}} \mathbf{y}.
\]

The idea of this relaxation comes from graph-filtering techniques which have found connections with decoders used in the semi-supervised learning on graphs \cite{puy2018random}.  In our problem, we use the penalty term $ \gamma \widetilde{\mathbf {x}}^\top g(L)\widetilde{\mathbf{x}}$ to incorporate the smoothness of initial signal.  The next theorem provides bounds for the error  between the original signal $\mathbf{x}$ and the solution of \eqref{eqn:obj2}, where \cite[Theorem 3.2]{puy2018random} becomes a special case for $T=1$. 

\begin{theorem}\label{thm:main_regular}
Let $\Omega^{(j)}$ be a sampling set which obtained according to the sampling distribution $\mathbf{p}^{(j)}$.  $P_{\Omega^{(j)}}$, $S_j$, $W_j$  are provided in Table \ref{tab:var_SR}  associated with   $\mathbf{p}^{(j)}$.  
$M_{\max}>0$ is a constant such that $\|P_{\Omega^{(j)}}^{-1/2} W_{j}^{\frac{1}{2}}S_j\|_2\leq M_{\max}$. Let $\varepsilon ,\delta\in(0,1)$ and suppose that $m\geq \frac{3}{\delta^2}(\nu_{j,\mathbf{p}^{(j)}}^{k,T})^2\log(\frac{2k}{\varepsilon})$ for $j= 1, 3$ (and $m_t\geq \frac{3}{\delta^2}\left(\mathbf{\nu}_{2,\mathbf{p}^{(2)}}^{k,T}(t) \right)^2\log(\frac{2k}{\varepsilon})$ for $j=2$). With probability at least $1-\varepsilon$, the following holds for all $\mathbf{x}\in\text{span}({U}_{k})$ and all $\mathbf e\in\mathbb{R}^{M}$, all $\gamma>0$, and all nonnegative and nondecrasing polynomial functions $g$ such that $g(\sigma_{k+1})>0$, where $\sigma_{k+1}$ is the ($k+1$)th singular value of  $L$.

 Let $\mathbf{x}^*$ be the solution of Problem \eqref{eqn:obj2} with $\mathbf{y}=S_j(\pi_{ A,T}(\mathbf{x}))+\mathbf e$. Then 
    \begin{equation}\label{eqn:thm:main_regular1}
    \begin{aligned}
   &\|\alpha^*-\mathbf{x}\|_2\leq \frac{1}{\sqrt{1-\delta}f_{T}(\lambda_{k})}\left(2+\frac{M_{\max}f_{T}(\lambda_{k+1}) }{\sqrt{\gamma g(\sigma_{k+1}) }}\right)\|P_{\Omega^{(j)}}^{-1/2}W_{j}^{\frac{1}{2}}\mathbf{e}\|_2+\\
        &\frac{1}{\sqrt{1-\delta}f_{T}(\lambda_{k})}\left(M_{\max}f_{T}(\lambda_{k+1})\sqrt{\frac{g(\sigma_{k})}{g(\sigma_{k+1})}}+\sqrt{\gamma g(\sigma_{k})} \right)\|\mathbf{x}\|_2
        \end{aligned}
    \end{equation}
    and
       \begin{equation}\label{eqn:thm:main_regular2}
        \|\beta^*\|_2\leq\frac{1}{\sqrt{\gamma g(\sigma_{k+1})}}\|P_{\Omega^{(j)}}^{-1/2}W_{j}\mathbf{e}\|_2+\left( \sqrt{\frac{g(\sigma_{k})}{g(\sigma_{k+1})}}\|\mathbf{x}\|_2 \right),
    \end{equation}
    where $\alpha^*:=U_kU_k^\top\mathbf{x}^*$, $\beta^*:=(I_n-U_kU_k^\top)\mathbf{x}^*$ and $f_T$ is defined in \eqref{map:f_T}.
\end{theorem}
Notice that when $T=1$, the bounds in Theorem \ref{thm:main_regular} are exactly the bounds in \cite[Theorem 3.2]{puy2018random}, due to that $f_{T}(\lambda_k)=f_{T}(\lambda_{k+1})=1$ and $W_{j}=\frac{1}{m}I_{n}$.
In the estimates,  $x^*=\alpha^*+\beta^*$. 
To find a bound for $\|x^*-x\|_2$, one could simply use the triangular inequality and the bounds in \eqref{eqn:thm:main_regular1}, \eqref{eqn:thm:main_regular2}.

In the absence of noise, we thus have 
\begin{equation*}
\|x^*-x\|_2\leq \left(\left(\frac{M_{\max}f_{T}(\lambda_{k+1})}{\sqrt{1-\delta}f_{T}(\lambda_{k})}+1\right)\sqrt{\frac{g(\sigma_{k})}{g(\sigma_{k+1})}}+\frac{\sqrt{\gamma g(\sigma_{k})}}{\sqrt{1-\delta}f_{T}(\lambda_{k})} \right)\|\mathbf{x}\|_2.
 \end{equation*}
If $g(\sigma_{k})=0$, we notice that we obtain a perfect reconstruction. Note that  $g$ is supposed to be nondecreasing and nonnegative. In addition $\sigma_1\leq\cdots\leq\sigma_{n}$. Thus $g(\sigma_{k})=0$ implies that we also have $g(\sigma_{1})=\cdots=g(\sigma_{k-1})=0$. If $g(\sigma_{k})\neq 0$, the above bound shows that we should choose $\gamma$ as close as possible to $0$ and seek to minimize the ratio $\frac{g(\sigma_{k})}{g(\sigma_{k+1})}$ to minimize the upper bound on the reconstruction error. Notice that if $g(L)=L^{\ell}$ for $\ell\in\mathbb{N}^*$, then the ratio $g(\sigma_{k})/g(\sigma_{k+1})$ decreases as $\ell$ increases. Thus increasing the power of $L$ and taking $\gamma$ sufficiently small to compensate the potential growth of $g(\sigma_{k})$ is a simple solution to improve the reconstruction quality in the absence of noise. {In addition, notice that  $f_T$ is an increasing function with respect to $T$ (see \eqref{map:f_T}). Thus $1/f_T(\lambda_k)$ will be decreasing as $T$ is increasing. Additionally, $\lambda_k\geq \lambda_{k+1}$,  we have that $f_T(\lambda_{k+1})/f_T(\lambda_k)$ is  a nonincreasing function with respect to $T$. Therefore, for fixed $g$, the error bound will be a nonincreasing function with respect to $T$. }

In the presence of noise, for a fixed function $g$, the upper bound on the reconstruction error is minimized for a value of $\gamma$ proportional to $\|P_{\Omega^{(j)}}W_j\mathbf{e}\|_2/\|\mathbf{x}\|_2$. To optimize the result further, one should seek to have $g(\sigma_{k})$ as small as possible and $g(\sigma_{k+1})$ as large as possible.

Therefore,  the reconstruction of the original bandlimited graph signal can be summarized in Algorithm \ref{ALGO:reconstruction}.
 \begin{algorithm}[h!]
\caption{  Procedure to reconstruct the original signal}\label{ALGO:reconstruction}
\begin{algorithmic}[1]
\STATE\textbf{Input:} {  The signal evolution operator ${A}\in\R^{n\times n}$, the samples $\{\mathbf{y}_{t}^{(j)}\}_{t=0}^{T-1}$,  sampling sets $\Omega^{(j)}$, and the probability distributions $\mathbf{p}^{(j)}$.}
\STATE Generate the sampling matrix  $S_j$, the weighting matrix $W_{j}$ from the sampling set $\Omega^{(j)}$ and matrix $P_j$ from the probability distribution $\mathbf{p}^{(j)}$. 
\STATE Stack the samples together to a column vector and denote it as $\mathbf{y}^{(j)}$
 \IF{The support $U_k$ of $\mathbf{x}$ is known} 
 \STATE Set ${\mathbf{x}}^*=\arg\min\limits_{\widetilde{\mathbf{ x}}\in\text{span}({U}_{k})}\|P_{\Omega^{(j)}}^{-1/2}W_{j}^{\frac{1}{2}}(S_j\pi_{ A,T}(\widetilde{\mathbf{x}})-\mathbf{y})\|_2$ as   estimate of $\mathbf{x}$.
\ELSE
\STATE Choose some constant $\gamma$ and regularizization function $g(z)$.
\STATE Set ${\mathbf{x}}^*=\arg\min\limits_{\widetilde{\mathbf{ x}}\in\mathbb{R}^{n}}\|P_{\Omega^{(j)}}^{-1/2}W_{j}^{\frac{1}{2}}(S_j\pi_{ A,T}(\widetilde{\mathbf{ x}})-\mathbf{y})\|_2^2+\gamma \widetilde{\mathbf{x}}^{\top} g(L)\widetilde{\mathbf{ x}}$ as  estimate of $\mathbf{x}$.
 \ENDIF
\STATE\textbf{Output: $\mathbf{x}^*$. }
\end{algorithmic}
\end{algorithm}

\section{Numerical Experiments}
\label{sec:experiments}

\begin{figure}[h]
\centering
\includegraphics[width=.18\linewidth, keepaspectratio]{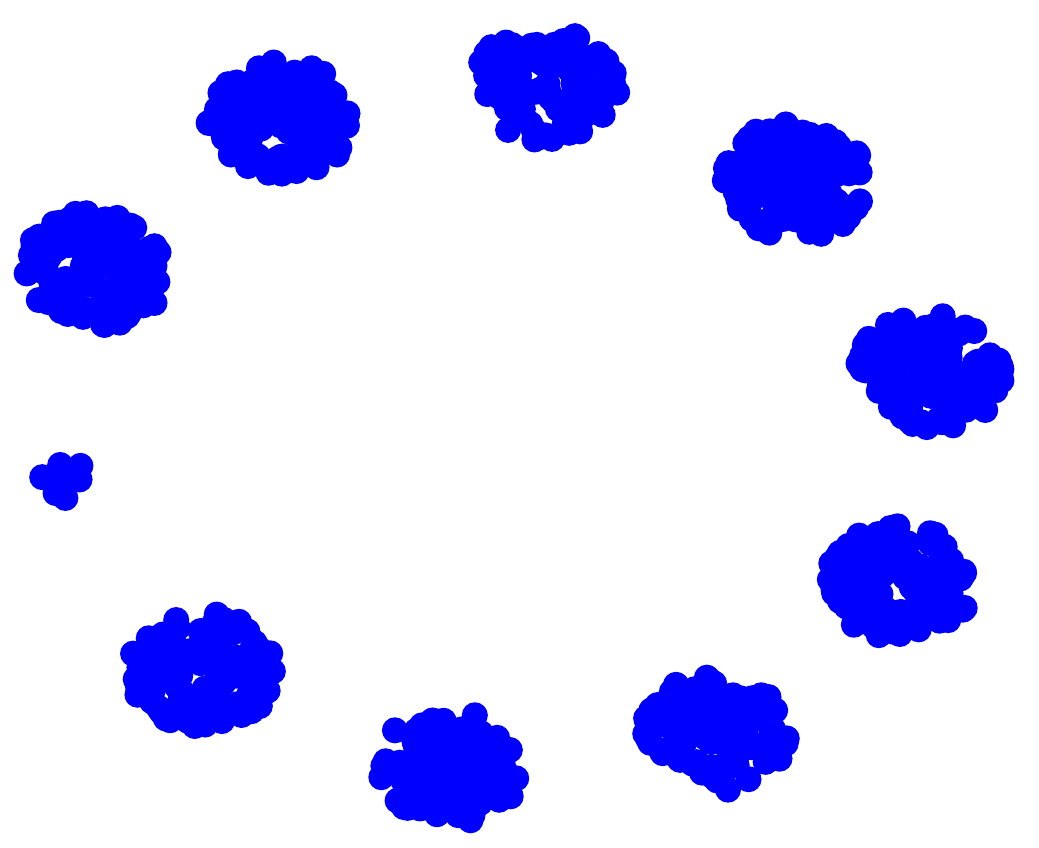}
\includegraphics[width=.18\linewidth, keepaspectratio]{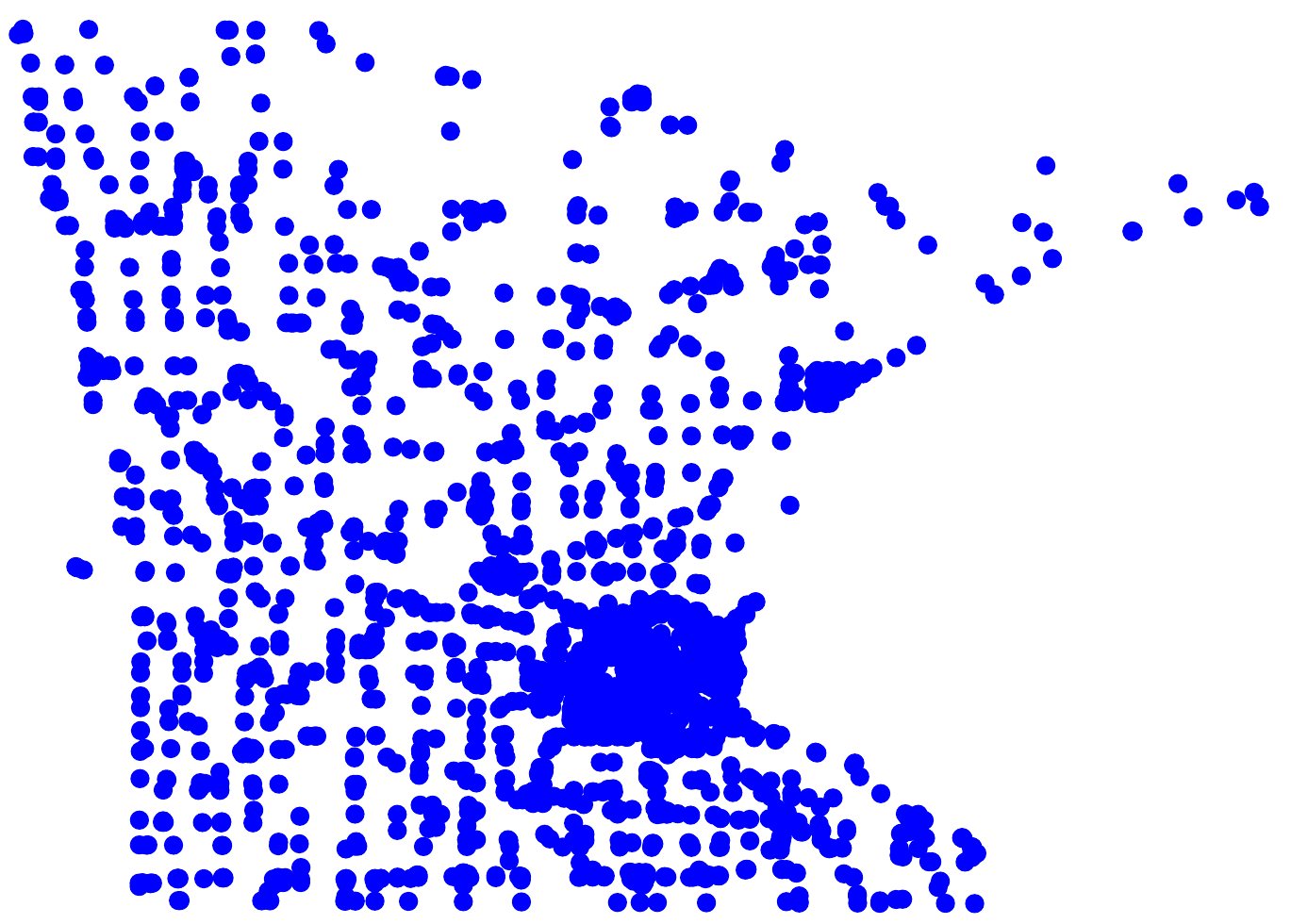}
\includegraphics[width=.18\linewidth, keepaspectratio]{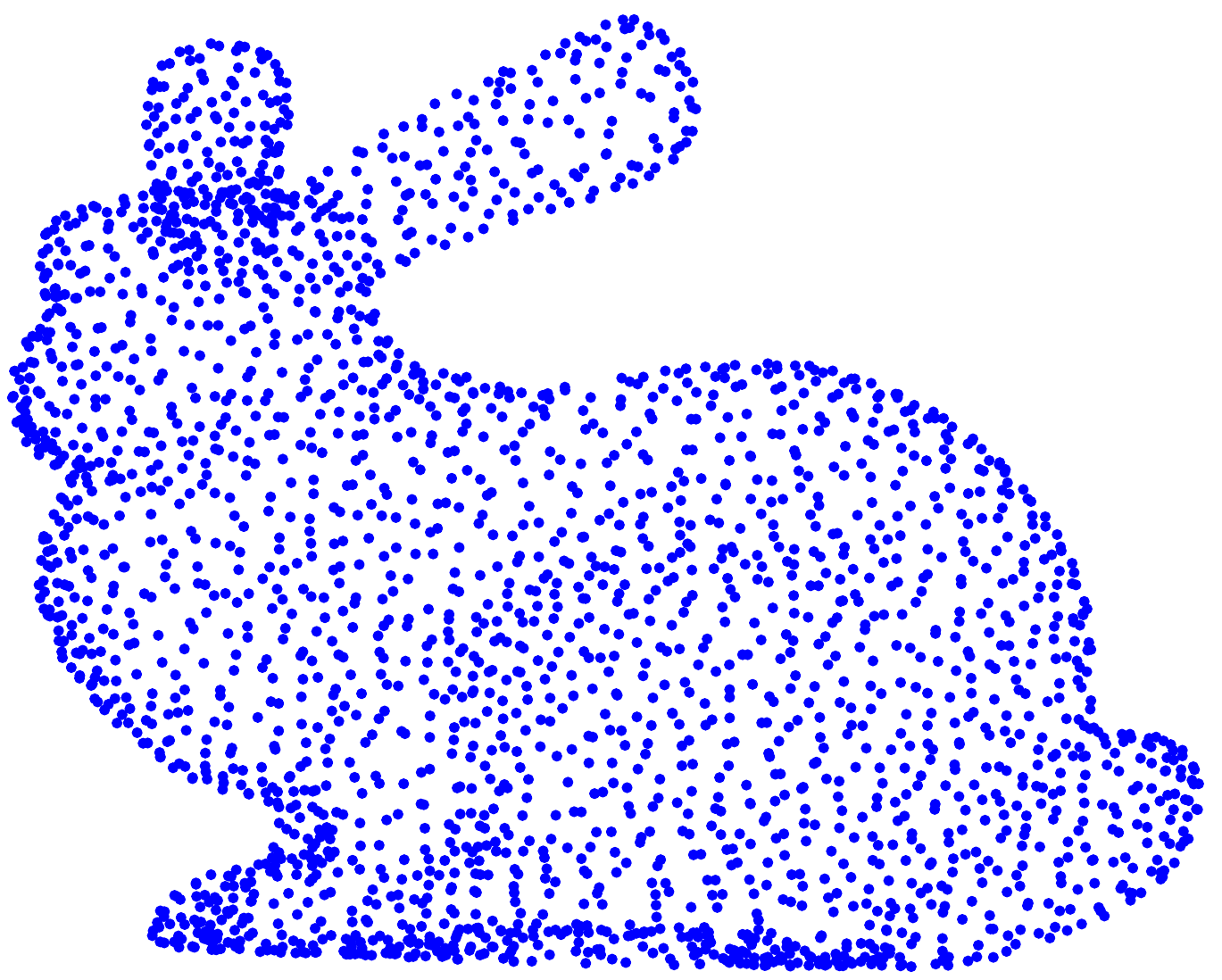}
\includegraphics[width=.18\linewidth, keepaspectratio]{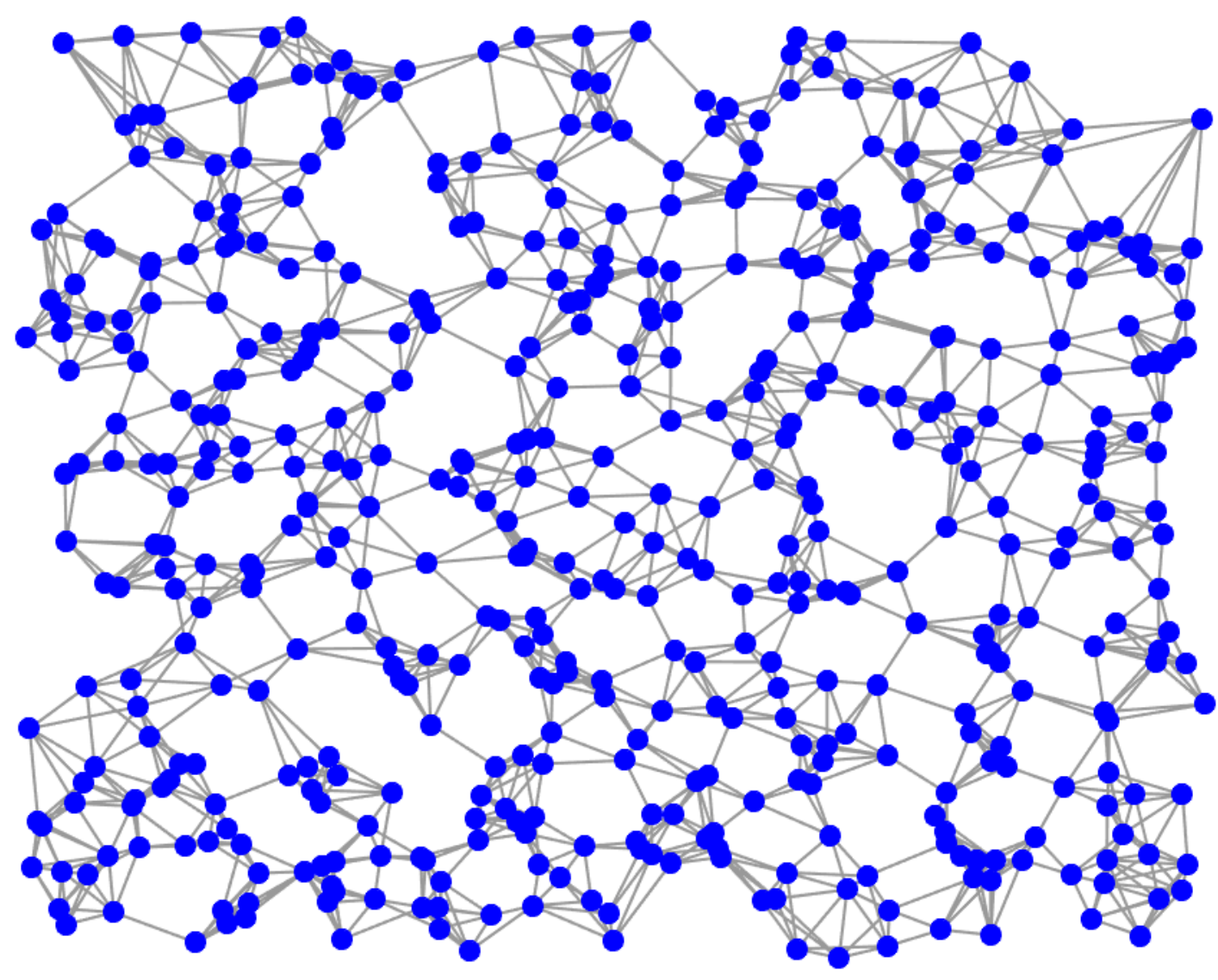}
\includegraphics[width=.18\linewidth, keepaspectratio]{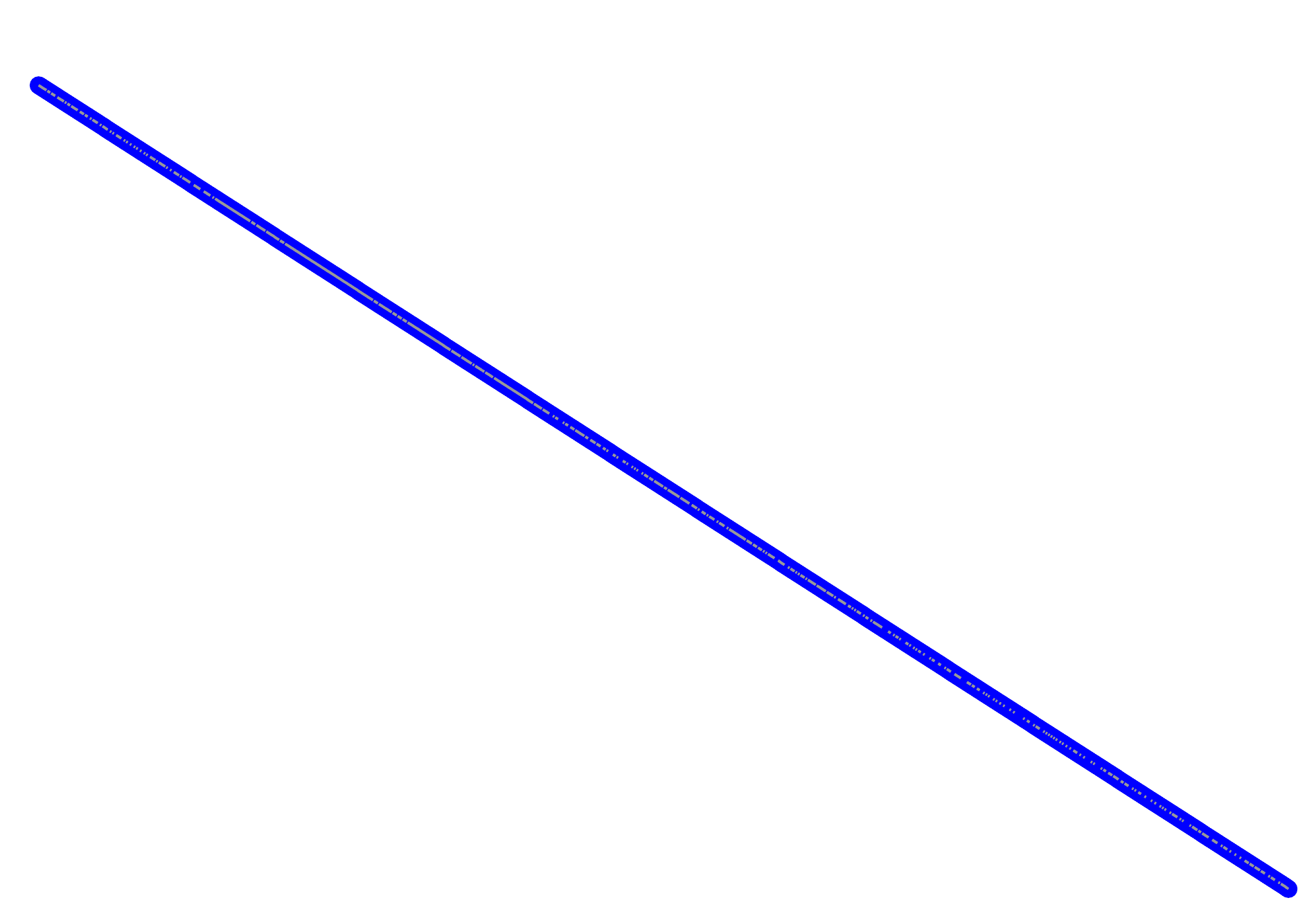}
\caption{\footnotesize\label{fig:graphs} The five different graphs used in the simulations: the figures are community graph, Minnesota graph, bunny graph, sensor graph and ring graph from left to right.}
\vspace{-0.1in}
\end{figure}

In this section, we  empirically illustrate the  theoretical findings and demonstrate the performance of the proposed algorithms.

\paragraph{Graphs} We  use almost the same set of graphs in \cite{puy2018random}.  It consists of {five} different types of graphs, and is  all available in the GSP toolbox \cite{perraudin2014} which is  presented in Fig.~\ref{fig:graphs}: a) different community-type graphs of size $V = 1000$; b)Minnesota road graph of size $V = 2642$; c) Stanford bunny graph of size $V = 2503$; d) an unweighted sensor graph of size $V = 1000$ \footnote{We replace the binary tree graph in \cite{puy2018random} with the sensor graph, since the latter is more relevant to applications considered in our paper.}; e) a path graph of size $V = 1000$. Each node 
is connected to its left and right neighbours except for the
boundary nodes which only have one neighbour.

\paragraph{Heat diffusion processes} We use the normalized Laplacian $L$ and consider the continuous heat diffusion processes over the graphs in all experiments with the evolution operator $A=\exp(-\Delta t L)$. Note that in \cite{puy2018random}, the combinatorial Laplacian was used, so we also run the experiments for $T=1$ for comparison  with static case, instead of using the numerical results in \cite{puy2018random}. The  specific parameters are summarized in Table~\ref{tab:table-diffusion}.

\begin{table}[th]
\footnotesize
\centering
\begin{tabular}{|l|l|l|l|l|l|l|}
\hline
Type       & Community & Minesota & Bunny & Sensor & Path & Real images \\ \hline
$\Delta t$ &  4         &    30      &    30   &      10  &    4  &      1       \\ \hline
$T$ &   20        &     20     &   20    &      10  &    20  &    6         \\ \hline
\end{tabular}
\caption{\label{tab:table-diffusion} Parameters in heat diffusion processes }
\vspace{-0.2in}
\end{table}

\paragraph{Space-time sampling} All space-time 
samples are taken  with replacement as stated in   Theorem~\ref{thm:embedding_all}.  The total number of space-time samples is denoted by $M$ and $m_t$ for $t=0,\cdots,T-1$ are set to be equal for regime 2. 
 
\subsection{Effect of the space-time sampling distributions on $M$}
First we show how $\mathbf{p}^{(j)}$ on the proposed   sampling regimes
affect the required total number of   samples $M$  to ensure   stable embedding \eqref{eqn:frame}. 
 We compute the lower embedding constant 
 \begin{equation*}
 \begin{aligned}
\small
\underline{\delta}_{k,\mathbf{p}^{(j)}} :&=   \min_{x\in \mathrm{span}(U_k),\|x\|=1} \| P_{\Omega^{(j)}}^{-\frac{1}{2}}W_j^{\frac{1}{2}}S_j\pi_{ A,T}(\mathbf{x})\|_2^2\\
&= \lambda_{\text{min}}(U_k^\top\pi_{ A,T}^\top S_j^\top W^{\frac{1}{2}}P_{\Omega^{(j)}}^{-1}W^{\frac{1}{2}}S_j\pi_{ A,T}U_k)
 \vspace{-1mm}
\end{aligned}
 \end{equation*}
for {different  $M$}. 
We compute $\underline{\delta}_{k,\mathbf{p}^{(j)}}$ for 250 independent draws of the matrices $S_j$s.

\subsubsection{Using community graphs}

Inspired by \cite{puy2018random}, we use  three types of community graphs\footnote{The community graphs $C_1,C_2,C_3$ used in our paper correspond to $C_1,C_3,C_5$ in \cite{puy2018random}}, denoted by $C_1, C_2, C_3$,  to study the effect of the size of the communities  on sampling distributions.   All  these graphs have  $10$ communities with  $9$ of them of approximately equal size and  the last community being reduced size.  Specifically: (i) the graphs of type $C_1$ have $10$ communities of size $100$; (ii) the graphs of type $C_2$ have $1$ community of size $25$, $8$ communities of size $108$, and $1$ community of size $111$;
(iii) the graphs of type $C_3$ have $1$ community of size $13$, $8$ communities of size $109$, and $1$ community of size $115$.

\begin{figure}[th]
\centering
\begin{minipage}{.3\linewidth} \centering \small \hspace{2mm} $\underline{\delta}_{10,\pi^{(1)}}$  \end{minipage}
\begin{minipage}{.3\linewidth} \centering \small \hspace{2mm}  $\underline{\delta}_{10,\pi^{(2)}}$ \end{minipage}
\begin{minipage}{.3\linewidth} \centering \small \hspace{2mm}  $\underline{\delta}_{10,\pi^{(3)}}$  \end{minipage}\\
\includegraphics[width=.3\linewidth, height=0.25\linewidth
]{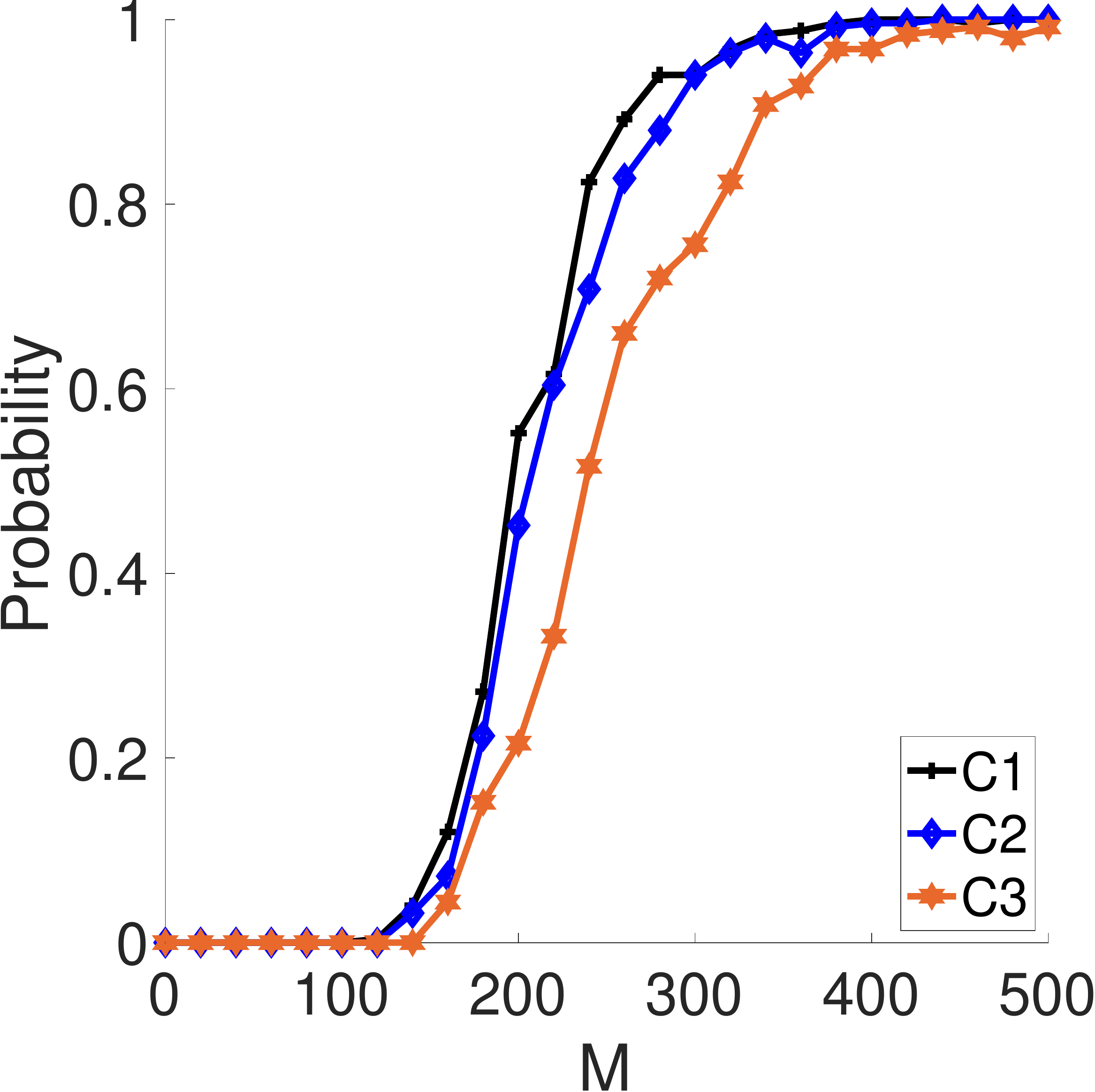}
\includegraphics[width=.3\linewidth, height=0.25\linewidth
]{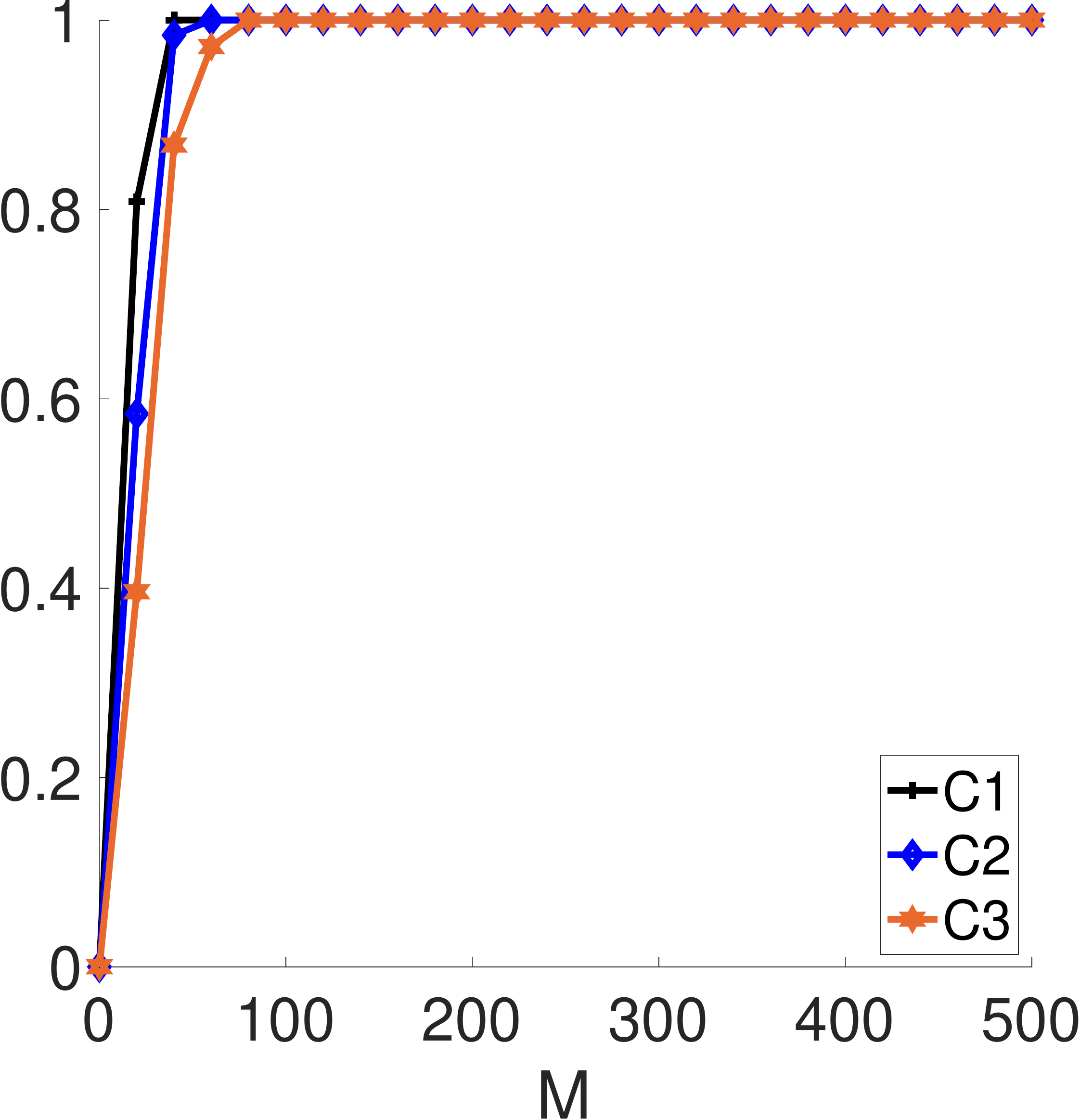}
\includegraphics[width=.3\linewidth, height=0.25\linewidth
]{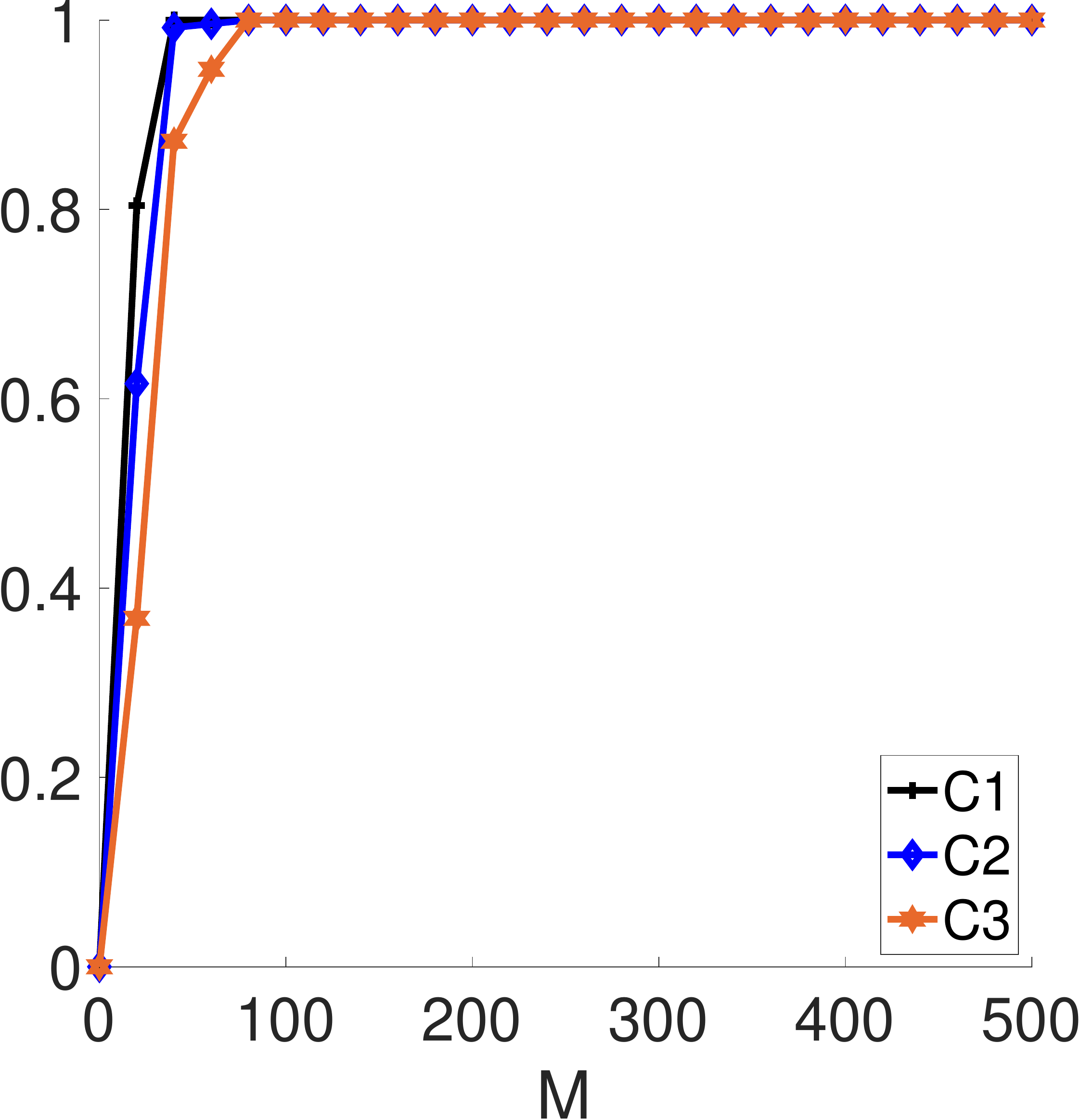}\\
\begin{minipage}{.3\linewidth} \centering \small \hspace{2mm} $\underline{\delta}_{10, \mathbf{p}_{\mathrm{opt}}^{(1)}}$  \end{minipage}
\begin{minipage}{.3\linewidth} \centering \small \hspace{2mm}    $\underline{\delta}_{10, \mathbf{p}_{\mathrm{opt}}^{(2)}}$ \end{minipage}
\begin{minipage}{.3\linewidth} \centering \small \hspace{2mm}   $\underline{\delta}_{10,  \mathbf{p}_{\mathrm{opt}}^{(3)}}$  \end{minipage}\\
\includegraphics[width=.3\linewidth, height=0.25\linewidth
]{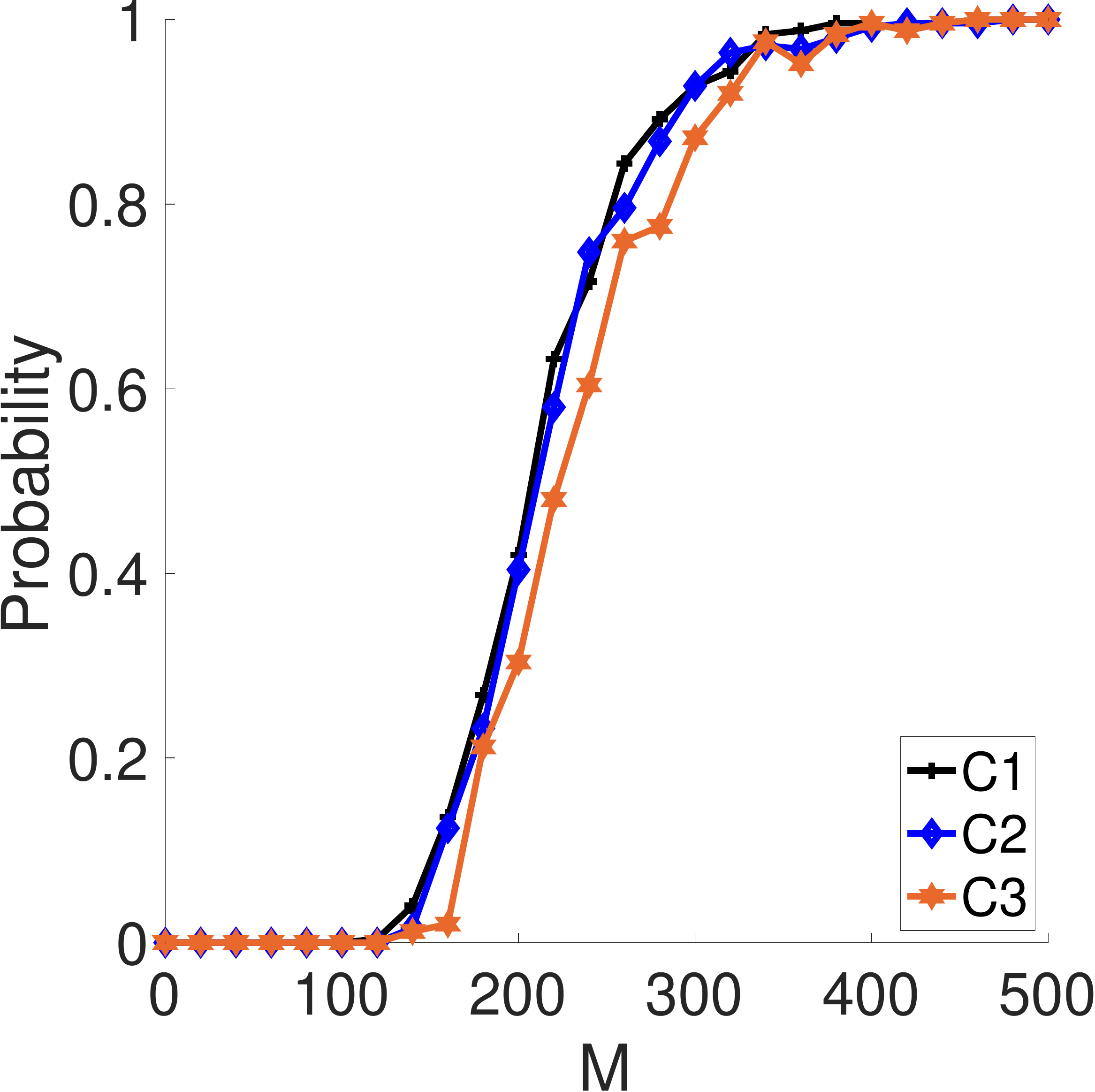}
\includegraphics[width=.3\linewidth, height=0.25\linewidth
]{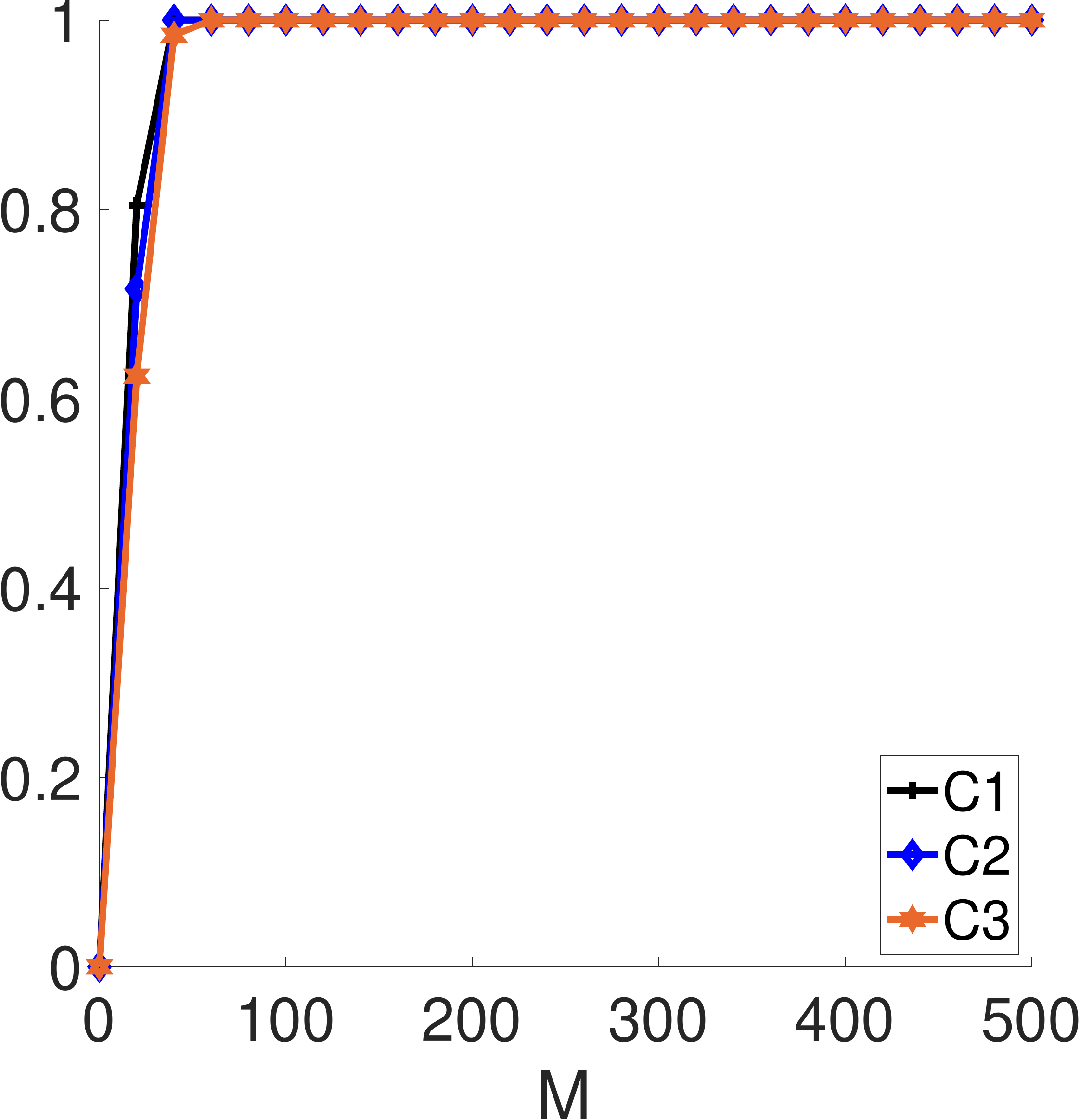}
\includegraphics[width=.3\linewidth, height=0.25\linewidth
]{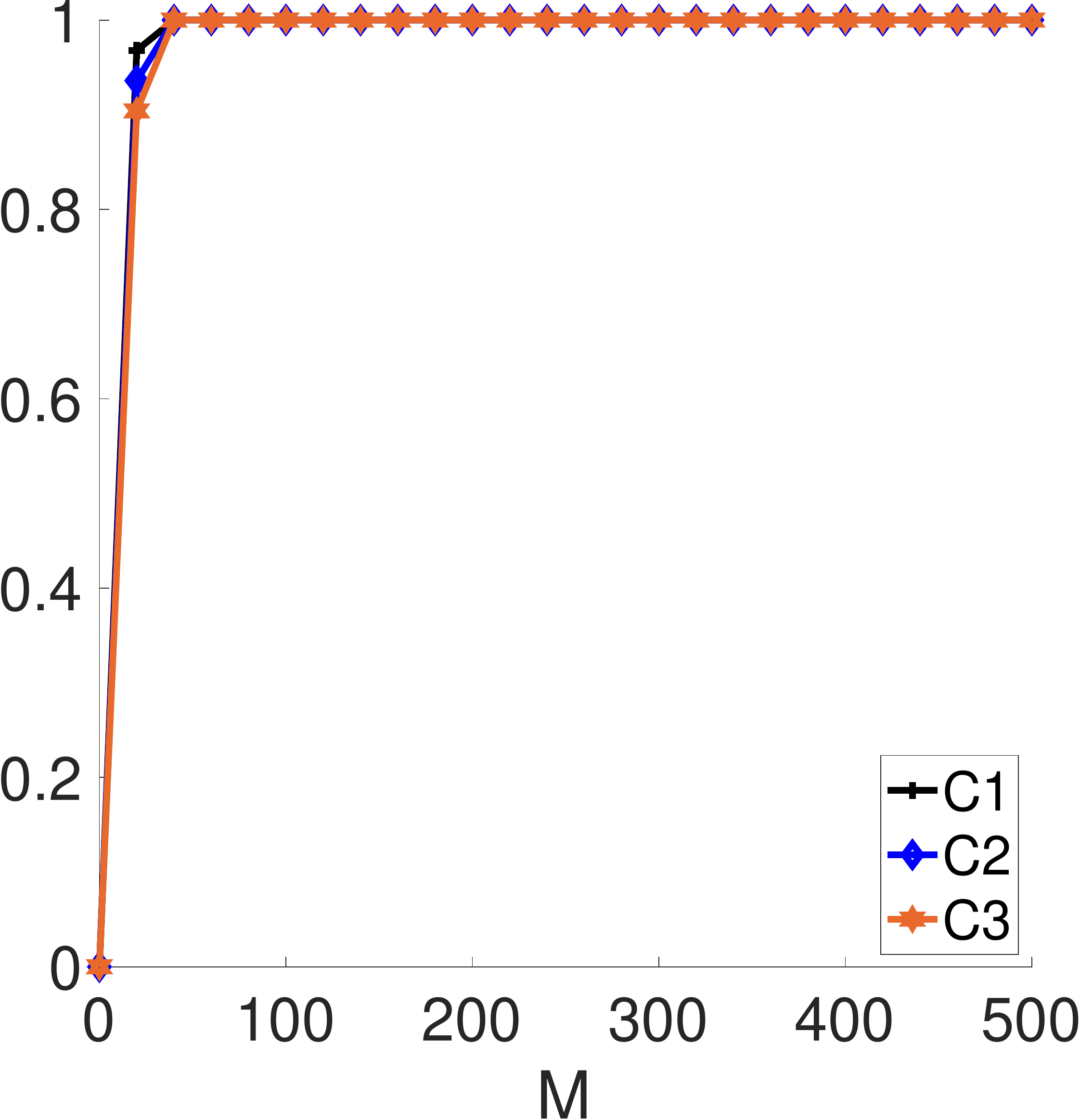}\\
\caption{\footnotesize\label{fig:eig_community} Probability that $\underline{\delta}_{10}$ is great than $0.005$ as a function of $M$ for $5$ different types of community graphs: $C_1$ in black plus, $C_2$ in blue diamond, $C_3$ in orange hexagon. We choose $ M=20:20:500$. 
}
\vspace{-0.2in}
\end{figure}
 
\begin{figure}[th]
\centering
\begin{minipage}{.3\linewidth} \centering \small \hspace{2mm} Regime 1  \end{minipage}
\begin{minipage}{.3\linewidth} \centering \small \hspace{2mm}  Regime 2  \end{minipage}
\begin{minipage}{.3\linewidth} \centering \small \hspace{2mm}  Regime 3  \end{minipage}\\
\includegraphics[width=.30\linewidth,height=0.25\linewidth
]{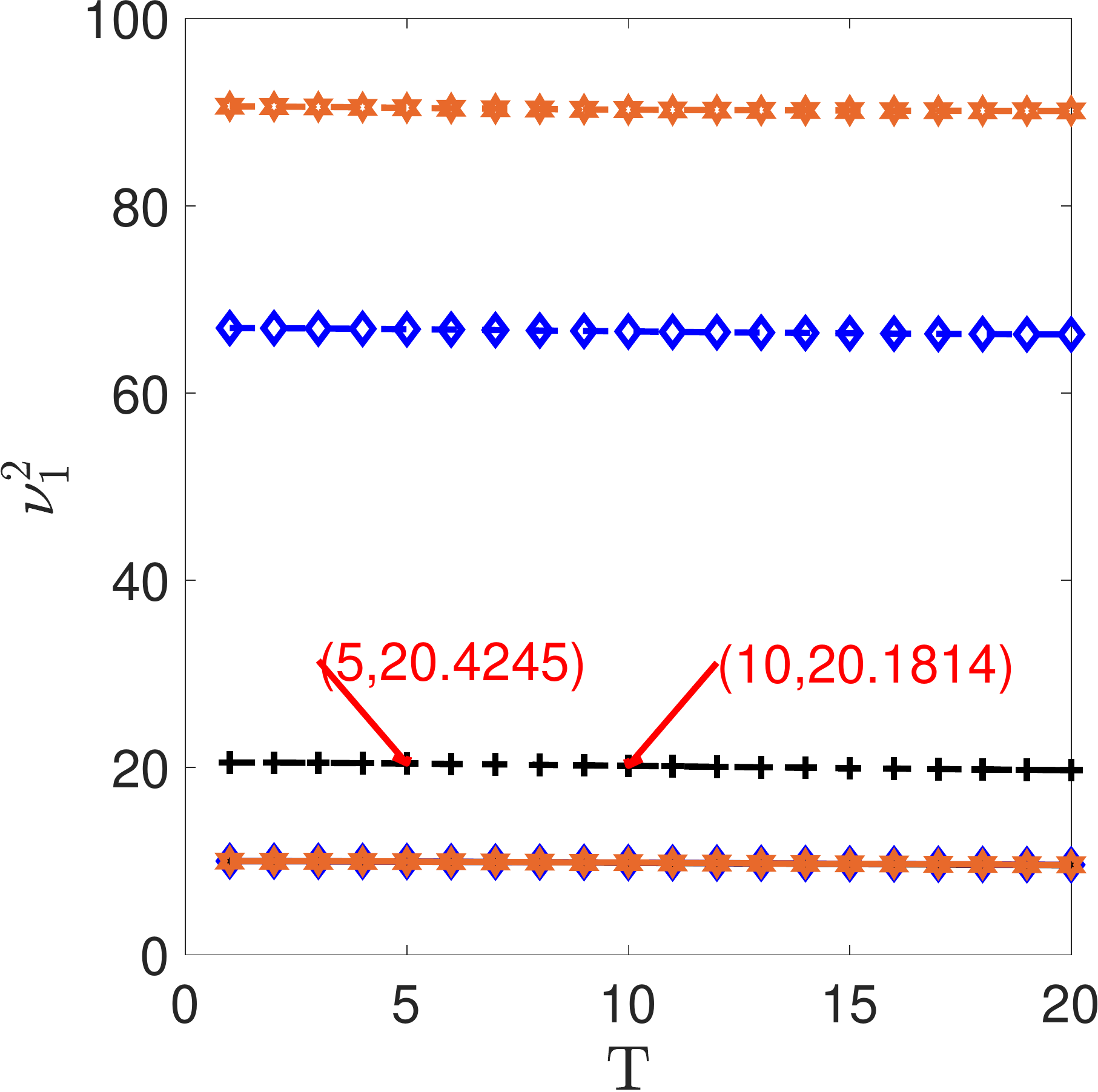}
\includegraphics[width=.30\linewidth, height=0.25\linewidth
]{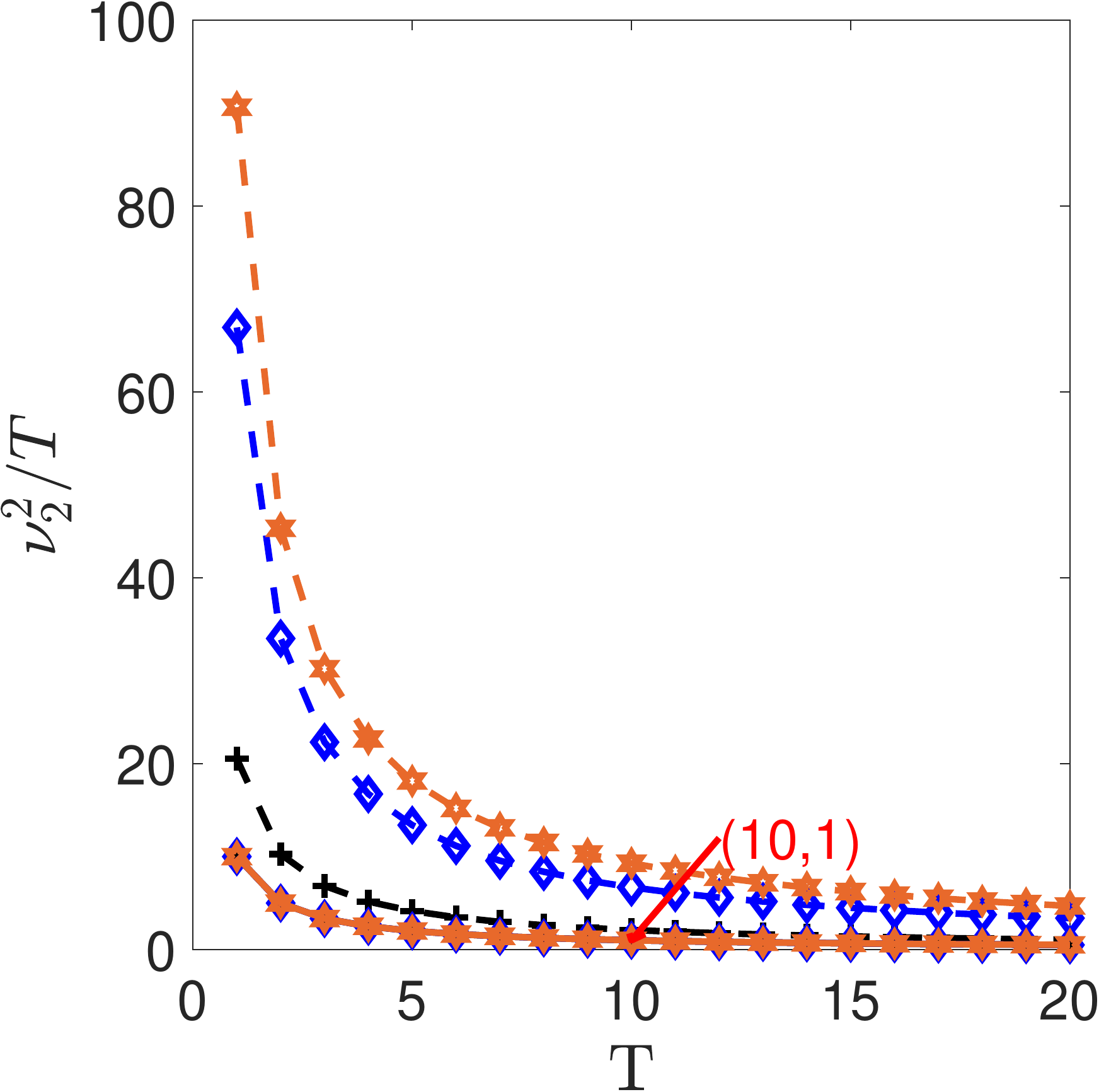}
\includegraphics[width=.30\linewidth,height=0.25\linewidth
]{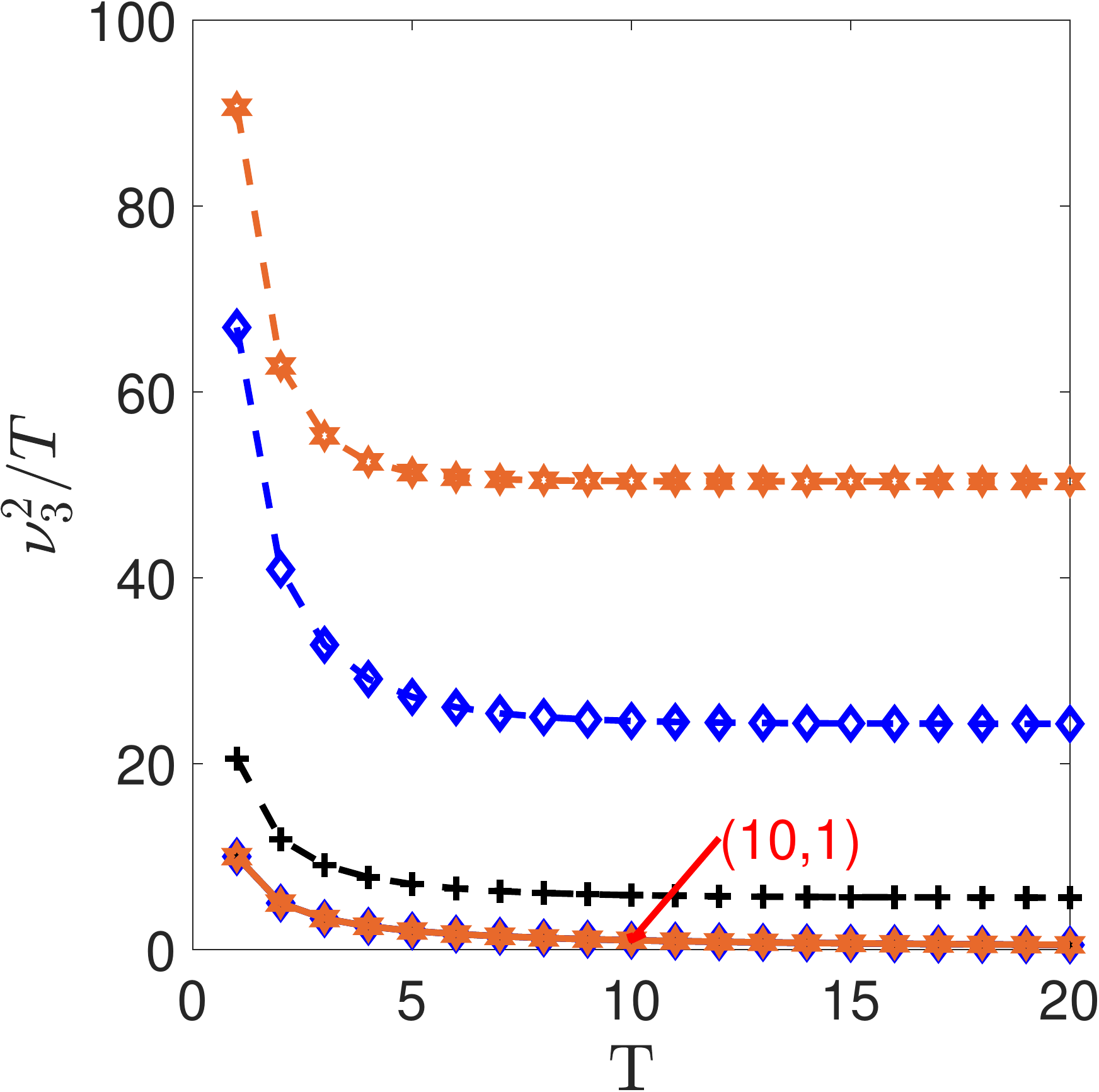}\\
\caption{\footnotesize\label{fig:coherence_vs_Time} The relation between spectral graph weighted coherence of order $(k,T)$   and $T$ with $k=10$ and $T$ ranging from 1 to 20 for community graphs.
 We present results for  $C_1$ in black plus, $C_2$ in blue diamond, and $C_3$ in orange hexagon; uniform sampling distribution in solid curve and optimal sampling distribution in dashed curve.
For regime 1 ,2 and 3, we present  the values of $(\nu_{1,\bf{p}^{(1)}}^{k,T})^2$, $\sum_{t=0}^{T-1}(\nu_{\bf{p}_t^{(2)}})^2/T$ and $(\nu_{3,\bf{p}^{(3)}}^{k,T})^2/T$, respectively.  All figures show that the value are decreasing as $T$ increases, with very slow rate in regime 1 (see two representative data pairs).  For the case of optimal sampling distribution, the curves in regime 2 and regime 3 are in fact $10/T$ (see the representative data pair (10,1)). }
\vspace{-0.3in}
\end{figure}

  We present in Fig.~\ref{fig:eig_community} the probabilities that $\underline{\delta}_{10, \mathbf{p}^{(j)}}$ is greater than 0.005, i.e, $\mathrm{Prob}(\underline{\delta}_{10,\mathbf{p}^{(j)}} \geq 0.005)$, estimated over $250$ trials. 
  Let $M_{i,\mathbf{p}^{(j)}}^*$ be the required number of space-time samples  to ensure $\mathrm{Prob}(\underline{\delta}_{10,\mathbf{p}^{(j)}} \geq 0.005) = 0.9$ for $i$-th graph-type and sampling distribution $\mathbf{p}^{(j)}$ of  each sampling regime. Theorem~\ref{thm:embedding_all} predicts that $M_{i, \mathbf{p}^{(j)}}^*$ scales linearly with $T(\mathbf{\nu}_{1,\mathbf{p}^{(1)}}^{k,T})^2$, $\sum_{t=0}^{T-1}(\mathbf{\nu}_{2,\mathbf{p}^{(2)}}^{k,T}(t))^2$ and $(\mathbf{\nu}_{3,\mathbf{p}^{(3)}}^{k,T})^2$  for $j=1,2,3$ respectively of community graph $i$. 
For the uniform distribution $\mathbf{p}:=\pi$, the first figure from the top panel in Fig.~\ref{fig:eig_community} shows that $M_{1, \mathbf{p}^{(j)}}^*\geq M_{2, \mathbf{p}^{(j)}}^* \geq M_{3, \mathbf{p}^{(j)}}^*$. This matches their ordering of $\mathbf{\nu}_{j,\mathbf{p}^{(j)}}^{k,T}$ presented in  Fig.~\ref{fig:coherence_vs_Time}, and is in accordance with Theorem \ref{thm:embedding_all}. For the optimal sampling distribution $\mathbf{p}_{\mathrm{opt}}$, we have $\mathbf{\nu}_{j,\mathbf{p}_{\textbf{opt}}^{(j)}}^{k,T}$ approximately the same for all types of graphs and $j=1,2,3$.  Therefore $M_{i,\mathbf{p}^{(j)}}^*$ must be nearly identical for all graph-types, as observed in the second panel of Fig.~\ref{fig:eig_community}. For each type of graph, in Proposition \ref{gcoherence}, we showed that for all distributions, 
\vspace{-1.5mm}
\[\sum_{t=0}^{T-1}(\mathbf{\nu}_{2,\mathbf{p}_t^{(2)}}^{k,T})^2\leq (\mathbf{\nu}_{3,\mathbf{p}^{(3)}}^{k,T})^2 \leq T(\mathbf{\nu}_{1,\mathbf{p}^{(1)}}^{k,T})^2,\]  this matches the numerical results in  Fig.~\ref{fig:coherence_vs_Time}.

  In addition, we also computed the spectral graph weighted coherence for $C_1$, $C_2$, $C_3$ with respect to different sampling distributions and  total sampling time $T$ in Fig.~\ref{fig:coherence_vs_Time}. The results show that the number of total space-time samples required in regime 2 and regime 3 are the same as $T$ varies and thus one can expect use fewer spatial samples as $T$ increase with approximately 1-1 trade-off. While in regime 1, the temporal information only brings a tiny ratio of reduction in spatial samples.

\subsubsection{The results for other graphs}  We also perform the  same experiments for  four other graphs that have wide applications:   Minnesota, bunny, path, and sensor graphs (cf.~\cite{perraudin2014}). For the first three graphs, the experiments are performed for bandlimited signals with bandwidth $100$.  For the sensor graph, we set the bandwidth  $50$. Here our goal is to test the performance with respect to signals with different bandwidth.   For the path graph, some eigenvalues have  multiplicities larger than 1. We choose the bandwidth to ensure that $\lambda_k <\lambda_{k+1}$ as required in our assumptions. The results for $\text{Prob}(\underline{\delta}_{k,\mathbf{p}^{(i)}}>0.005)$ \textit{v.s.} $M$ are summarized in Fig.~\ref{fig:eig_other}. 

\subsection{Comparisons}

\subsubsection{Optimal distribution versus uniform distribution}  Compared with the uniform distribution $\pi^{(j)}$, the optimal distribution $\mathbf{p}_{\textbf{opt}}^{(j)}$ for each regime takes the graph topology and temporal information into account so that one can use fewer space-time samples yet yield  robust reconstruction. In the case $k=100$,  the advantage of using $\mathbf{p}_{\textbf{opt}}^{(j)}$  is most  obvious for the Minnesota and bunny graphs in regime 3. For the Minnesota graph, we reach only a probability of $0.008$ at $M = 800$ with $\pi^{(j)}$, whereas $M = 400$  are sufficient to reach a probability $1$ with $\mathbf{p}_{\textbf{opt}}^{(j)}$. $\pi^{(j)}$ exhibits a poorer performance for the bunny graph:  we reach only a probability of $0.016$ at $M = 1400$ with $\pi^{(j)}$, whereas $M = 400$   are sufficient to reach a probability 1 with $\mathbf{p}_{\textbf{opt}}^{(j)}$. 

Uniform sampling is not working well in regime 3 for the Minnesota and  bunny graphs,  since there exist few eigenmodes of the bandlimited diffusion field  whose energy is highly concentrated on few  space-time nodes.  In other words, we have   $ \|\widetilde{U}_{100,T}\delta_i\| \approx 1$ for some space-time node $i$. In this case, the corresponding $\nu_{j,\mathbf{\pi}^{(3)}}^{k,T}$ is approximately $\sqrt{Tn}$.  In contrast,     $\nu_{3,\mathbf{p}_{\textbf{opt}}^{(3)}}^{k,T}$ is $\sqrt{k}$ for $\mathbf{p}_{\textbf{opt}}^{(j)}$. In regime 2, thanks to the the same number of spatial samples taken at consecutive times, the gap between the uniform sampling and optimal sampling is much smaller.

For more regular graphs like the sensor graph and path graph, as our analysis in Section \ref{graphdiffusion} shows, the optimal distributions are close to the uniform distributions, and therefore they have comparable performances.

\textbf{Dynamical sampling versus static sampling.} The static sampling ($T=1$) results using uniform distributions are presented  in  Table \ref{tab:table-uniformdistribution}.  For path graph and community graph, we observed that the embedding produced by space-time samples chosen according to regime 2 and 3  are as good as the one produced by the same number of spatial samples, we even do better in community graph. For the bunny graph and Minnesota graph,  if we count the average spatial samples at each time instance, then we use much fewer spatial samples than the static setting to achieve a comparable performance. From the perspective of sampling, this means that one can use fewer number of sensors and dynamical sampling provides a much cheaper alternative for data acquisition. We also refer to the captions of Figure \ref{fig:rec_nonoise} and Figure \ref{fig:rec_noisy} for the comparison 
of reconstruction with the static case.

\begin{table}[h]
\footnotesize
\centering
\begin{tabular}{|l|l|l|l|l|l|l|l|l|l|l|l|l|}
\hline
Type   & \multicolumn{3}{l|}{Community} & \multicolumn{3}{l|}{Minnesota} & \multicolumn{3}{l|}{Bunny}          & \multicolumn{3}{l|}{Path} \\ \hline
Regime & 1         & 2        & 3       & 1        & 2        & 3       & 1          & 2          & 3         & 1       & 2      & 3      \\ \hline
$M$    & 400       & 80       & 80      & 3000     & 1600     & 1600    & 3200       & 2400       & 2600      & 3200    & 400    & 400    \\ \hline
$M/T$  & 20        & 4        & 4       & 150      & 80       & 80      & 200        & 120        & 120       & 200     & 20     & 20     \\ \hline
Static & \multicolumn{3}{l|}{130}       & \multicolumn{3}{l|}{450}     & \multicolumn{3}{l|}{$320$} & \multicolumn{3}{l|}{360}  \\ \hline
\end{tabular}
\caption{\label{tab:table-uniformdistribution}For the uniform distributions, the number of samples that required to reach a probability at least 0.95.}
\vspace{-0.2in}
\end{table}

\begin{figure}[th]
\centering
\begin{minipage}{.3\linewidth} \centering \small \hspace{2mm} $\text{Minnesota, } \underline{\delta}_{100,\mathbf{p}^{(1)}}$  \end{minipage}
\begin{minipage}{.3\linewidth} \centering \small \hspace{2mm}  $\text{Minnesota, } \underline{\delta}_{100,\mathbf{p}^{(2)}}$ \end{minipage}
\begin{minipage}{.3\linewidth} \centering \small \hspace{2mm}  $\text{Minnesota, } \underline{\delta}_{100,\mathbf{p}^{(3)}}$  \end{minipage}\\
\includegraphics[width=.3\linewidth, height=0.25\linewidth
]{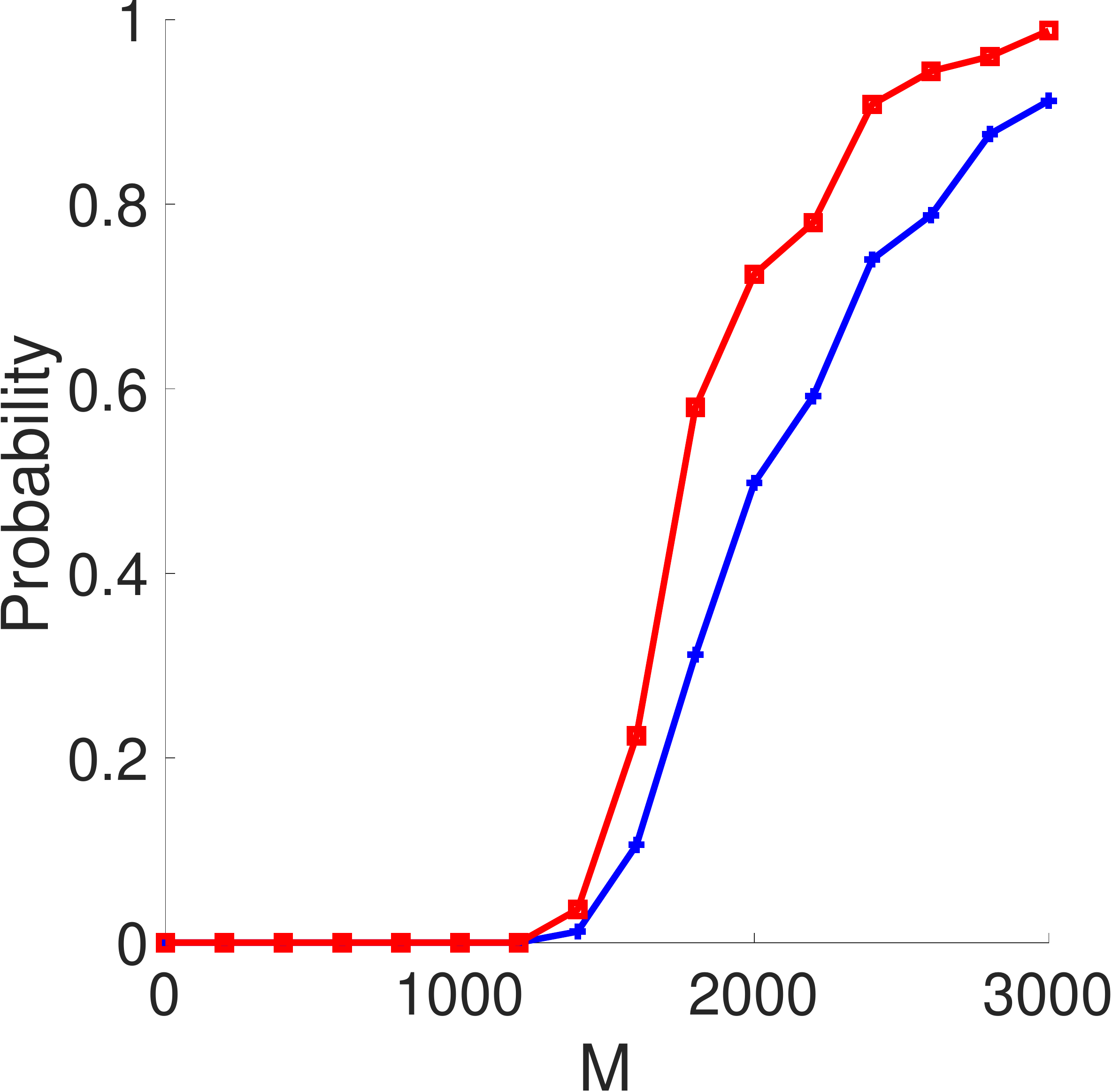}
\includegraphics[width=.3\linewidth, height=0.25\linewidth
]{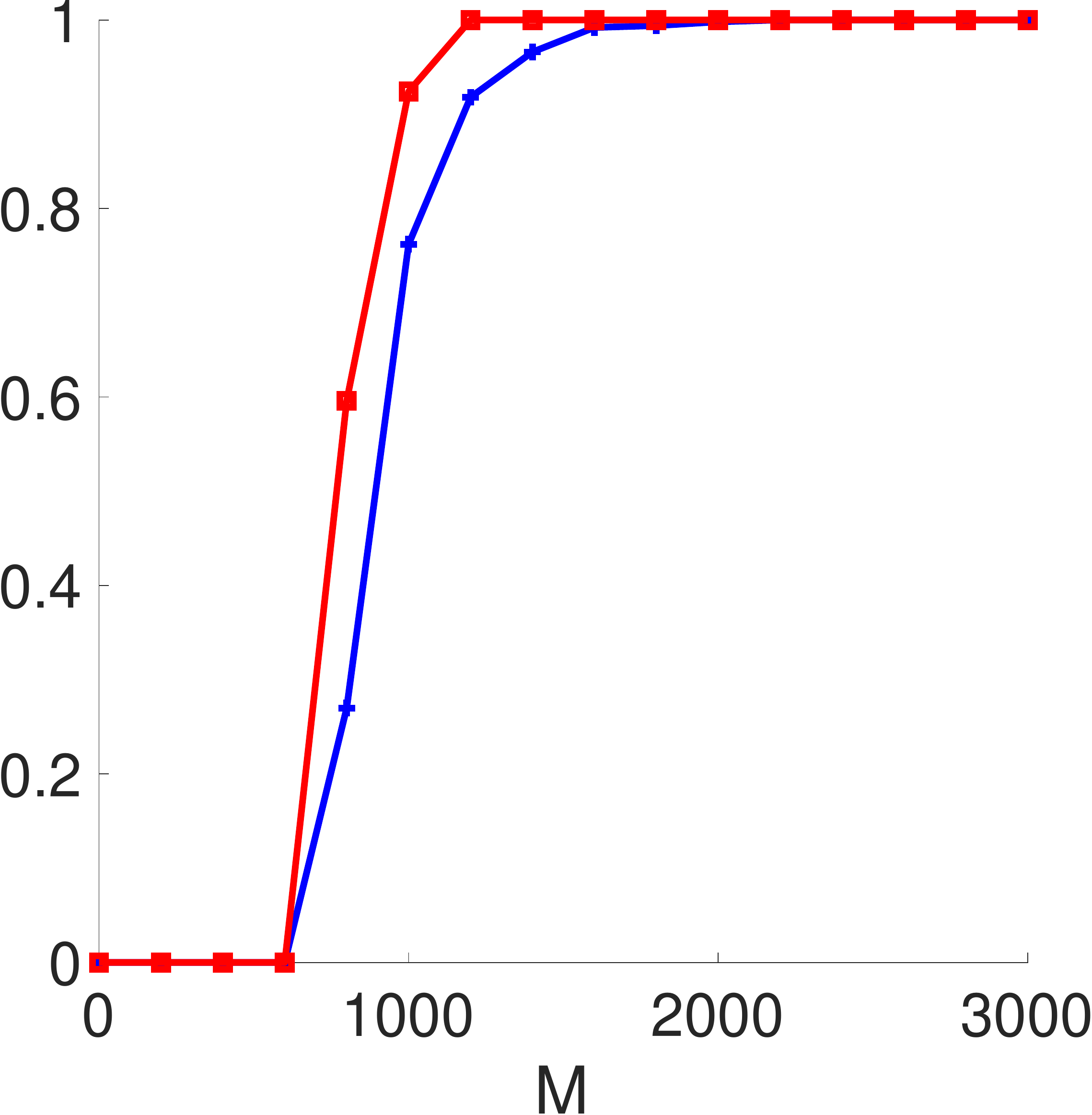}
\includegraphics[width=.3\linewidth, height=0.25\linewidth
]{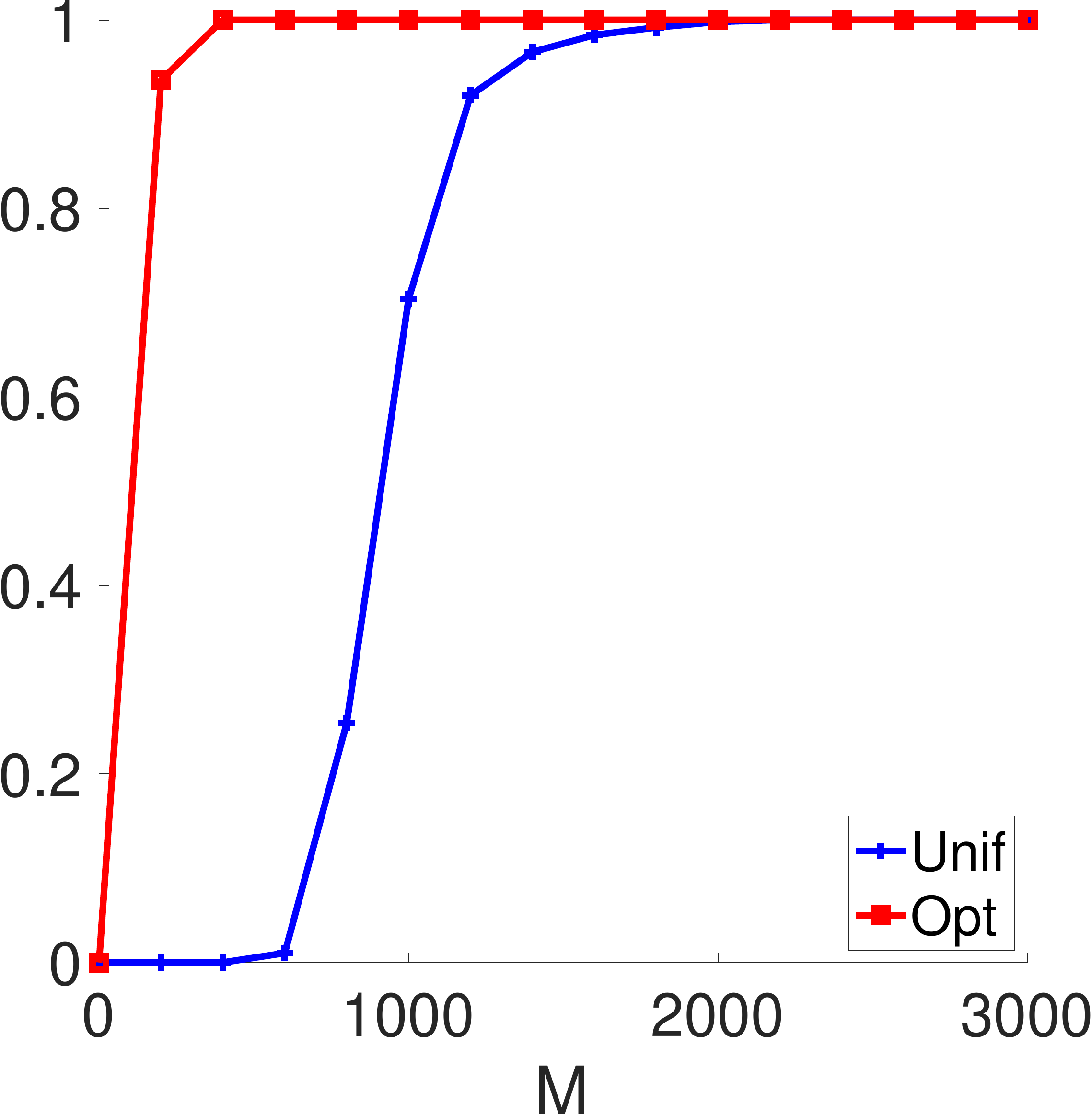}\\
\begin{minipage}{.3\linewidth} \centering \small \hspace{2mm} $\text{Bunny, } \underline{\delta}_{100,\mathbf{p}^{(1)}}$  \end{minipage}
\begin{minipage}{.3\linewidth} \centering \small \hspace{2mm}  $\text{Bunny, } \underline{\delta}_{100,\mathbf{p}^{(2)}}$ \end{minipage}
\begin{minipage}{.3\linewidth} \centering \small \hspace{2mm}   $\text{Bunny, } \underline{\delta}_{100,\mathbf{p}^{(3)}}$  \end{minipage}\\
\includegraphics[width=.3\linewidth, height=0.25\linewidth
]{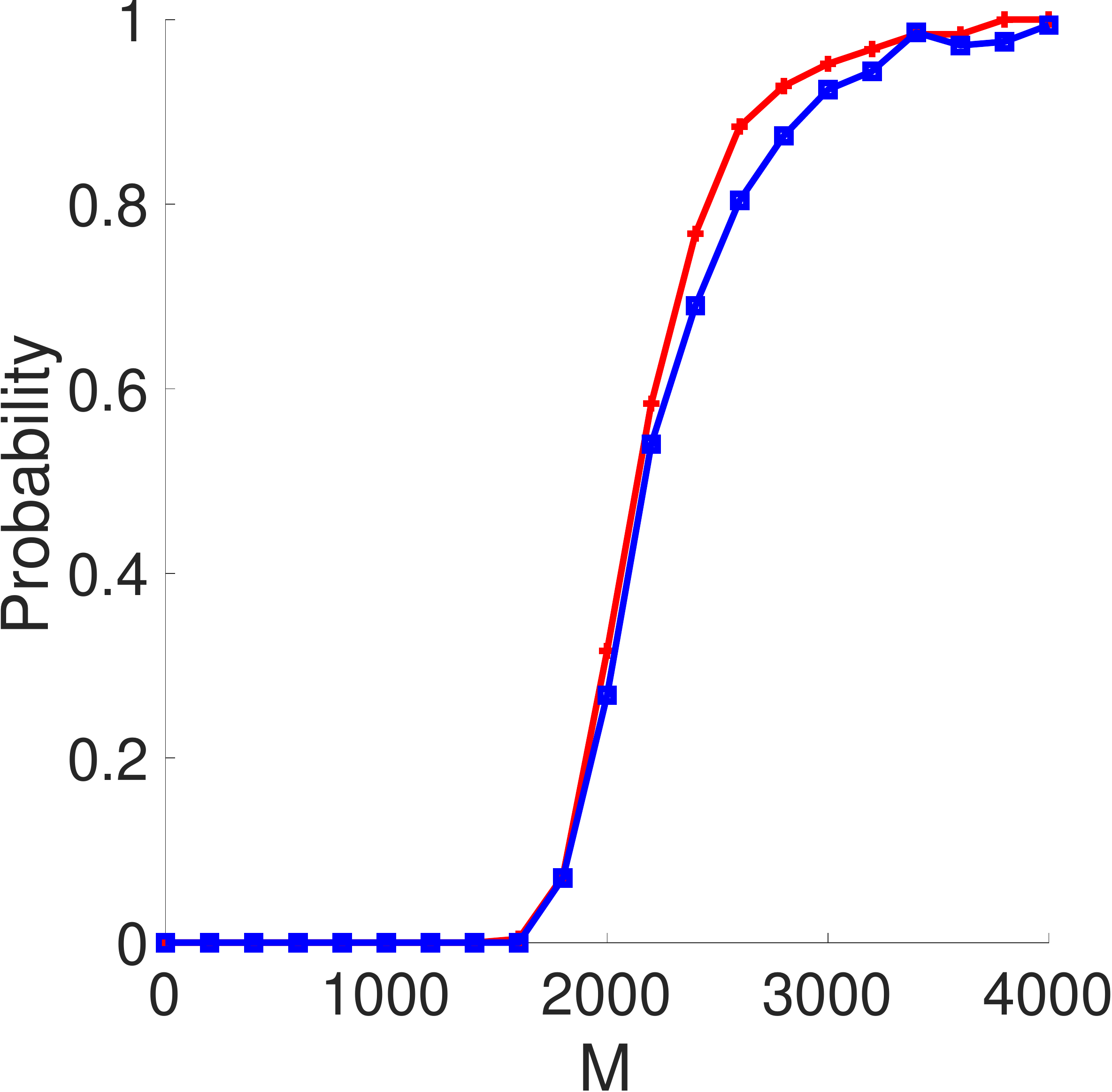}
\includegraphics[width=.3\linewidth, height=0.25\linewidth
]{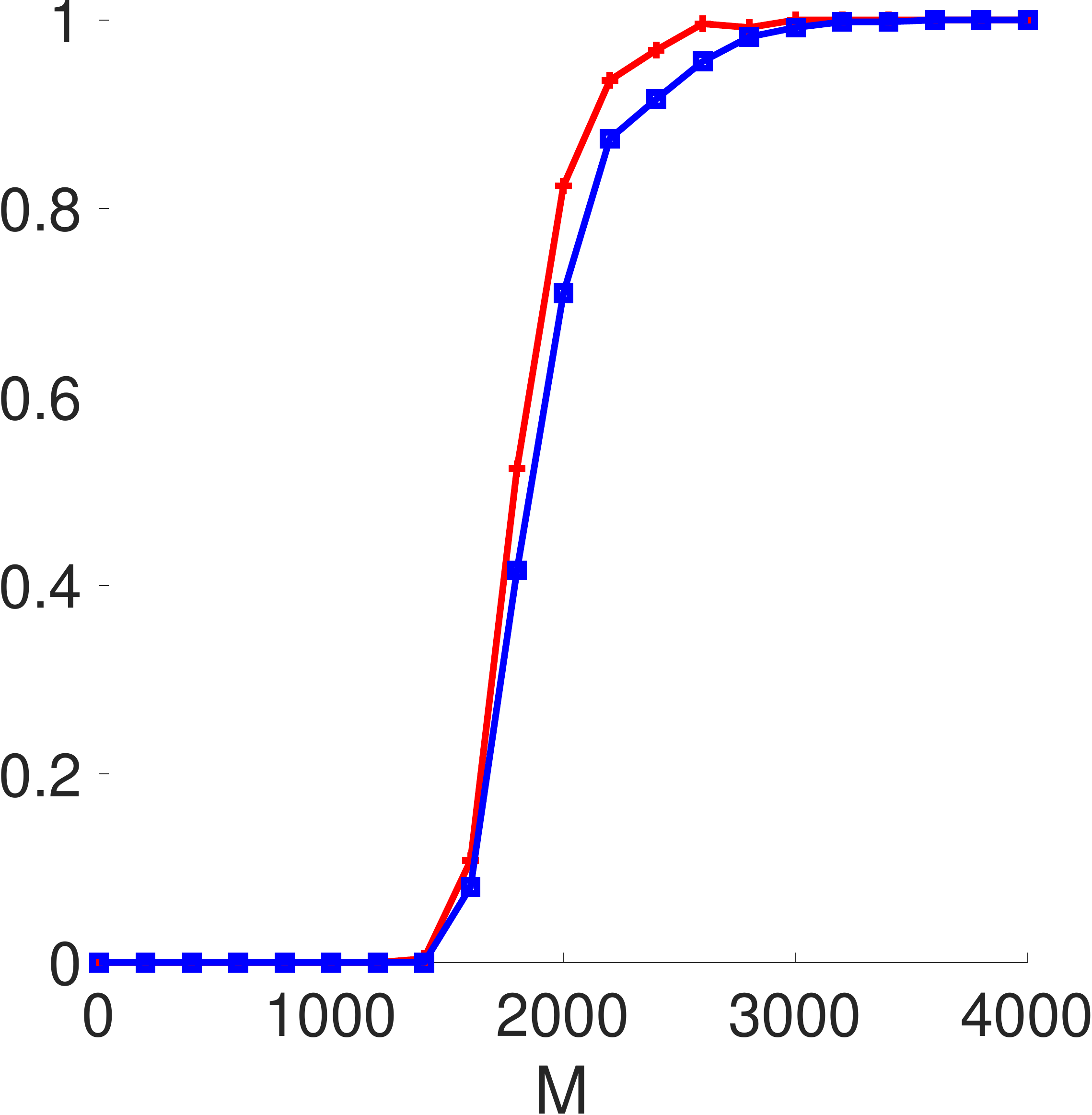}
\includegraphics[width=.3\linewidth, height=0.25\linewidth
]{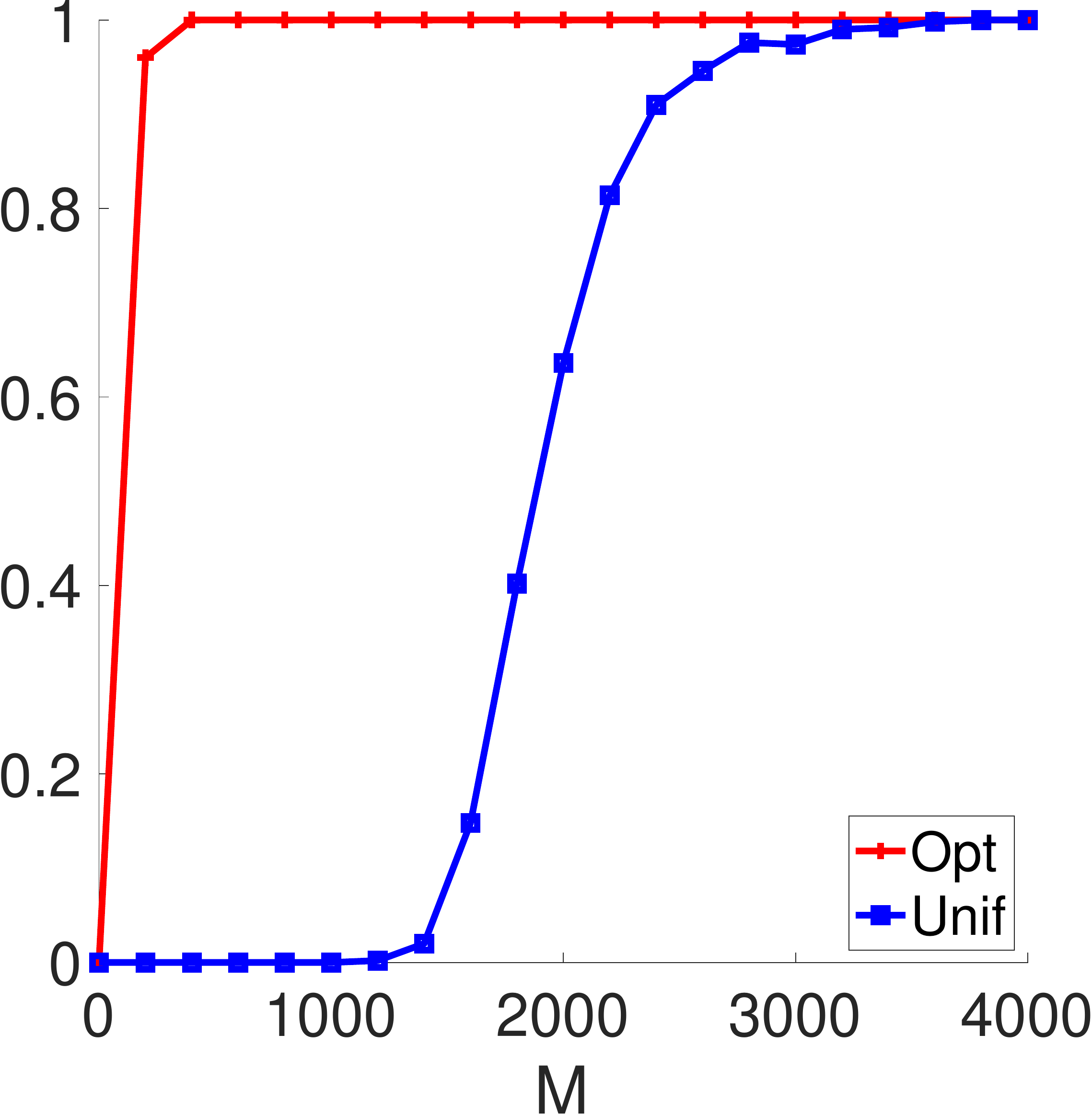}\\

\begin{minipage}{.3\linewidth} \centering \small \hspace{2mm} $\text{Sensor, } \underline{\delta}_{50,\mathbf{p}^{(1)}}$  \end{minipage}
\begin{minipage}{.3\linewidth} \centering \small \hspace{2mm}  $\text{Sensor, } \underline{\delta}_{50,\mathbf{p}^{(2)}}$ \end{minipage}
\begin{minipage}{.3\linewidth} \centering \small \hspace{2mm}   $\text{Sensor, } \underline{\delta}_{50,\mathbf{p}^{(3)}}$  \end{minipage}\\
\includegraphics[width=.3\linewidth, height=0.25\linewidth
]{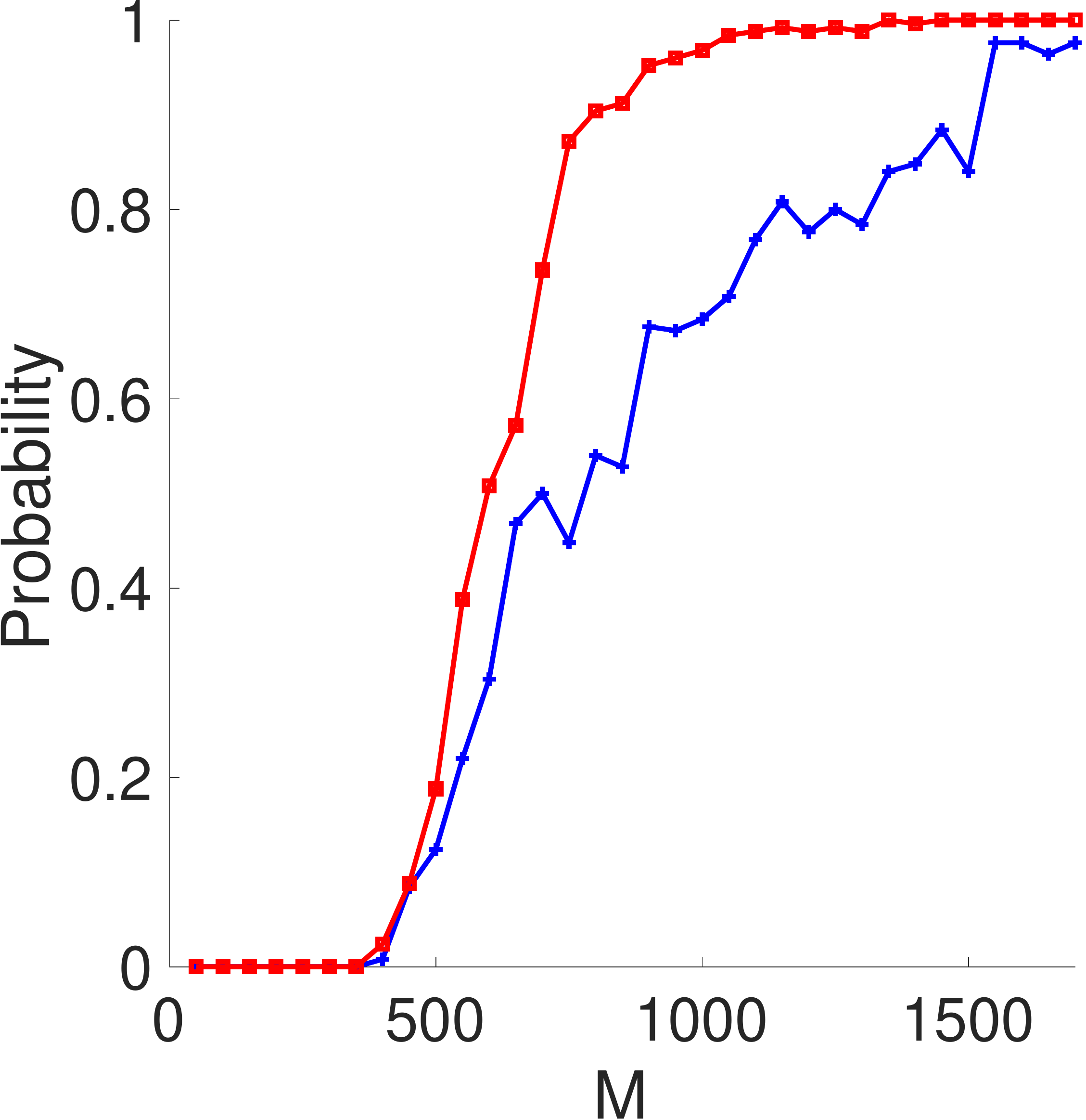}
\includegraphics[width=.3\linewidth, height=0.25\linewidth
]{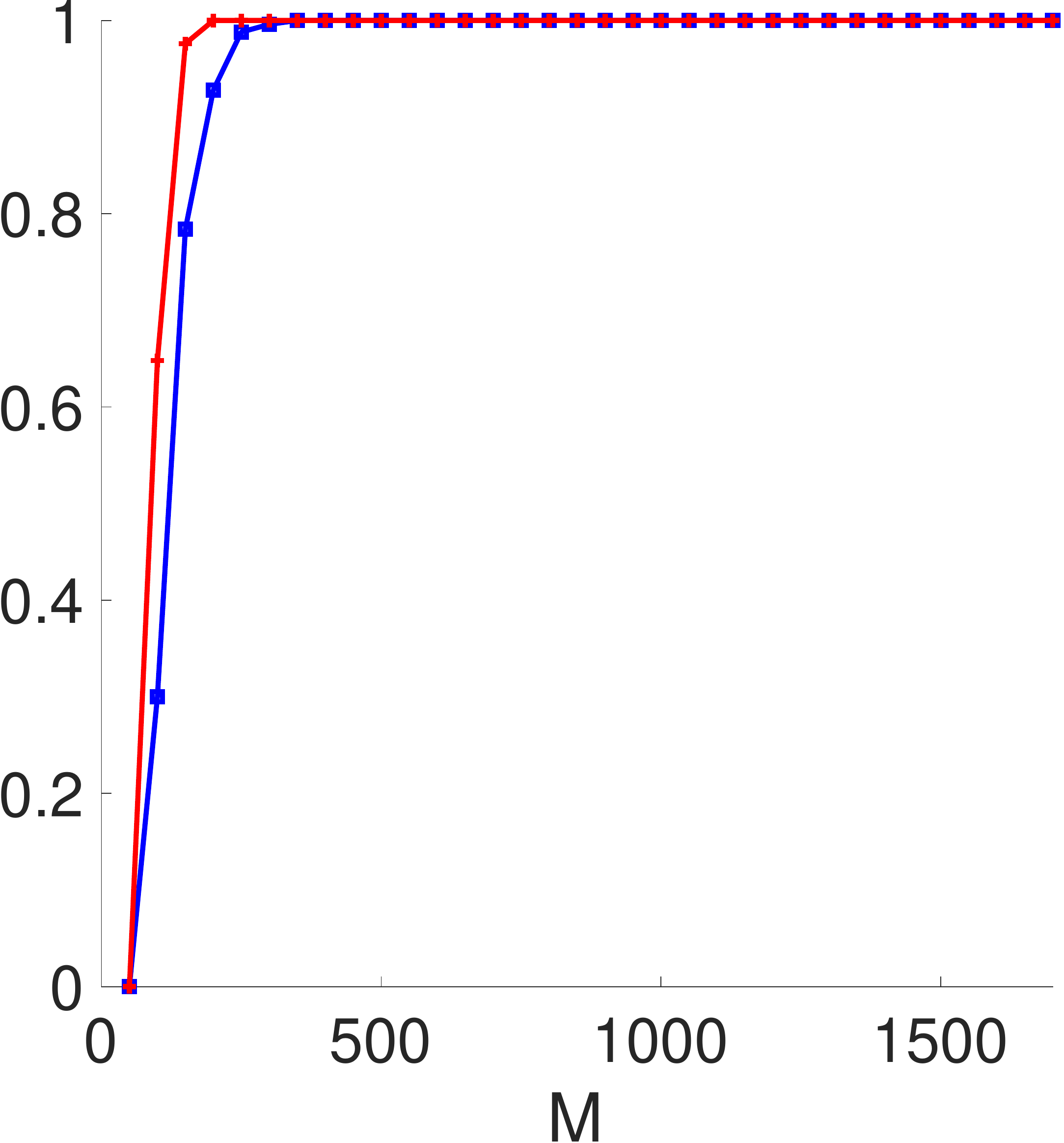}
\includegraphics[width=.3\linewidth, height=0.25\linewidth
]{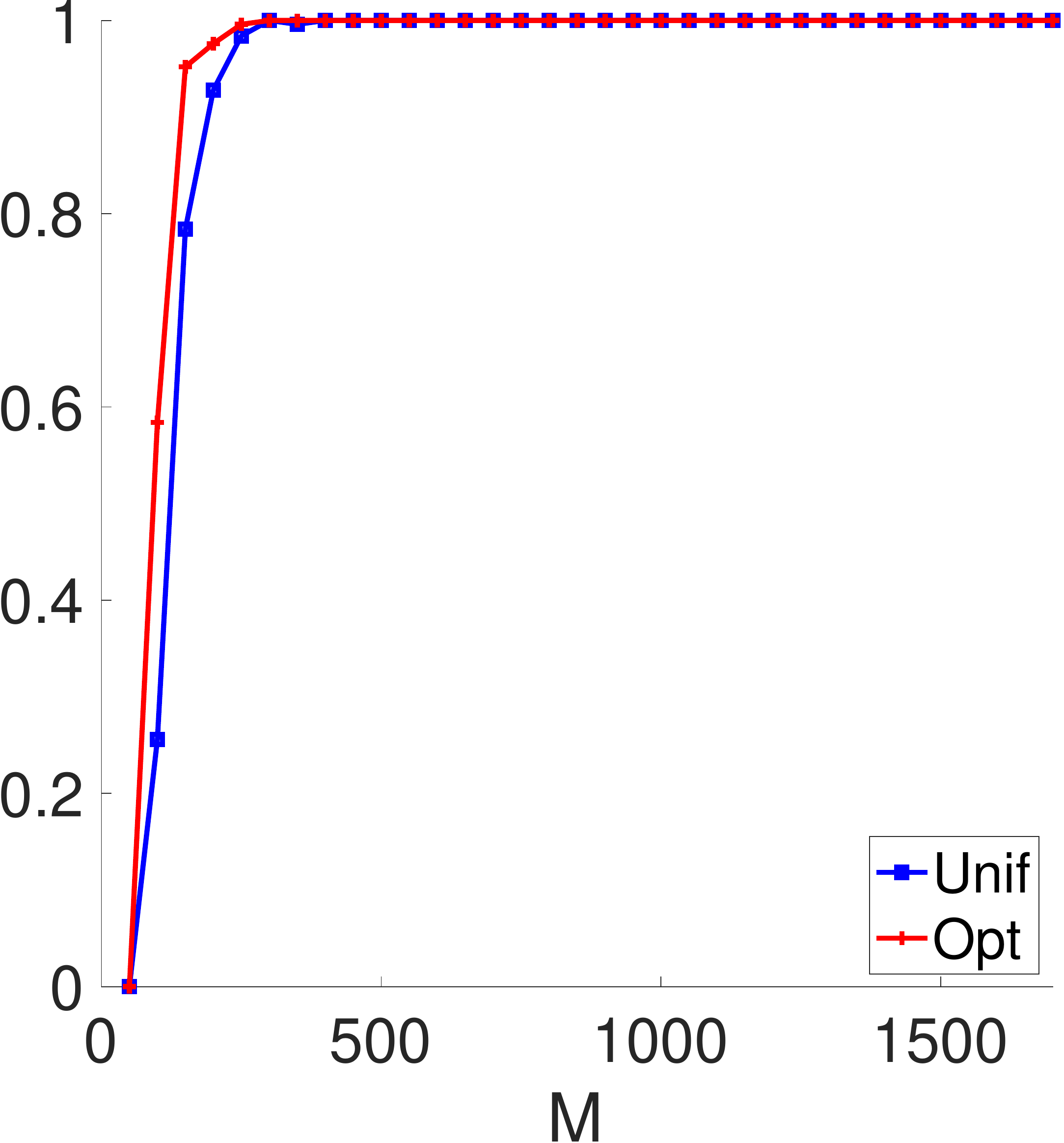}\\
\begin{minipage}{.3\linewidth} \centering \small \hspace{2mm} $\text{Path, } \underline{\delta}_{100,\mathbf{p}^{(1)}}$  \end{minipage}
\begin{minipage}{.3\linewidth} \centering \small \hspace{2mm}  $\text{Path, } \underline{\delta}_{100,\mathbf{p}^{(2)}}$ \end{minipage}
\begin{minipage}{.3\linewidth} \centering \small \hspace{2mm}   $\text{Path, } \underline{\delta}_{100,\mathbf{p}^{(3)}}$  \end{minipage}\\
\includegraphics[width=.3\linewidth, height=0.25\linewidth
]{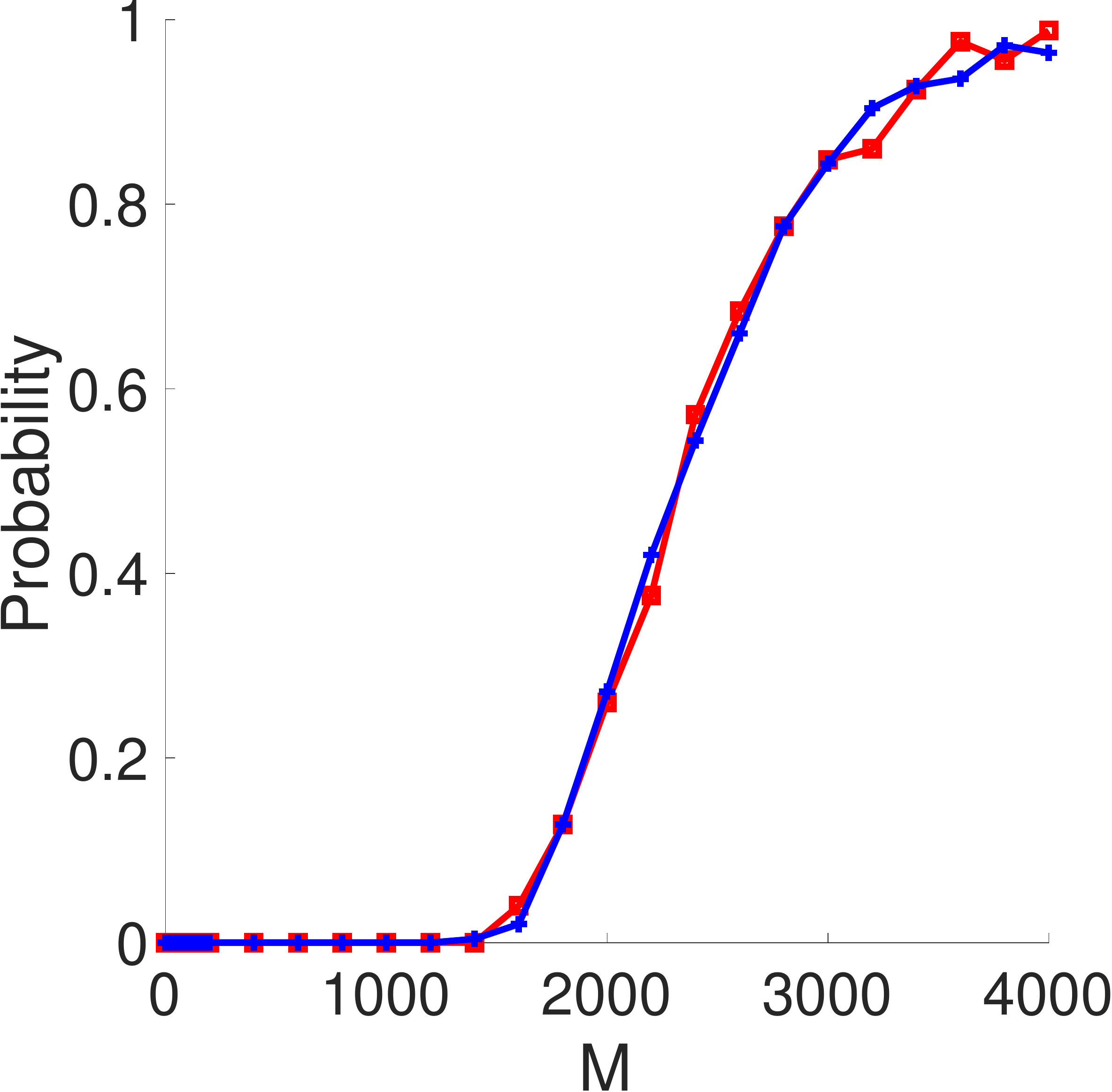}
\includegraphics[width=.3\linewidth, height=0.25\linewidth
]{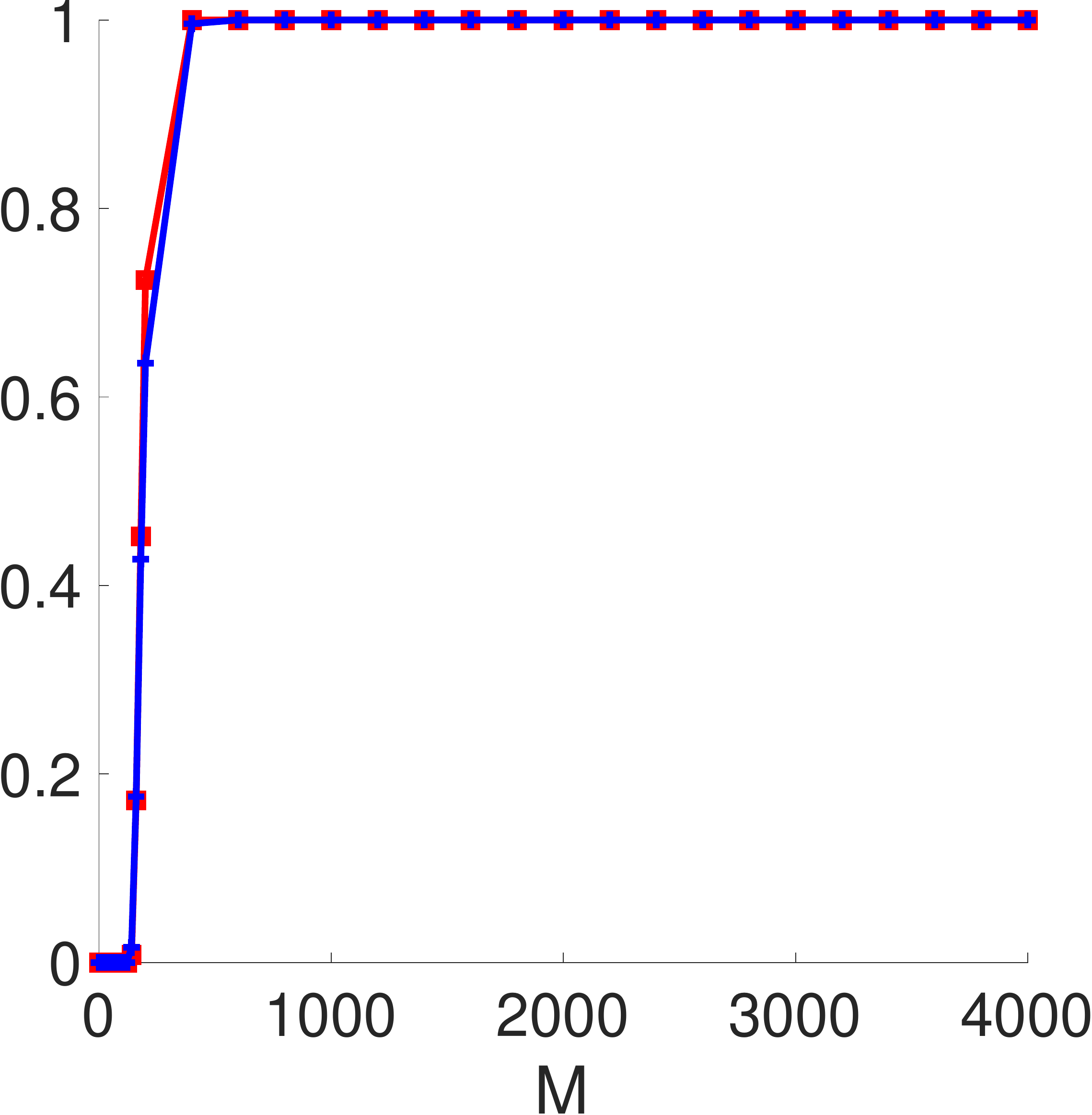}
\includegraphics[width=.3\linewidth, height=0.25\linewidth
]{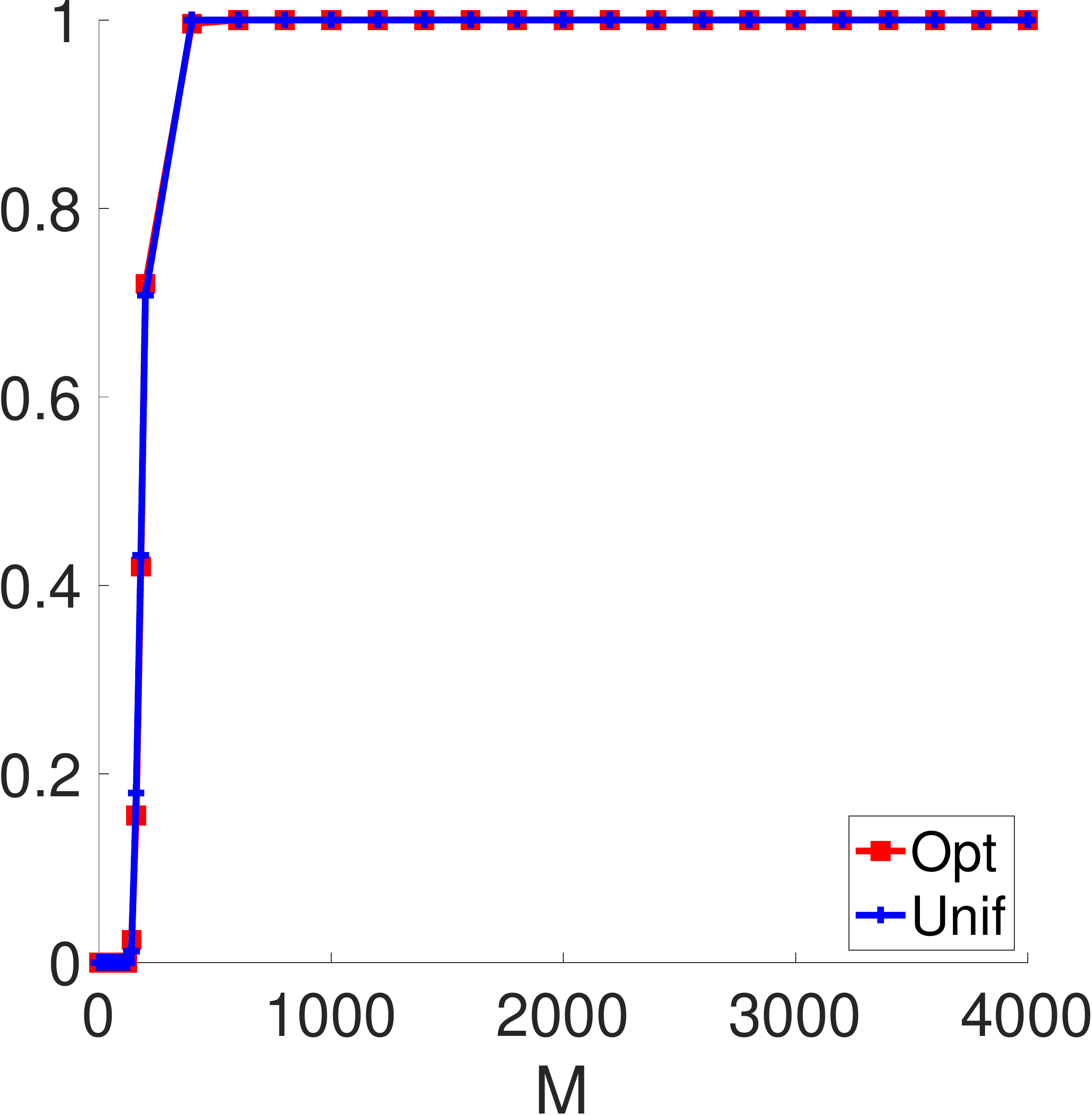}\\
\caption{\footnotesize\label{fig:eig_other} Probability that $\underline{\delta}_{10}$ is great than $0.005$ as a function of $M$. 
The panels show the results at $k=100$ for the Minnesota, bunny and path graphs, and at $k=50$ for the sensor graph using optimal and uniform distributions. The gap for $M$ is 200. }
\vspace{-0.2in}
\end{figure}

\textbf{Comparison of sampling regimes.} Let $M_{\mathbf{p}^{(j)}}^*$  denote the required number of space-time samples  in regime $j$ such that $\underline{\delta}_{k,\mathbf{p}^{(j)}}> 0.005$.  We have found that $M_{\mathbf{p}^{(1)}}^*$ is  largest  for all graphs used in the experiments. 
And regime 2 has the best performance in terms of the sample complexity. That is to say, we have $ M_{\mathbf{p}^{(2)}}^* \leq  M_{ \mathbf{p}^{(3)}}^*\leq  M_{ \mathbf{p}^{(1)}}^*$,
which  indicates the relationship
$T(\nu_{1,\mathbf{p}^{(1)}}^{K,T})^2 \geq  (\nu_{3,\mathbf{p}^{(3)}}^{K,T})^2\geq \sum_{t=0}^{T-1}(\nu_{2,\mathbf{p}^{(2,t)}}^{K,T})^2 $ for uniform and optimal distributions.  These observations coincide with the respective coherence and  are consistent with Theorem~\ref{thm:embedding_all} and the analysis in Proposition \ref{prop:gcoherence}. %

\subsection{Reconstruction of bandlimited signals} 

In this part, we examine the performance of the reconstruction of initial signals from space-time samples on   community graphs of type $C_3$, the Minnesota graphs and  bunny graphs. We consider the recovery of $k$-bandlimited signals with $k=10$, in contrast with reconstruction results for static sampling setting in \cite{puy2018random}.
We choose $M = 200$ for $\mathbf{p}_{\textbf{opt}}^{(j)}$ with $j=1,2,3$.

The experiments are performed with and without noise on the sampled data. In the presence of noise, the noise vector  $\mathbf{e}$ is random with i.i.d. mean 0 and variance $\sigma^2$ Gaussian entries \footnote{For nodes sampled multiple times, the realization of the noise is thus different each time for the same sampled node. Thus, the noise vector  contains no duplicated entry.}.
$\sigma$   are chosen from $\{0, 1.5 \times 10^{-3}, 3.7 \times 10^{-3}, 8.8 \times 10^{-3}, 2.1 \times 10^{-2}, 5.0 \times 10^{-2} \}$. The signals are reconstructed by solving \eqref{eqn:obj2} for different values of the 
parameter $\gamma$ and different functions $g$. For the community graph and the bunny graph, the regularisation  parameter $\gamma$ varies between $10^{-3}$ and $10^2$. For the Minnesota graph, it varies between $10^{-1}$ and $10^{5}$. {For each $\sigma$, $10$ independent random signals of unit norm are drawn, sampled and reconstructed with  all 
possible  $(\gamma, g)$}. Then the mean  errors\footnote{See Theorem \ref{thm:main_regular} for the definition of ${\alpha}^*$ and ${\beta}^*$.} $\|x^* - x\|_2$, $\|{\alpha}^* - x\|_2$ and $\|{\beta}^*\|_2$ over these $10$ signals are computed.
 
 \begin{figure}[]
 \vspace{-0.1in}
\centering
\begin{minipage}{.3\linewidth} \centering \small \hspace{2mm} Regime 1  \end{minipage}
\begin{minipage}{.3\linewidth} \centering \small \hspace{2mm}  Regime 2  \end{minipage}
\begin{minipage}{.3\linewidth} \centering \small \hspace{2mm} Regime3  \end{minipage}\\
\begin{minipage}{.3\linewidth} \centering \small \hspace{2mm} Community  \end{minipage}
\begin{minipage}{.3\linewidth} \centering \small \hspace{2mm}  Community \end{minipage}
\begin{minipage}{.3\linewidth} \centering \small \hspace{2mm} Community  \end{minipage}\\
\includegraphics[width=.3\linewidth, height=0.25\linewidth
]{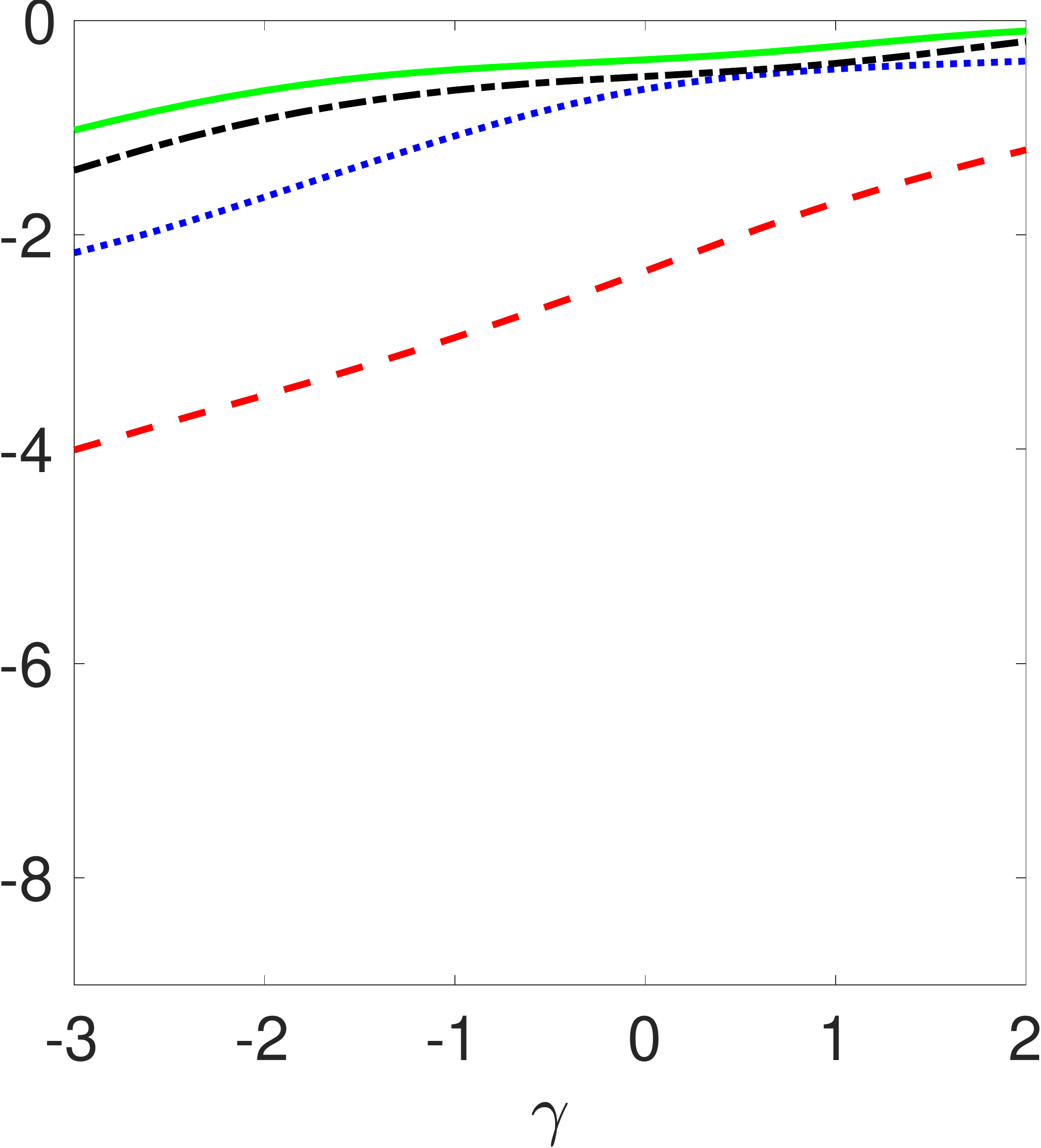}
\includegraphics[width=.3\linewidth, height=0.25\linewidth
]{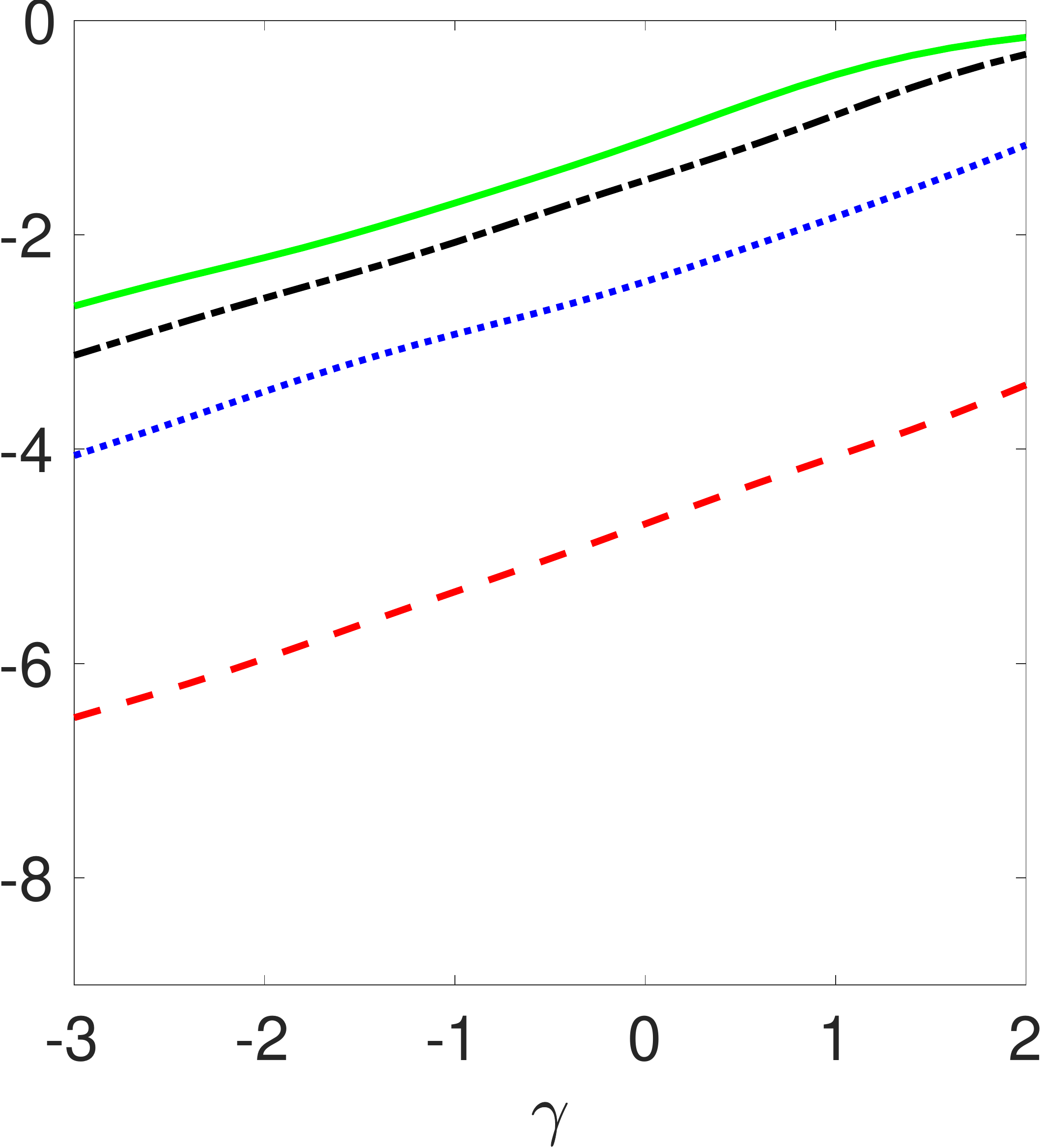}
\includegraphics[width=.3\linewidth, height=0.25\linewidth
]{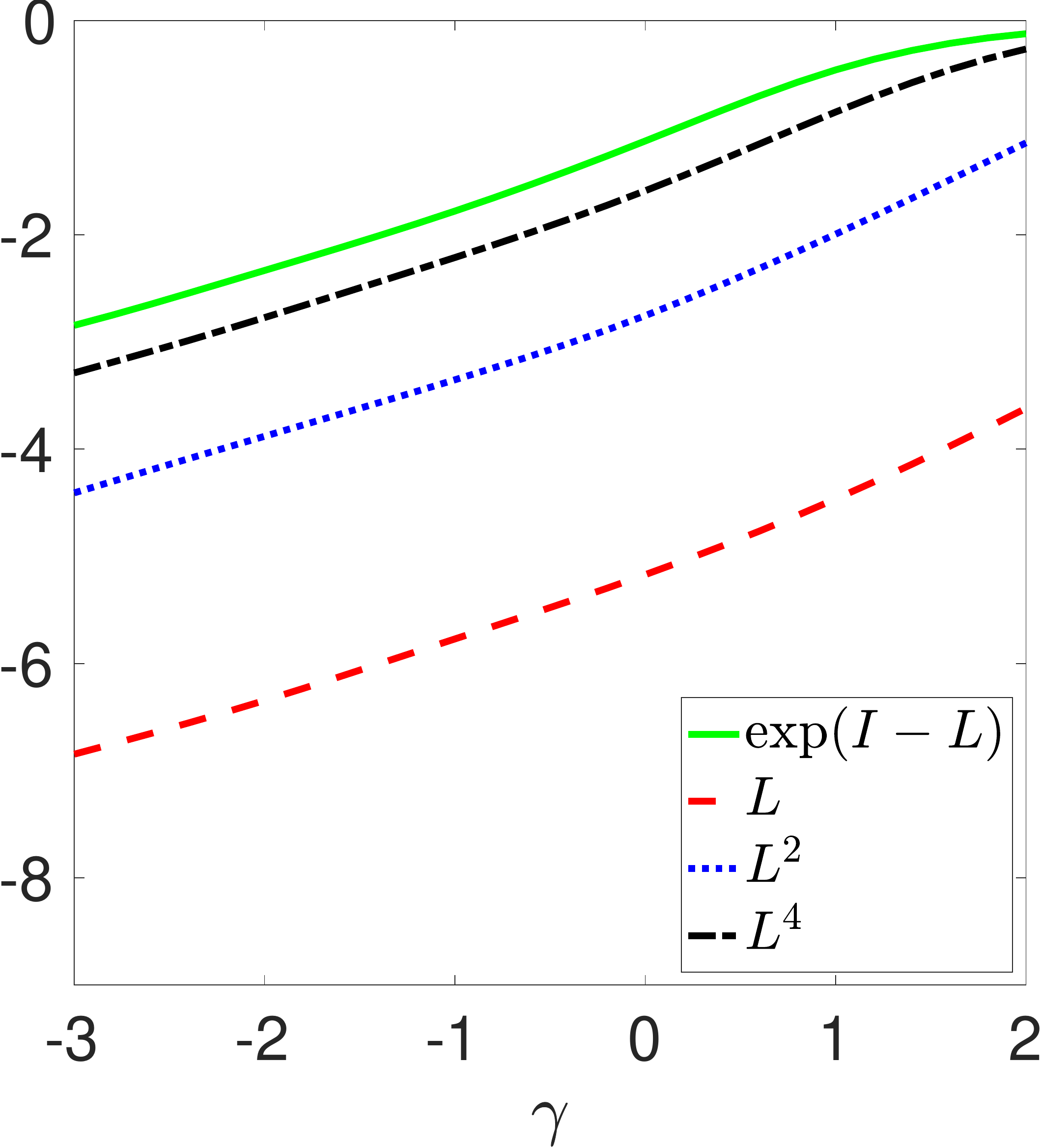}\\
\begin{minipage}{.3\linewidth} \centering \small \hspace{2mm} Bunny  \end{minipage}
\begin{minipage}{.3\linewidth} \centering \small \hspace{2mm}  Bunny  \end{minipage}
\begin{minipage}{.3\linewidth} \centering \small \hspace{2mm}   Bunny  \end{minipage}\\
\includegraphics[width=.3\linewidth, height=0.25\linewidth
]{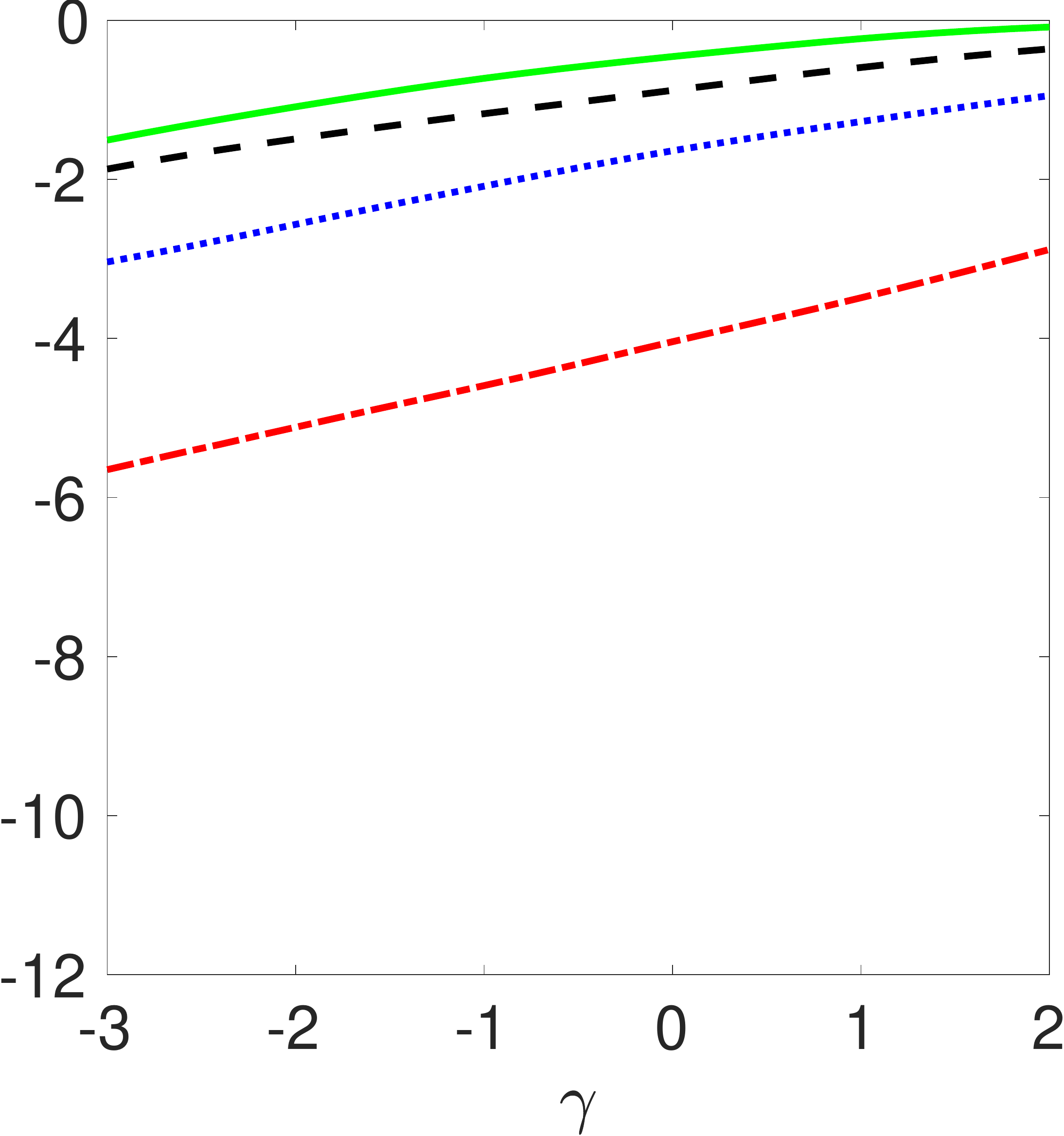}
\includegraphics[width=.3\linewidth, height=0.25\linewidth
]{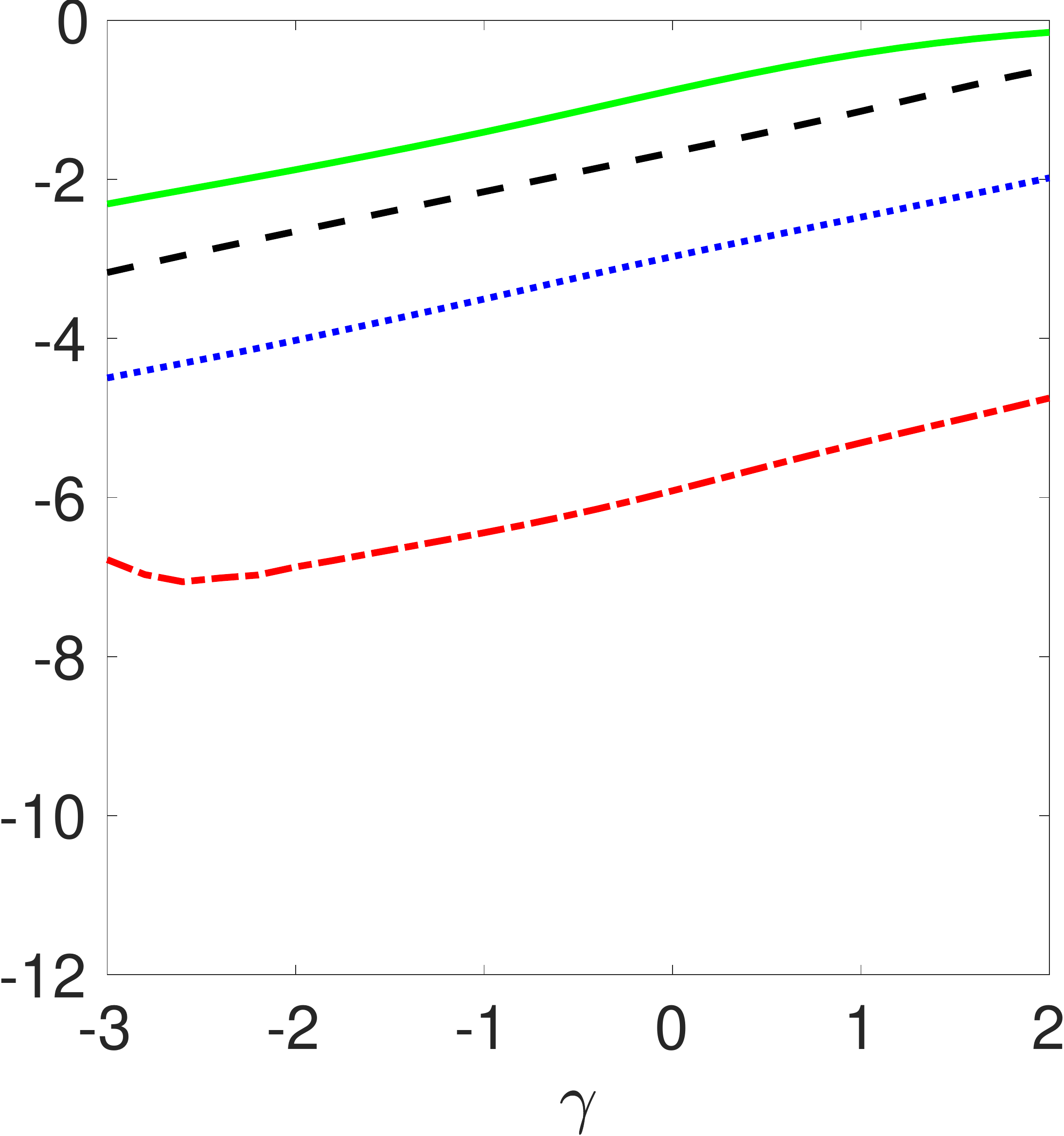}
\includegraphics[width=.3\linewidth, height=0.25\linewidth
]{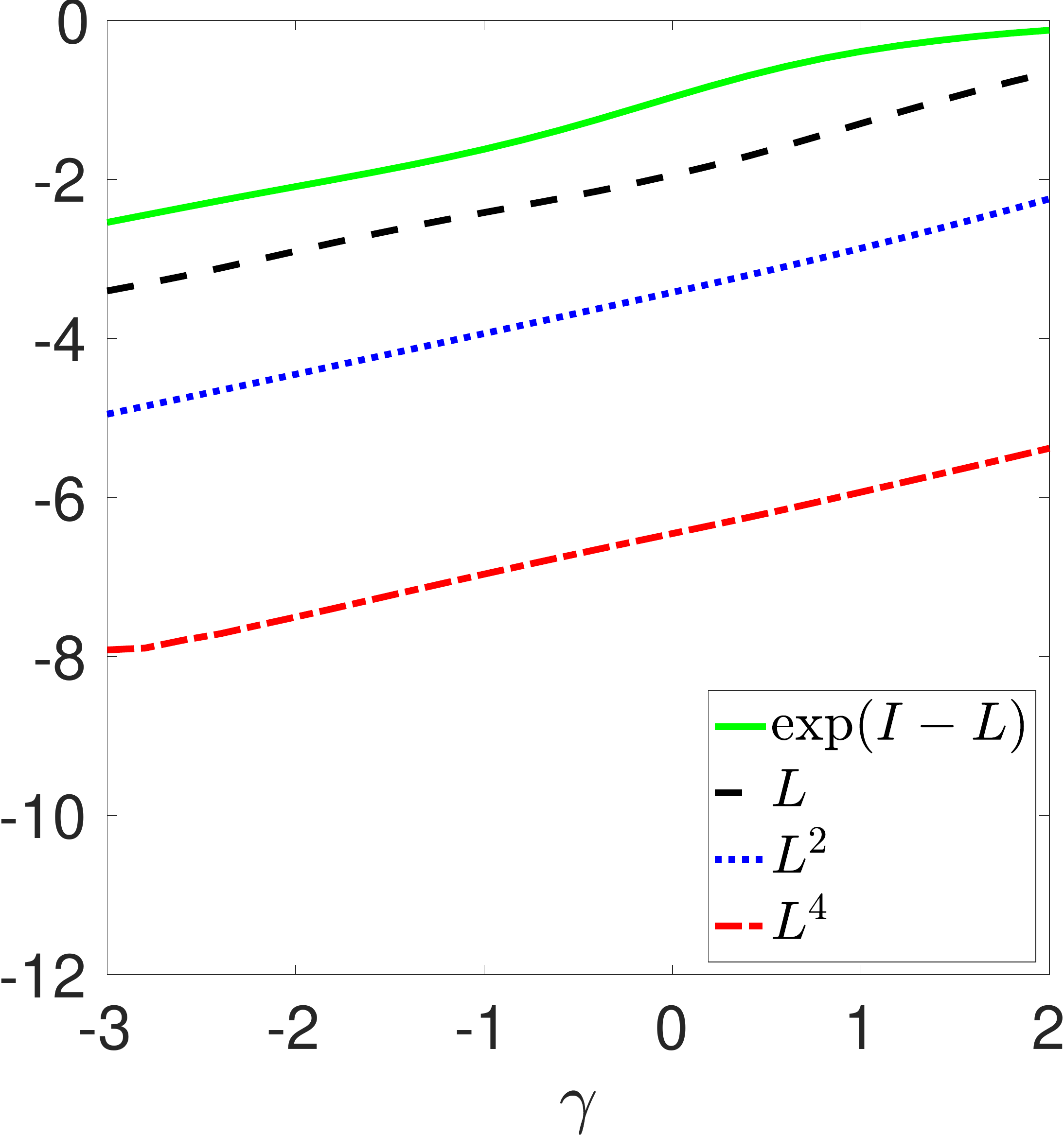}\\
\begin{minipage}{.3\linewidth} \centering \small \hspace{2mm} Minesota  \end{minipage}
\begin{minipage}{.3\linewidth} \centering \small \hspace{2mm}  Minesota  \end{minipage}
\begin{minipage}{.3\linewidth} \centering \small \hspace{2mm}   Minesota  \end{minipage}\\
\includegraphics[width=.3\linewidth, height=0.25\linewidth
]{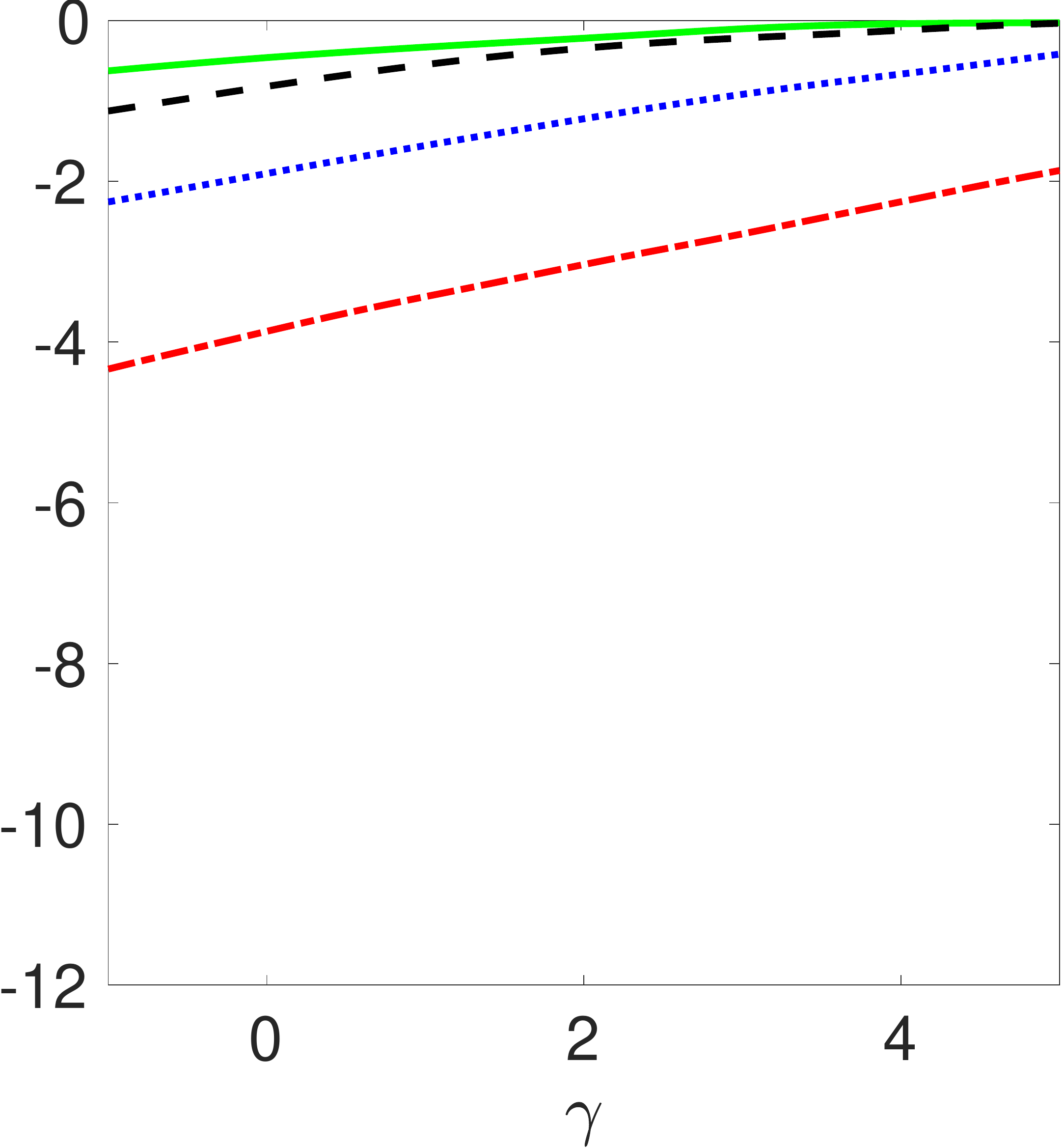}
\includegraphics[width=.3\linewidth, height=0.25\linewidth
]{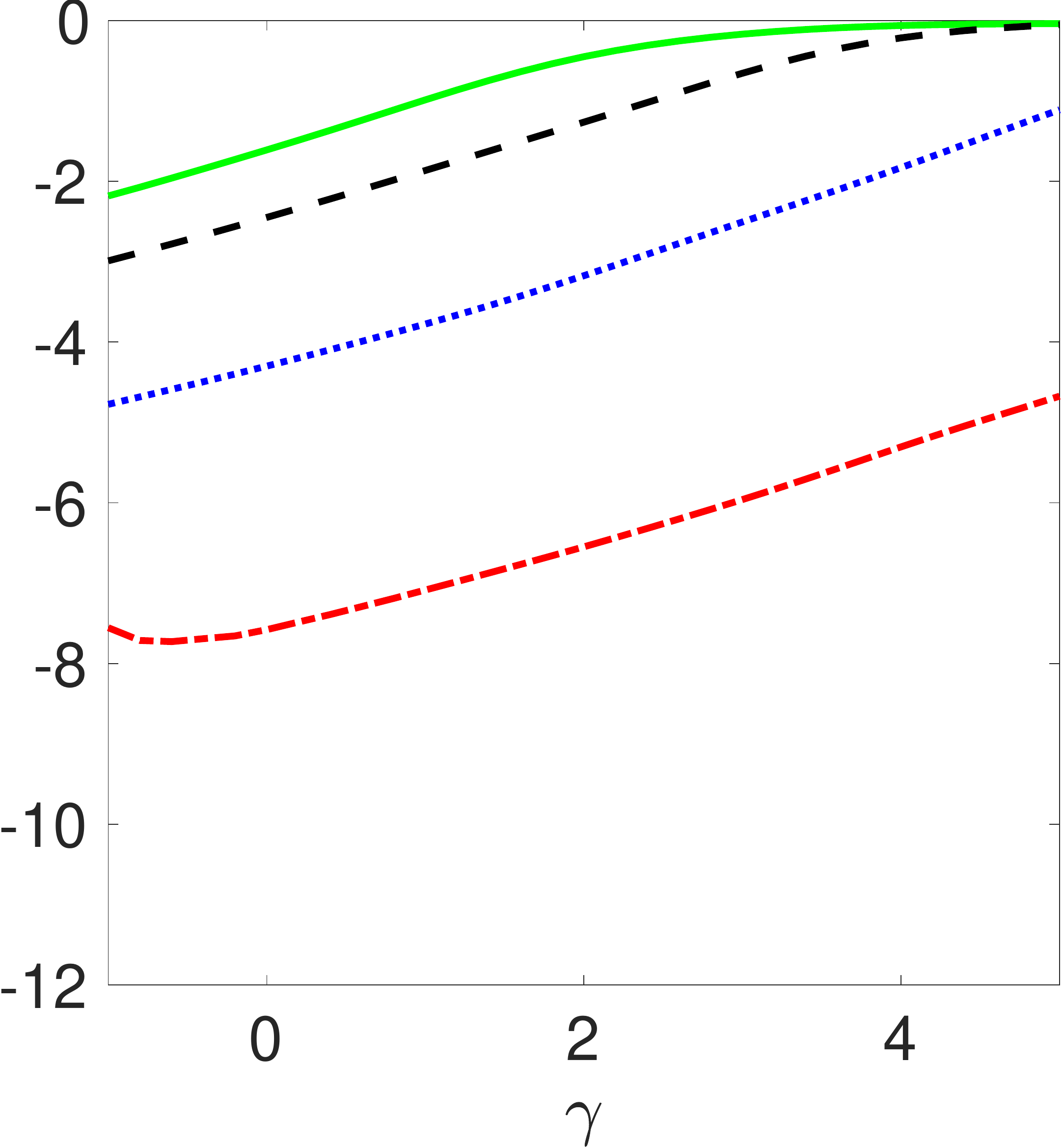}
\includegraphics[width=.3\linewidth, height=0.25\linewidth
]{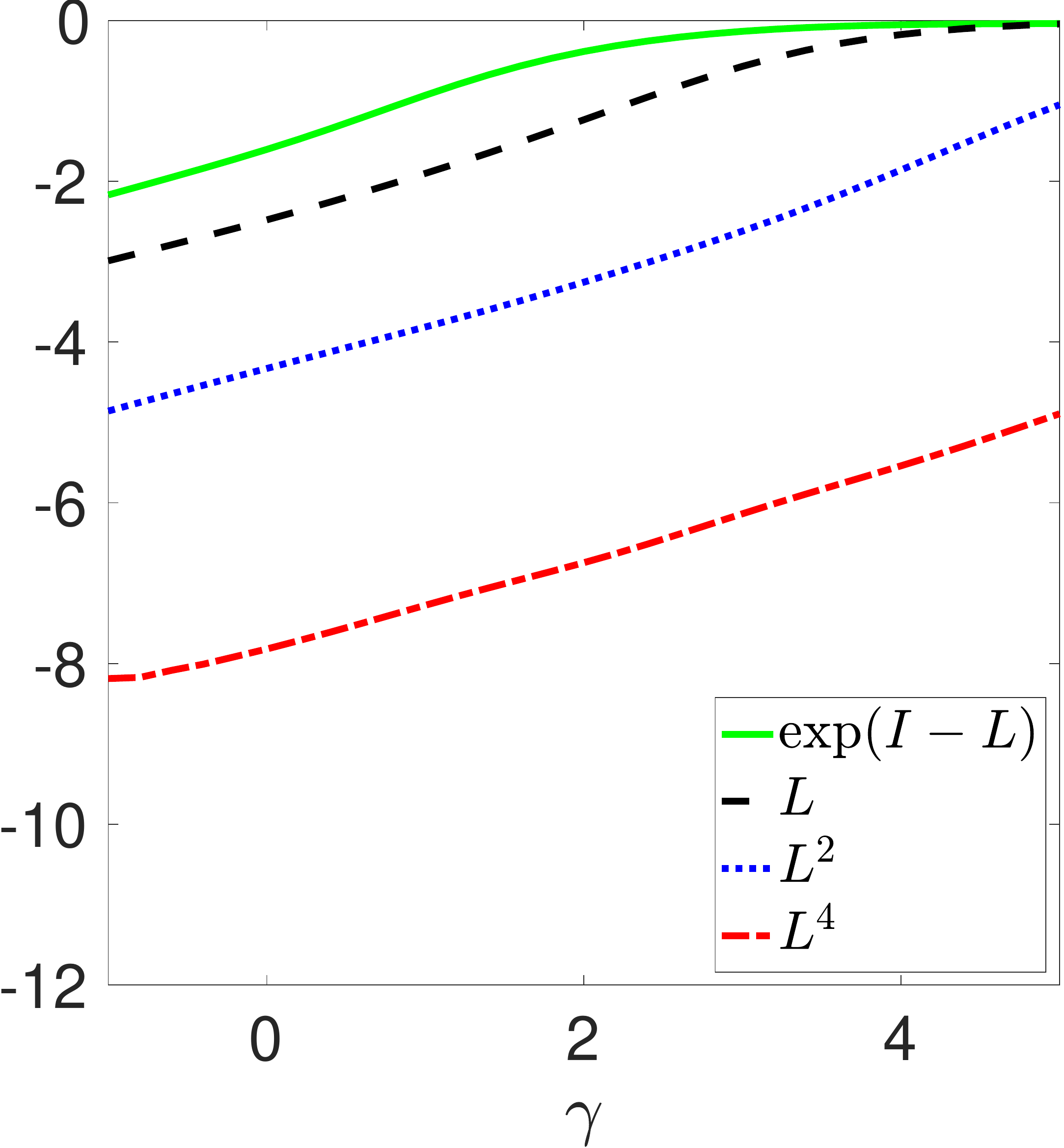}\\
\caption{\footnotesize Mean reconstruction errors $\log_{10}\|x-x^*\|$ of  10 bandlimited signals as a function of $\gamma$  for the community graph $C_3$, Bunny graph and Minnesota graph. We take $M=200$ noise free space-time samples for three regimes. The black dash curve represents results with $g(L)=L$. The blue dotted curves indicate the results with $g(L)=L^2$. The red dashed-dot curves indicate the results with $g(L)=L^4$.  The green solid curve indicate the results with $g(L)=\exp(I-L)$. The first, second and third columns show the results for regime 1, regime 2 and regime 3 using the optimal sampling distribution respectively. We refer to the reconstruction errors of $\|x-\alpha^*\|$ and $\|\beta^*\|$ in log scale to {the supplementary information (SI)} section \ref{addnum} (see their definitions in Theorem \ref{thm:main_regular}). Compared to the static case (Figure \ref{fig:rec_nonoise_static} in SI), our regime 1 achieved  comparable (cf. community graph) or slightly better (cf. bunny and Minesota graphs)  performance, and we observed significantly improvement in regime 2 and regime 3. 
}\label{fig:rec_nonoise}
\vspace{-0.3in}
\end{figure}
 
The mean reconstruction errors are presented in Fig.~\ref{fig:rec_nonoise}  when the measurements are noise-free. 
In these experiments, the signals are reconstructed by setting $g(L) = L, L^2, L^4, \exp(I-L)$.  
Let's recall that the ratio $g(\sigma_{k})/g(\sigma_{k+1})$ decreases as the power of $L$ increases. We observe that all reconstruction errors, $\|x^* - x\|_2$, $\|{\alpha}^* - x\|_2$ and $ \|{{\beta}^*}\|_2$ decrease when the ratio $g(\sigma_{k})/g(\sigma_{k+1})$ {decreases} in the range of small $\gamma$ (cf.~plots in SI section~\ref{addnum}), as predicted by the upper bounds  in Theorem~\ref{thm:main_regular}.  Additionally, Compared to the static case (Figure \ref{fig:rec_nonoise_static} in SI), our regime 1 achieved  comparable (even better)  performance, and we observed significantly improvement in regime 2 and 3.

We present the mean reconstruction errors when the   measurements are noisy 
in Fig.~\ref{fig:rec_noisy}. In these experiments, we reconstruct the signals using $g(L) = L^4$. As expected the best regularisation parameter $\gamma$ increases with the noise level. Compared with  static sampling setting (cf. Fig~\ref{fig:rec_static} in SI section \ref{addnum}), our estimates in regime 2 and regime 3 achieve comparable accuracy. The estimates obtained in Regime 1 did not perform as well as other regimes, as indicated by our analysis: one needs to have more space-time samples due to the large spectral graph weighted coherence.

\begin{figure}[!t]
\centering
\begin{minipage}{.3\linewidth} \centering \small \hspace{2mm} Regime 1  \end{minipage}
\begin{minipage}{.3\linewidth} \centering \small \hspace{2mm}  Regime 2  \end{minipage}
\begin{minipage}{.3\linewidth} \centering \small \hspace{2mm} Regime3  \end{minipage}\\
\begin{minipage}{.3\linewidth} \centering \small \hspace{2mm} Community  \end{minipage}
\begin{minipage}{.3\linewidth} \centering \small \hspace{2mm}  Community \end{minipage}
\begin{minipage}{.3\linewidth} \centering \small \hspace{2mm} Community  \end{minipage}\\
\includegraphics[width=.3\linewidth,height=0.25\linewidth]{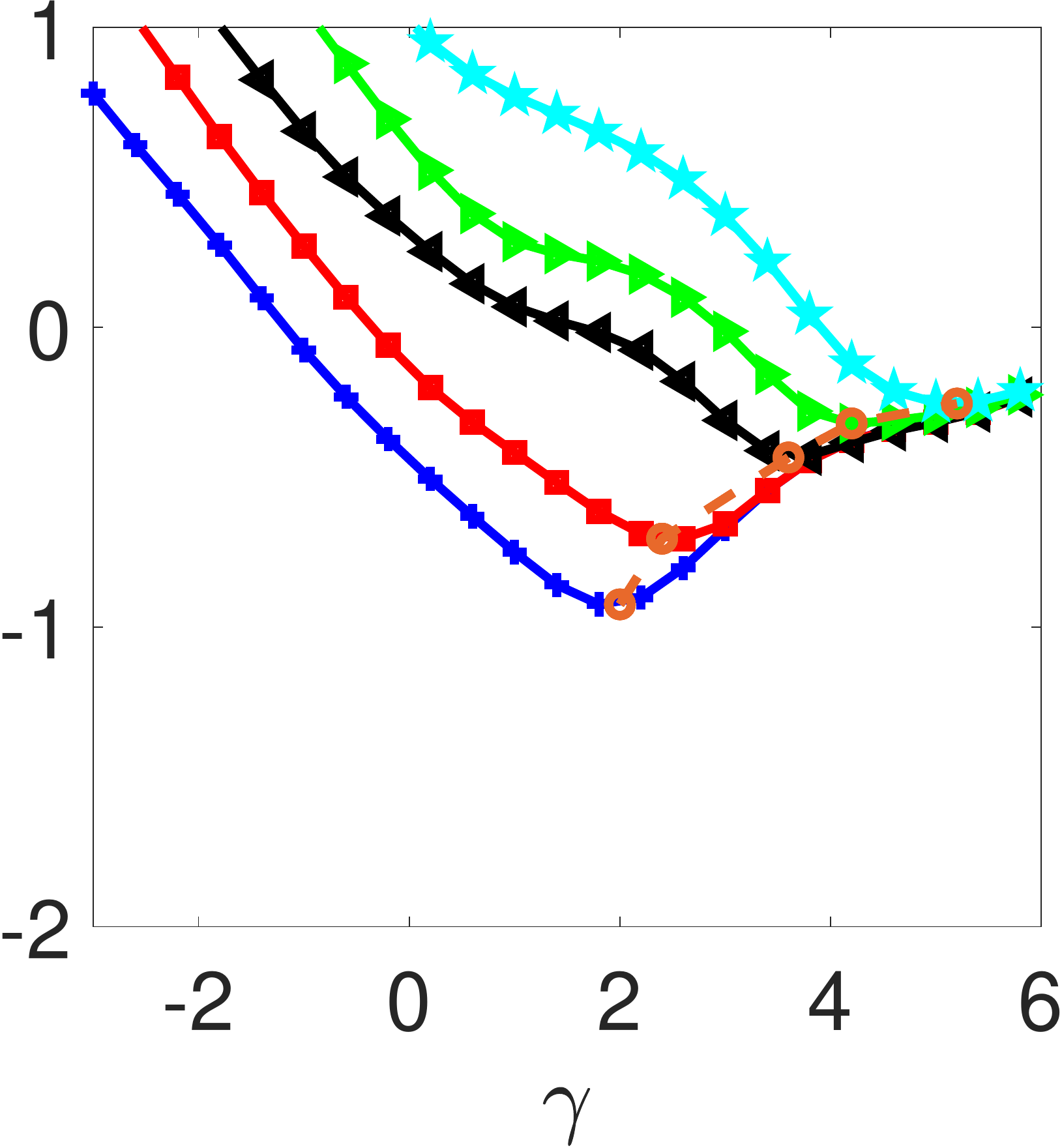}
\includegraphics[width=.3\linewidth,height=0.25\linewidth]{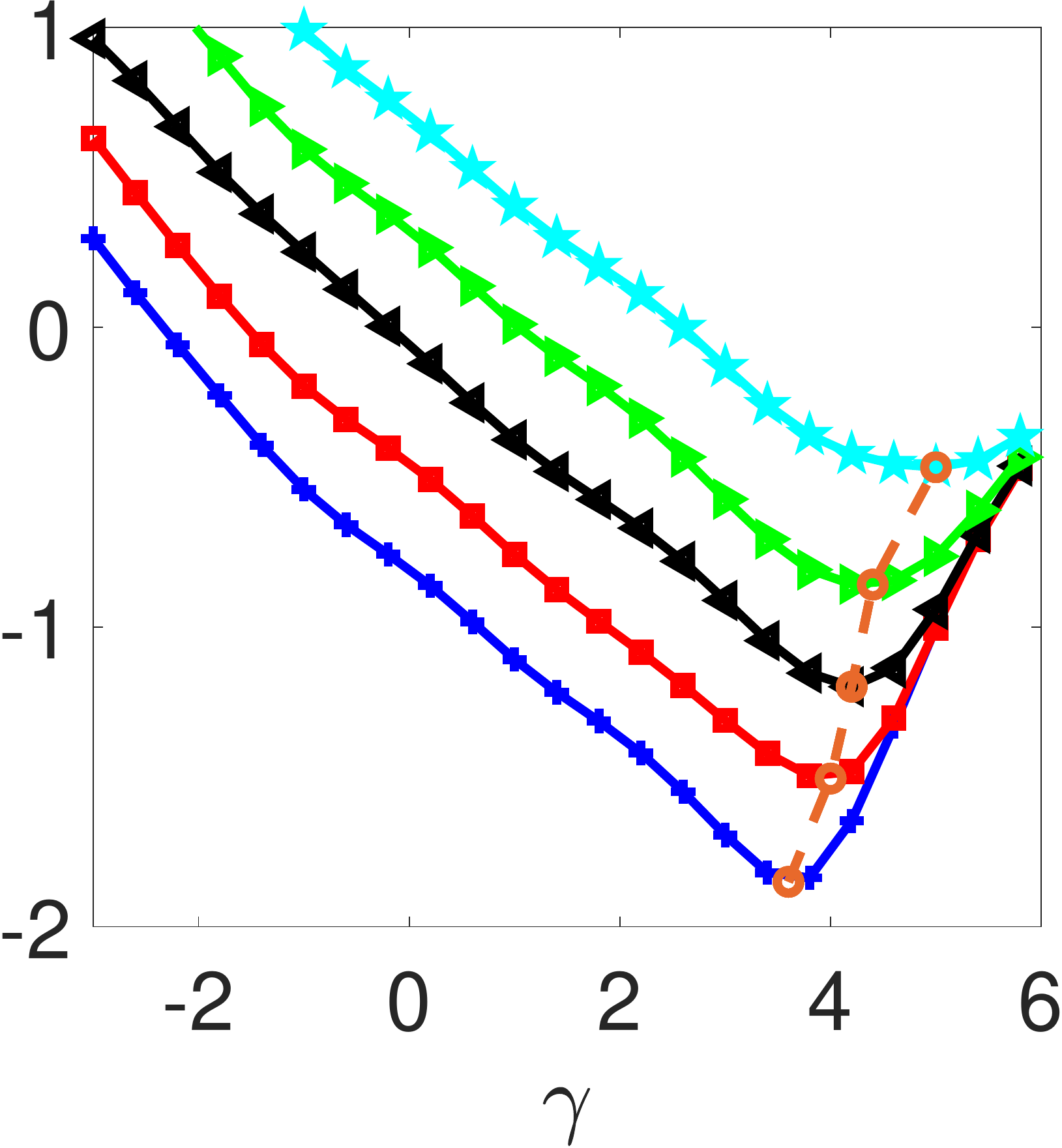}
\includegraphics[width=.3\linewidth,height=0.25\linewidth]{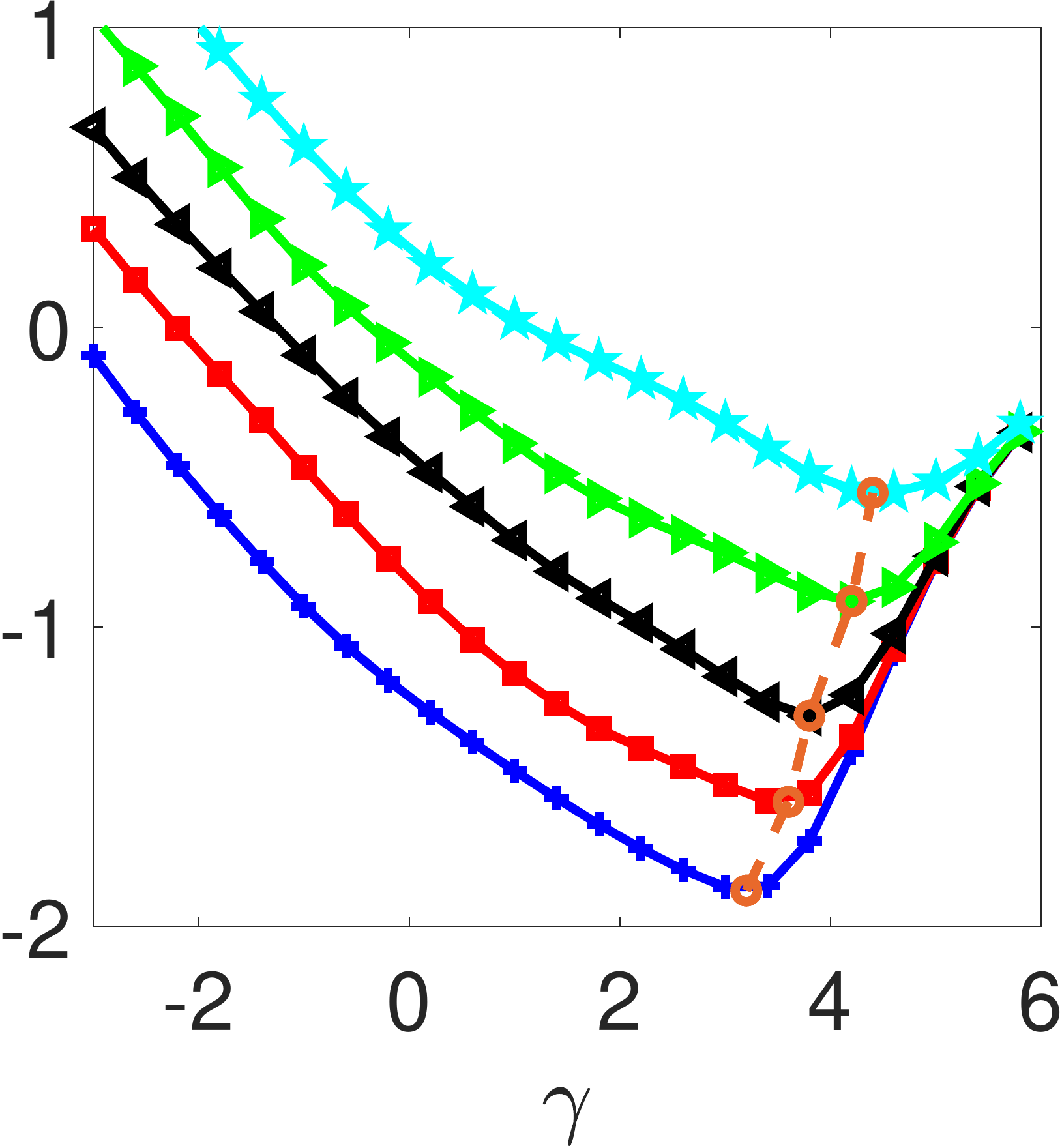}\\
\begin{minipage}{.3\linewidth} \centering \small \hspace{2mm} Bunny  \end{minipage}
\begin{minipage}{.3\linewidth} \centering \small \hspace{2mm}  Bunny \end{minipage}
\begin{minipage}{.3\linewidth} \centering \small \hspace{2mm} Bunny  \end{minipage}\\
\includegraphics[width=.3\linewidth,height=0.25\linewidth]{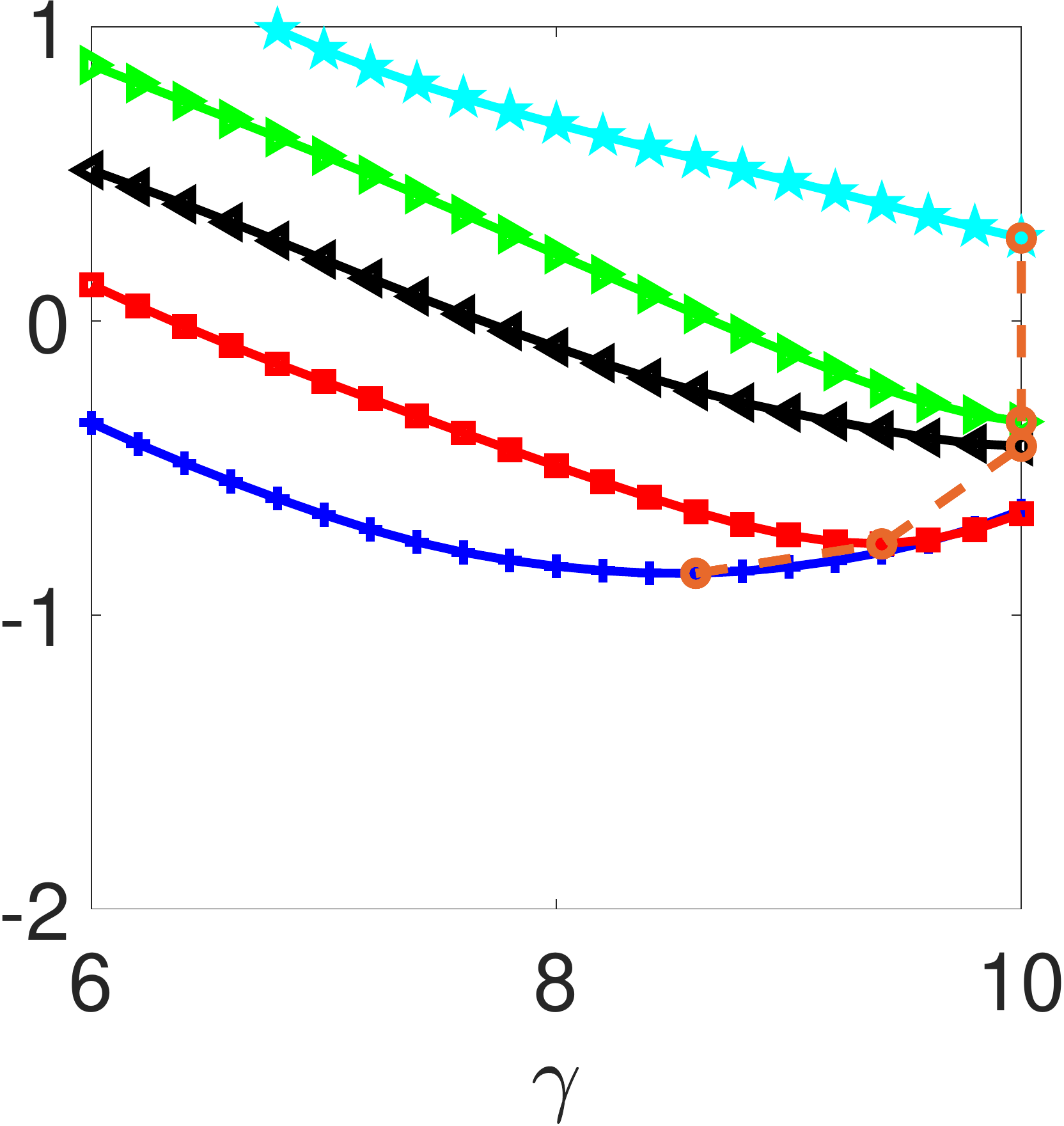}
\includegraphics[width=.3\linewidth,height=0.25\linewidth]{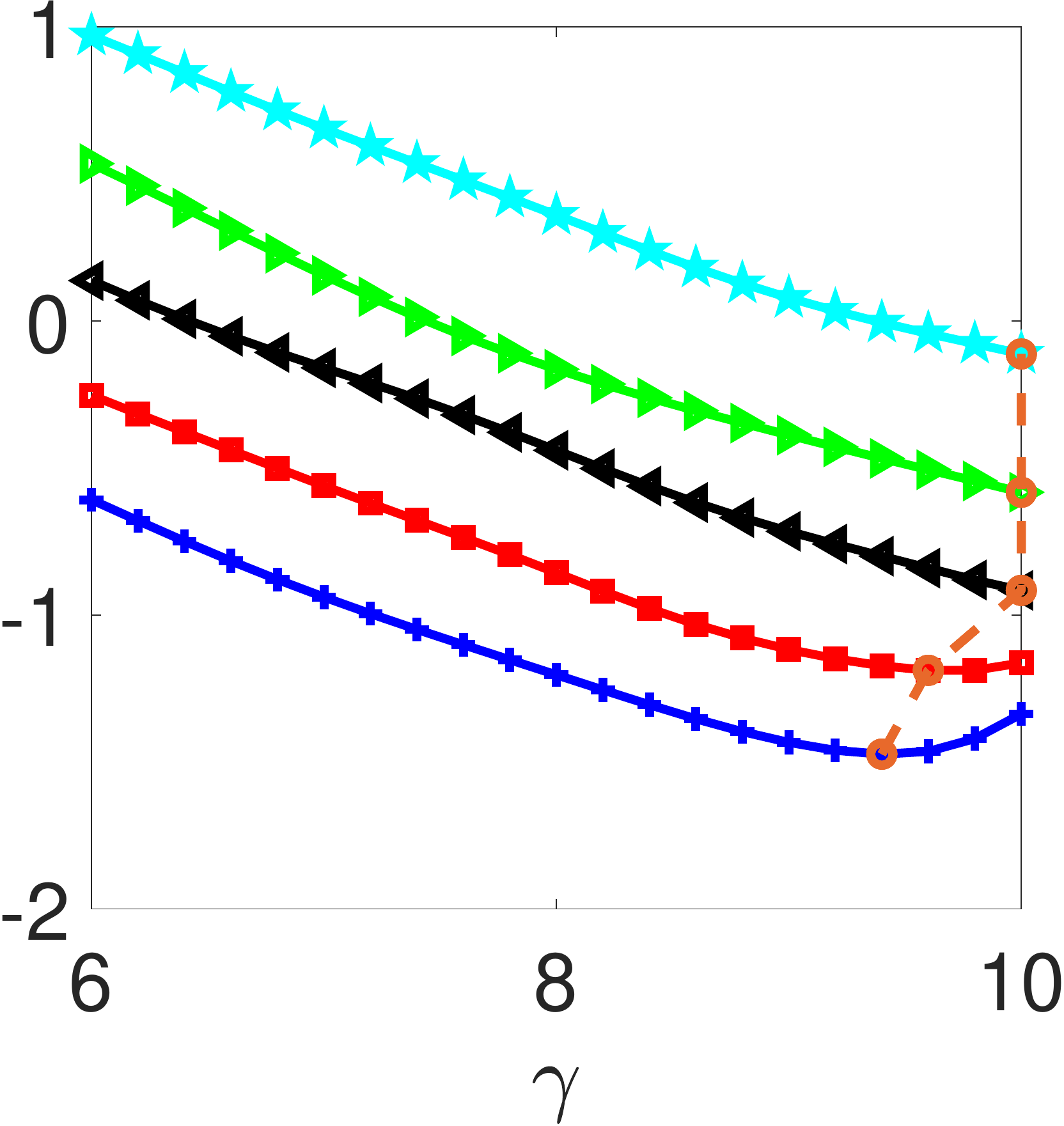}
\includegraphics[width=.3\linewidth,height=0.25\linewidth]{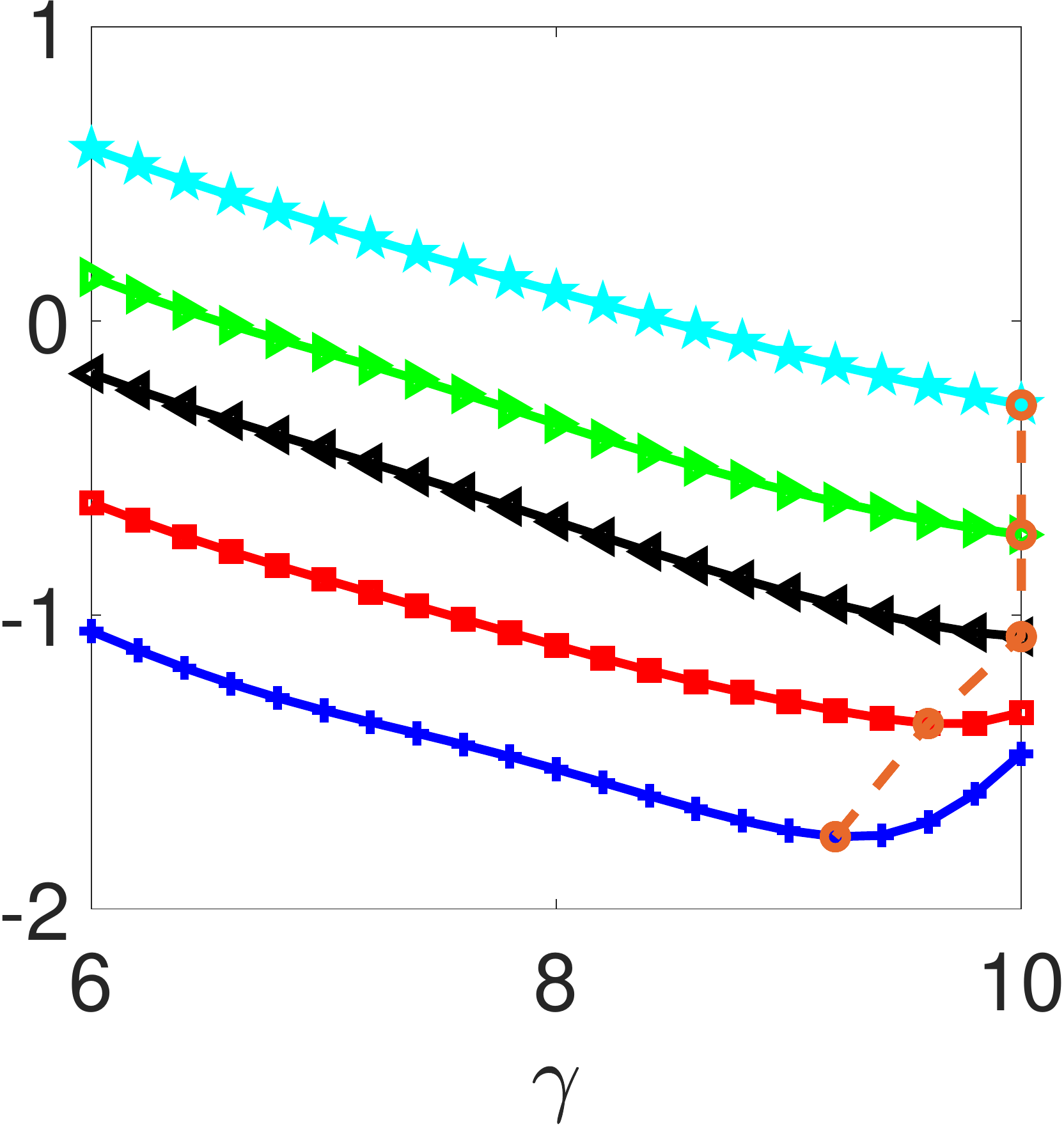}\\
\begin{minipage}{.3\linewidth} \centering \small \hspace{2mm} Minesota  \end{minipage}
\begin{minipage}{.3\linewidth} \centering \small \hspace{2mm}  Minesota \end{minipage}
\begin{minipage}{.3\linewidth} \centering \small \hspace{2mm} Minesota  \end{minipage}\\
\includegraphics[width=.3\linewidth,height=0.25\linewidth]{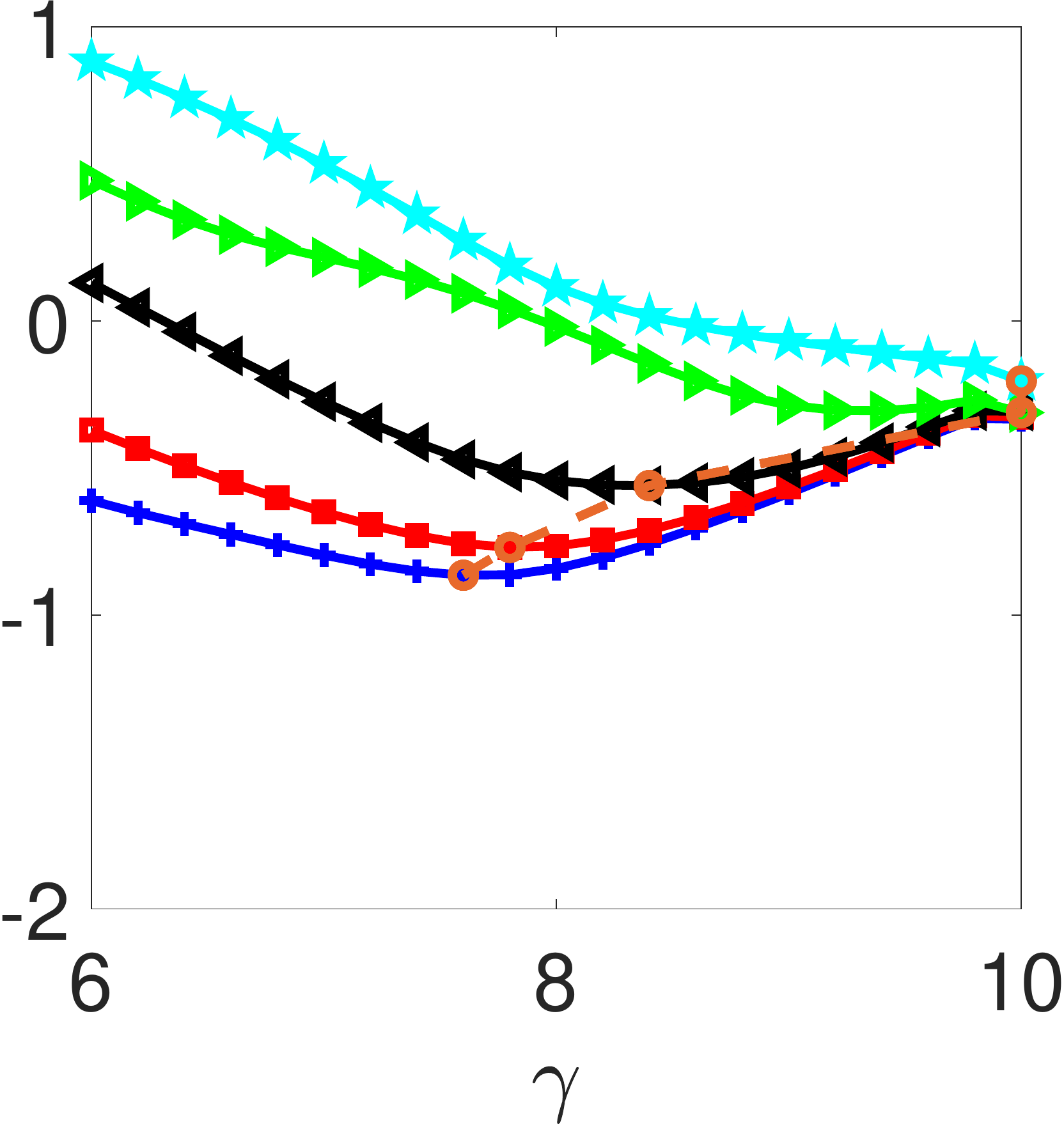}
\includegraphics[width=.3\linewidth,height=0.25\linewidth]{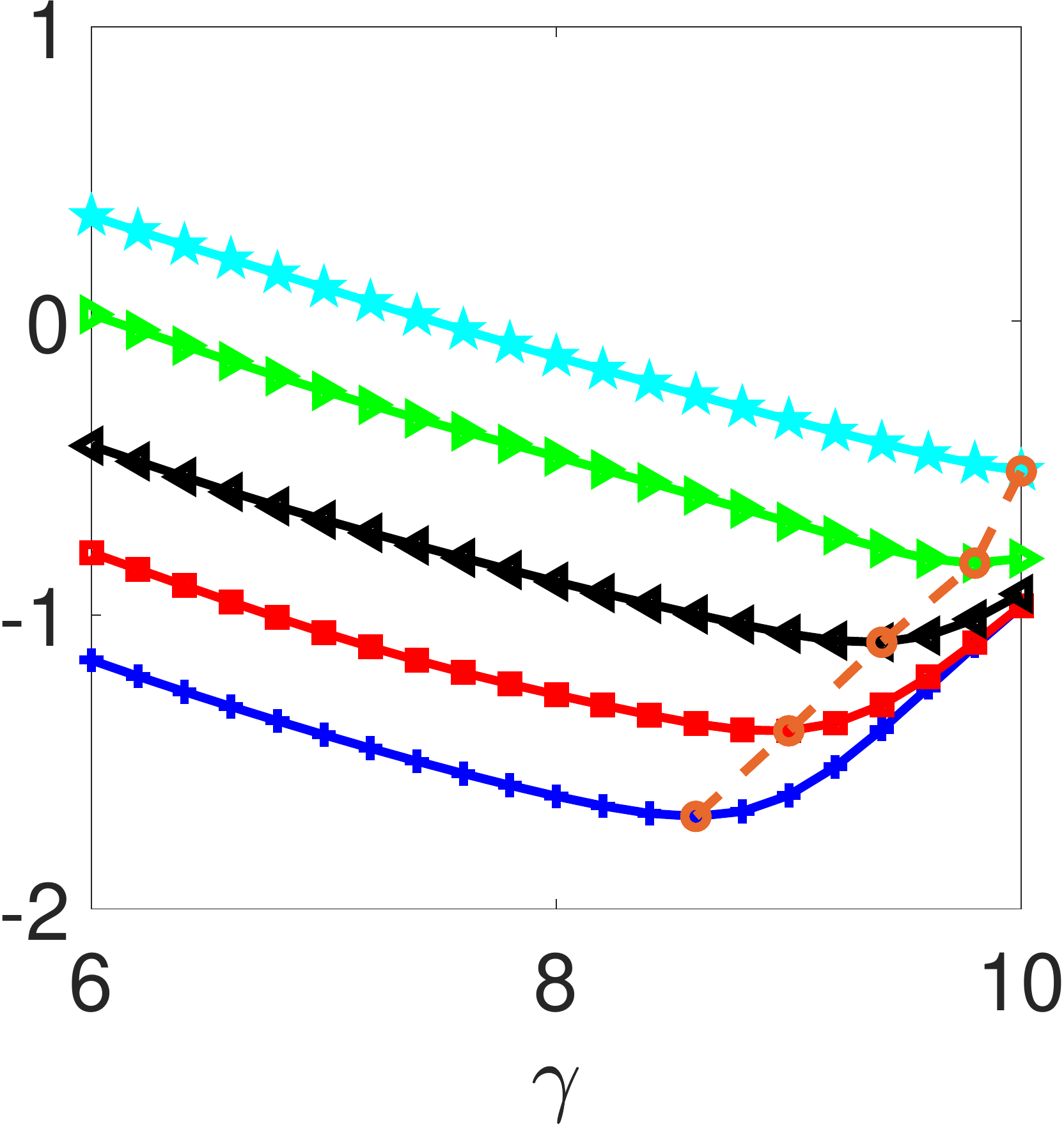}
\includegraphics[width=.3\linewidth,height=0.25\linewidth]{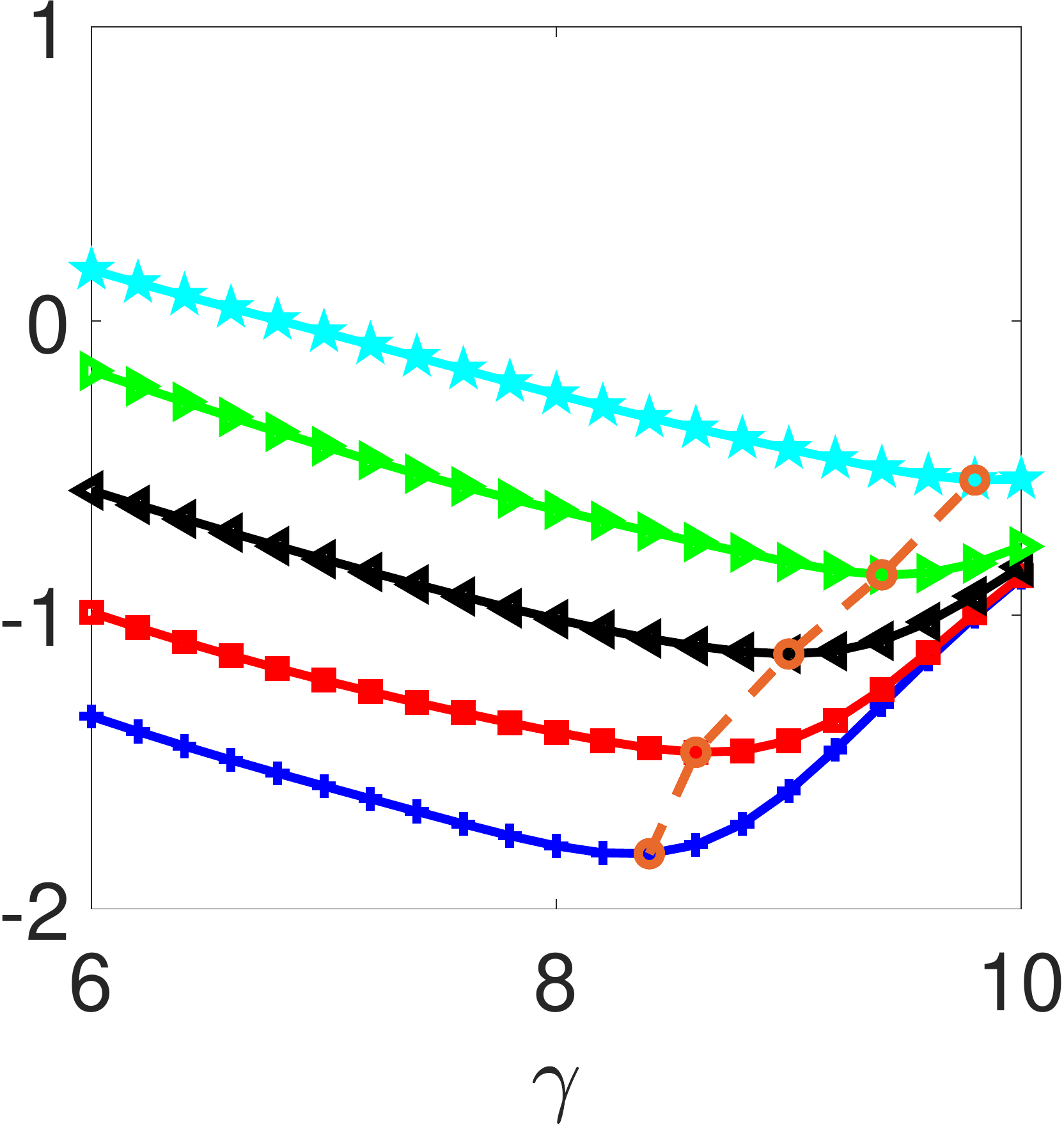}\\
\caption{\footnotesize Mean reconstruction error $\log_{10}\|x^* - x\|_2$ of $10$-bandlimited signals as a function of $\gamma$ with $g(L) = L^4$. The simulations are performed in presence of noise. The standard deviation of the noise is $0.0015$ (blue plus), $0.0037$ (red square), $0.0088$ (black left-pointing triangle), $0.0210$ (green right-pointing triangle), $0.0500$ (cyan pentagram). The best reconstruction errors are indicated by orange circles. The first, second and third columns show the results for a community graph of type $C_5$, and the Minnesota graph, respectively.  The reconstruction accuracy achieved in regime 2 and regime 3 are satisfactory, compared with the static sampling results  summarized in Fig.~\ref{fig:rec_static} in SI, since we only use 10 spatial samples on average.  For bunny graph, we may need to choose bigger $\gamma$, and we expect that the optimal error will also be as comparable with the static case in  Fig.~\ref{fig:rec_static}.  }\label{fig:rec_noisy}
\vspace{-0.2in}
\end{figure}

\subsection{Illustration: space-time sampling of a series  of blurring real images}

 \begin{figure}[h]
\vspace{-0.1in}
\centering
\subfloat{\includegraphics[width=.6\linewidth]{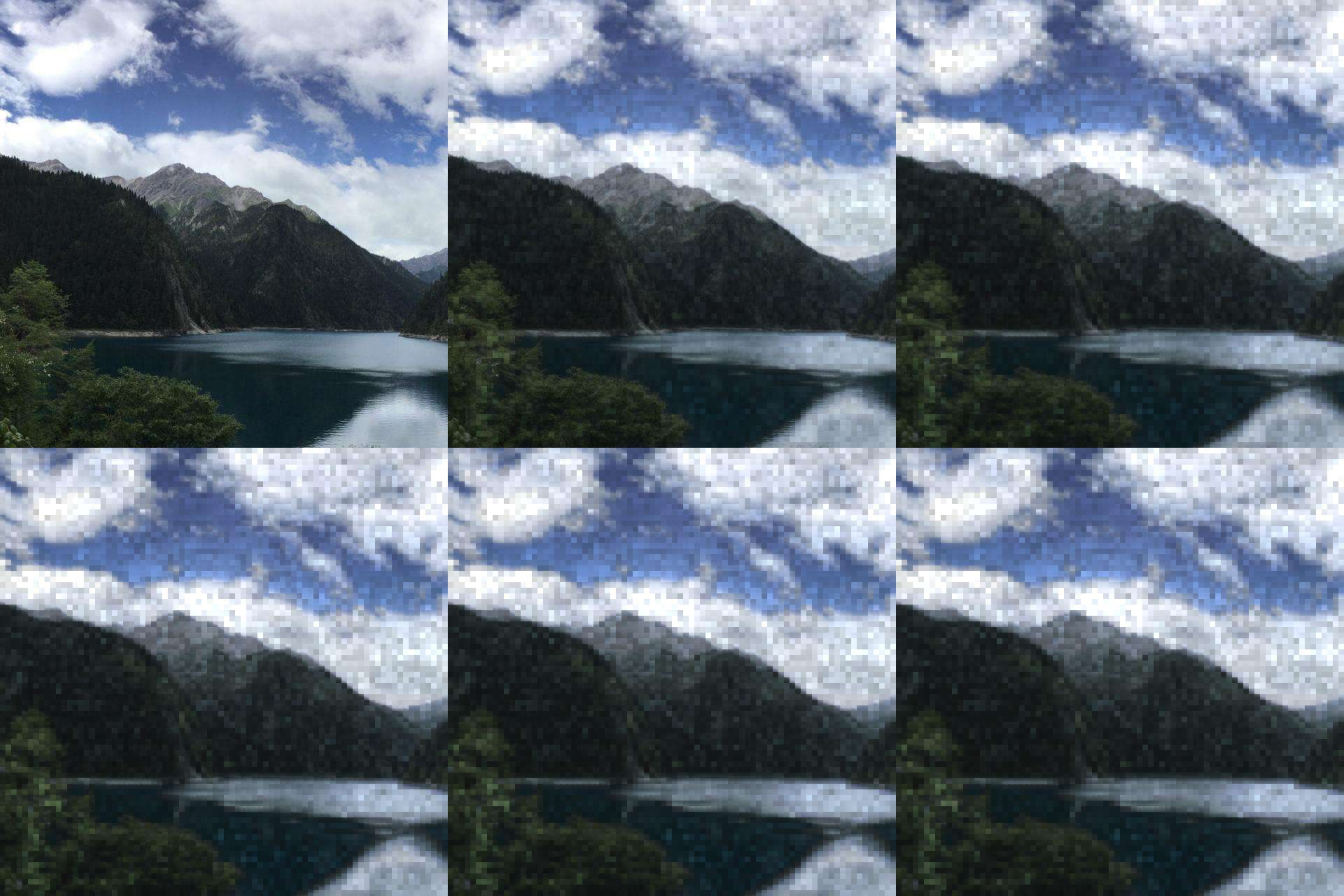}} 
\caption{\footnotesize The evolved images of the Jiuzhaigou over the first 6 time instances.
} \label{fig:evolved_images}
\vspace{-0.2in}
\end{figure}

The experimental section is finished with an example of de-blurring an image from the space-time samples. We use the photo of  Jiuzhaigou park in Sichuan, China.  The original photo is (approximately) bandlimited in a nearest neighborhood graph. We consider the heat diffusion process over the graph so we have a series of blurring images. We illustrate the dynamics in  Fig.~\ref{fig:evolved_images}.

This RGB image contains $640\times 640$ pixels. We divide this image into patches of $8 \times 8$ pixels, thus obtaining $80\times 80$ patches of $64$ pixels per RGB channel.  Each patch is denoted by $\mathbf{q}_{i,j,\ell} \in\mathbb{R} ^{64}$
with $i \in\{1,\cdots,80\}$, $j \in\{ 1,\cdots,80\}$, and $\ell \in \{1,2,3\}$. The pair of indices $(i, j)$ encodes the spatial location of the patch and $\ell$ encodes the color channel. Using these patches, we obtain the following matrix
\[X:=\begin{bmatrix}
\mathbf{q}_{1,1,1} & \mathbf{q}_{2,1,1} &\cdots&\mathbf{q}_{2,1,1} &\cdots &\mathbf{q}_{80,80,1}\\ 
\mathbf{q}_{1,1,2} & \mathbf{q}_{2,1,2} &\cdots&\mathbf{q}_{2,1,2} &\cdots &\mathbf{q}_{80,80,2}\\
\mathbf{q}_{1,1,3} & \mathbf{q}_{2,1,3} &\cdots&\mathbf{q}_{2,1,3} &\cdots &\mathbf{q}_{80,80,3}
\end{bmatrix}\in\mathbb{R}^{192\times n}, n=6400.
\]
 Each column of $X$ represents a color patch of the original image at a given position.

Motivated by the simulations in \cite{puy2018random},  we build a graph by  modelling the similarity between the columns of $X$. Let $\mathbf{x}_i:=X(:,i) \in\mathbb{R}^{192}$. 
For each $i$, we search for the $20$ nearest
neighbours of $\mathbf{x}_i$ among all other columns of $X$. Let $\mathbf{x}_j \in\mathbb{R}^{192}$ be a vector connected to $\mathbf{x}_i$. The weight $W_{ij}$ of the weighted adjacency matrix $W \in\mathbb{R}^{n\times n}$ satisfies
\begin{equation*} 
W_{ij}:=\exp\left(-\frac{\|\mathbf{x}_i-\mathbf{x}_j\|_2^2}{2\sigma^2} \right),
\end{equation*}
where $\sigma> 0$ is the standard deviation of all Euclidean distances between pairs of connected columns/patches. We then symmetrize  $W$. Each column of $X$ is thus connected to at least $20$ other columns after symmetrization. The construction of the graph is finished by computing the normalized Laplacian $L \in\mathbb{R}^{n\times n}$ associated to $W$. Set the evolution operator $A=\exp(-L)$. We let $A$ act on $X^\top$ iteratively with $T=6$ and we treat each column of $X^\top$ a bandlimited graph signal with bandwidth 300.  For this experiment,  we take $M=7200$ measurements using the uniform and optimal distributions 
for these three sampling regimes and denote the samples as matrix $Y$ which means the samples at each time instance is less than $20\%$ for regime 1 and 2.

The sampled images are presented in Fig.~\ref{fig:samples}, where all non-sampled pixels are in black. 
The image is reconstructed by solving \eqref{eqn:obj2} for each column of Y with  $\gamma =  10^{-5}$ and $g(L) = L^4$. Fig.~\ref{fig:real_image}  shows that a very accurate reconstruction of the original image is obtained.
The SNRs between the original and the reconstructed images are presented in Table \ref{tab:SNRs}.  The optimal sampling distribution allows a better image reconstruction quality than the results of the uniform sampling distribution. Additionally, the reconstructions  from sampling regimes 2 and 3 have no big difference but both are better than the results from sampling regime 1. The reconstruction results from space-time samples show that incorporating the temporal dynamics is very effective in reducing the spatial sampling density, and one achieved much better performance than the static setting with 4 times  higher spatial sampling densities.

\begin{figure}[h]
\vspace{-0.08in}
\centering
\subfloat{\includegraphics[width=.3\linewidth]{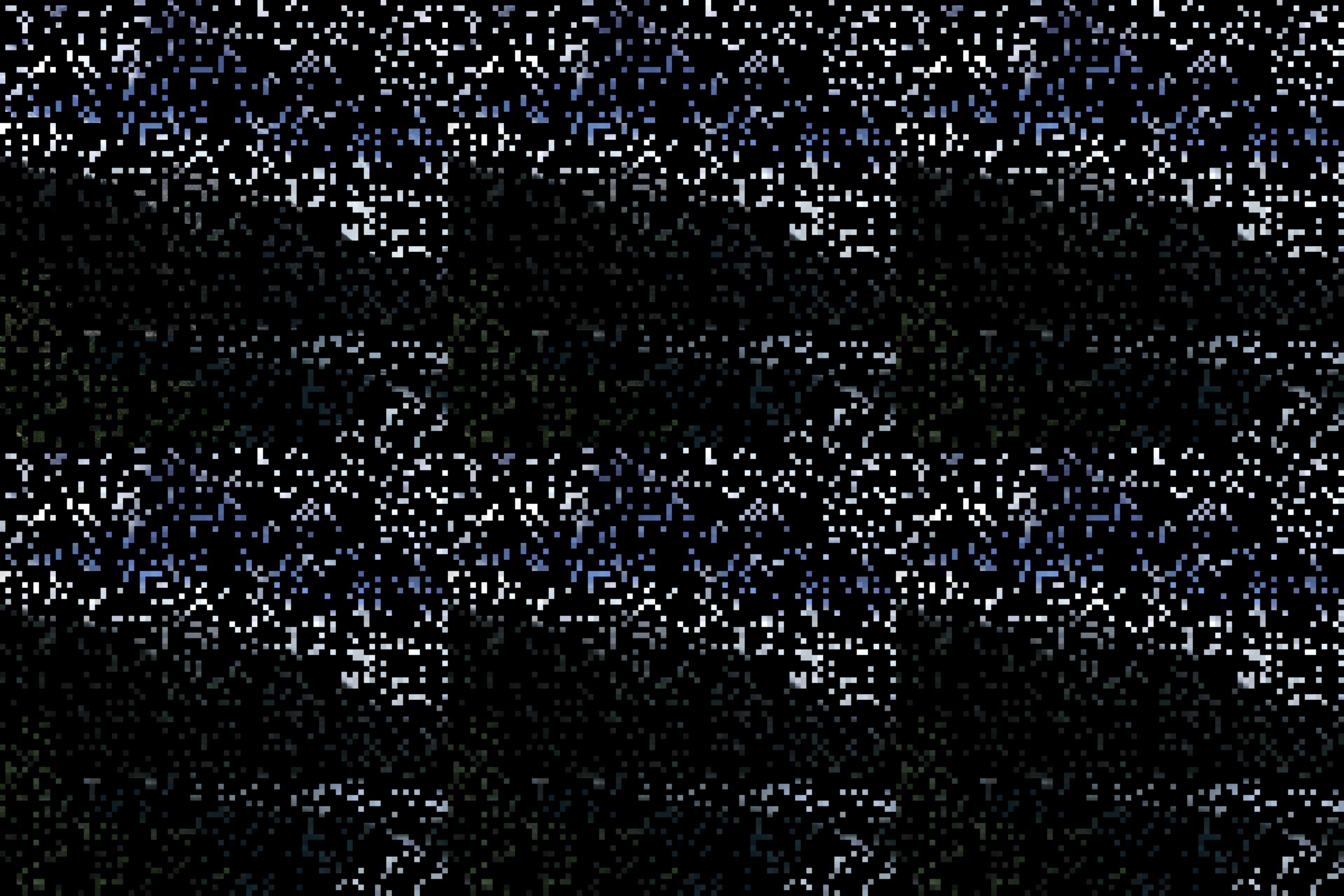}}\hspace{0.005in}
\subfloat{\includegraphics[width=.3\linewidth]{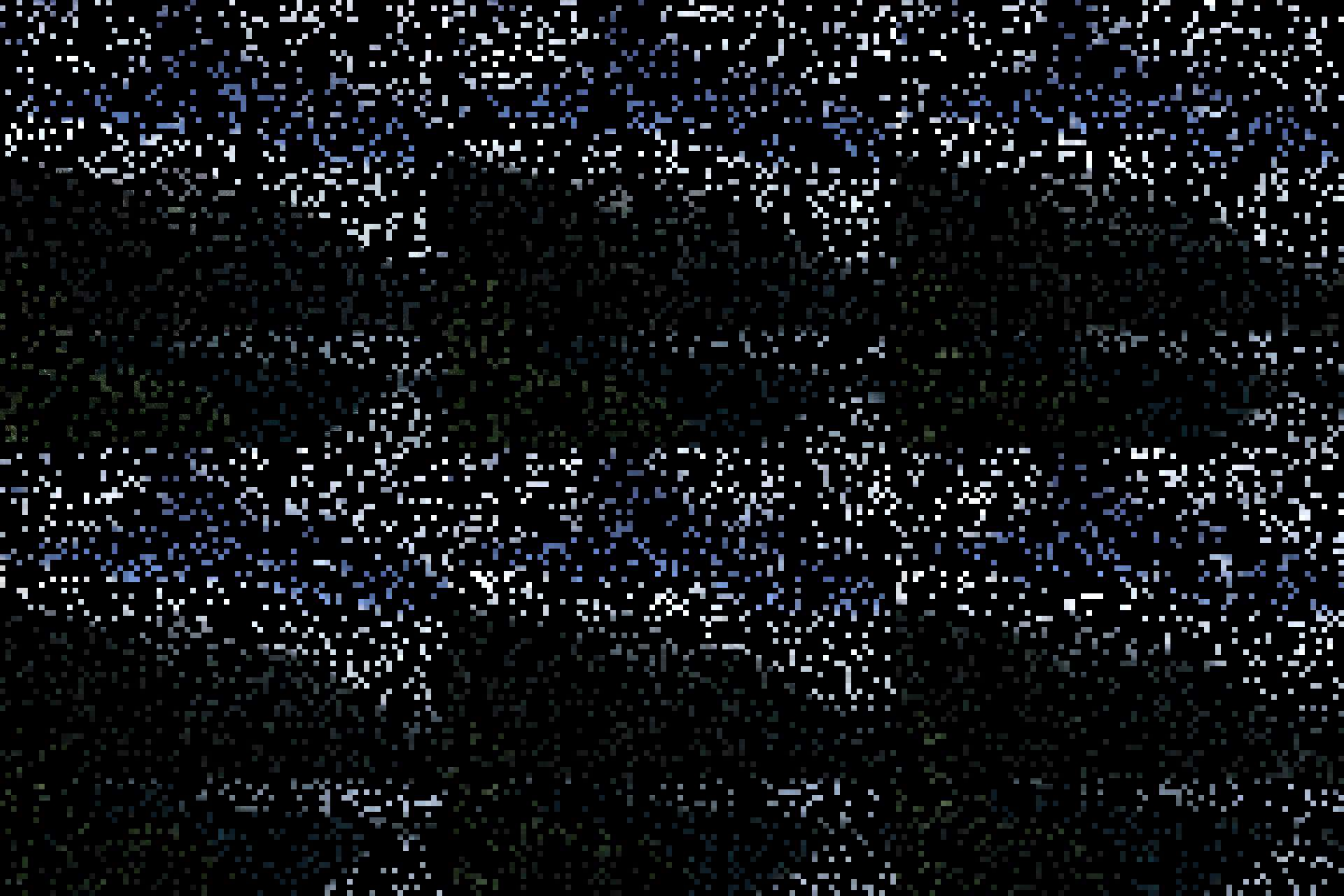}}\hspace{0.005in}
\subfloat{\includegraphics[width=.3\linewidth]{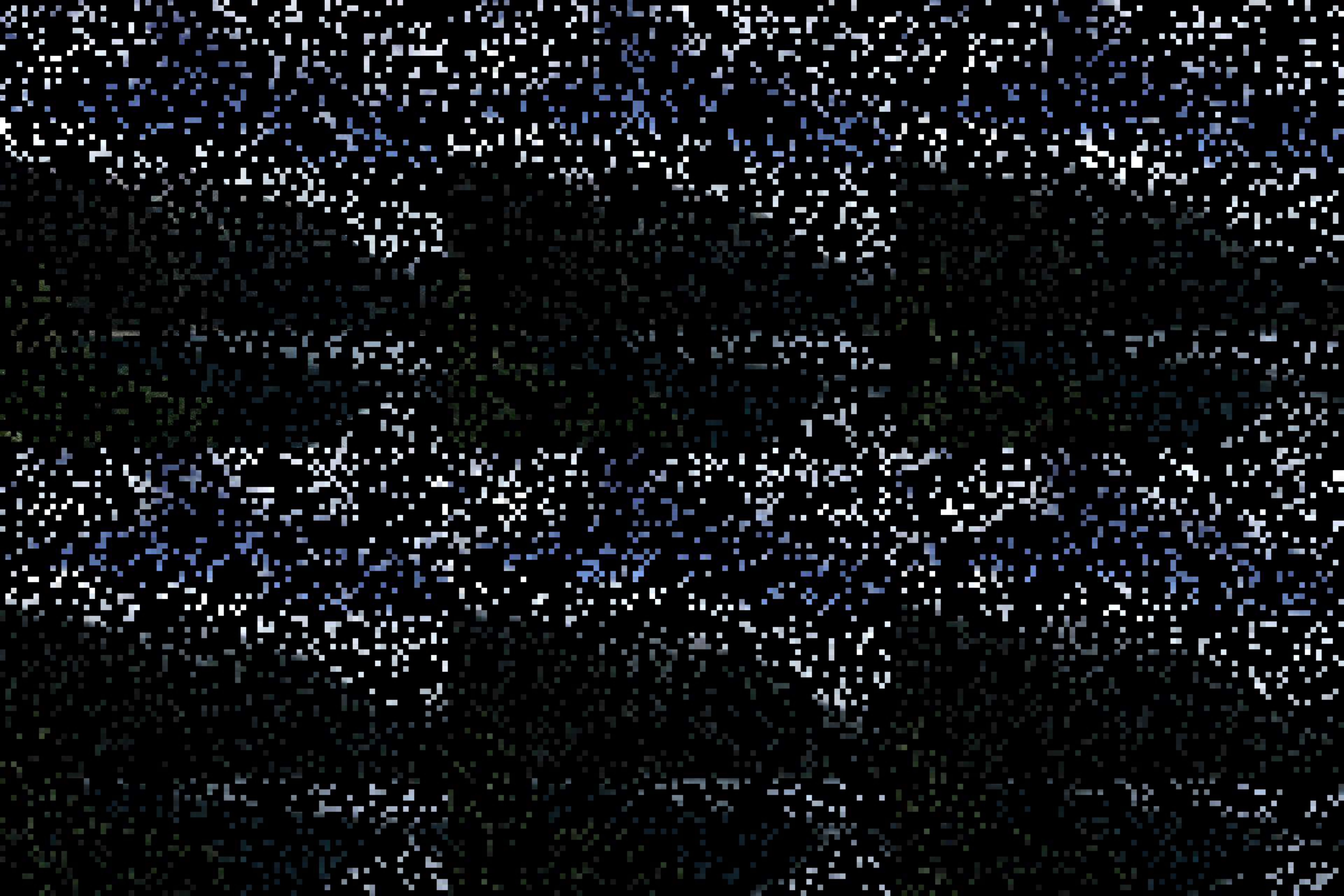}}\\
\vspace{-0.1in}
\subfloat[Regime 1]{\includegraphics[width=.3\linewidth]{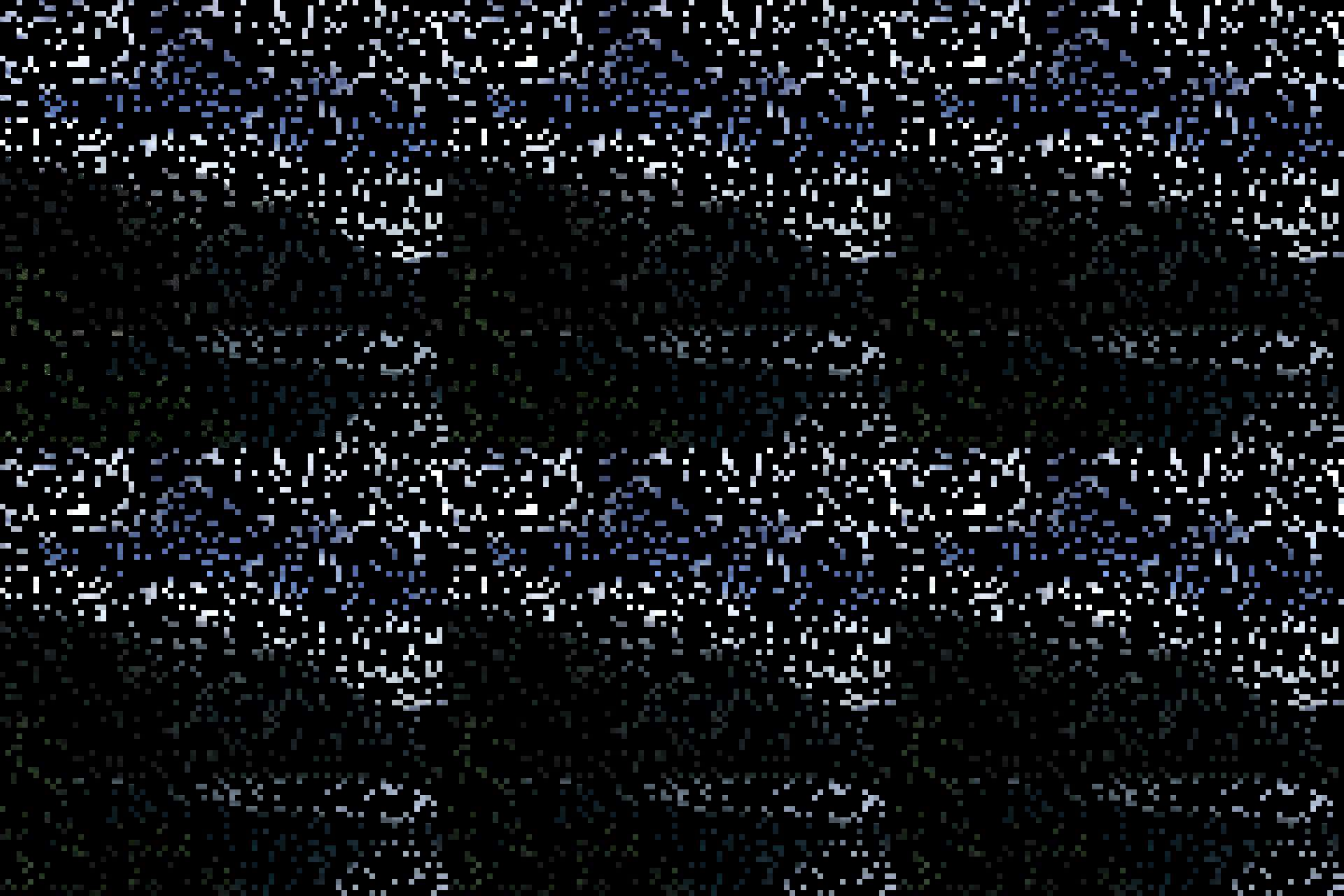}}\hspace{0.005in}
\subfloat[Regime 2]{\includegraphics[width=.3\linewidth]{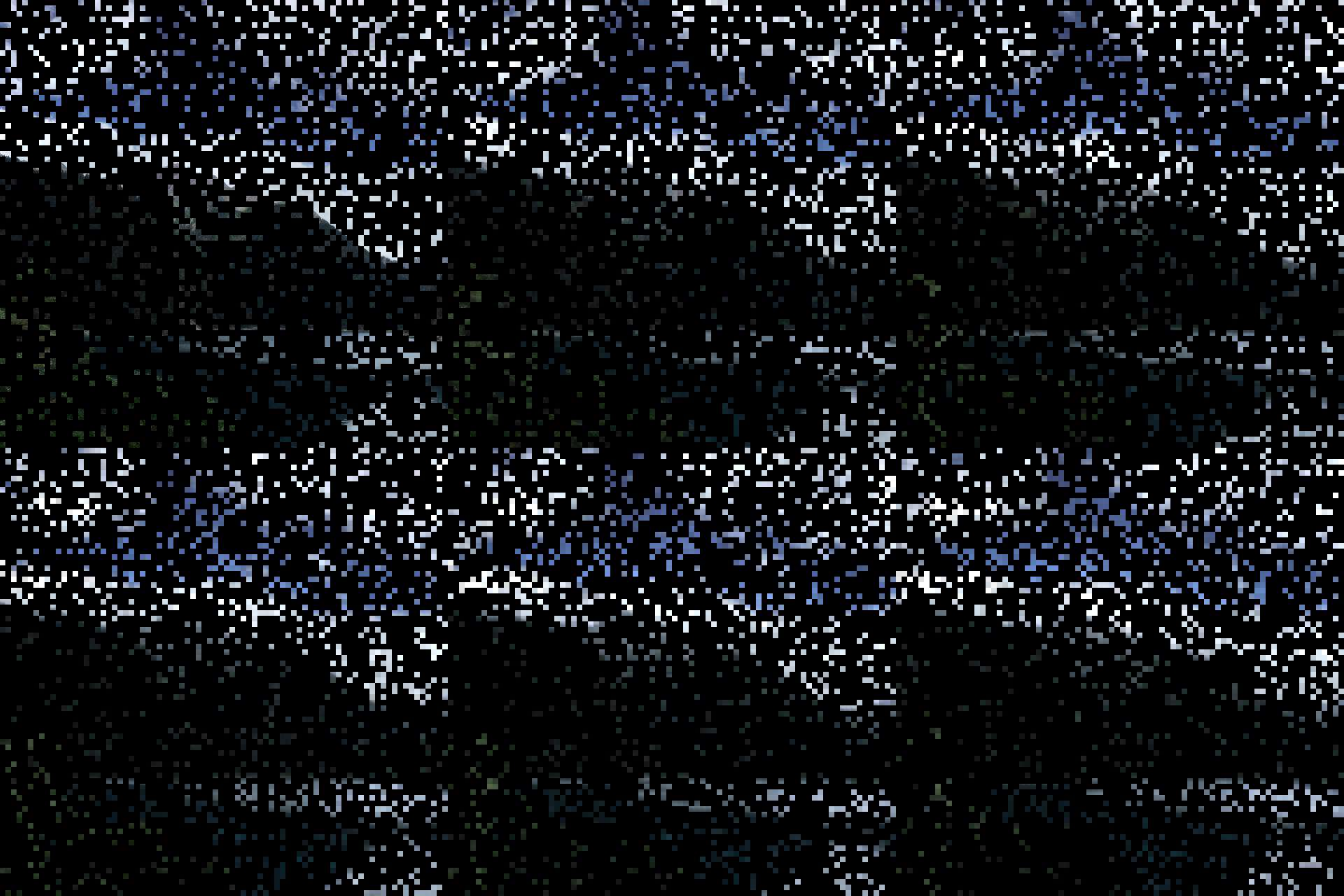}}\hspace{0.005in}
\subfloat[Regime 3]{\includegraphics[width=.3\linewidth]{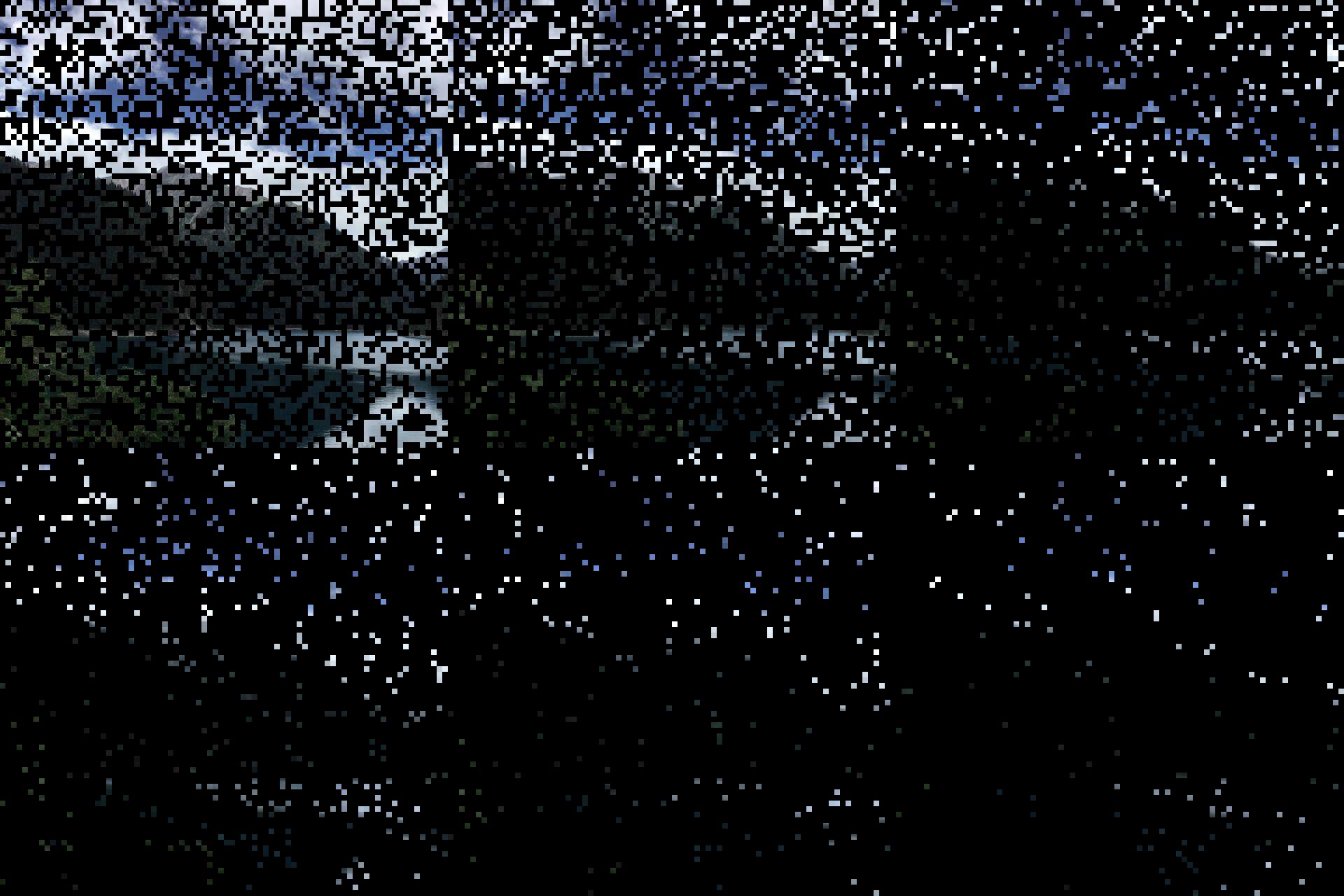}}
\caption{\footnotesize 
Samples  of the  real image by using these three sampling regimes: 
\textbf{Top row:} uniform sampling. \textbf{Bottom row:} optimal sampling.
} \label{fig:samples}
\vspace{-0.1in}
\end{figure}

\begin{figure}[h]
\centering
\subfloat{\includegraphics[width=.2\linewidth]{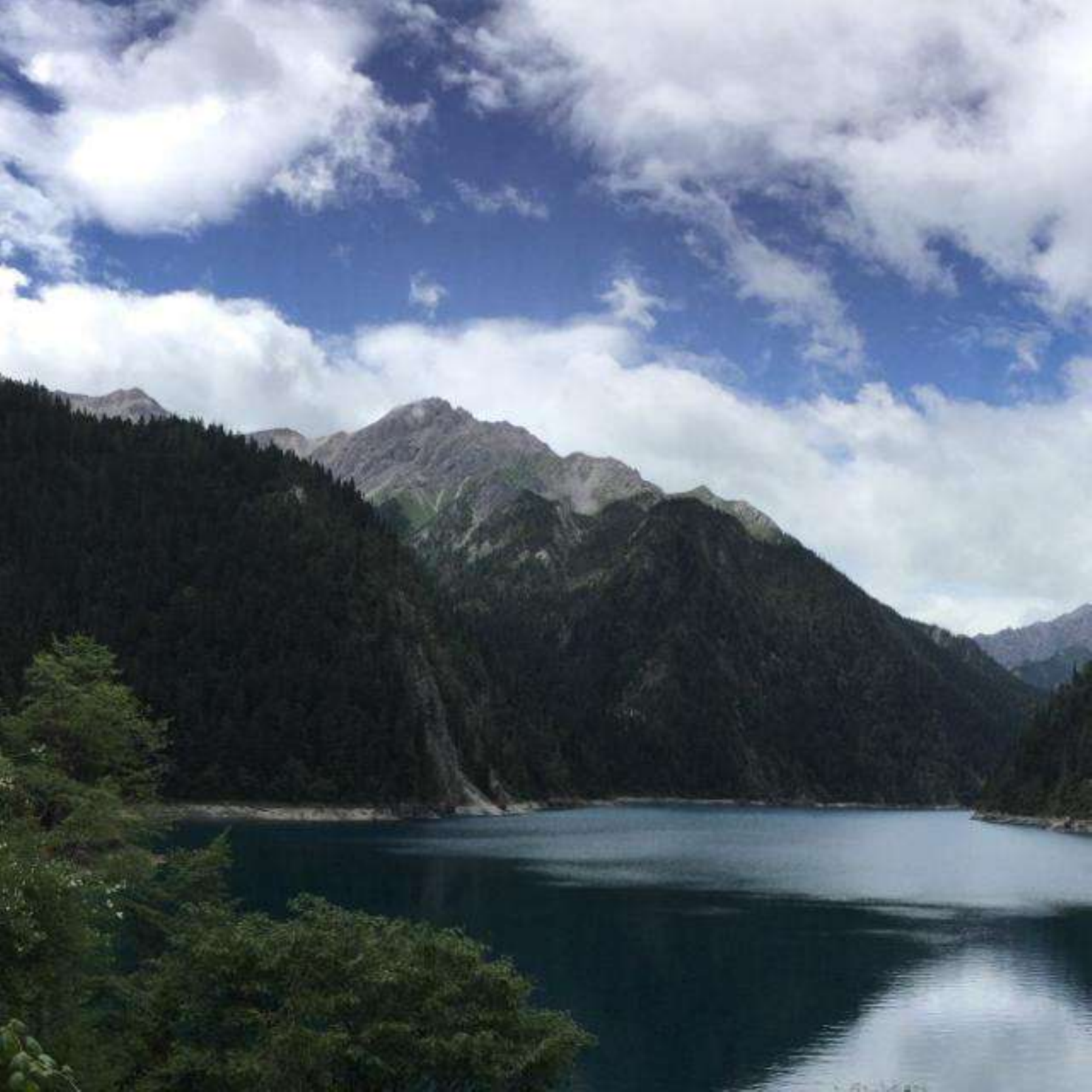}} \hspace{0.005in}
\subfloat{\includegraphics[width=.2\linewidth]{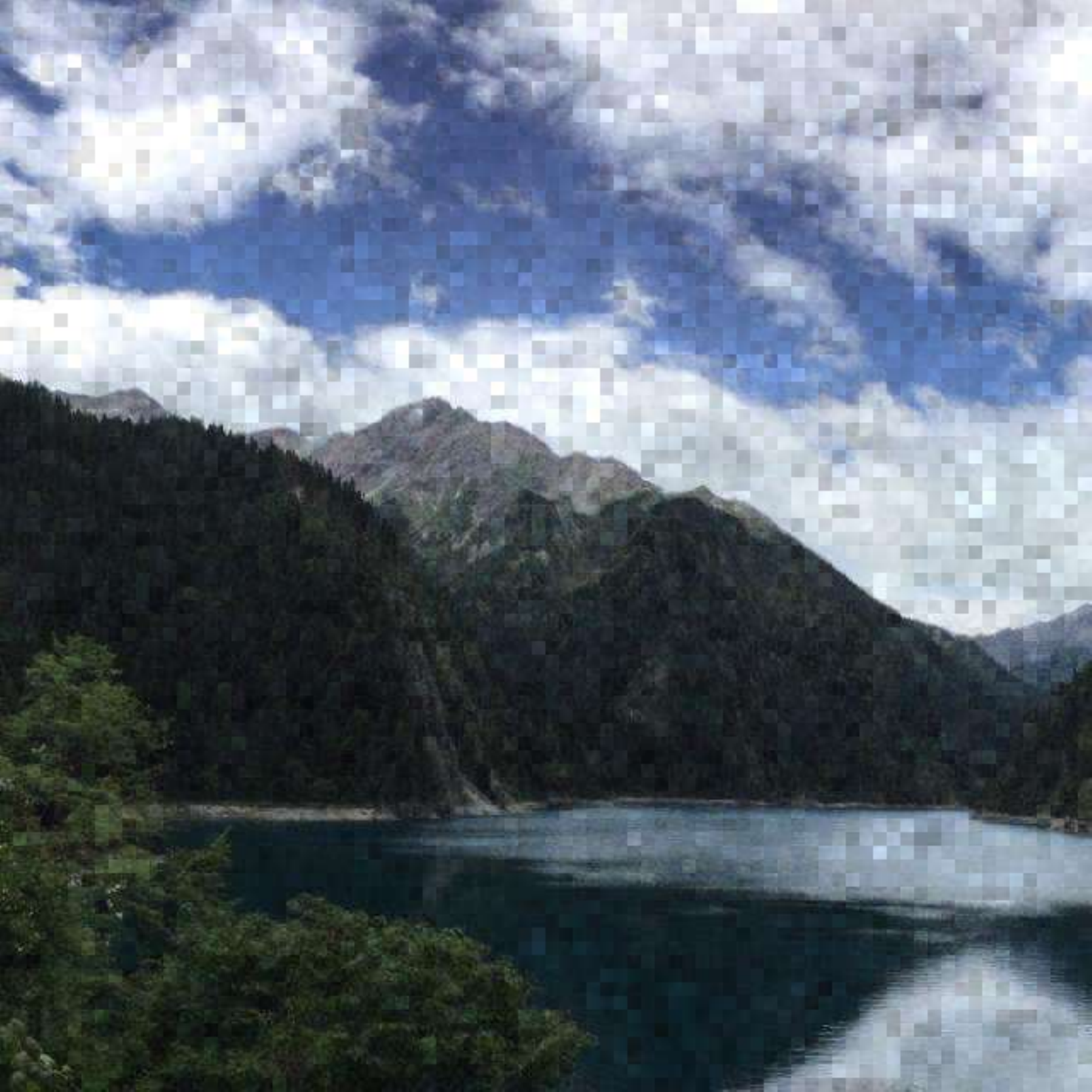}}\hspace{0.005in}
\subfloat{\includegraphics[width=.2\linewidth]{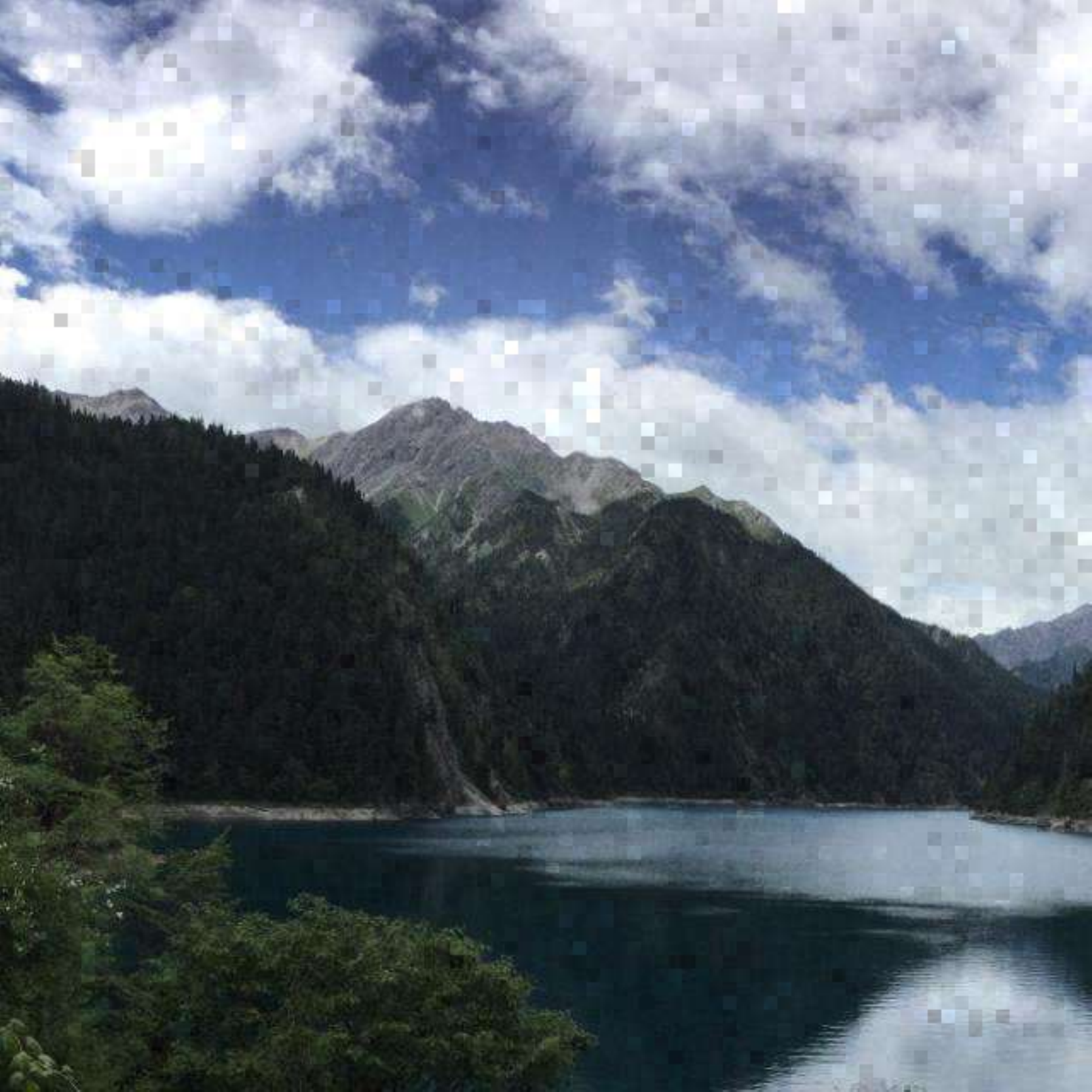}}\hspace{0.005in}
\subfloat{\includegraphics[width=.2\linewidth]{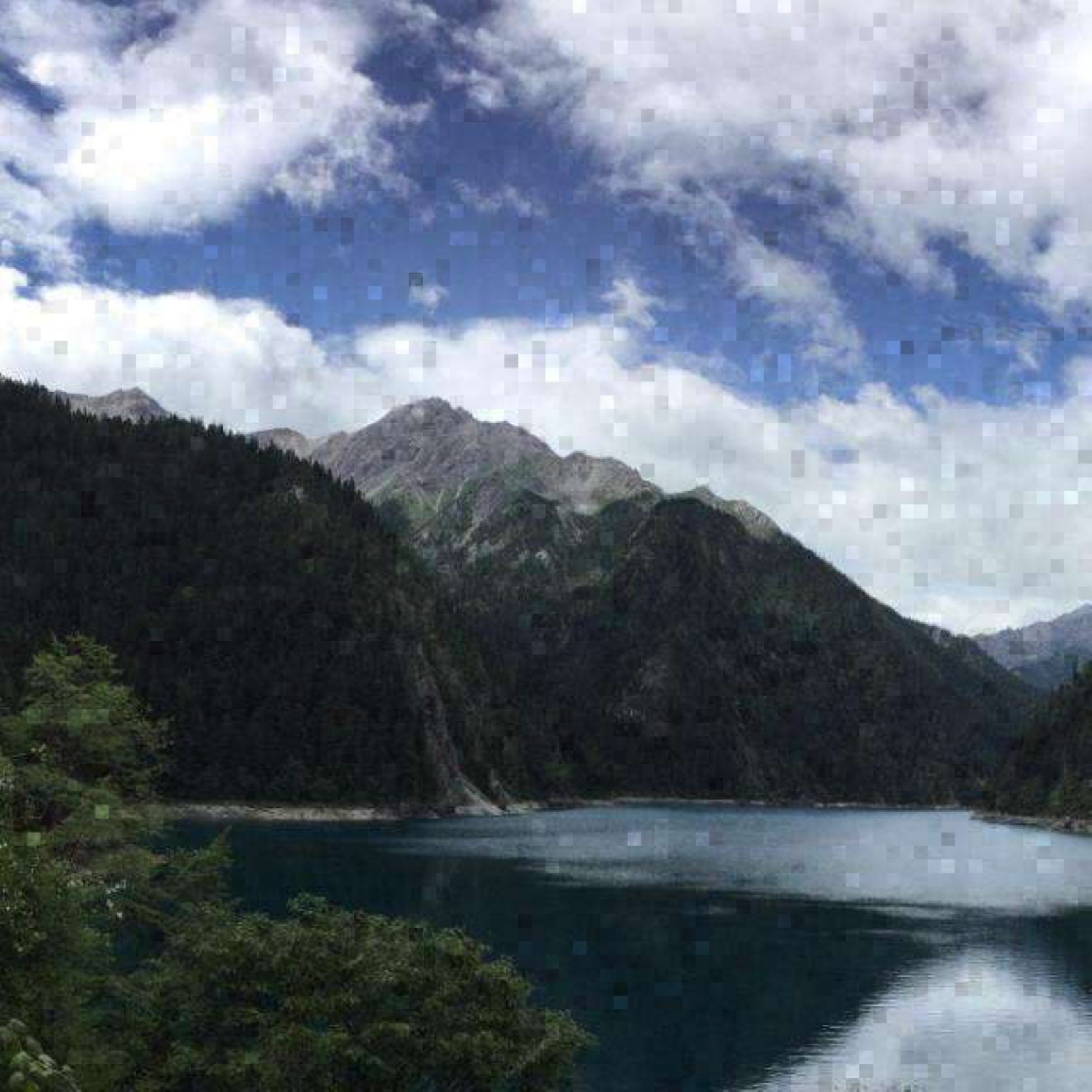}}
\vspace{-.1in}
\subfloat[Original]{\includegraphics[width=.2\linewidth]{figs/Jiuzhaigou.pdf}} \hspace{0.005in}
\subfloat[Regime 1]{\includegraphics[width=.2\linewidth]{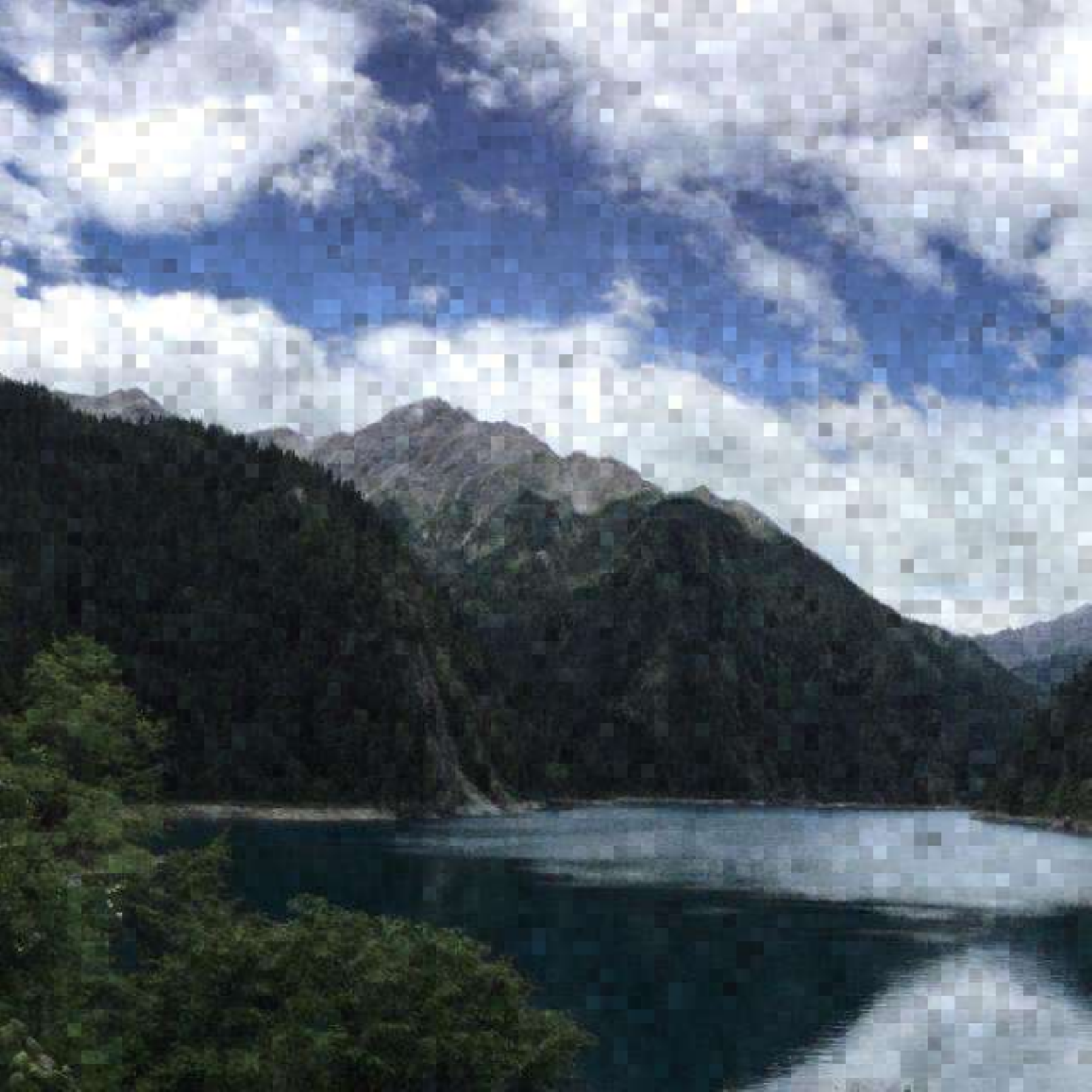}}\hspace{0.005in}
\subfloat[Regime 2]{\includegraphics[width=.2\linewidth]{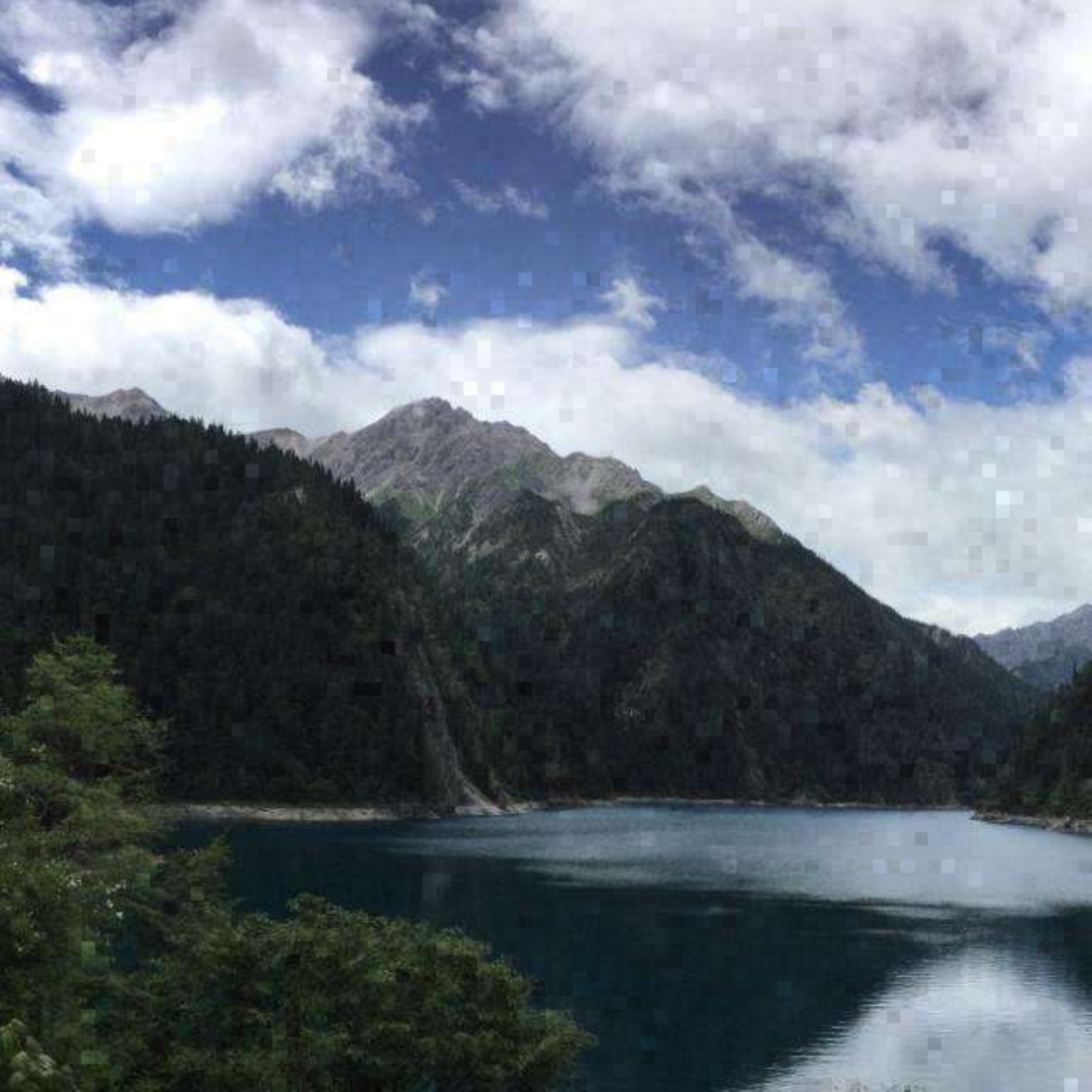}}\hspace{0.005in}
\subfloat[Regime 3]{\includegraphics[width=.2\linewidth]{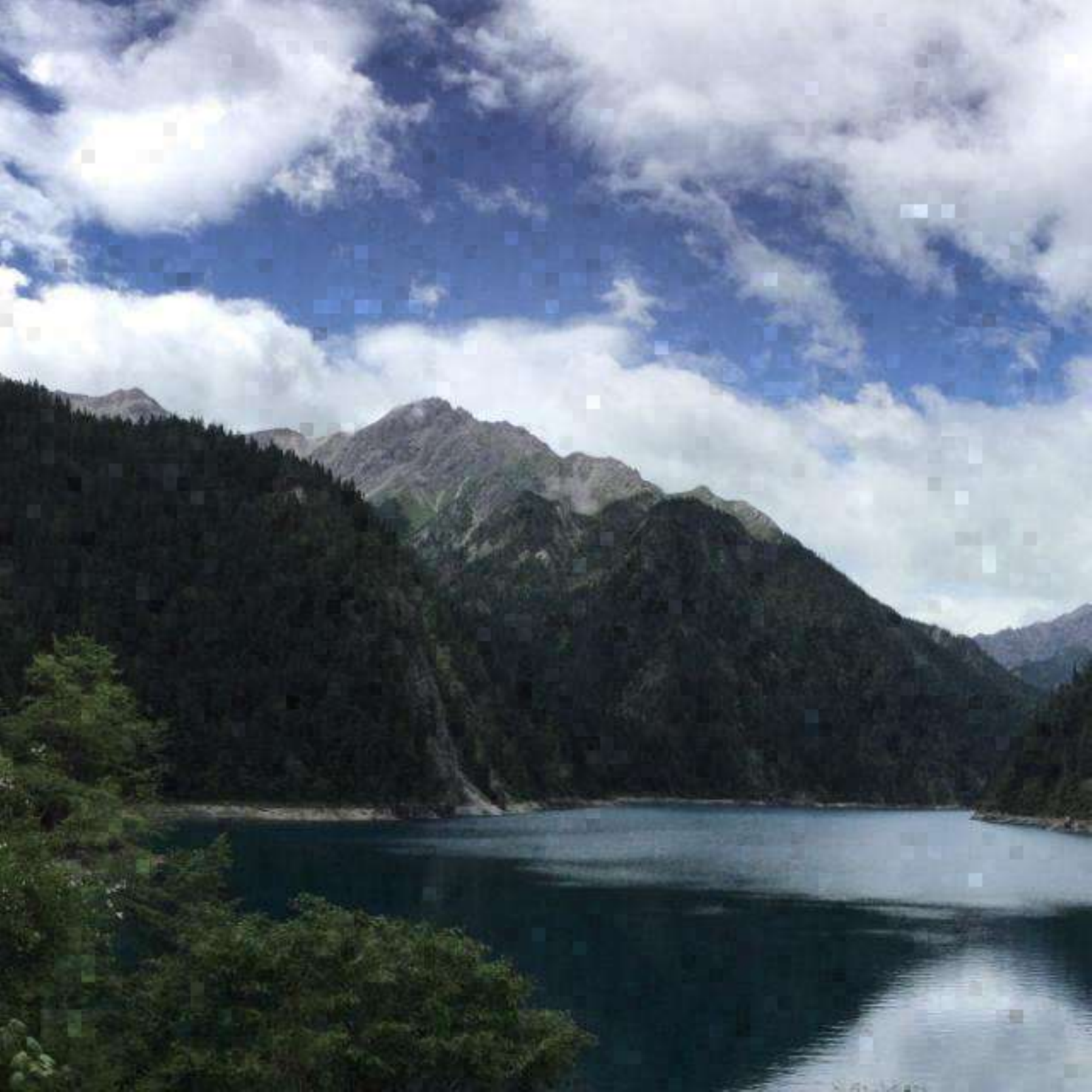}}
\caption{\footnotesize
Comparisons of the reconstruction   of the  real image by using these three sampling regimes with regularizer $\gamma=10^{-5}$: 
\textbf{Top row:} Uniform sampling. \textbf{Bottom row:} Optimal sampling.
} \label{fig:real_image}
\vspace{-0.1in}
\end{figure}

\begin{table}[!h]
\caption{\footnotesize
The SNR between the original and the reconstructed image: The  number of samples per time instance for regime 1,2,3 is 1200 and the total number of samples for static setting is 4800. }\label{tab:SNRs}
 \centering
 \begin{tabular}{ |c||c|c|c|c|} 
\hline
 ~             & \multicolumn{4}{c|}{ $\gamma =10^{-5} $}   \cr
 \cline{2-5}

~             & Regime 1                           & Regime 2                           & Regime 3  & Static   \cr

 \hhline {|=||=|=|=|=|}

\textbf{Uniform}  &$22.59$              & $ 27.92$ & $28.14$ & $21.47$                    \cr
\hline
\textbf{Optimal}    &$22.93$              & $30.18$                  &   $30.87$   & $21.37$  \cr
\hline
\end{tabular}
\end{table}


\bibliographystyle{siamplain}
\bibliography{ref}

\clearpage
\section{Supplementary material}
\subsection{Proofs}
\begin{proof}[Proof of Proposition \ref{prop:gcoherence}]
\begin{enumerate}
    \item [(a)] First let's prove (a.2).
    Notice that 
\begin{align*}
    \sum_{\ell=1}^{k}  \frac{ \lambda_\ell^{2t}}{f_T^2(\lambda_\ell) } =&~  \sum_{\ell=1}^{k} \sum_{i=1}^{n} (U(i,\ell))^2 \frac{ \lambda_\ell^{2t}}{f_T^2(\lambda_\ell) }\\
    =&~\sum_{i=1}^{n} \mathbf{p}_t^{(2)}(i)\cdot  \frac{1}{ \mathbf{p}_t^{(2)}(i)}\|\widetilde{U}_{k,T}^{\top}\delta_{tn+i}\|_2^2
    \leq\left(\nu_{2,\mathbf{p}^{(2)}}^{k,T}(t)\right)^2\cdot\sum_{i=1}^{n} \mathbf{p}_t^{(2)}(i), 
\end{align*}
the equation holds when $\frac{1}{ \mathbf{p}_t^{(2)}(i)}\|\widetilde{U}_{k,T}^{\top}\delta_{tn+i}\|_2^2=\left(\nu_{2,\mathbf{p}^{(2)}}^{k,T}(t)\right)^2$ for all $i\in[n]$.
Recall that $\sum_{i=1}^{n}\mathbf{p}_t^{(2)}(i)=1$, then we have $\left(\nu_{2,\mathbf{p}^{(2)}}^{k,T}(t)\right)^2 \geq  \sum_{\ell=1}^{k}  \frac{ \lambda_\ell^{2t}}{f_T^2(\lambda_\ell) }$. Thus equation holds when $\mathbf{p}_t^{(2)}(i)= (\|\widetilde{U}_{k,T}^{\top}\delta_{tn+i}\|_2^2)/ (\sum_{\ell=1}^{k}  \frac{ \lambda_\ell^{2t}}{f_T^2(\lambda_\ell) })$ for all $i\in[n]$. In particular, we have $\sum_{t=0}^{T-1}\left(\nu_{2,\mathbf{p}^{(2)}}^{k,T}(t)\right)^2 \geq \sum_{t=0}^{T-1} (\sum_{\ell=1}^{k}  \frac{ \lambda_\ell^{2t}}{f_T^2(\lambda_\ell) })=k$.

The proofs for (a.1) and (a.3) are similar to (a.2) by noticing the fact that
\begin{align*}
    \sum_{i=1}^{n}\left\|\widetilde{U}_{k,T}((i:n:Tn),:)\right\|_2^2
    =&~\sum_{i=1}^{n} \mathbf{p}^{(1)}(i)\cdot  \frac{1}{ \mathbf{p}^{(1)}(i)}\left\|\widetilde{U}_{k,T}((i:n:Tn),:)\right\|_2^2\\
    \leq&~\sum_{i=1}^{n} \mathbf{p}^{(1)}(i)\cdot \nu_{1, \mathbf{p}^{(1)}}^{k,T}=\nu_{1,\mathbf{p}^{(1)}}^{k,T}
\end{align*}
 and 
 \begin{align*}
    k=\sum_{i=1}^{nT}  \|\widetilde{U}_{k,T}^{\top}\delta_{i}\|_2^2~=&~\sum_{i=1}^{nT} \mathbf{p}^{(3)}(i)\cdot  \frac{1}{ \mathbf{p}^{(3)}(i)}\|\widetilde{U}_{k,T}^{\top}\delta_{i}\|_2^2\\
    \leq&~\sum_{i=1}^{nT} \mathbf{p}^{(3)}(i)\cdot \nu_{3, \mathbf{p}^{(3)}}^{k,T}=\nu_{3,\mathbf{p}^{(3)}}^{k,T}.
\end{align*}
\item [(b)] Let's prove this result when $\mathbf{p}$ is the uniform sampling distribution. The results for the optimal distributions can be derived from part (a) directly.
Let's first  compare $\left(\nu_{1,\mathbf{p}^{(1)}}^{k,T}\right)^2$ and $\left(\nu_{3,\mathbf{p}^{(3)}}^{k,T}\right)^2$. 
   {Recall the definitions of $\nu_{j,\mathbf{p}^{(j)}}^{k,T}$, $\widetilde{U}_{k,T}$, and $F_{k,T}$, $\Lambda_{k}$ in Definition \ref{def:gswc}, \eqref{eqn:tilde_UkL}, and \eqref{map:f_T} respectively, we have}
    \begin{equation*}\label{eqn:rel_coher_13}
    \begin{aligned}
     \left(\nu_{1,\mathbf{p}^{(1)}}^{k,T}\right)^2:=&~ 
     \max_{1\leq i\leq n} \frac{1}{\mathbf{p}^{(1)}(i)}\left\| \widetilde{U}_{k,T}((i:n:Tn),:)  \right\|_2^2\\
     =&~
     \max_{1\leq i\leq n} n\left\| \sum_{t=0}^{T-1}( F_{k,T}\Lambda_k^tU_k^{\top}\delta_{i}\delta_{i}^{\top} U_k \Lambda_k^tF_{k,T}) \right\|_2\\
     \stack{\ding{172}}{\geq}&~  \max_{1\leq i\leq n} \left\{ n\cdot \max_{0\leq t\leq T-1}\left\| F_{k,T}\Lambda_k^tU_k^{\top}\delta_{i}\delta_{i}^{\top} U_k \Lambda_k^tF_{k,T} \right\|_2\right\}\\
     =&~n\max_{1\leq i\leq n}\max_{0\leq t\leq T-1}
     \left\|\left[\frac{\lambda_1^{t}\cdot U_k(i,1)}{f_{T}(\lambda_1)},   \cdots,\frac{\lambda_k^{t}\cdot U_k(i,k)}{f_{T}(\lambda_k)}\right] \right\|_2^2\\
    =&~\frac{1}{T}\left( nT\max_{i,t}\|\widetilde{U}_{k,T}\delta_{tn+i}\|_2^2 \right)
     =\frac{\left(\nu_{3,\mathbf{p}^{(3)}}^{k,T}\right)^2 }{T},
    \end{aligned}
\end{equation*}
where \ding{172} is coming from that $F_{k,T}\Lambda_k^tU_k^{\top}\delta_{i}\delta_{i}^{\top} U_k \Lambda_k^tF_{k,T}$ is  positive semidefinite matrix for all $t$.

Secondly, let's compare $\nu_{3,\mathbf{p}^{(3)}}^{k,T}$ and $\nu_{2,\mathbf{p}^{(2)}}^{k,T}$. Notice that 

\begin{equation}
\begin{aligned}
     \left(\nu_{2,\mathbf{p}^{(2)}}^{k,T}(t)\right)^2 =&~\max_{1\leq i\leq n}  \frac{1}{ \mathbf{p}_t^{(2)}(i)}\|\widetilde{U}_{k,T}^{\top}\delta_{tn+i}\|_2^2
     =\max_{1\leq i\leq n}  n\|\widetilde{U}_{k,T}^{\top}\delta_{tn+i}\|_2^2\\
     \leq &~ \frac{1}{T}\max_{1\leq j\leq Tn} Tn\|\widetilde{U}_{k,T}^{\top}\delta_{j}\|_2^2
     = \frac{1}{T} \left(\nu_{3,\mathbf{p}}^{k,T}\right)^2
\end{aligned}
\end{equation}
for all $ t=0,\cdots,T-1$.
Thus, the conclusions of part (b) are derived.
\end{enumerate}
\end{proof}

In order to prove Theorem \ref{thm:embedding_all}, we need the lemma from \cite[Theorem 1.1]{tropp2012user}\label{lmm:tropp}:
\begin{lemma}
Let $\{X_k\}$ be   a finite sequence of independent, random, self-adjoint, positive semi-definite matrices with dimension $d$. Assume that each random matrix satisfies  $\lambda_{\max}(X_k) \leq R$ almost surely.
Define
\vspace{-0.1in}
\[\mu_{\min} := \lambda_{\min}(\sum_{i}\mathbb{E}(X_i))
\text{ and } \mu_{\max} := \lambda_{\max} (\sum_{i}\mathbb{E}(X_i)).\]
Then
\vspace{-0.1in}
\[\mathbb{P}\left\{\lambda_{\min}(\sum_{i} X_i)\leq(1-\delta)\mu_{\min} \right\}\leq d\left[ \frac{e^{-\delta}}{(1-\delta)^{1-\delta}}\right]^{\mu_{\min}/R}, \forall \delta\in[0,1]
\]
and
\vspace{-0.1in}
\[\mathbb{P}\left\{\lambda_{\max}(\sum_{i} X_i)\geq(1+\delta)\mu_{\max} \right\}\leq d\left[ \frac{e^{\delta}}{(1+\delta)^{1+\delta}}\right]^{\mu_{\max}/R}, \forall \delta\geq 0.
\]
\end{lemma}
\begin{proof}[Proof of Theorem \ref{thm:embedding_all}] Let's first  prove the result for regime 2.
$\pi_{ A,T}(x)\in\text{span}(\widetilde{U}_{k,T})$, we thus only need to study the property of the matrix $P_{\Omega^{(2)}}^{-1/2}W_{2}^{\frac{1}{2}}S_2\widetilde{U}_{k,T}$. Notice that
\begin{equation*}
    \begin{aligned}
        &   (P_{\Omega^{(2)}}^{-1/2}W_{2}^{\frac{1}{2}}S_2\widetilde{U}_{k,T})^\top P_{\Omega^{(2)}}^{-1/2} W_{2}^{\frac{1}{2}}S_2\widetilde{U}_{k,T}\\
        =& \sum_{t=0}^{T-1}\sum_{i=1}^{m_t} \frac{(F_{k,T}\Lambda_k^tU_k^{\top}\delta_{t,\omega_{t,i}}\delta_{t,\omega_{t,i}}^{\top} U_k \Lambda_k^tF_{k,T}) }{m_t\mathbf{p}_{t}^{(2)}(\omega_{t,i}^{(2)})}.
    \end{aligned}
\end{equation*}
Set
$X_{t,i}= \frac{F_{k,T}\Lambda_k^tU_k^{\top}\delta_{t,\omega_{t,i}}\delta_{t,\omega_{t,i}}^{\top} U_k \Lambda_k^tF_{k,T} }{m_t\mathbf{p}_{t}^{(2)}(\omega_{t,i}^{(2)})}$ 
and $X:= \sum_{t=0}^{T-1}\sum_{i=1}^{m_t} X_{t,i}$. 

Thus $X$ is a sum of $M:=\sum_{t=0}^{T-1}m_t$ independent, random, self-adjoint, positive semi-definite matrices.
Now let's estimate $\mu_{\min}$ and $\mu_{\max}$ by computing $\mathbb{E}(X_{t,i})$:
\begin{equation*}
    \begin{aligned}
     \mathbb{E}(X_{t,i}) 
     =&\mathbb{E}\left(  \frac{F_{k,T}\Lambda_k^tU_k^{\top}\delta_{t,\omega_{t,i}}\delta_{t,\omega_{t,i}}^{\top} U_k \Lambda_k^tF_{k,T} }{m_t\mathbf{p}_{t}^{(2)}(\omega_{t,i})}\right)\\
     =& \sum_{k=1}^{n}\mathbf{p}_{t}^{(2)}(k)\cdot  \frac{(F_{k,T}\Lambda_k^tU_k^{\top}\delta_{t,k}\delta_{t,k}^{\top} U_k \Lambda_k^tF_{k,T}) }{m_t\mathbf{p}_{t}^{(2)}(k)}
      =\frac{1}{m_t} \Lambda_{k}^{2t}F_{k,T}^2.
    \end{aligned}
\end{equation*}
Thus $\sum_{t=0}^{T-1}\sum_{i=1}^{m_t} \mathbb{E}(X_{t,i})=
\diag(\frac{\sum_{t=0}^{T-1} \lambda_1^{2t}}{f_{T}^2(\lambda_1)},\frac{\sum_{t=0}^{T-1} \lambda_2^{2t}}{f_{T}^2(\lambda_2)},\cdots,\frac{\sum_{t=0}^{T-1}  \lambda_k^{2t}}{f_{T}^2(\lambda_k)})=I_k$.
Therefore, $\mu_{\min}=\mu_{\max}=1$. 
Additionally, for all $X_{t,i}$, we have
\begin{equation*}
    \begin{aligned}
     \lambda_{\max}(X_{t,i}) =&\left\|X_{t,i}\right\|_2
     =\left\| \frac{F_{k,T}\Lambda_k^tU_k^{\top}\delta_{t,\omega_{t,i}}\delta_{t,\omega_{t,i}}^{\top} U_k \Lambda_k^tF_{k,T} }{m_t\mathbf{p}_{t}^{(2)}(\omega_{t,i})}\right\|_2\\
      \leq&     \max_{0\leq t\leq T-1} \max_{1\leq i\leq n}  \frac{\|F_{k,T}\Lambda_k^tU_k^{\top}\delta_{t,i}\|_2^2}{m_t\mathbf{p}_{t}^{(2)}(i)}=\max_{0\leq t\leq T-1}\frac{\left(\nu_{2,\mathbf{p}_t^{(2)}}^{k,T}(t)\right)^2}{m_{t}}=:\nu.
    \end{aligned}.
\end{equation*}
 Lemma \ref{lmm:tropp} yields, for any $\delta\in(0,1)$, we have
 $\mathbb{P}\left\{\lambda_{\min}(X)\leq 1-\delta   \right\}\leq k\left[ \frac{e^{-\delta}}{(1-\delta)^{1-\delta}}\right]^{1/\nu}\leq k\exp\left( -\frac{\delta^2}{3\nu}\right)$
and $\mathbb{P}\left\{\lambda_{\max}(X)\geq 1+\delta\right\}\leq k\left[ \frac{e^{\delta}}{(1+\delta)^{1+\delta}}\right]^{1/\nu}\leq k\exp\left( -\frac{\delta^2}{3\nu }\right)$ for all $ \delta\in[0,1]$,
here we have used the fact that $\delta+(1-\delta\log(1-\delta))\geq \frac{\delta^2}{3}$ and $-\delta+(1+\delta\log(1+\delta))\geq \frac{\delta^2}{3}$.
Therefore, for any $\delta\in(0,1)$, 
\begin{equation}
    1-\delta\leq \lambda_{\min}(X)\leq\lambda_{\max}(X)\leq 1+\delta
\end{equation}
holds with probability at least $1-\epsilon$ 
when $\frac{1}{\nu}\geq \frac{3}{\delta^2} \log\left(\frac{2k}{\epsilon} \right)$.
Notice that $m_t\geq \frac{3 }{\delta^2} \left(\nu_{2,\mathbf{p}_t^{(2)}}^{k,T}(t)\right)^2 \log\left(\frac{2k}{\epsilon} \right)$ for all $t=0,\cdots,T-1$ can ensure $\frac{1}{\nu}\geq \frac{3}{\delta^2} \log\left(\frac{2k}{\epsilon} \right)$. Thus, \[ (1-\delta)\|\pi_{ A,T}(x)\|^2\leq  \|P_{\Omega^{(2)}}^{-1/2}W_{2}^{\frac{1}{2}}S_2\pi_{ A,T}(x)\|_2^2\leq (1+\delta) \|\pi_{ A,T}(x)\|_2^2
\]
holds  with probability $1-\epsilon$, provided that $m_{t}\geq \frac{3}{\delta^2} \left(\nu_{2,\mathbf{p}_t^{(2)}}^{k,T}(t)\right)^2\log\left(\frac{2k}{\epsilon} \right)$.
Combining Lemma \ref{lmm: embedding}, we derive the conclusion.

For the results of regime 1 and regime 3, notice that
\begin{equation*}
    \begin{aligned}
         (P_{\Omega^{(1)}}^{-1/2}W_{1}^{\frac{1}{2}}S_1\widetilde{U}_{k,T})^\top P_{\Omega^{(1)}}^{-1/2} W_{1}^{\frac{1}{2}}S_1\widetilde{U}_{k,T}
        =&  \sum_{i=1}^{m}\frac{1}{m\mathbf{p}^{(1)}(i)}\sum_{t=0}^{T-1}F_{k,T}\Lambda_k^tU_k^{\top}\delta_{i}\delta_{i}^{\top} U_k \Lambda_k^tF_{k,T}
    \end{aligned}
\end{equation*}
and 
\begin{equation*}
    \begin{aligned}
         (P_{\Omega^{(3)}}^{-1/2}W_{3}^{\frac{1}{2}}S_3\widetilde{U}_{k,T})^\top P_{\Omega^{(3)}}^{-1/2} W_{3}^{\frac{1}{2}}S_3\widetilde{U}_{k,T}
        =&  \sum_{i=1}^{m}\frac{1}{m\mathbf{p}^{(3)}(i)}\widetilde{U}_{k,T}^{\top}\delta_{i}\delta_{i}^{\top} \widetilde{U}_{k,T}
    \end{aligned}
\end{equation*}
Set $X_i^{(1)}=\frac{1}{m\mathbf{p}^{(1)}(i)}\sum_{t=0}^{T-1}F_{k,T}\Lambda_k^tU_k^{\top}\delta_{i}\delta_{i}^{\top} U_k \Lambda_k^tF_{k,T}$, $X^{(1)}:=\sum_{i=1}^{m}X_i^{(1)}$, $X_i^{(3)}=\frac{1}{m\mathbf{p}^{(3)}(i)}\widetilde{U}_{k,T}^{\top}\delta_{i}\delta_{i}^{\top} \widetilde{U}_{k,T}$, and $X^{(3)}:=\sum_{i=1}^{m}X_i^{(3)}$. Then $X^{(1)}$ and $X^{(3)}$ are sum of $m$ independent, random, self-adjoint, positive semi-definite matrices. Similar to the proof for  regime 2, we can get the results for regime 1 and 3.
\end{proof}

\begin{proof}[Proof of Theorem \ref{thm: err_dirct}]
Note that  $ {\mathbf{x}}^*$ is a solution to \eqref{eqn:obj1}. Thus by optimality of ${\mathbf{x}}^*$, we obtain
\begin{equation}\label{eqn:thm_main1_eqn1}
\|P_{\Omega^{(j)}}^{-1/2}W_{j}^{\frac{1}{2}}(S_j\pi_{ A,T}({\mathbf{x}}^*)-\mathbf{y})\|_2\leq\|P_{\Omega^{(j)}}^{-1/2}W_{j}^{\frac{1}{2}}(S_j\pi_{ A,T}(\widetilde{\mathbf {x}})-\mathbf{y})\|
\end{equation}
for any $\widetilde{\mathbf{x}}\in\text{span}(U_k)$. In particular, for $\widetilde{\mathbf{ x}}=\mathbf{x}$, we have
\begin{equation}
 \small \|P_{\Omega^{(j)}}^{-1/2}W_{j}^{\frac{1}{2}}S_j\pi_{ A,T}({\mathbf{x}}^*)-\mathbf{y})\|_2\leq\|P_{\Omega^{(j)}}^{-1/2}W_{j}^{\frac{1}{2}}(S_j\pi_{ A,T}(\mathbf{x})-\mathbf{y})\|=\|P_{\Omega^{(j)}}^{-1/2}W_{j}^{\frac{1}{2}} \mathbf e\|.
\end{equation}
Notice that
\begin{equation}\label{eqn:thm_main1_eqn2}
    \begin{aligned}
    ~&\|P_{\Omega^{(j)}}^{-1/2}W_{j}^{\frac{1}{2}}S_j\pi_{ A,T}({\mathbf{x}}^*)-\mathbf{y})\|_2\\ =&\|P_{\Omega^{(j)}}^{-1/2}W_{j}^{\frac{1}{2}}S_j\pi_{ A,T}({\mathbf{x}}^*)-P_{\Omega^{(j)}}^{-1/2}W_{j}^{\frac{1}{2}}S_j\pi_{ A,T}(\mathbf{x})-P_{\Omega^{(j)}}^{-1/2}W_{j}^{\frac{1}{2}}\mathbf e\|_2\\
    \geq&\|P_{\Omega^{(j)}}^{-1/2}W_{j}^{\frac{1}{2}}S_j\pi_{ A,T}({\mathbf{x}}^*-\mathbf{x})\|_2-\|P_{\Omega^{(j)}}^{-1/2}W_{j}^{\frac{1}{2}}\mathbf e\|_2\\
    \geq&\sqrt{1-\delta}f_{T}(\lambda_k)\|{\mathbf{x}}-\mathbf{x}^*\|_2-\|P_{\Omega^{(j)}}^{-1/2}W_{j}^{\frac{1}{2}}\mathbf e\|_2.
    \end{aligned}
\end{equation}
Combing \eqref{eqn:thm_main1_eqn1}--\eqref{eqn:thm_main1_eqn2}, the following inequality holds 
\[ \|P_{\Omega^{(j)}}^{-1/2}W_{j}^{\frac{1}{2}} \mathbf{e}\|\geq \sqrt{1-\delta}f_{T}(\lambda_k)\|{\mathbf{x}}-\mathbf{x}^*\|_2-\|P_{\Omega^{(j)}}^{-1/2}W_{j}^{\frac{1}{2}} \mathbf{e}\|_2.\] 
Thus, we have the first bound 
\[\|\mathbf{x}^*- {\mathbf{x}}\|_2\leq \frac{2}{\sqrt{1-\delta}f_{T}(\lambda_k)}\|P_{\Omega^{(j)}}^{-1/2}W_{j}^{\frac{1}{2}}\mathbf e\|_2.\]

To prove the second bound, let us choose $\mathbf e_0=S_j\pi_{ A,T}(\mathbf x_0)$ for some $\mathbf x_0\in\text{span}(U_k)$. Therefore, $\mathbf{y}=S_j(\pi_{ A,T}(\mathbf x_0+\mathbf{x}))$ and ${\mathbf{x}}^*=\mathbf{x}+\mathbf{x}_0$ is a solution for \eqref{eqn:obj1} in this case. By Theorem  \ref{thm:embedding_all},  we have
\begin{equation*}
    \begin{aligned}
    \|{\mathbf{x}}-\mathbf{x}^*\|_2=\|\mathbf x_0\|_2\geq& \frac{1}{\sqrt{1+\delta} f_{T}(\lambda_1)}\|P_{\Omega^{(j)}}^{-1/2}W_{j}S_j\pi_{ A,T}(\mathbf x_0)\|_2\\
    =&\frac{1}{\sqrt{1+\delta} f_{T}(\lambda_1)}\|P_{\Omega^{(j)}}^{-1/2}W_{j}^{\frac{1}{2}}S_j\pi_{ A,T}(\mathbf{x}_0)\|_2\\
    =&\frac{1}{\sqrt{1+\delta} f_{T}(\lambda_1)}\|P_{\Omega^{(j)}}^{-1/2}W_{j}^{\frac{1}{2}}\mathbf e_0\|_2.
    \end{aligned}
\end{equation*}
The proof of Theorem \ref{thm: err_dirct} is completed.
\end{proof}

\begin{proof}[Proof of Theorem \ref{thm:main_regular} ]
Since ${\mathbf{x}}^*$ is a solution to \eqref{eqn:obj2}, we have for all $\widetilde{\mathbf{x}}\in\mathbb{R}^{n}$
\begin{equation}\label{eqn:obj2-1}
\begin{aligned}
   &\|P_{\Omega^{(j)}}^{-1/2}W_{j}^{\frac{1}{2}}(S_j\pi_{ A,T}(\mathbf{x}^*)-\mathbf{y})\|_2^2+\gamma (\mathbf{x}^*)^{T}g(L)\mathbf{x}^*\\
   \leq& \|P_{\Omega^{(j)}}^{-1/2}W_{j}^{\frac{1}{2}}(S_j\pi_{ A,T}(\widetilde{\mathbf{x}})-\mathbf{y})\|_2^2+\gamma \widetilde{\mathbf{x}}^\top g(L)\widetilde{\mathbf{x}}
\end{aligned}.
\end{equation}
 We also have $\mathbf{x}^*=\mathbf \alpha^*+\mathbf \beta^*$ with $\mathbf \alpha^*\in\text{span}(U_k)$ and $\mathbf\beta^*\in\text{span}(U_k)^{\perp}=\text{span}(\overline{U}_k)$, where the matrix $\overline{U}_k$ is defined as
$\overline{U}_k:=(\mathbf u_{k+1},\cdots,\mathbf u_n)\in\mathbb{R}^{n\times (n-k)}$.

Choosing $\widetilde{\mathbf{x}}=\mathbf{x}$ in \eqref{eqn:obj2-1} and using the fact that $U_k^{\top}\beta^*=0$, $\overline{U}_k^{\top}\alpha^*=0$,  $\overline{U}_k^{\top}\mathbf{x}=0$, and the eigen-decomposition of $g(L)$ is  $g(L)=Ug(\Sigma)U^{\top}$, we obtain that
\begin{equation}\label{eqn:reg_pre}
\begin{aligned}
   ~& \|P_{\Omega^{(j)}}^{-1/2}W_{j}^{\frac{1}{2}}(S_j\pi_{ A,T}(\mathbf{x}^*)-\mathbf{y})\|_2^2+\gamma (\mathbf{x}^*)^{T}g(L)\mathbf{x}^* \\
    =&\|P_{\Omega^{(j)}}^{-1/2}W_{j}^{\frac{1}{2}}(S_j\pi_{ A,T}(\mathbf{x}^*)-\mathbf{y})\|_2^2+\gamma (U_k^{\top}\alpha^*)^{T}g( \Sigma_k)(U_k^{\top}\alpha^*)\\
    &~+\gamma (\overline{U}_k^{\top}\beta^*)^{T}g( \overline{\Sigma}_k)(\overline{U}_k^{\top}\beta^*)\\
    \leq& \|P_{\Omega^{(j)}}^{-1/2}W_{j}^{\frac{1}{2}}(S_j\pi_{ A,T}(\mathbf{x})-\mathbf{y})\|_2^2+\gamma \mathbf{x}^{T}g(L)\mathbf{x} \\
      =& \|P_{\Omega^{(j)}}^{-1/2}W_{j}^{\frac{1}{2}} \mathbf{e}\|_2^2+\gamma (U_k^{\top}\mathbf{x})^{T}g( \Sigma_k)(U_k^{\top}\mathbf{x}),
\end{aligned}
\end{equation}
where 
$g( \Sigma_k)=\diag(g(\sigma_1),\cdots,g(\sigma_k))$ and  $g(\overline{\Sigma}_k)=\diag(g(\sigma_{k+1}),\cdots,g(\sigma_n))$.

Notice that 
\[(U_k^{\top}\alpha^*)^{T}g(\Sigma_k)(U_k^{\top}\alpha^*)\geq 0,\] 
\[(\overline{U}_k^{\top}\beta^*)^{T}g(\overline{\Sigma}_k)(\overline{U}_k^{\top}\beta^*)\geq g(\sigma_{k+1})\|\overline{U}_k^{\top}\beta^*\|_2^2=g(\sigma_{k+1})\|\beta^*\|_2^2\]
and 
\[(U_k^{\top}\mathbf{x})^{T}g( \Sigma_k)(U_k^{\top}\mathbf{x})\leq g( \sigma_k)\|U_k^{\top}\mathbf{x}\|_2^2=g( \sigma_k)\|\mathbf{x}\|_2^2.\] 
Combining \eqref{eqn:reg_pre}, we thus have 
\begin{equation}\label{eqn:obj2-2}
\small
\|P_{\Omega^{(j)}}^{-1/2}W_{j}^{\frac{1}{2}}S_j\pi_{ A,T}(\mathbf{x}^*)-\mathbf{y})\|_2^2+\gamma g( \sigma_{k+1})\|\beta^*\|_2^2\leq \|P_{\Omega^{(j)}}^{-1/2}W_{j}^{\frac{1}{2}} \mathbf{e}\|_2^2+\gamma g( \sigma_k)\|\mathbf{x}\|_2^2.
\end{equation}
As the left hand side of   \eqref{eqn:obj2-2} is a sum of two positive quantities, we also have
\begin{equation}\label{eqn:obj2-3}
\begin{aligned}
&\|P_{\Omega^{(j)}}^{-1/2}W_{j}^{\frac{1}{2}}(S_j\pi_{ A,T}(\mathbf{x}^*)-\mathbf{y})\|_2^2\\
\leq &\|P_{\Omega^{(j)}}^{-1/2}W_{j}^{\frac{1}{2}}\mathbf{e}\|_2^2+\gamma g( \sigma_k)\|\mathbf{x}\|_2^2
\leq  (\|P_{\Omega^{(j)}}^{-1/2} W_{j}^{\frac{1}{2}}\mathbf e\|_2+\sqrt{\gamma g( \sigma_k)}\|\mathbf{x}\|_2)^2
\end{aligned}
\end{equation}
and
\begin{equation}\label{eqn:obj2-4}
\gamma g( \sigma_{k+1})\|\beta^*\|_2^2
\leq  (\|P_{\Omega^{(j)}}^{-1/2} W_j^{\frac{1}{2}} \mathbf{e}\|_2+\sqrt{\gamma g( \sigma_k)}\|\mathbf{x}\|_2)^2.
\end{equation}
Thus, we have
$\|P_{\Omega^{(j)}}^{-1/2}W_{j}^{\frac{1}{2}}S_j\pi_{ A,T}(\mathbf{x}^*)-\mathbf{y})\|_2 \leq  \|P_{\Omega^{(j)}}^{-1/2}W_{j}^{\frac{1}{2}}\mathbf{e}\|_2+\sqrt{\gamma g( \sigma_k)}\|\mathbf{x}\|_2$ and 
$ \sqrt{\gamma g( \sigma_{k+1})}\|\beta^*\|_2 \leq  \|P_{\Omega^{(j)}}^{-1/2}W_{j}^{\frac{1}{2}}\mathbf{e}\|_2+\sqrt{\gamma g( \sigma_k)}\|\mathbf{x}\|_2$. 
The second inequality yields  \eqref{eqn:thm:main_regular2}.
It remains to prove \eqref{eqn:thm:main_regular1}. Recall that 
  \[  \begin{aligned}
        &\|P_{\Omega^{(j)}}^{-1/2}W_{j}^{\frac{1}{2}}(S_j\pi_{ A,T}(\mathbf{x}^*)-\mathbf{y})\|_2\\   
        =&\|P_{\Omega^{(j)}}^{-1/2}W_{j}^{\frac{1}{2}}(S_j\pi_{ A,T}(\alpha^*-\mathbf{x}))-P_{\Omega^{(j)}}^{-1/2}W_{j}^{\frac{1}{2}}(S_j\pi_{ A,T}(\beta^*)+\mathbf{e})\|_2\\
         \geq &\|P_{\Omega^{(j)}}^{-1/2}W_{j}^{\frac{1}{2}}(S_j\pi_{ A,T}(\alpha^*-\mathbf{x}))\|_2-\|P_{\Omega^{(j)}}^{-1/2}W_{j}^{\frac{1}{2}}\mathbf{e}\|_2-\|P_{\Omega^{(j)}}^{-1/2}W_{j}^{\frac{1}{2}} S_j\pi_{ A,T}(\beta^*)\|_2\\
        \geq &\sqrt{1-\delta}f_{T}(\lambda_k)\|\alpha^*-\mathbf{x}\|_2-\|P_{\Omega^{(j)}}^{-1/2}W_{j}^{\frac{1}{2}}\mathbf{e}\|_2-M_{\max}\|\pi_{ A,T}(\beta^*)\|_2.
    \end{aligned}
    \]
Since $\beta^*\in\text{span}(\overline{U}_k)$, we have $\|\pi_{ A,T}(\beta^*)\|_2\leq f_{T}(\lambda_{k+1})\|\beta^*\|_2$. Therefore, 
\begin{equation}\label{eqn:obj2-5}
   \begin{aligned}
        &\|P_{\Omega^{(j)}}^{-1/2}W_{j}^{\frac{1}{2}}(S_j\pi_{ A,T}(\mathbf{x}^*)-\mathbf{y})\|_2\\
        \geq &\sqrt{1-\delta}f_{T}(\lambda_{k})\|\alpha^*-\mathbf{x}\|_2-\|P_{\Omega^{(j)}}^{-1/2}W_{j}^{\frac{1}{2}}\mathbf{e}\|_2-M_{\max}f_{T}(\lambda_{k+1})\|\beta^*\|_2.
    \end{aligned}
\end{equation}
Combing \eqref{eqn:obj2-3} and \eqref{eqn:obj2-5}, we have
\[\begin{aligned}
&   \|P_{j}^{-1/2}W_j^{\frac{1}{2}}\mathbf{e}\|_2+\sqrt{\gamma g( \sigma_k)}\|\mathbf{x}\|_2\\
 \geq &   \sqrt{1-\delta}f_{T}(\lambda_{k})\|\alpha^*-\mathbf{x}\|_2-\|P_{j}^{-1/2}W_{j}^{\frac{1}{2}}\mathbf{e}\|_2-M_{\max}f_{T}(\lambda_{k+1})\| \beta^*\|_2\\
 \geq &\sqrt{1-\delta}f_{T}(\lambda_{k})\|\alpha^*-\mathbf{x}\|_2-\|P_{j}^{-1/2}W_{j}^{\frac{1}{2}}\mathbf{e}\|_2-\\
 ~&M_{\max}f_{T}(\lambda_{k+1})\left(
\frac{1}{\sqrt{\gamma g(\sigma_{k+1})}}\|P_{\Omega^{(j)}}^{-1/2}W_{j}\mathbf{e}\|_2+\left( \sqrt{\frac{g(\sigma_{k})}{g(\sigma_{k+1})}}\|\mathbf{x}\|_2 \right)\right).
\end{aligned}
\]
Thus the conclusion \eqref{eqn:thm:main_regular1}
holds.
\end{proof}

\subsection{Additional numerical results}\label{addnum}
In this section, we present results for the  reconstruction errors $\|x-\alpha^*\|_2$ and $\|\beta^*\|$ for all  graphs, in the spirit of Theorem \ref{thm:main_regular}, see Fig.~\ref{fig:rec_bunny_appendix} and \ref{fig:rec_minesota_appendix}, complementary to Fig.~\ref{fig:rec_nonoise}.

 We present mean reconstruction errors of  10 bandlimited signals as a function of $\gamma$. We take $M=200$ noise free space-time samples for three regimes. The back curve represents results with $g({L})={L}$. The blue curves indicate the results with $g({L})={L}^2$. The red curves indicate the results with $g({L})={L}^4$.  The green curve indicates the results with $g({L})=\exp(I-L)$. The first, second and third columns show the results for regime 1, regime 2 and regime 3 using the optimal sampling distribution respectively. The first and second rows show the mean reconstruction errors $\|x-\alpha^*\|$ and $\|\beta^*\|$ respectively. 
 
 \paragraph{Comparison to the static case} We also plot the reconstruction results by changing $T=1$.

\begin{figure}[h]
\centering
\begin{minipage}{.3\linewidth} \centering \small \hspace{2mm} Regime 1  \end{minipage}
\begin{minipage}{.3\linewidth} \centering \small \hspace{2mm}  Regime 2  \end{minipage}
\begin{minipage}{.3\linewidth} \centering \small \hspace{2mm} Regime3  \end{minipage}\\

\includegraphics[width=.3\linewidth, height=0.25\linewidth]{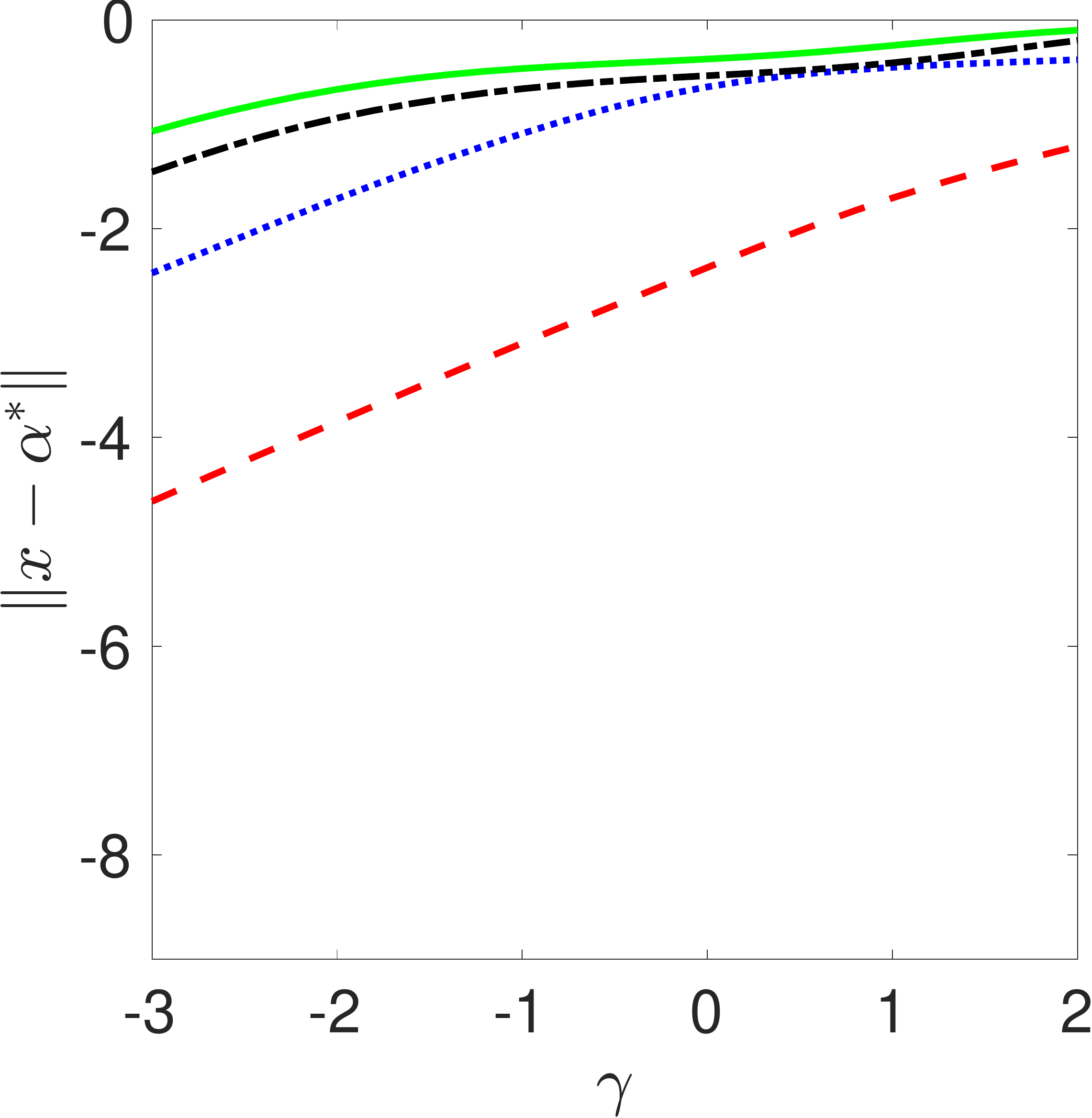}
\includegraphics[width=.3\linewidth, height=0.25\linewidth]{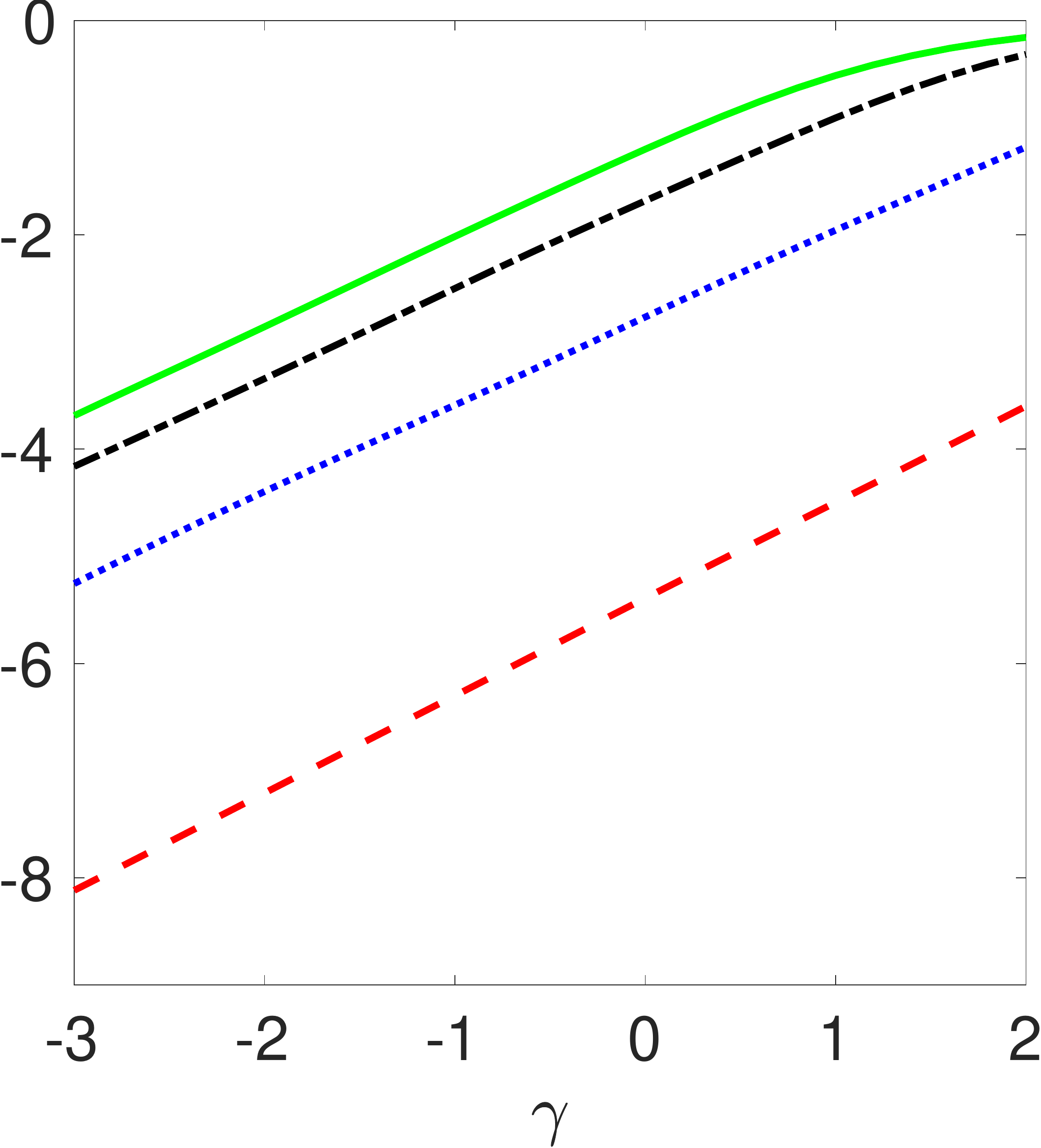}
\includegraphics[width=.3\linewidth, height=0.25\linewidth]{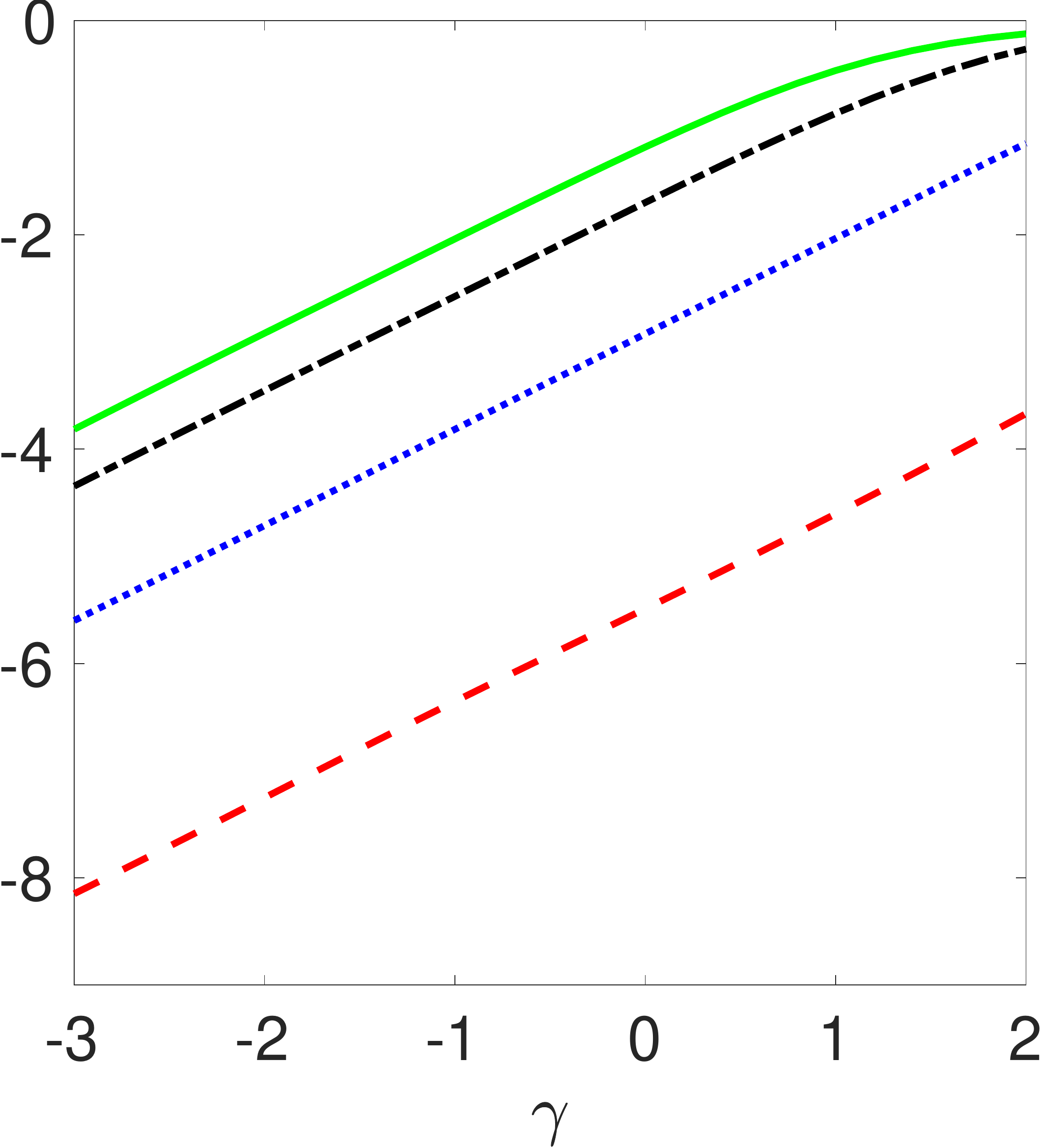}\\
\includegraphics[width=.3\linewidth, height=0.25\linewidth]{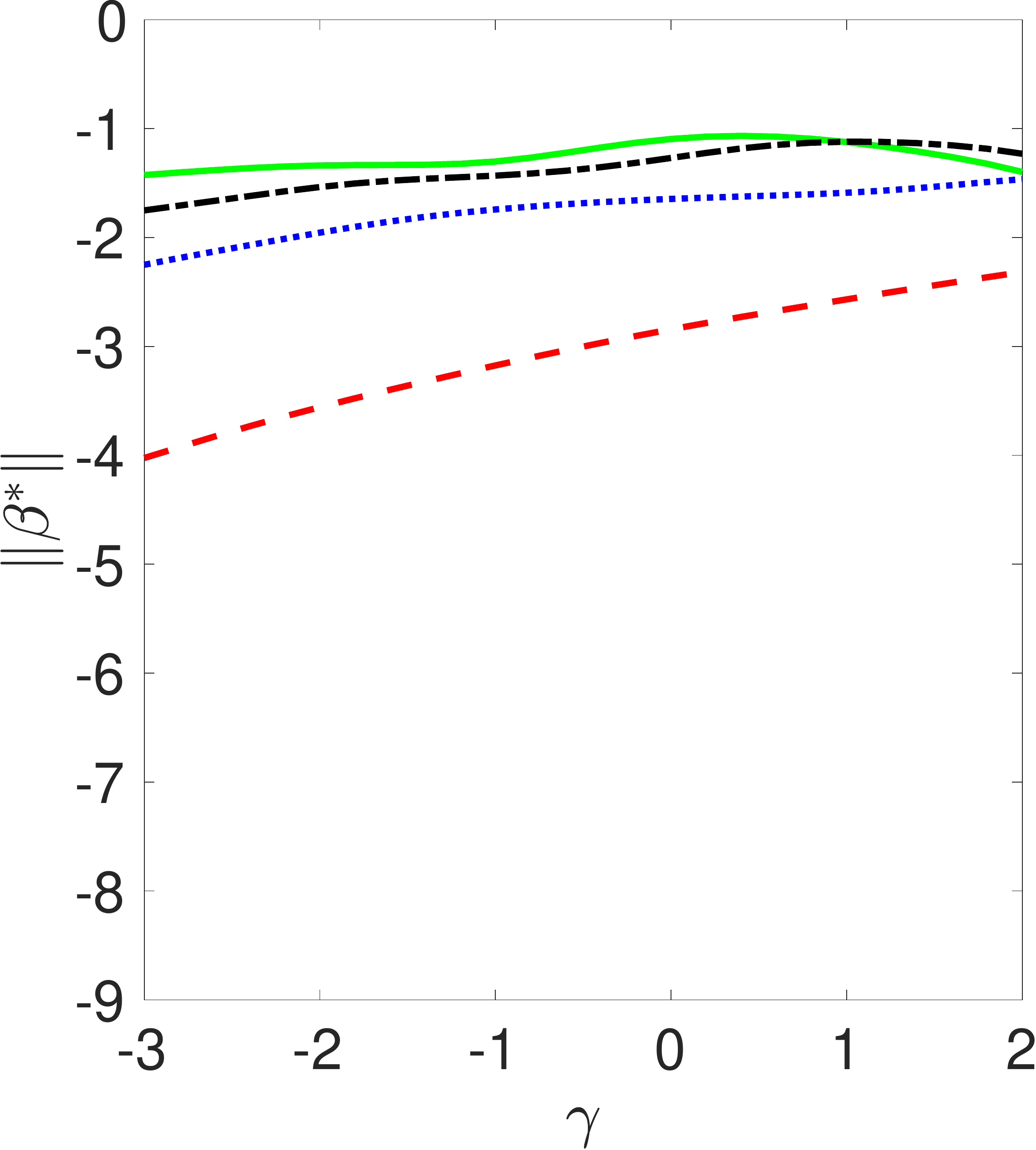}
\includegraphics[width=.3\linewidth, height=0.25\linewidth]{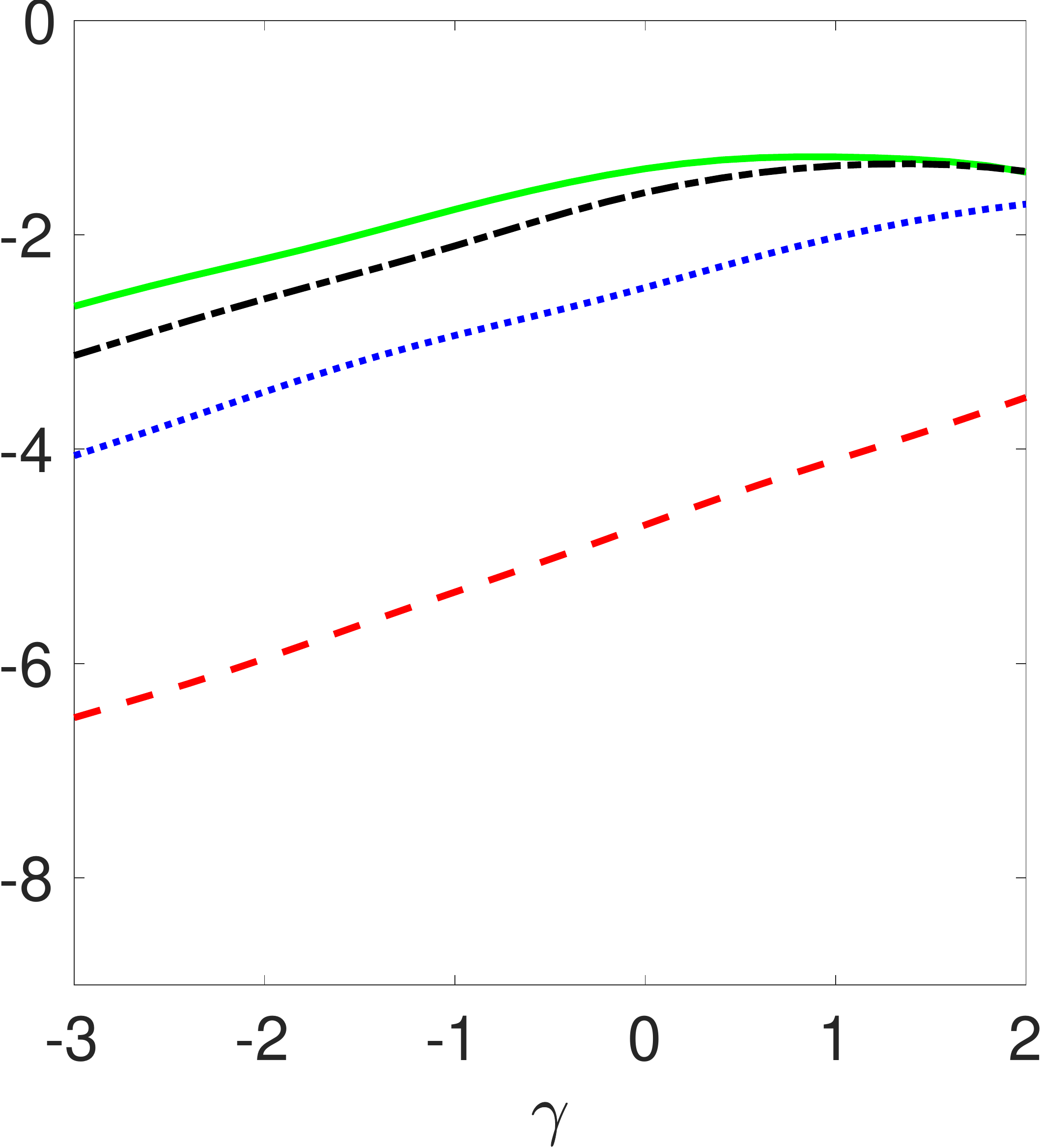}
\includegraphics[width=.3\linewidth, height=0.25\linewidth]{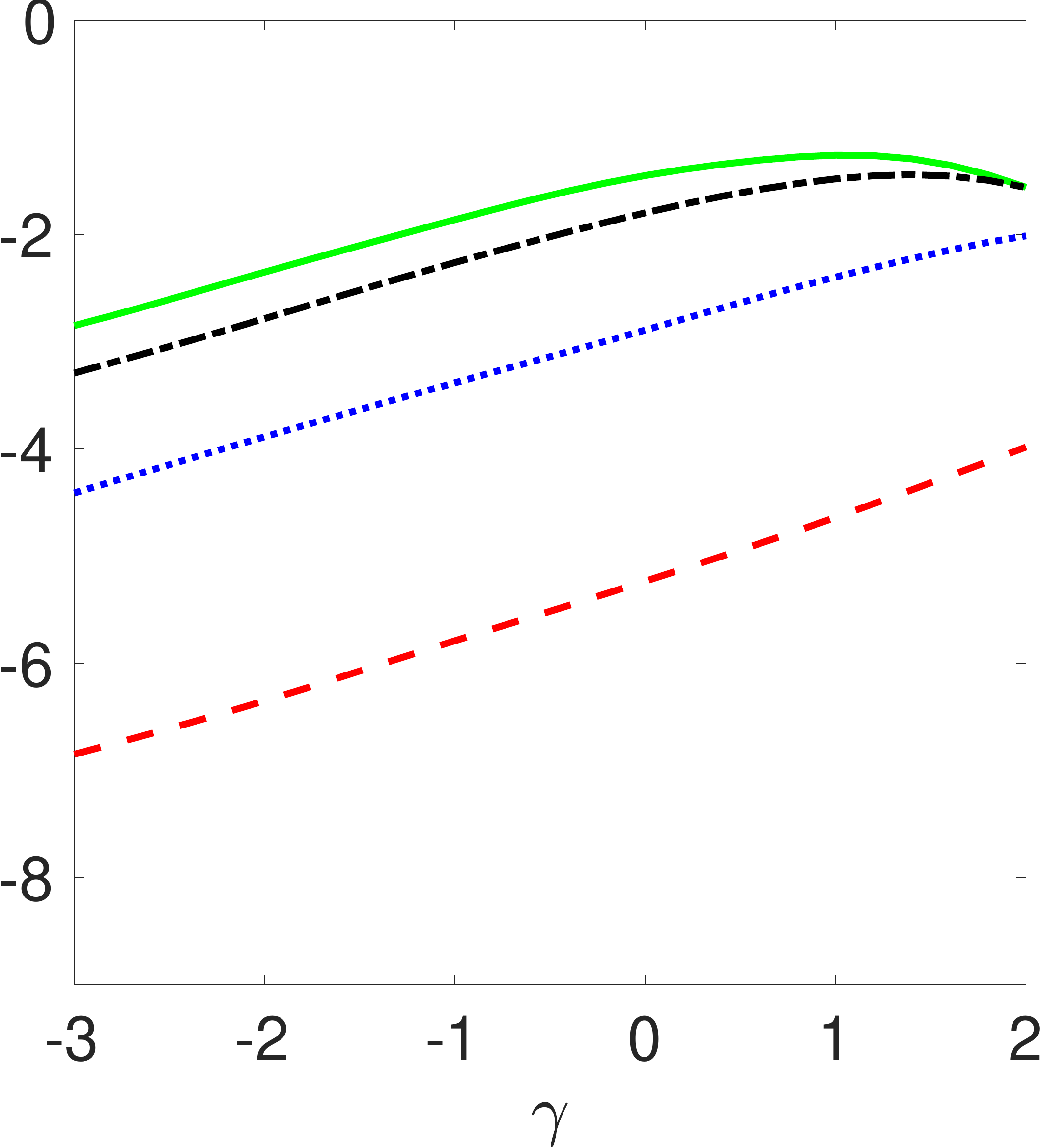}\\
\caption{\footnotesize \textbf{Community graph $C_5$}  }\label{fig:rec_community}
\end{figure}

\begin{figure}[h]
\centering
\begin{minipage}{.3\linewidth} \centering \small \hspace{2mm} Regime 1  \end{minipage}
\begin{minipage}{.3\linewidth} \centering \small \hspace{2mm}  Regime 2  \end{minipage}
\begin{minipage}{.3\linewidth} \centering \small \hspace{2mm} Regime3  \end{minipage}\\
\includegraphics[width=.3\linewidth, height=0.25\linewidth]{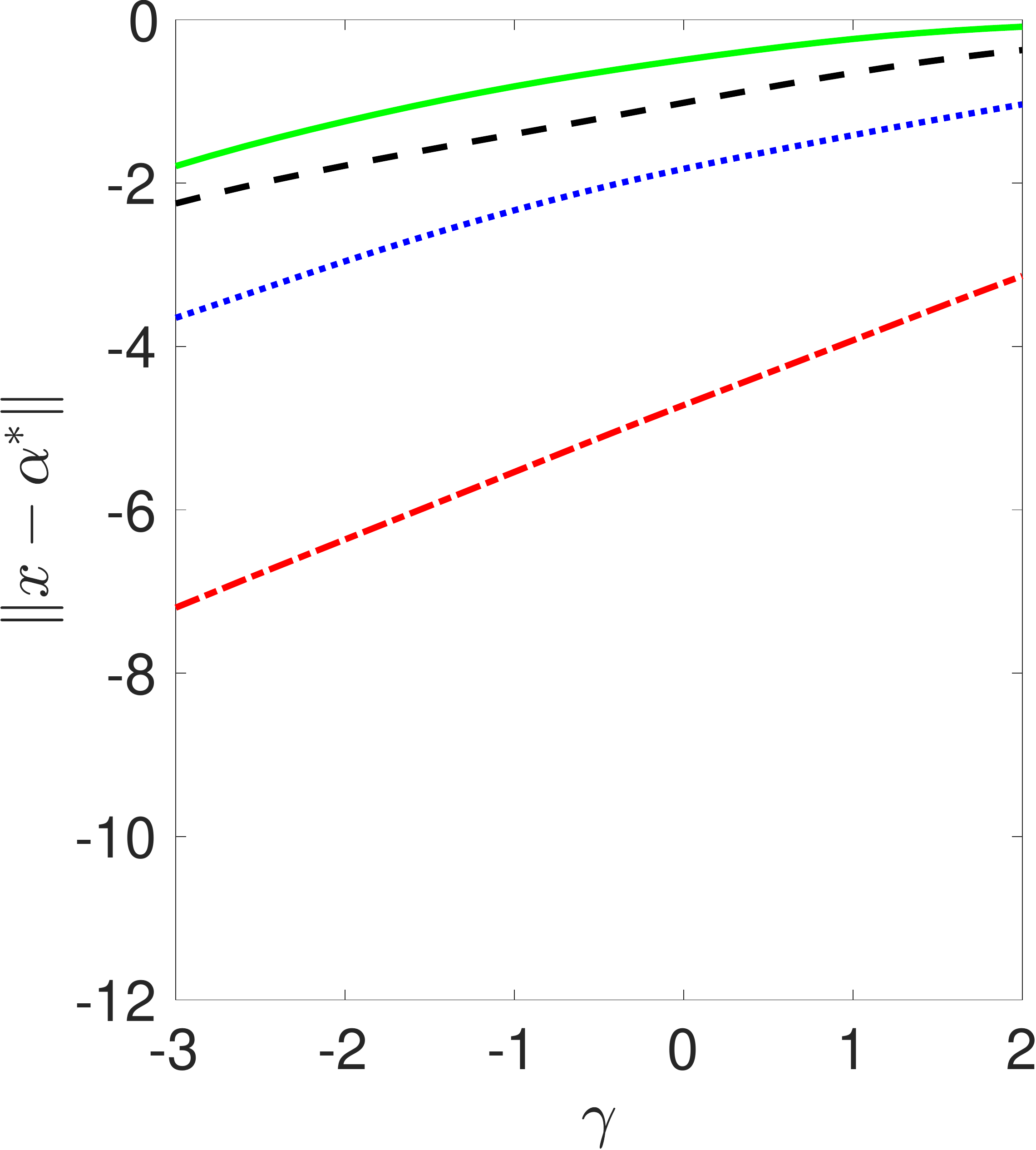}
\includegraphics[width=.3\linewidth, height=0.25\linewidth]{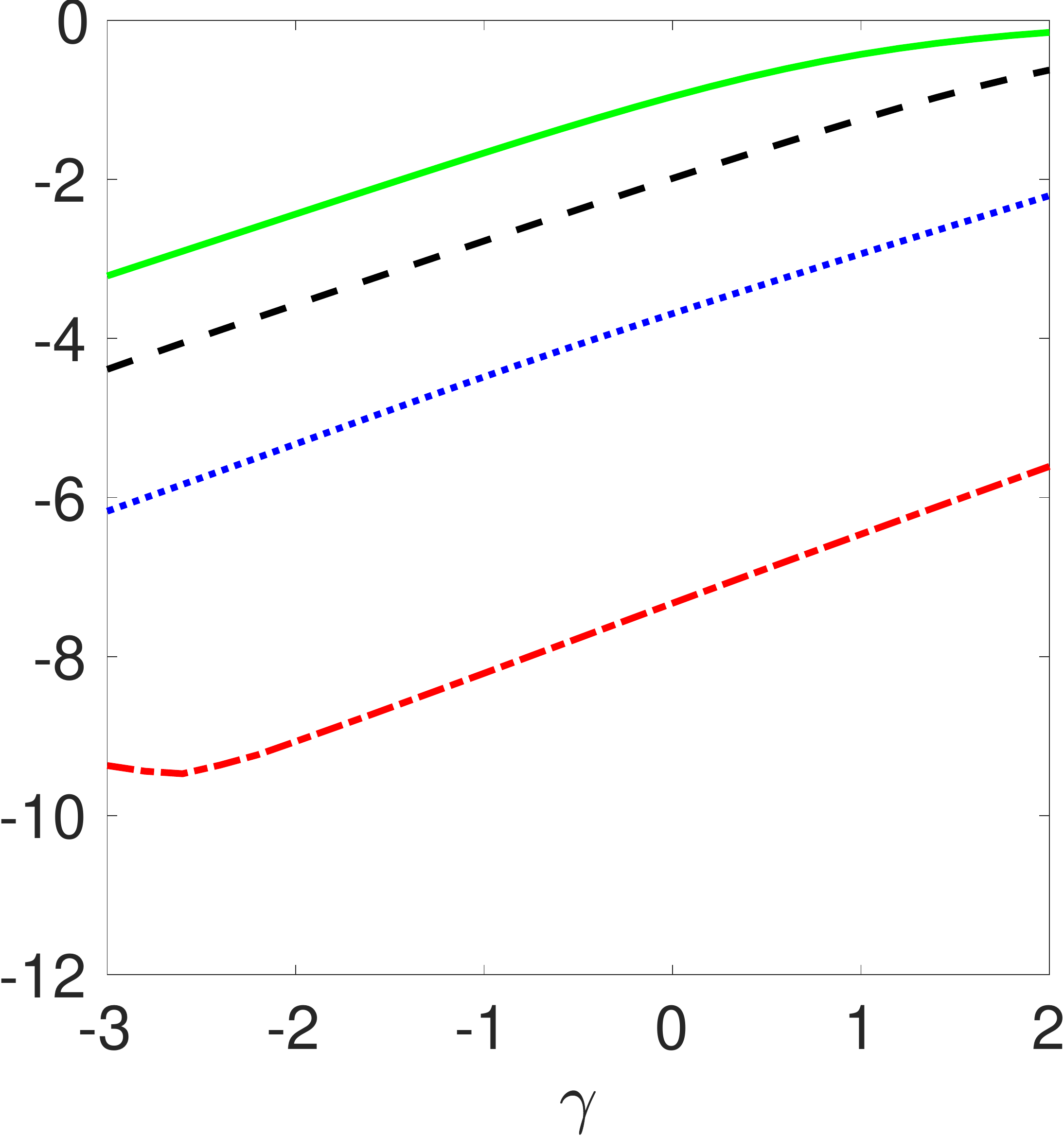}
\includegraphics[width=.3\linewidth, height=0.25\linewidth]{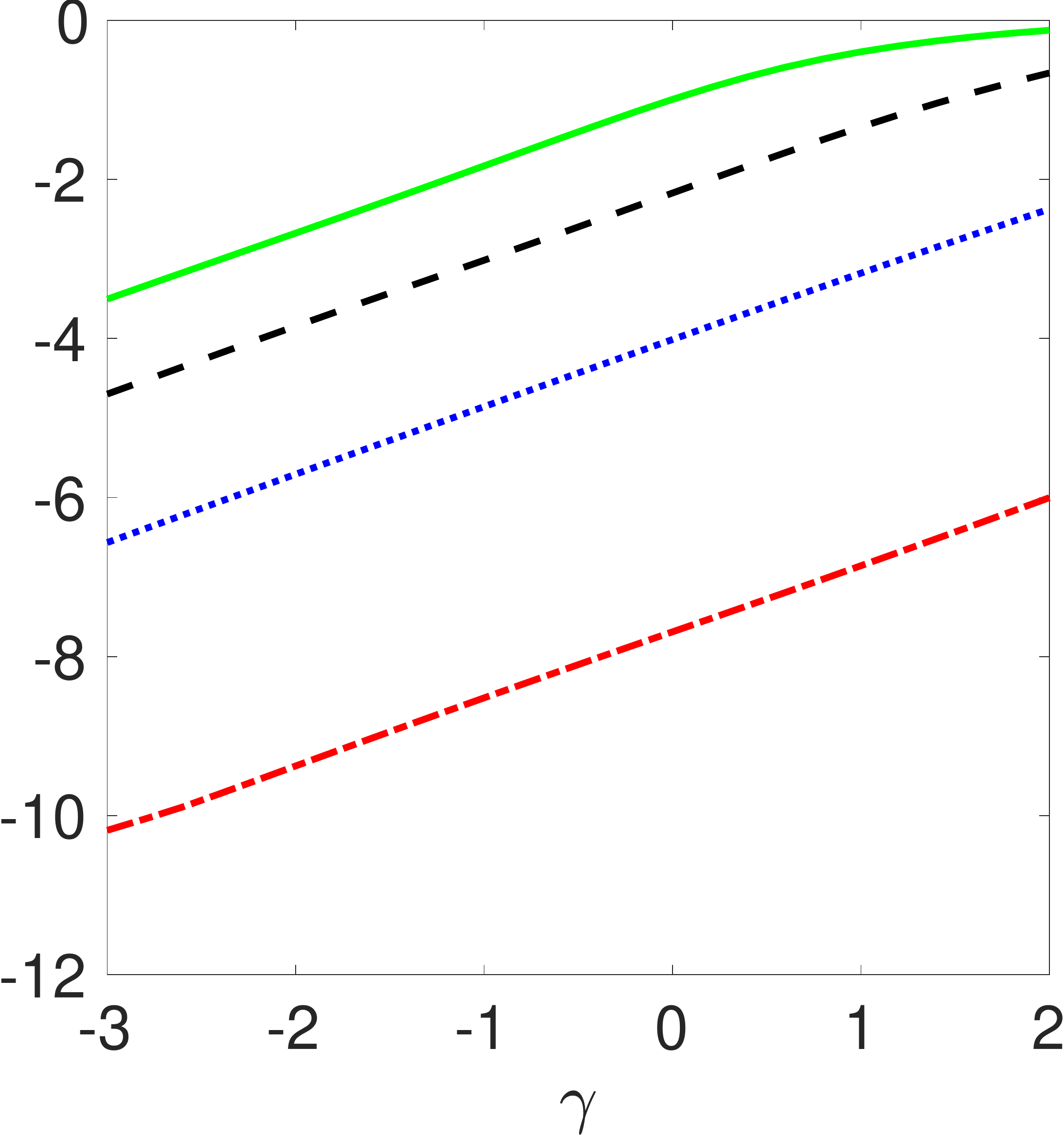}\\
\includegraphics[width=.3\linewidth, height=0.25\linewidth]{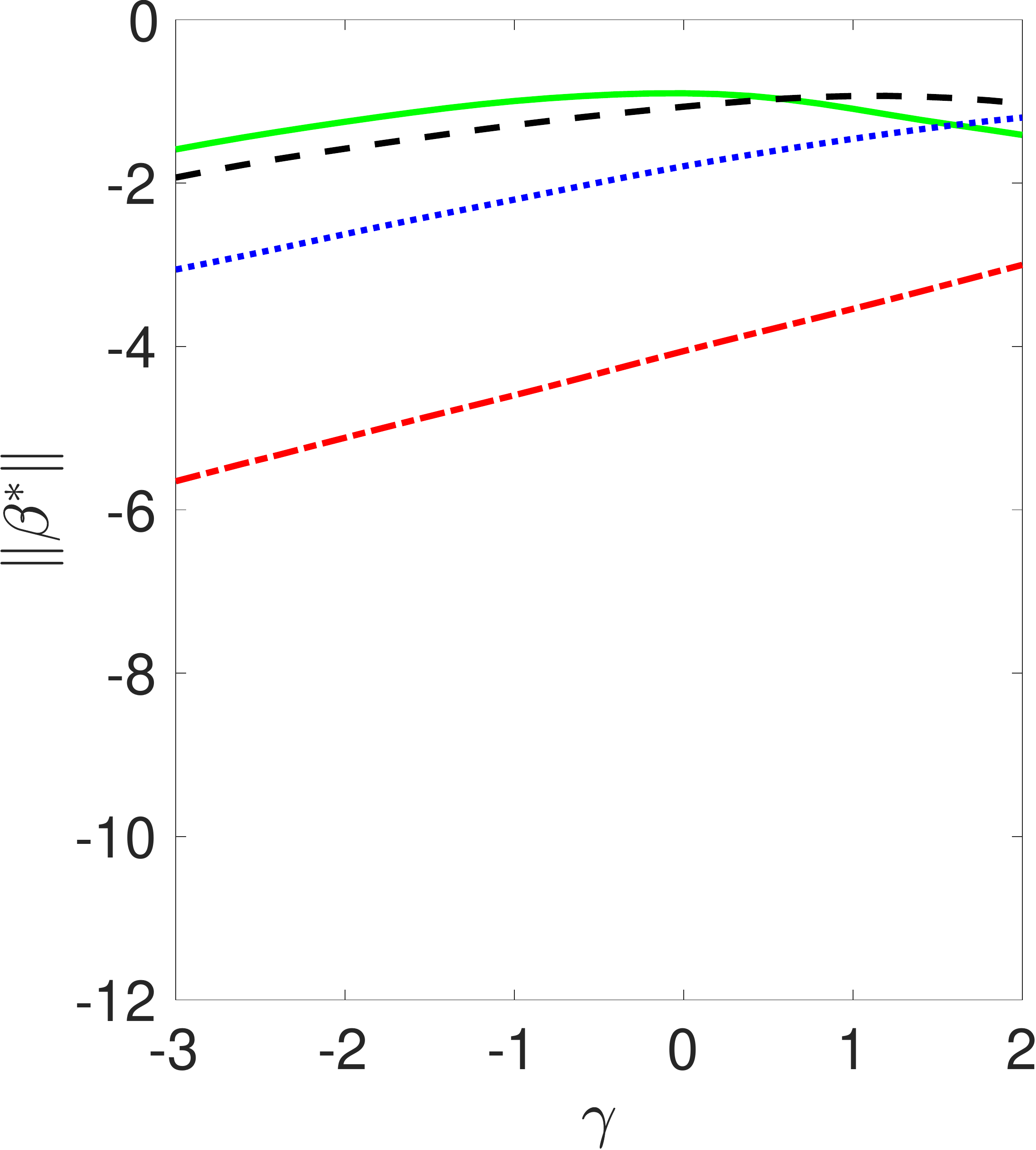}
\includegraphics[width=.3\linewidth, height=0.25\linewidth]{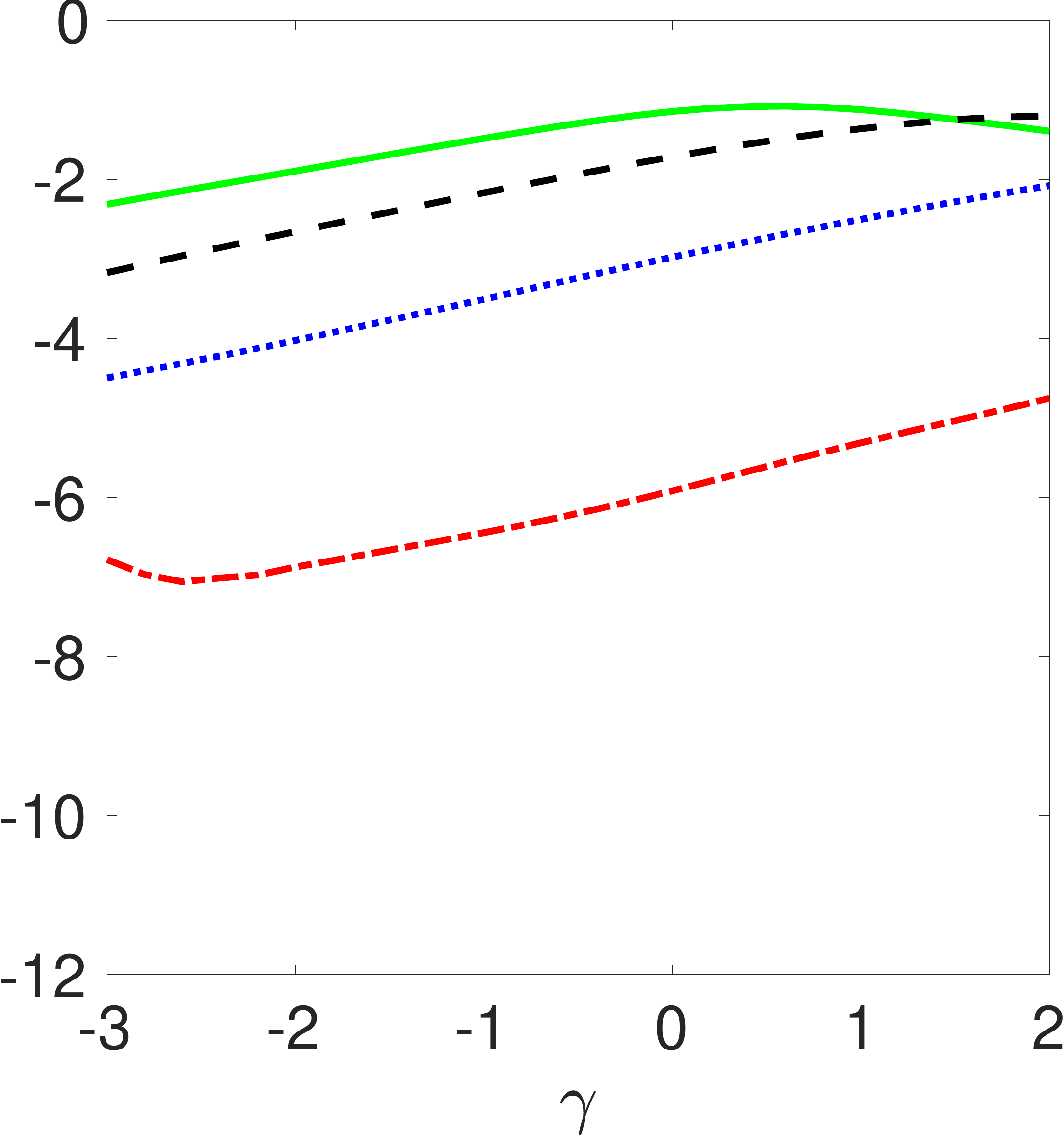}
\includegraphics[width=.3\linewidth, height=0.25\linewidth]{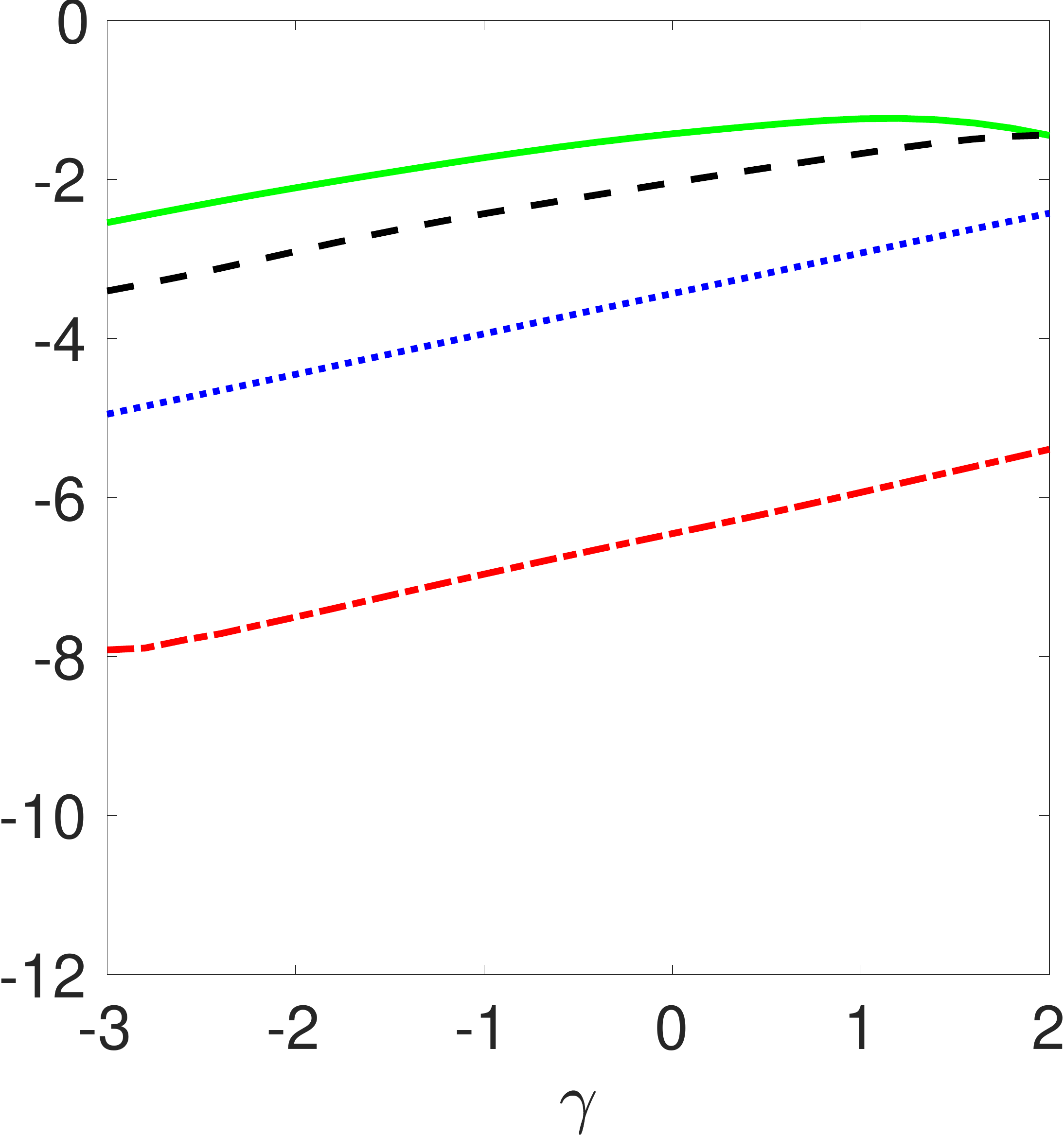}\\
\caption{\footnotesize\label{fig:rec_bunny_appendix} \textbf{Bunny graph}. }
\end{figure}

\begin{figure}[h]
\centering
\begin{minipage}{.3\linewidth} \centering \small \hspace{2mm} Regime 1  \end{minipage}
\begin{minipage}{.3\linewidth} \centering \small \hspace{2mm}  Regime 2  \end{minipage}
\begin{minipage}{.3\linewidth} \centering \small \hspace{2mm} Regime3  \end{minipage}\\
\includegraphics[width=.3\linewidth, height=0.25\linewidth]{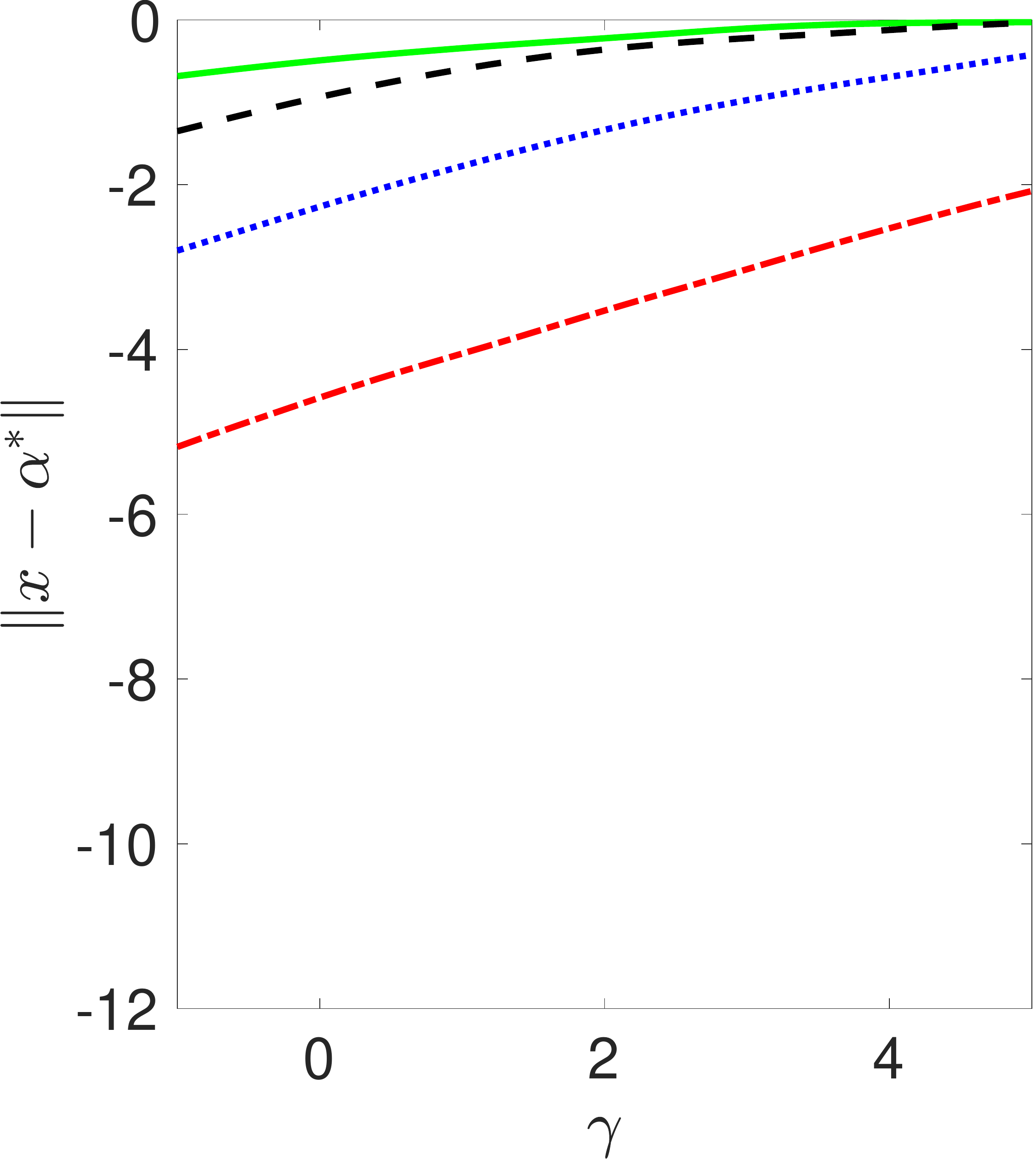}
\includegraphics[width=.3\linewidth, height=0.25\linewidth]{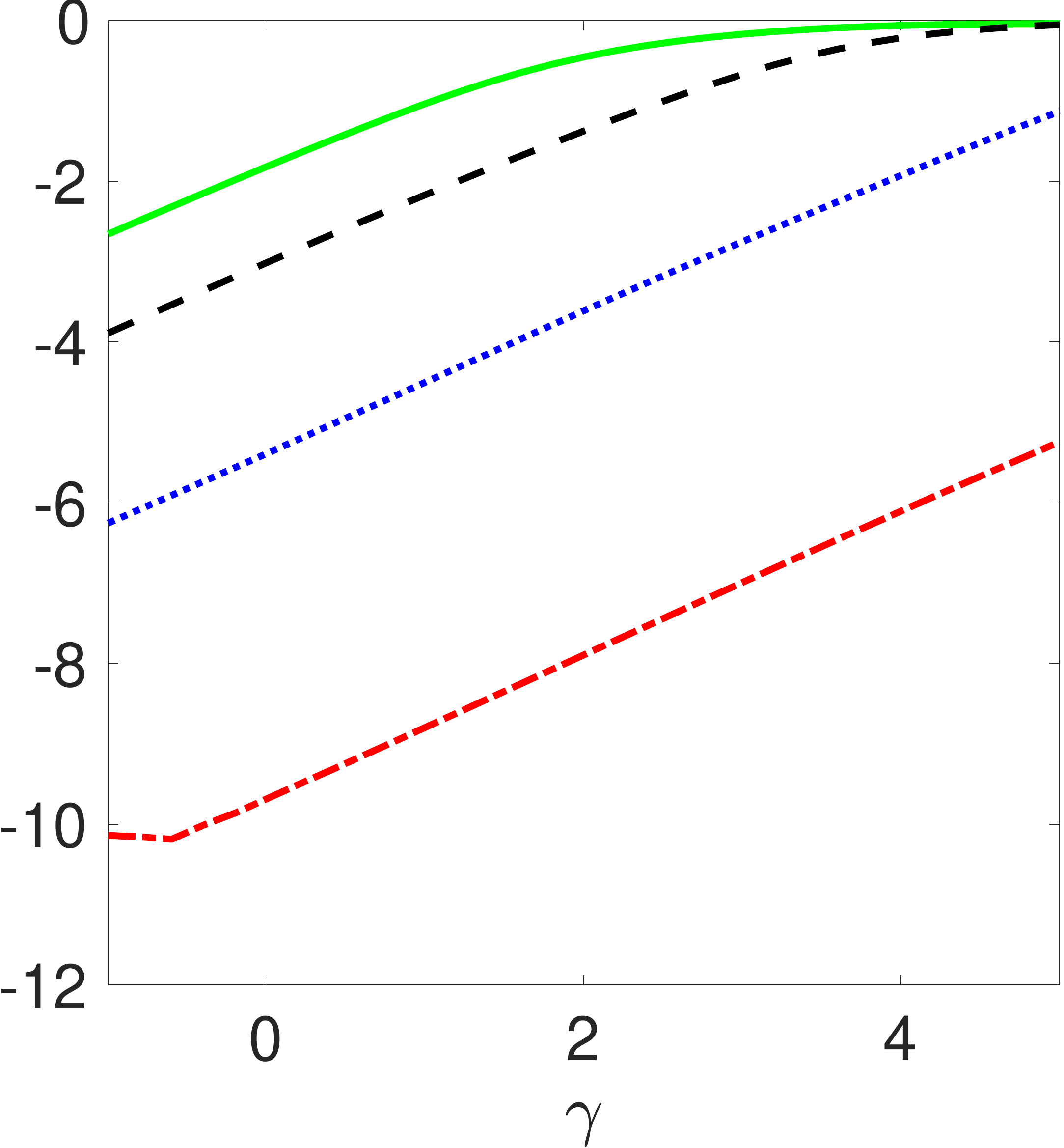}
\includegraphics[width=.3\linewidth, height=0.25\linewidth]{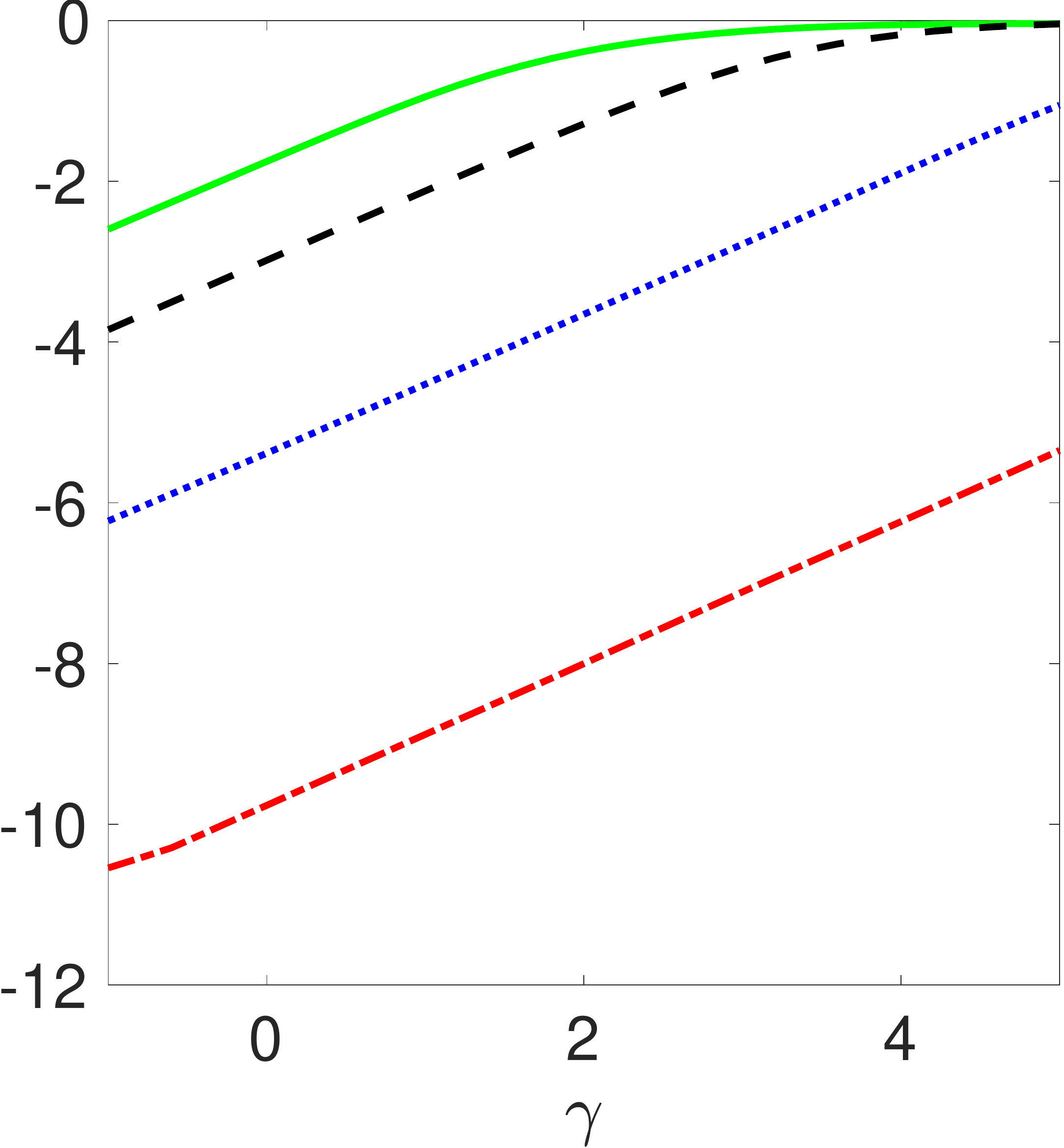}\\
\includegraphics[width=.3\linewidth, height=0.25\linewidth]{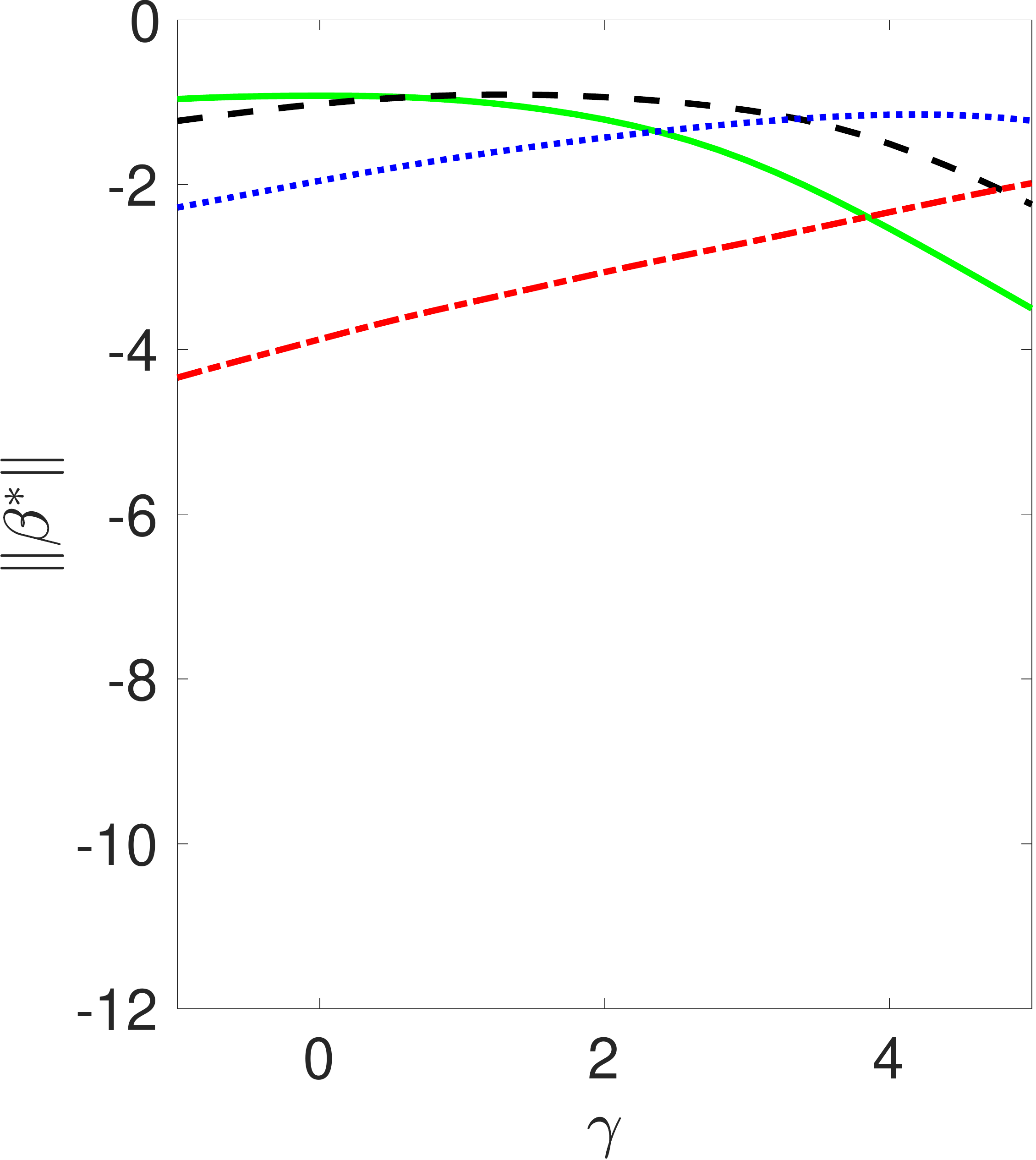}
\includegraphics[width=.3\linewidth, height=0.25\linewidth]{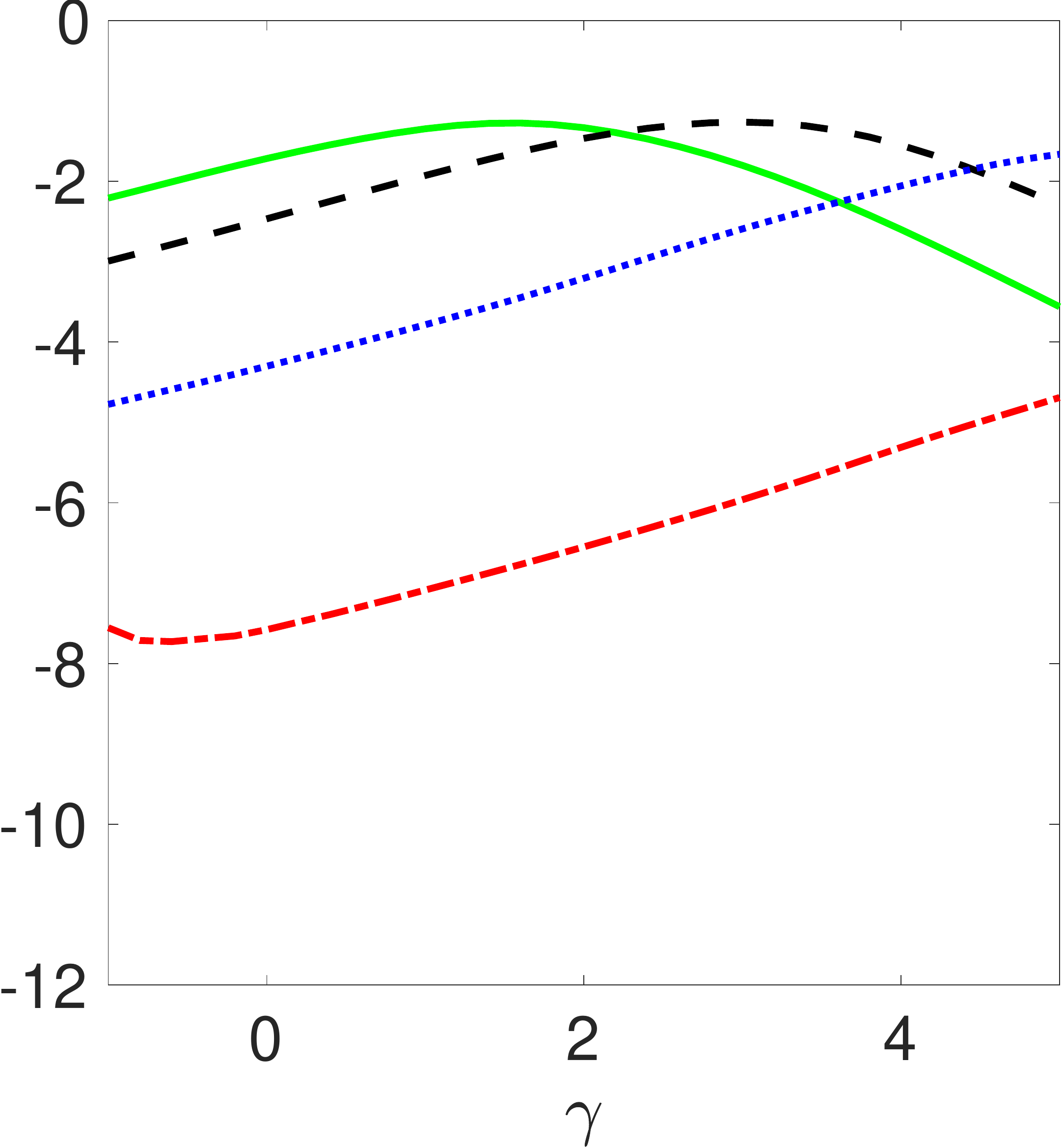}
\includegraphics[width=.3\linewidth, height=0.25\linewidth]{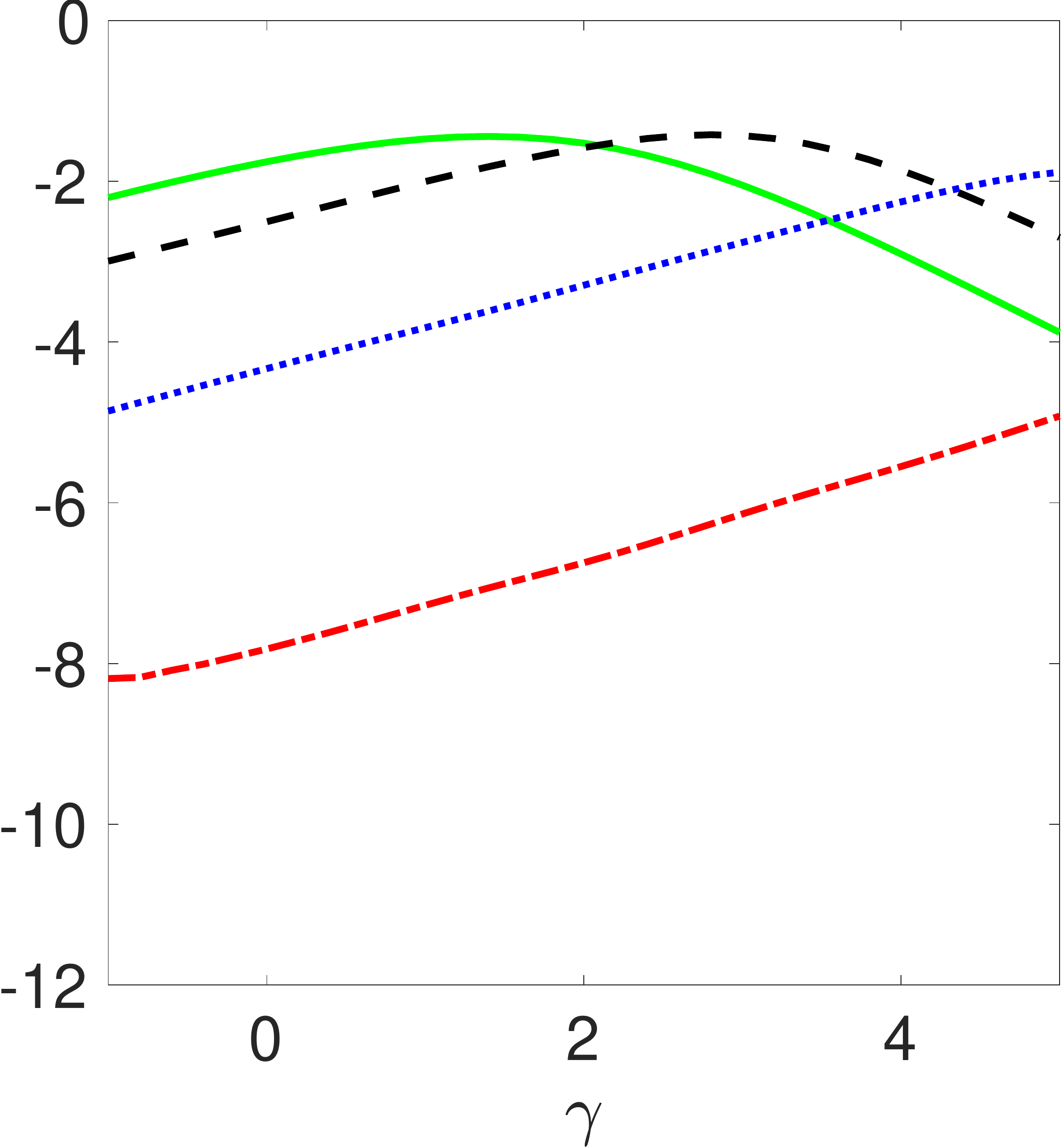}\\
\caption{\footnotesize \textbf{Minnesota  graph} }\label{fig:rec_minesota_appendix}
\end{figure}

\begin{figure}[h]
\centering
\begin{minipage}{.3\linewidth} \centering \small \hspace{2mm}  Community   \end{minipage}
\begin{minipage}{.3\linewidth} \centering \small \hspace{2mm} Bunny  \end{minipage}
\begin{minipage}{.3\linewidth} \centering \small \hspace{2mm}  Minesota   \end{minipage}\\
\includegraphics[width=.3\linewidth, height=0.25\linewidth]{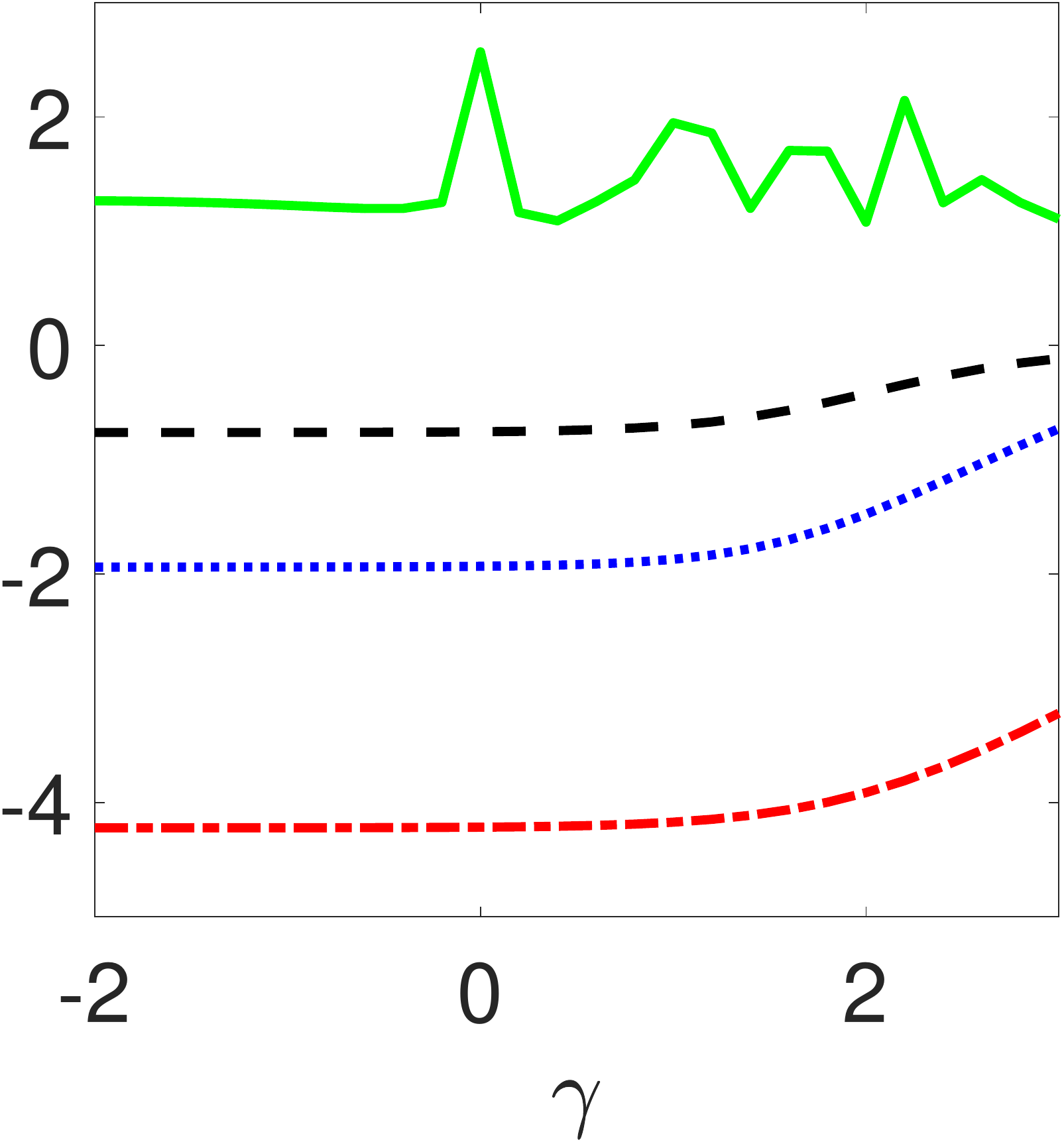}
\includegraphics[width=.3\linewidth, height=0.25\linewidth]{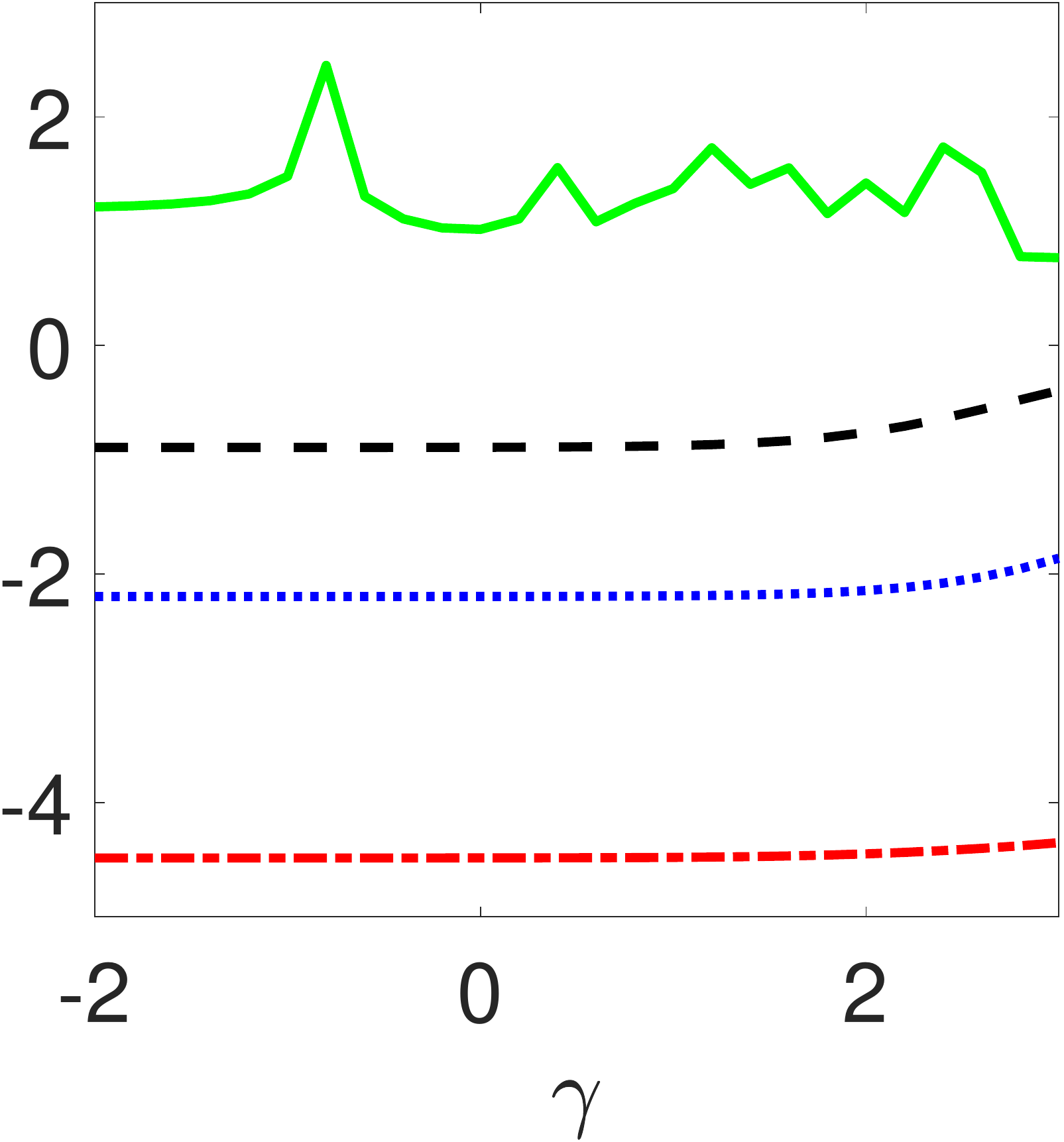}
\includegraphics[width=.3\linewidth, height=0.25\linewidth]{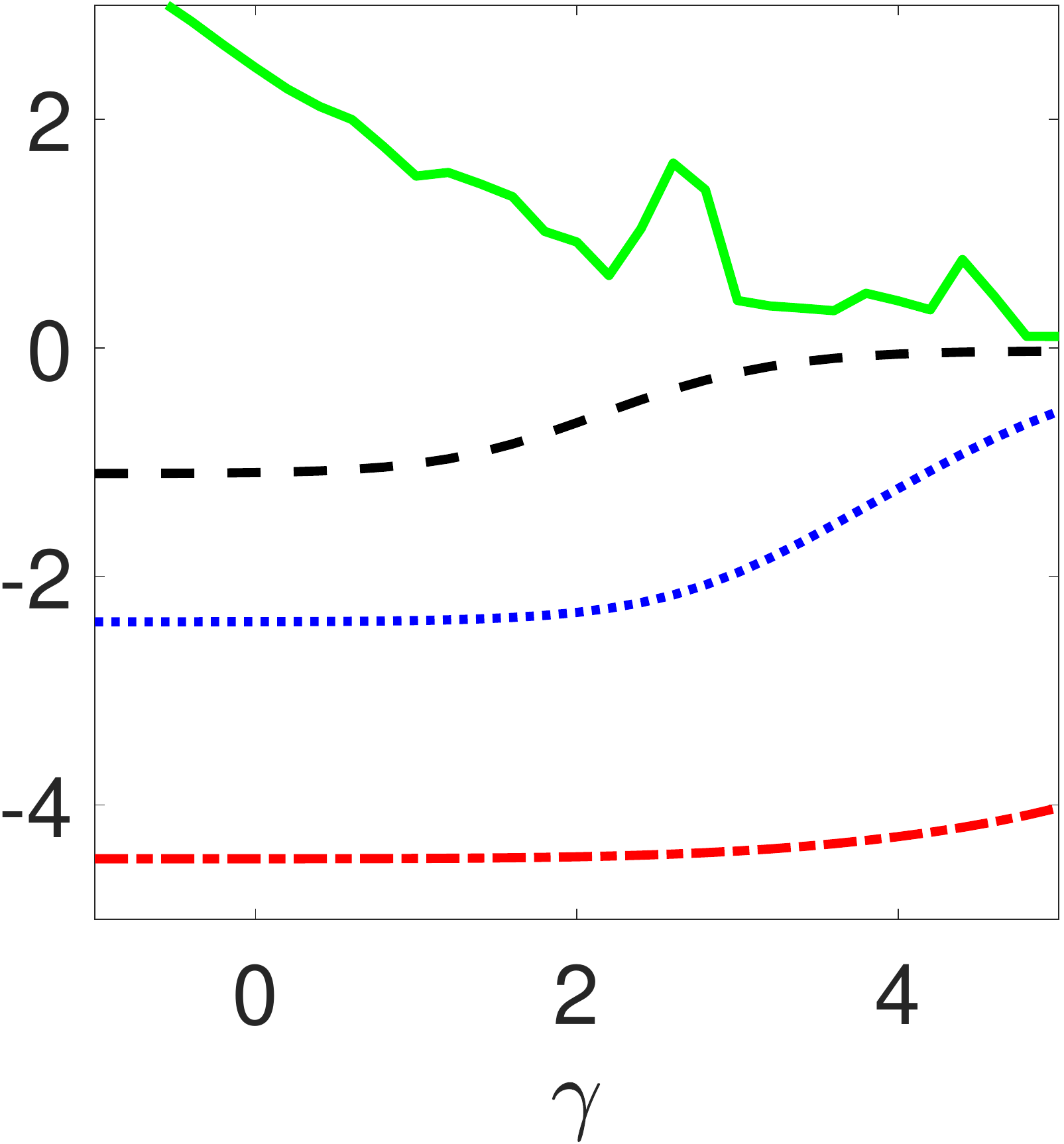}\\
\caption{\footnotesize \textbf{Mean reconstruction error $\log_{10}\|x^*-x\|$ in static case $T=1$ from 200  spatial samples}. All other parameters are the same with Fig.~\ref{fig:rec_nonoise}. }\label{fig:rec_nonoise_static}
\end{figure}

\begin{figure}[h]
\centering
\begin{minipage}{.3\linewidth} \centering \small \hspace{2mm}  Community   \end{minipage}
\begin{minipage}{.3\linewidth} \centering \small \hspace{2mm} Bunny  \end{minipage}
\begin{minipage}{.3\linewidth} \centering \small \hspace{2mm}  Minesota   \end{minipage}\\
\includegraphics[width=.3\linewidth, height=0.25\linewidth]{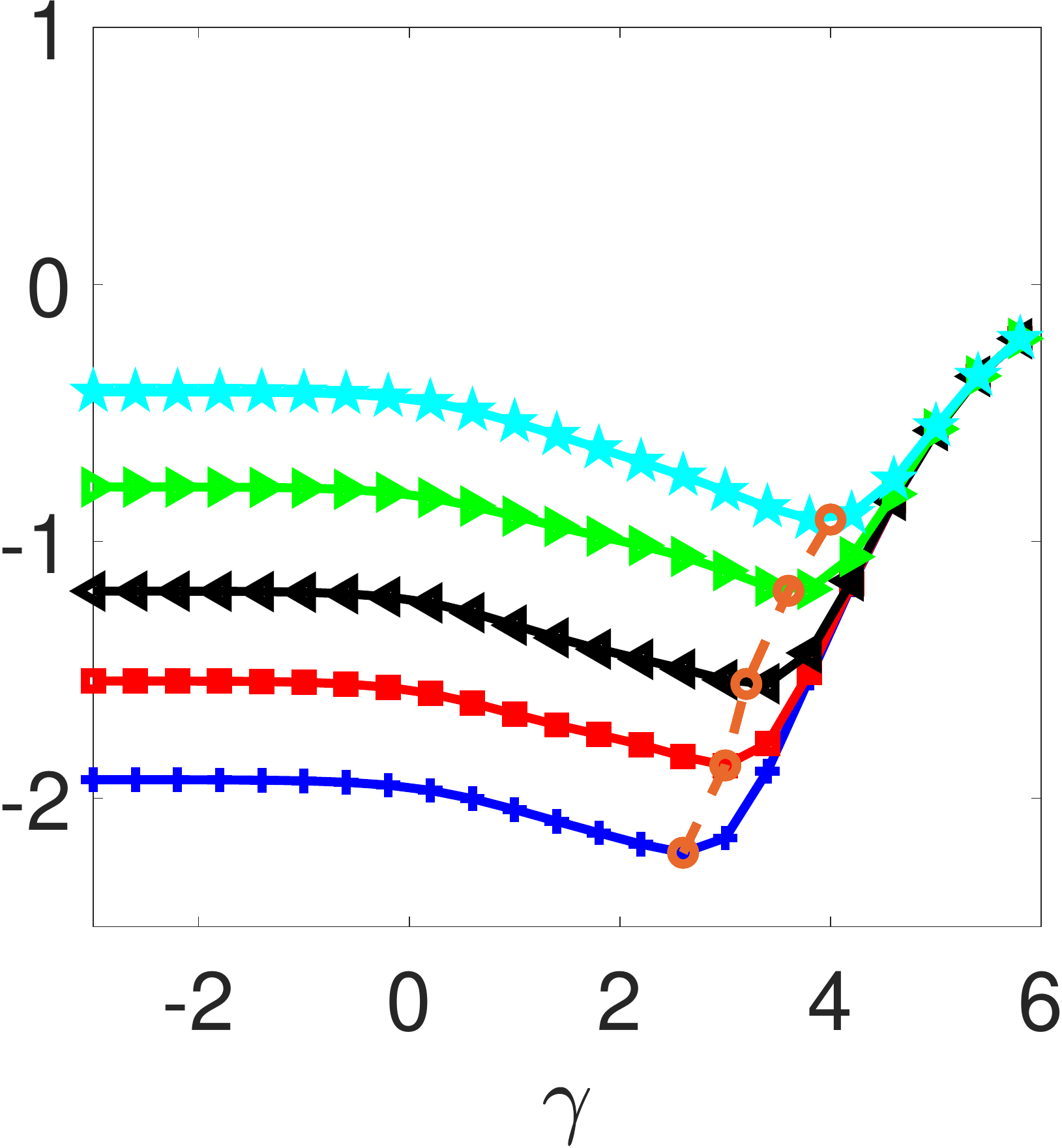}
\includegraphics[width=.3\linewidth, height=0.25\linewidth]{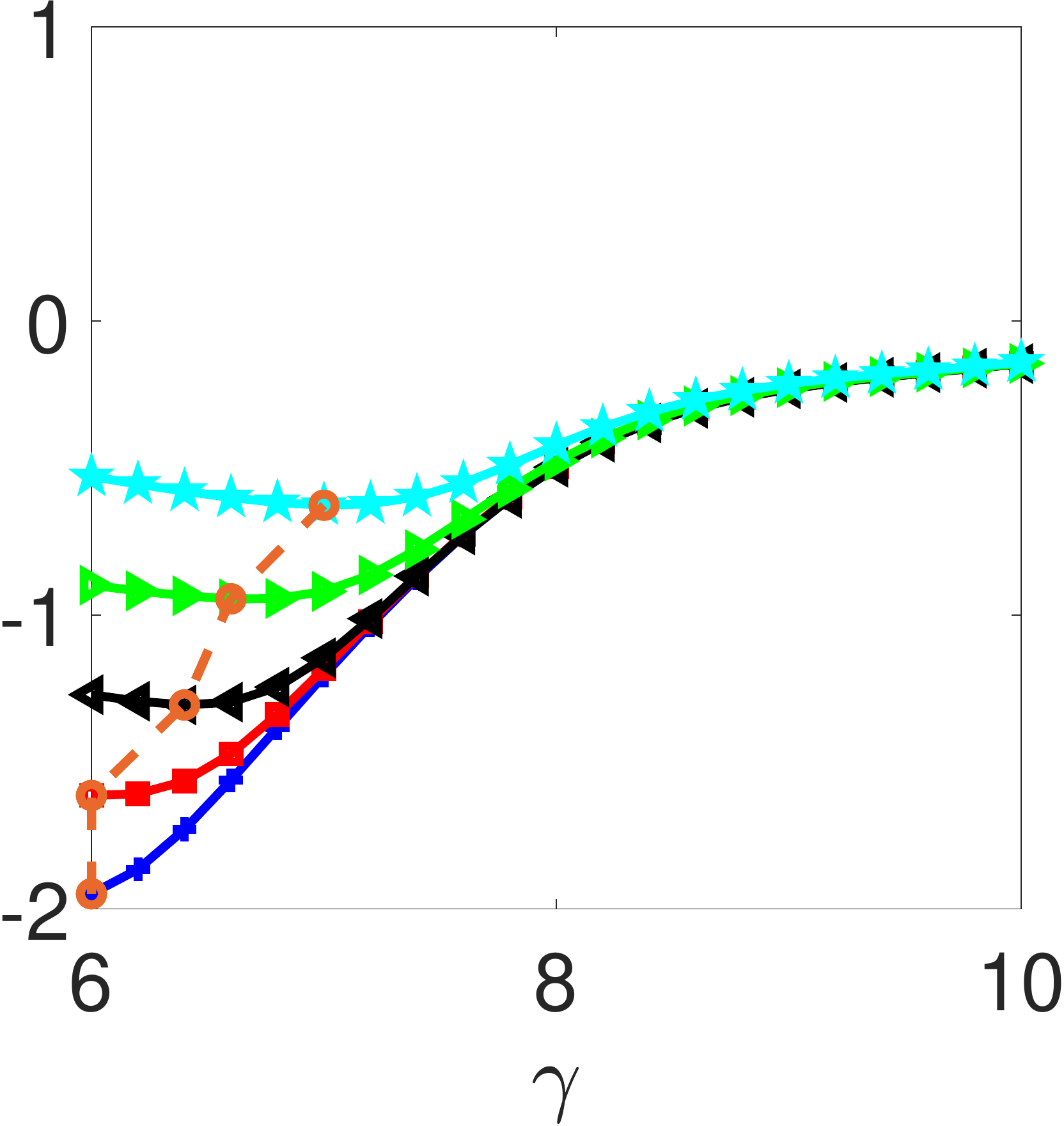}
\includegraphics[width=.3\linewidth, height=0.25\linewidth]{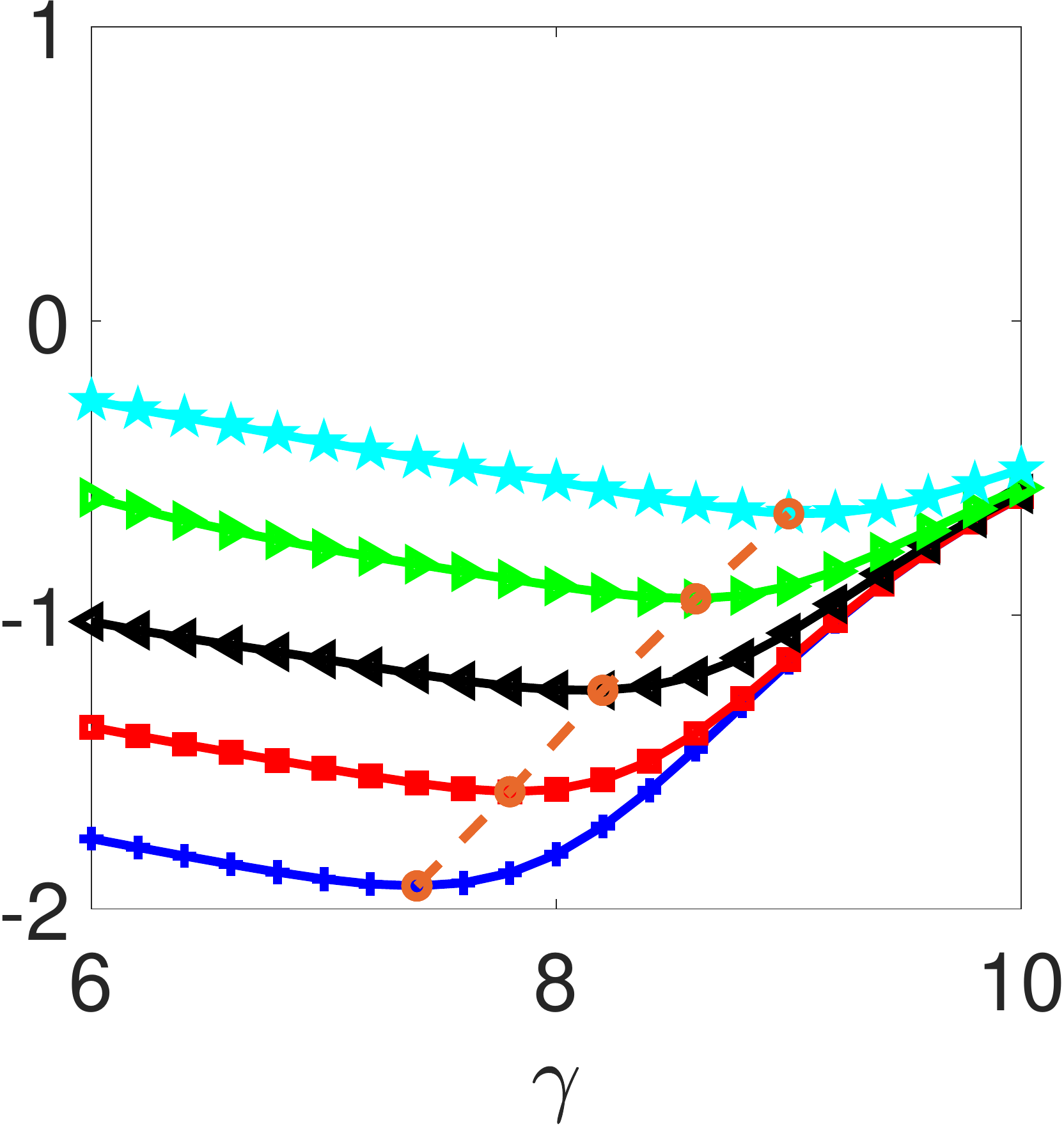}
\\
\caption{\footnotesize \textbf{Mean reconstruction error $\log_{10}\|x^*-x\|$ in static case $T=1$ from 200 noisy spatial samples}. All other parameters are the same with Fig.~\ref{fig:rec_noisy}. }\label{fig:rec_static}
\end{figure}

\end{document}